\documentclass[12pt,a4paper,oneside,oldfontcommands]{memoir}

\usepackage[utf8]{inputenc}
\usepackage[OT2,T1]{fontenc}
\usepackage{stmaryrd}
\usepackage{graphicx,epstopdf,epsfig,psfrag}
\usepackage[usenames,dvipsnames]{xcolor}

\usepackage{mathrsfs}
\usepackage{commath}
\usepackage{amsmath,amssymb,amsfonts,amsthm}
\usepackage{mathtools}

\usepackage{appendix}
\usepackage{etoolbox}
\usepackage{comment}
\usepackage{import}

\usepackage{pgf,tikz,pgfplots}
\pgfplotsset{compat=1.14}
\usetikzlibrary{calc}
\usetikzlibrary{decorations.pathmorphing}
\usetikzlibrary{shapes.geometric}
\usetikzlibrary{arrows,decorations.markings}
\graphicspath{{figures/}}

\usepackage[draft]{todonotes} %disable/draft
\usepackage[final]{showlabels}	%final/draft

\tikzset{->-/.style={decoration={
			markings,
			mark=at position #1 with {\arrow{>}}},postaction={decorate}}}	
\tikzset{-<-/.style={decoration={
			markings,
			mark=at position #1 with {\arrow{<}}},postaction={decorate}}}

\newcommand{\E}{\mathrm{e}}
\newcommand{\I}{\mathrm{i}}

\renewcommand{\Re}{{\mathrm{Re}}}
\DeclareUnicodeCharacter{2212}{a}

\newcommand{\C}{{\mathbb C}}
\newcommand{\N}{{\mathbb N}}
\newcommand{\Q}{{\mathbb Q}}
\newcommand{\R}{{\mathbb R}}
\newcommand{\Z}{{\mathbb Z}}

\newcommand{\W}{W}

\newcommand{\cA}{{\mathcal A}}

\newcommand{\cF}{{\mathcal F}}
\newcommand{\cG}{{\mathcal G}}
\newcommand{\cI}{{\mathcal I}}
\newcommand{\cJ}{{\mathcal J}}
\newcommand{\cK}{{\mathcal K}}
\newcommand{\cL}{{\mathcal L}}
\newcommand{\cH}{{\mathcal H}}
\newcommand{\cM}{{\mathcal M}}
\newcommand{\cN}{{\mathcal N}}

\newcommand{\cO}{{\mathcal O}}
\newcommand{\cP}{{\mathcal P}}
\newcommand{\cT}{{\mathcal T}}

\newcommand{\cD}{{\mathcal D}}
\newcommand{\cC}{{\mathcal C}}

\newcommand{\UU}{\mathrm{U}}
\newcommand{\SU}{\mathrm{SU}}

\newcommand{\SL}{\mathrm{SL}}
\newcommand{\SO}{\mathrm{SO}}

\newcommand{\su}{{\mathfrak{su}}}

\newcommand{\so}{{\mathfrak{so}}}
\newcommand{\g}{{\mathfrak{g}}}

\newcommand{\la}{\langle}
\newcommand{\ra}{\rangle}
\newcommand{\dd}{{\mathrm{d}}}
\newcommand{\up}{{\,\uparrow\,}}
\newcommand{\down}{{\,\downarrow\,}}

\newcommand{\tr}{{\mathrm{Tr}}}
\newcommand{\f}{\frac}
\newcommand{\tl}{\widetilde}
\def\nn{\nonumber}
\def\pp{\partial}

\newcommand{\id}{\mathbb{I}}
\def\act{\triangleright}

\newcommand{\ch}{\text{ch}}
\newcommand{\sh}{\text{sh}}
\def\eps{\epsilon}
\def\vJ{\vec{J}}
\def\vcC{\vec{\cC}}
\def\vu{\vec{u}}
\def\vsigma{\vec{\sigma}}

\newcommand{\mat} [2] {\left ( \begin{array}{#1}#2\end{array} \right ) }

\newcommand{\GG}{\mathbf{G}}

\newcommand{\DD}{\mathbf{D}}
\newcommand{\CC}{\mathbf{C}}
\newcommand{\bfa}{\mathbf{a}}

\usepackage{amsthm} % \proof and other thrm environement
\newtheorem{thmcnter}{thmcnter}[chapter]

\newtheorem*{Thrm*}{Theorem}
\newtheorem*{Prop*}{Proposition}

\newtheorem{prop}[thmcnter]{Proposition}
\newtheorem{lemma}[thmcnter]{Lemma}

\theoremstyle{definition}
\newtheorem*{Def*}{Definition}
\newtheorem{Definition}[thmcnter]{Definition}

\usepackage{colortbl}

\usepackage{indentfirst} % retrait en d�but de paragraphe	
\usepackage{layout}% one inch = 2.54cm
\usepackage{changepage}
\usepackage{pdfpages}
\usepackage{setspace}
\setsecnumdepth{subsection}
\chapterstyle{ger}

\makeatletter
\renewcommand\part{%
	\if@openright
	\cleardoublepage
	\else
	\clearpage
	\fi
	\thispagestyle{empty}%
	\if@twocolumn
	\onecolumn
	\@tempswatrue
	\else
	\@tempswafalse
	\fi
	\null\vfil
	\secdef\@part\@spart}
\makeatother

\makepagestyle{Contents}
\makeevenhead{Contents}{\small{\it{Contents}}}{}{\thepage}
\makeoddhead{Contents}{\small{\it{Contents}}}{}{\thepage}

\makepagestyle{Introduction}
\makeevenhead{Introduction}{\small{\it{Introduction}}}{}{\thepage}
\makeoddhead{Introduction}{\small{\it{Introduction}}}{}{\thepage}

\makepagestyle{Conclusion}
\makeevenhead{Conclusion}{\small{\it{Conclusion}}}{}{\thepage}
\makeoddhead{Conclusion}{\small{\it{Conclusion}}}{}{\thepage}

\makepagestyle{Regular}
\makeevenhead{Regular}{{\it{\leftmark}}}{}{\thepage}
\makeoddhead{Regular}{{\it{\rightmark}}}{}{\thepage}

\makepagestyle{Bibliography}
\makeevenhead{Bibliography}{\small{\it{Bibliography}}}{}{\thepage}
\makeoddhead{Bibliography}{\small{\it{Bibliography}}}{}{\thepage}

\nouppercaseheads

\usepackage[pagestyles,raggedright]{titlesec}
\usepackage{titletoc}

\newcommand{\secmark}{}
\newcommand{\marktotoc}[1]{\renewcommand{\secmark}{#1}}

\titleformat{\section}{\normalfont\Large\bfseries}
{\makebox[1.5em][l]{\llap{\secmark}\thesection}}{0.4em}{}
\titlecontents{section}[3.7em]{}
{\contentslabel[\llap{\secmark}\thecontentslabel]{2.3em}}
{\hspace*{-2.3em}}{\titlerule*{.}\contentspage}

\titleformat{\subsection}{\normalfont\large\bfseries}
{\makebox[2.3em][l]{\llap{\secmark}\thesubsection}}{0.4em}{}
\titlecontents{subsection}[6.7em]{}
{\small\contentslabel[\llap{\secmark}\thecontentslabel]{3.1em}}
{\hspace*{-2.3em } }{\titlerule*{.}\contentspage}

\begin{document}
\includepdf[pages=-,offset=0cm -1.5cm]{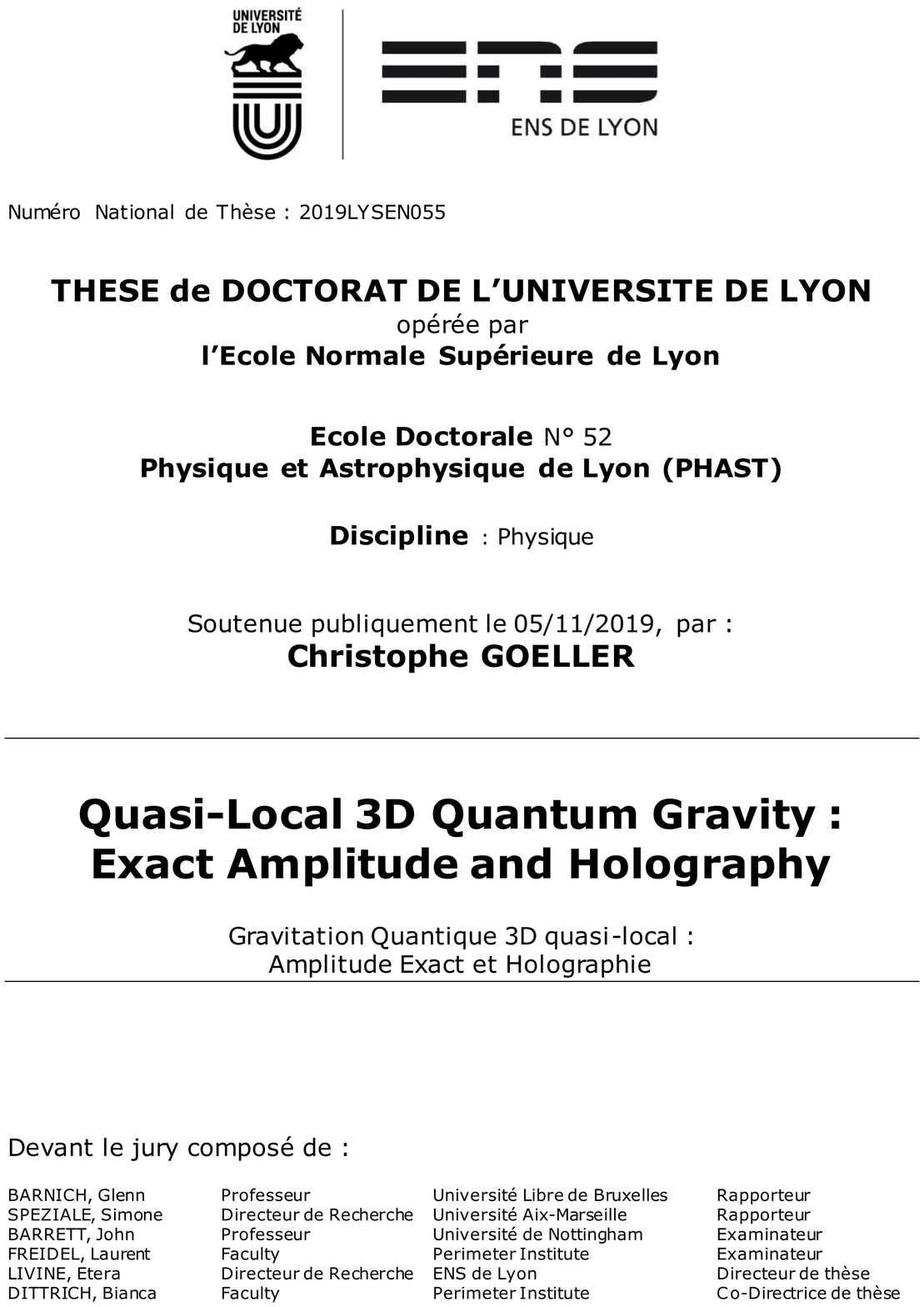} 

%%%%%%%%%%%%%%%%%%%%%%%%%%%%%%%%%%%%%%%%%%%%%%%%%%%%%%%%%%%%%%%%%%%%%%%%%%%%%%%%%%%%
\frontmatter
%%%%%%%%%%%%%%%%%%%%%%%%%%%%%%%%%%%%%%%%%%%%%%%%%%%%%%%%%%%%%%%%%%%%%%%%%%%%%%%%%%%%

\thispagestyle{empty}

%%%%%%%%%%%%%%%%%%%%%%%%%%%%%%%%%%%%%
\newpage
~
\thispagestyle{empty}
%%%%%%%%%%%%%%%%%%%%%%%%%%%%%%%%%%%%%

%%%%%%%%%%%%%%%%%%%%%%%%%%%%%%%%%%%%%
%%%%%%%%%% INTRO AND ACKNO%%%%%%%%%%%
%%%%%%%%%%%%%%%%%%%%%%%%%%%%%%%%%%%%%

\vspace*{4cm}

%% \emph{\lipsum[11-11]} 
\emph{Dirth'ena enasalin}
\cleardoublepage

\section*{Acknowledgement}
\addcontentsline{toc}{section}{Acknowledgements}

First and foremost, I would like to thank both my supervisors. In France, Etera was the first one to welcome a few years ago and he has been helping me navigate in the universe of theoretical physics and quantum gravity ever since. It has been a real pleasure working with him, even though I won't deny that it has also been challenging, trying to keep up with his very quick mind. Not a single day went, and still goes, by without him pushing me forward, and I will always be grateful to him for that. Overseas, at the Perimeter Institute, where I spent most of my time as a PhD student, Bianca took time to make sure I was adjusting to the life at PI and helped me find some organisation into the twists and turns of my thought. And due to the sheer amount of seminars and colloquiums organized at PI, it was not an easy task. Their complementary approaches to handling a PhD student was perfect for me.

I am also grateful to all the members of my committee. To Glenn Barnich and Simone Speziale, who took the time to be my examiners and read my manuscript. To John Barrett and Laurent Freidel for being part of my jury. I must also thank Laurent for the time he took at PI to talk about physics and gravity in all its form, as much as philosophically that mathematically. My time at PI would not have been the same without him.

I spent a lot of time during my PhD working in close collaboration with Aldo Riello. While Bianca and Etera answered lots of questions, Aldo was the one taking charge of a ridiculous amount of questions, from now basics quantum gravity knowledge to purely speculative questions. I cannot thank him enough for the time he took patiently with me. In the same vein, I must thank Cl\'ement Delcamp, who, as a fellow PhD student of Bianca, was well-suited to help me. I cannot forget Sylvain Carrozza and Marc Geiller, who, in their own way, keep reminding me to be more confident and trust myself. I must emphasize that Marc also had to deal with me during the last months of my PhD here in France, where he took a lot of time to answer some "last-minute" questions.

My time at PI would not have been the same without my office mates, Adrian and Nitica, who both kept me company in the office pretty late at night when needed, either to work or to play, as well as my time playing D\&D and magic, that helped me go through the hard time of the thesis.

Finally, I must of course thank my friends and family. Even while I was in Canada, they were always there to support me when needed. As a final note, I have to thank one last person. She came to Canada for me, and supported me, encouraged me, from the beginning to the end of my thesis. To you, thank you for everything you did, and you still do.

\cleardoublepage
\vspace*{-2cm}
\section*{Résumé de la Thèse (English version below)}

Cette thèse est consacrée à l’étude du rôle des frontières en gravitation quantique pour une région compacte de l’espace-temps et explore en détail le cas en trois dimensions d'espace-temps. Cette étude s'inscrit dans le contexte du principe holographique qui conjecture que la géométrie d'une région de l’espace et sa dynamique peuvent être entièrement décrites par une théorie vivant sur la frontière de cette région. La réalisation la plus étudiée de ce principe est la correspondance AdS/CFT, où les fluctuations quantiques d'une géométrie asymptotiquement AdS sont décrites par une théorie conforme sur la frontière à l'infini, invariante sous le groupe de Virasoro. La philosophie appliquée ici diffère d’AdS/CFT. Je me suis intéressé à une holographie quasi-locale, c’est-à-dire pour une région bornée de l’espace avec une frontière à distance finie. L'objectif est de clarifier la relation bulk-boundary dans le cadre du modèle de Ponzano-Regge, qui définit la gravitation quantique euclidienne en 3D par une intégrale de chemin discrète. 

Je présente les premiers calculs approximatifs et exacts des amplitudes de Ponzano-Regge avec un état quantique de frontière 2D. Après présentation générale du calcul de l'amplitude 3D en fonction de l'état quantique de bord 2D, on se concentre sur le cas d'un tore 2D, qui trouve application dans l'étude de la thermodynamique des trous noirs BTZ. Dans un premier temps, la frontière 2D est décrite par des états de spin networks semi-classiques. L'approximation par phase stationnaire permet de retrouver, dans la limite asymptotique, la formule de l'amplitude de la gravité quantique 3D en tant que caractère du groupe BMS des symétries d'un espace-temps asymptotiquement plat. Puis dans un second temps, on introduit de nouveaux états quantiques cohérents, qui permettent une évaluation analytique exacte des amplitudes de gravité quantique 3D à distance finie sous la forme d'une régularisation complexe du caractère BMS. La possibilité de ce calcul exact suggère l’existence de structures (quasi-) intégrables liées aux symétries de la gravité quantique 3D en présence de frontières 2D bornées.

\newpage

\section*{Summary of the Thesis}
\addcontentsline{toc}{section}{Summary of the Thesis}

This thesis is dedicated to the study of quasi-local boundary in quantum gravity. In particular, we focus on the three-dimensional space-time case. This research takes root in the holographic principle, which conjectures that the geometry and the dynamic of a space-time region can be entirely described by a theory living on the boundary of this given region. The most studied case of this principle is the AdS/CFT correspondence, where the quantum fluctuations of an asymptotically AdS space are described by a conformal field theory living at spatial infinity, invariant under the Virasoro group. The philosophy applied in this thesis however differs from the AdS/CFT case. I chose to focus on the case of quasi-local holography, i.e. for a bounded region of space-time with a boundary at a finite distance. The objective is to clarify the bulk-boundary relation in quantum gravity described by the Ponzano-Regge model, which defined a model for 3D gravity via a discrete path integral.

I present the first perturbative and exact computations of the Ponzano-Regge amplitude on a torus with a 2D boundary state. After the presentation of the general framework for the 3D amplitude in terms of the 2D boundary state, we focus on the case of the 2D torus, that found an application in the study of the thermodynamics of the BTZ black hole. First, the 2D boundary is described by a coherent spin network state in the semi-classical regime. The stationary phase approximation allows to recover in the asymptotic limit the usual amplitude for 3D quantum gravity as the character of the symmetry of asymptotically flat gravity, the BMS group. Then we introduce a new type of coherent boundary state, which allows an exact evaluation of the amplitude for 3D quantum gravity. We obtain a complex regularization of the BMS character. The possibility of this exact computation suggests the existence of a (quasi)-integrable structure, linked to the symmetries of 3D quantum gravity with 2D finite boundary.

%\lipsum[2-4]

%%%%%%%%%%%%%%%%%%%%%%%%%%%%%%%%%%%%%
%%%%%%%%%% TABLE OF CONTENTS %%%%%%%%
%%%%%%%%%%%%%%%%%%%%%%%%%%%%%%%%%%%%%

\pagestyle{Contents}
\renewcommand\cftsubsectionfont{\small}
\newpage

\tableofcontents*

\newpage
~
\thispagestyle{empty}

%%%%%%%%%%%%%%%%%%%%%%%%%%%%%%%%%%%%%%%%%%%%%%%%%%%%%%%%%%%%%%%%%%%%%%%%%%%%%%%%%%%%
\mainmatter
%%%%%%%%%%%%%%%%%%%%%%%%%%%%%%%%%%%%%%%%%%%%%%%%%%%%%%%%%%%%%%%%%%%%%%%%%%%%%%%%%%%%
\pagestyle{Introduction}
\chapter*{Introduction}
\addcontentsline{toc}{part}{Introduction}

General Relativity is one of the cornerstones of today's theoretical physics, and it has been the case since its discovery by Einstein more than a century ago. From Newtonian gravity to Special Relativity and finally to General Relativity, Einstein's key insight at the time was what we now know as the Equivalence Principle. This principle tells us that, locally, gravitation and inertia are equivalent. The physics in a uniformly accelerating frame of reference cannot be distinguished from the physics in a constant gravitational field. From this principle, Einstein deduces that gravitation is better understood not as a force but rather through the geometry of space and time. Since then, a tremendous amount of work has been done on this subject. Through the years, General Relativity was experimentally confirmed, for example with the prediction of the gravitational lensing, i.e. the impact of a massive celestial object on light. A few years ago, on September 14, 2015, a direct observation of General Relativity was finally uncovered with the detection of gravitational waves \cite{Abbott:2016blz}. In a sense, this was the culmination of the theory of General Relativity. Beyond the research of an extension of General Relativity at a classical level and the knowledge of quasi-local symmetry group, what we are still missing today is a coherent description of what we call quantum gravity, i.e. a consistent description of the law of gravity at the quantum level (if it exists).
The thesis finds its home in this context. In particular, we have chosen to focus on a toy model describing flat quantum gravity in a three-dimensional Euclidean space-time manifold called the Ponzano-Regge model \cite{PR1968}. This model is intrinsically discrete and coincides with the spin foam formulation of quantum gravity in three dimensions without cosmological constant.
More precisely, we will look at the construction and the exact computation of the quasi-local partition function of flat three-dimensional gravity. By quasi-local, we mean that we focus on a finite region of space-time bounded by a finite boundary. This work is a starting point to a more in-depth study of quasi-local quantities in non-trivial topology for three-dimensional flat gravity.

\section*{From General Relativity to Quantum Gravity}

General Relativity is a very complex, and at the same time, a very simple theory. The reason behind its simplicity is that it is the {\it only} possible theory of gravitation. This property is encoded into the so-called Lovelock's theorem \cite{Lovelock:1971yv}, which tells us that the only possible three or four-dimensional local theory depending on a metric up to second order derivatives is described by General Relativity. The number of admissible theories becomes enormous if we relax at least one constraint, see  \cite{Crisostomi:2017ugk,Langlois:2018dxi,Heisenberg:2018vsk} for some reviews on the subject. In particular, in more than four dimensions, General Relativity provides only a part of the full admissible action for gravity. Such a generalization of General Relativity is not the subject of this thesis, and we restrict ourselves to the usual Einstein formulation of gravity. The starting point is then the Einstein-Hilbert action and Einstein equations. 

In this work we are interested in studying Euclidean quantum gravity in three dimensions. By Euclidean, we mean that all three dimensions are of the same nature, i.e. they have the same sign (positive, for example) in the metric signature. At first sight, it seems to be an unexpected choice since it is clear that such a theory cannot really describe any real physical situation. This would have been different considering three-dimensional gravity in a Lorentzian space-time. In that case, we are merely suppressing one dimension of space while keeping a time coordinate and two spaces coordinates. This is something commonly done in condensed matter physics for example, where one (or more) dimension of space does not matter due to its size with respect to some length parameters. This question is nevertheless still open in the context of gravity and the three-dimensional Euclidean approach is therefore just a toy model. However, it is a very useful one. To understand the interest, and in some sense, the necessity of studying this case, let us focus on some of the problems arising for four-dimensional gravity.

In the last fifty years or so, one of the most common points among the works conducted on theoretical and fundamental physics is to be found in the use of quantum field theory \cite{Weinberg:1995mt}. Quantum field theory has first proven itself by unifying electromagnetism and the weak interaction. Also the tremendous phenomenological success of the standard model is for all to see. At first glance, it is therefore appealing to try and apply the same logic to General Relativity and gravity. Unfortunately, this is not as easy as one may think. The first, and maybe the biggest, obstacle is that General Relativity in four dimensions is {\it non-renormalizable} in the sense of quantum field theory. This means that the perturbative quantum theory associated to General Relativity involves an infinite number of undetermined coupling constants. Therefore, it does not seem possible to use quantum field theory to make any suitable and interesting prediction in the context of quantum gravity. The second obstacle comes from the very nature of what gravity is. With Einstein's theory, gravity is not described simply as a force, but rather by a theory on the geometry of space and time itself. That is, quantizing General Relativity means quantizing space-time itself, whatever that means. And we do not really know what it means. This comes with technical and philosophical difficulties. The first is about the evolution of a gravitational system. Quantum field theory tells us that the evolution is given by the Hamiltonian of the system. It is well-known that the natural Hamiltonian for a quantum theory of gravity is identically zero on-shell. But it is also clear that gravitational systems evolve... Secondly, by its very definition, quantum field theory is a local theory. Now, also by its very definition, the observables of General Relativity are necessarily diffeomorphism invariant, implying that they{\it are} non-local and hence not well described by quantum field theory. Thirdly, the notion of causality, which is cherished by a number of physicists, is also problematic. From the quantum field theory viewpoint, causality is an axiom. A quantum theory of gravity must however necessarily quantize both space and time. This implies that causality is subjected to quantum fluctuations that might{\it locally} break it. Last but not least, let us mention the problem of scattering amplitudes. In quantum field theory, scattering amplitudes are computed using the postulate that a region where the interactions become negligible does exist and thus the theory can safely be described by free fields. For gravity, this is already false at the classical level, since the gravitational coupling interaction never vanishes. This is a consequence of the fact that gravity, at least to the best of our knowledge, is not screened. Consequently, it is highly improbable that such a region exists at the quantum level. Naturally this is a non-exhaustive list of reasons why constructing a consistent theory of quantum gravity is a much harder problem than we would have thought at the beginning.

Despite these difficulties, lots of approaches to quantum gravity have been studied in the past decades. We do not claim to make a review of these approaches here, nor to talk about all of them since this is well beyond the scope of this work. See \cite{Carlip:2001wq,Esposito:2011rx} for more on the different approaches to quantum gravity. In this short paragraph, we only list a few of the most common approaches. The first one is, of course, string theory \cite{Polchinski:1998rq,Polchinski:1998rr,Tong:2009np}. It is important to note that string theory is not only a theory of gravity but follows the now centuries-old idea of unification in physics. It aims to find a global framework to describe nature. The basic idea of string theory is to describe fundamental particles not as point-like objects, but rather as one-dimensional strings which then interact with space-time. String theory provides a formalism to compute scattering amplitude in a perturbative approach, as in the usual perturbative quantum field theory. In turn, gravity can be included in this framework. 

The two other approaches we will mention are often associated to one another and are related to the work presented in this thesis. They are the loop quantum gravity approach and the spin foam approach. Loop Quantum Gravity \cite{Reisenberger:1996pu,Baez:1997zt,Ashtekar:2004eh,Dona:2010hm,Rovelli:2011eq} is a canonical approach to quantum gravity which aims at quantizing the theory directly from the Hamiltonian perspective. The notion of background is completely abandoned in favour of a non-perturbative approach. The main success of this construction is the rigorous mathematical definition of the Hilbert space of quantum geometry, representing the physical states of the theory, and the definition of geometrical observables, such as the volume and the area operators. This approach predicts that, at the Planck scale, the geometry is intrinsically discrete. This prediction is one of the reasons behind the emergence of the spin foam formulation \cite{Freidel:1998pt,Oriti:2003wf,Perez:2003vx,Livine:2010zx,Perez:2012wv}. In spin foam, we define the quantum theory via a path integral. To go beyond the formal definition of the path integral, we consider a representation of the space-time manifold. That is, we consider a cellular-decomposition to discretize the manifold and write the path integral as a state sum-model over the building blocks of the cellular decomposition. These building blocks are fundamentally interpreted as the building blocks of space-time arising from the intrinsic discrete nature at the Planck scale. In a sense, spin foams are a covariant version of Loop Quantum Gravity. Note that both loop quantum gravity and the spin foam approaches are also defined in dimensions other than four. In particular, we are interested in the spin foam formulation in three dimensions, provided by the Ponzano-Regge model.

\section*{Quantum Gravity in Three Dimensions}

As we already said, three-dimensional quantum gravity is a simple toy model to study gravity. While it is far simpler than the four-dimensional one, it still carries many of the same conceptual problems, namely, the construction of states and observables, the role of topology, the relation between the different quantization schemes, the coupling to particles and matter fields, the effect of the cosmological constant, the emergence or not of black holes, the holographic principle... The biggest downside on the one hand of three-dimensional gravity is the dynamics, which is significantly different from that of the four-dimensional case due to the absence of local degree of freedom and thus the absence of gravitons. Another problem lies in the fact that three-dimensional gravity does not have a good Newtonian limit in the sense that the force between point mass objects vanishes. Its biggest upside on the other hand is that it allows us to work on a mathematically rigorous non-perturbative approach to quantum gravity while staying exact at the same time.

One of the earlier works on three-dimensional gravity comes from Staruszkiewicz in 1963 \cite{Staruszkiewicz:1963zza}. In this paper, he described the behaviour of static solutions with point sources. The most seminal paper on three-dimensional gravity is credited however to Deser, Jackiw and 't Hooft in the 80s \cite{Deser:1983nh,Deser:1983tn} where they studied in detail the aspect of the static $\cN$ body solutions. Following this paper, the interest in the field started to grow, and still lasts until today and this thesis is one of the many works on it. A few years later, Achucarro and Townsend \cite{Achucarro:1987vz} showed that three-dimensional gravity can be represented as a Chern-Simons theory. Then this idea was  further developed by Witten \cite{Witten:1988hc,Witten:1989sx,Atiyah:1989vu}. He reformulated first-order gravity in terms of the Chern-Simon topological quantum field theory and showed how to extract expectation values of Wilson loop observables from a formal path integral formulation. In \cite{Reshetikhin:1991tc}, Reshetikhin and Turaev provided a mathematically rigorous formulation of Witten's Chern-Simons topological theory. At the same time, Turaev and Viro introduced a new way of constructing quantum invariants of three manifolds \cite{Turaev:1992hq}. Known as the Turaev-Viro model, this invariant is formulated in terms of a state-sum model. In a similar spirit to Witten, Horowitz \cite{Horowitz:1989ng} discussed the quantization of another class of topological quantum field theory, ever since known as $BF$ theories \cite{BLAU1991130}, to which three-dimensional first-order gravity belongs. The link between topological field theory, invariant of manifold and quantum gravity was further studied by Barrett and Crane \cite{Crane:1993vs,Barrett:1995mg,Barrett:1997gw}. 

In fact, the first model for three-dimensional Euclidean quantum gravity appeared earlier than all these works. It was formulated as far back as 1968 by Ponzano and Regge \cite{PR1968} and it coincides with a spin foam formulation of three-dimensional $BF$ quantum gravity.

\section*{Quantum Gravity and Holography}

One of the main points of interest in recent years on quantum gravity is the holographic principle. Roughly, this principle states that the gravitational phenomena occurring inside a region of space-time can be described by a "dual" theory living on the boundary of this space-time region. The apparition of the holographic principle is first due to 't Hooft in 1993 \cite{tHooft:1993dmi} then to Susskind in 1995 \cite{Susskind:1994vu}. One of the key strengths of holography is that it is compatible with the Bekenstein-Hawking formula for a black hole \cite{Hawking:1974rv,Hawking:1974sw}, stating that the entropy of a black hole is proportional to its area, instead of its volume. It is of primordial importance to note that this principle is not stated only for {\it asymptotic} region. In the following, we will differentiate two types of holography. The first one is the asymptotic holography, describing an infinite region and an asymptotic boundary. The second one is the quasi-local holography, describing holography for a {\it finite, bounded} region. While it is true that most of the recent works focus on asymptotic holography, it is necessary to study the quasi-local one for a theory such as gravity. We will come back to this point in the following. 

In practice, the first occurrence, and maybe the most famous one, of the (asymptotic) holographic principle is the well-known Anti De Sitter / Conformal Field Theory (AdS/CFT) duality proposed by Maldacena \cite{Maldacena:1997re}. It states that quantum gravity on a $D$-dimensional asymptotically AdS space is dual to a conformal field theory living on the boundary of the AdS space at spatial infinity. While there is still no proof of the full equivalence between these two theories, a good amount of works shows that gravity in a AdS asymptotic space with suitable fall-off conditions can indeed be described by a CFT on its boundary \cite{Witten:1998qj}. In the three-dimensional case, the emerging dual field theory is a specific conformal field theory, known as the Liouville theory \cite{Coussaert:1995zp,Carlip:2005zn}.

With this notion of holographic duality comes the question of symmetries. One of the first ingredients for the holographic correspondence is the matching of the symmetries between the theory in the D-dimensional manifold and its dual theory on the (D-1)-dimensional boundary. At the asymptotic level, there are two seminal works. The first one by Bondi, van der Burg, Metzner and Sachs \cite{Bondi:1962px,Sachs:1962wk} in 1962. They looked at the asymptotic symmetry group of asymptotically flat four-dimensional gravity. They expected to recover the isometries of flat space-time itself, the Poincaré group. Instead, they found something much more complex in what we call today the BMS group. While this group does contain the Poincaré group as a finite dimensional subgroup, it has an infinite number of generators called the "supertranslations", making the symmetry group infinite! The first implication of this discovery is that General Relativity does not reduce to special relativity for weak fields and long distance as one might have expected. In fact, it was shown later on that the BMS group can again be extended. A more in-depth study looking at a large radius expansion of the metric shows that another family of generators, known as the "superrotations", can be included. It then results in the so-called extended BMS group \cite{Barnich:2009se,Oblak:2015sea}. It extends both the translation group and the Lorentz group to infinite-dimensional counterparts. In recent years, lots of works have been done in the context of the BMS$_3$ group, the asymptotic symmetry group for three-dimensional gravity \cite{Barnich:2014kra,Barnich:2015uva,Oblak:2016eij}. In particular, it has been explicitly related to the partition function of asymptotically flat three-dimensional gravity \cite{Barnich:2015mui}.

The second famous work was done by Brown and Henneaux \cite{Brown:1986nw} in 1986 and is an harbinger of the AdS/CFT duality. They focused on the three-dimensional case with negative cosmological constant and looked at the asymptotic symmetry group for what is now known as the Brown-Henneaux fall-off conditions for the field. Similarly to the earlier work of Bondi, van der Burg and Sachs, they found that the asymptotic symmetry was not encoded into the usual AdS$_3$ isometry algebra $\so(2,2)$. Instead, it also gets enhanced in an infinite-dimensional algebra, whose associated group is the two-dimensional conformal group, i.e. two commuting copies of the Virasoro group. Similarly to the flat case,  it was explicitly shown \cite{Giombi:2008vd} that the partition function on such a space and with Brown-Henneaux fall-off conditions is related to a character of the Virasoro algebra.

Behind these computations, we emphasize that one of the key points is the choice of fall-off conditions for the field. In a sense, these fall-off conditions are natural because of the asymptotic boundaries: they impose that we recover either the usual flat or AdS metric at spatial infinity respectively.  However, a few conditions, such as imposing the recovery of the isometry group of globally flat or AdS space, make the choice of fall-off conditions not unique \cite{Compere:2013bya,Compere:2018aar}. From a more general perspective, all these asymptotics aspect are related to the so-called soft theorem and memory effect \cite{Strominger:2017zoo}. This is a large domain of research nowadays taking root in the fact that for a gauge theory on a manifold with boundary, the gauge invariance might be broken at the boundary. Degrees of freedom that were pure gauge in the bulk become relevant, physical degrees of freedom at the boundary. To recover the gauge invariance, it is necessary to add new fields at the boundary, representing the soft modes. See \cite{Strominger:2017zoo} and references therein for some recent development on this subject.

Up to now, we have only discussed asymptotic results. This work however is interested in the quasi-local aspects, i.e. in a local gravitational system within a finite region of space-time and with a finite boundary. Similarly to the asymptotic case, the notion of symmetry group is of primordial importance. In this setting, the most well-known result is found with the Chern-Simon theory, which is dual to a Wess-Zumino-Novikov-Witten conformal field theory. Its symmetry group is related to the conformal edge currents \cite{Balachandran:1991dw} and it leads to recovering the Kac-Moody algebra as the boundary symmetry group for Chern-Simon theory \cite{Banados:1994tn,Geiller:2017xad}. For finite boundary, the choice of boundary data is larger than in the asymptotic case since no natural choice arises. See \cite{Husain:1997fm} for an early study of this problem in gravity and $BF$ theory. Similar phenomena to the soft modes appear for finite boundary called the edge modes. See \cite{Freidel:2015gpa,Donnelly:2016auv,Freidel:2016bxd,Speranza:2017gxd,Geiller:2017xad,Freidel:2019ees} for recent discussions on this subject. The core of this work was to compute the amplitude for quasi-local region of space-time, which in turn must correspond to the character of the symmetry group for the quasi-local region. We will see that such a duality between partition function and symmetry group is already well-known in the continuum and for asymptotic boundaries.

\section*{Discrete Quantum Gravity and The Ponzano-Regge \\ Model}

First defined in 1968 by Ponzano and Regge, the Ponzano-Regge model \cite{PR1968} is the first model of Euclidean quantum gravity ever proposed. The model was first defined on a triangulated three-dimensional manifold as a state-sum using the Lie group $\SU(2)$. To each edge of the triangulation we associate a spin, that is an irreducible representation of $\SU(2)$. The model then assigns a quantum amplitude to each possible configuration of spins, and the amplitudes are finally summed over all possible admissible spins for all the edges. Hence the name of state-sum model. The Ponzano-Regge proposal as a model for 3D quantum gravity was somewhat fortuitous: while studying the asymptotic of the Wigner $3nj$-symbols from the recoupling theory of $\SU(2)$ representations, Ponzano and Regge realized that in the limit of large spins, the $6j$ symbol reproduced the complex exponential of a discrete version of the (boundary) action of General Relativity for a (flat) tetrahedron with edge lengths quantized in Planck units. This action had been proposed only a few years before by Regge himself \cite{Regge:1961px}, in the context of a discrete formulation of General Relativity, now commonly known as Regge calculus \cite{Hartle:1981}. Chapter \ref{chap3} is dedicated to the review of the model.

Another possible view on the model, more modern and recent, comes from its interpretation as a discretization of a $BF$ topological quantum field theory \cite{Rovelli:1993kc}. Its relation to Hamiltonian loop quantum gravity was rigorously shown in \cite{Noui:2004iy}. Moreover, its relation to Chern-Simons theory was shown to hold both at the covariant level \cite{Freidel:2004nb} and at the canonical one \cite{Freidel:2004vi,Noui:2004iz,Noui:2006ku,Meusburger:2008bs}. This formulation of the Ponzano-Regge model allows for a more general consideration on the discretization of the manifold, and opens up a different way of evaluating the amplitude of the model using group integrations instead of infinite sums. The $BF$ theory is a natural starting point from a non-perturbative approach, and thus is of importance for the spin foam formulation of quantum gravity. As a local state sum models, spin foams are suitable to consider finite boundaries. A first study of the one-loop partition function for three-dimensional gravity with finite boundaries, using perturbative Regge calculus \cite{Bonzom:2015ans}, revealed that a holographic dual can actually be defined even before taking an asymptotic limit. The possibility to make use of holographic dualities already for finite boundaries offers exciting perspectives. In particular, it allows us to get information on much more local properties of quantum gravity compared to those that can be encoded on an asymptotic boundary. Of course, this is also possible in other contexts, in the Chern-Simons formulation of three-dimensional AdS gravity \cite{Witten:1988hc} for example. In that case, the role of asymptotic condition is to automatically select a stricter set of boundary conditions, hence leading to a further specification of the boundary theory, from a Wess-Zumino-Novikov-Witten sigma model to scalar Liouville theory \cite{Coussaert:1995zp,Carlip:2005zn}

Being formulated in terms of a local state-sum, the Ponzano-Regge model allows to compute the amplitude of quantum gravitational processes within {\it finite}, quasi-local, regions. As a consequence of the topological nature of three-dimensional quantum gravity, the bulk variables can be exactly integrated out, modulo some global constraints coming from the topology of the chosen manifold. In the case of the sphere, which is topologically trivial, nothing remains. However, for non-trivial topology, for example for the torus, some information coming from the presence of a non-contractible cycle remains. Therefore, we are left with an amplitude defined purely on a two-dimensional object, the boundary, still carrying information about the bulk geometry. This is a perfect starting point for a study of the dual boundary theory. Obviously, the choice of boundary condition matters. Remember that due to the fact that we are working on a finite region, there are no preferential boundary conditions as in the asymptotic case. Each family of boundary conditions will define a class of possible dual boundary theory. This must be emphasized since it is the contrary that happens in the AdS/CFT duality. In this case, different boundary theories defined at spatial infinity are interpreted as given by duality {\it different} bulk theories. In our case, the bulk quantum gravity theory is instead considered to be known, and the dual theory, which is not necessarily a conformal field theory, emerges at a finite boundary as a reflection of a given boundary condition. This is completely analogous to the way Wess-Zumino-Novikov-Witten and Liouville theories emerge on the boundary of a Chern-Simons theory.

\section*{Plan of the Thesis}

This thesis is organized in six chapters divided in two parts. 

The first part of the thesis focuses on the review of gravity and quantum gravity in three dimensions. In chapter \ref{chap1} we quickly review the framework of General Relativity, and explain in more details why the three-dimensional case is peculiar. In particular, we focus on the first order formulation of General Relativity via the $BF$ action. This chapter ends with a short presentation of the AdS and flat asymptotic case. Chapter \ref{chap2} is a brief review of the computation of the partition function of three-dimensional Euclidean gravity done in the continuum context both in the AdS case and in the flat case. Finally, this part ends with chapter \ref{chap3} and the presentation of the Ponzano-Regge model both from a mathematical and historical point of view. Application of the Ponzano-Regge model to the trivial topology will also be briefly presented. 

The second part of the thesis is dedicated to the partition function of quantum gravity in the quasi-local setup. First, in chapter \ref{chap4}, based on \cite{Dittrich:2017hnl}, we introduce all the necessary ingredients for the computation of the Ponzano-Regge amplitude. Namely, we introduce the discretization of the torus we consider and the structure of the boundary Hilbert space. The chapter starts with a brief prelude focused on the quantum Regge calculus approach, where the partition function for quantum gravity in three dimensions was computed in a quasi-local region for the first time. And it ends with a first computation with the Ponzano-Regge model leading to a statistical model duality. Finally, the chapter \ref{chap5} (based on \cite{Dittrich:2017rvb}) and the chapter \ref{chap6} (based on \cite{Goeller:2019zpz}) are dedicated to the computation of the amplitude of quantum gravity in three dimensions considering two different classes of boundary states. Chapter \ref{chap5} focuses on a one-loop computation in the sense of the WKB approximation for the path integral whereas in chapter \ref{chap6}, we perform, for the first time, an {\it exact} computation of the amplitude of three-dimensional Euclidean quantum gravity on a non-trivial topology.

Note that the chapters \ref{chap1} and \ref{chap2} are mainly for the completeness of this thesis. We advise the reader to look at the references therein for the details.

\counterwithin{equation}{section}		\setcounter{equation}{0}
\counterwithin{thmcnter}{chapter}		\setcounter{thmcnter}{0}

\pagestyle{Regular}

\renewcommand{\afterpartskip}{}
\part*{Part I\\[.3cm]
	Quantum Gravity in Three Dimensions} 
\addcontentsline{toc}{part}{I\ Quantum Gravity in Three Dimensions} \label{part:partI}
\newpage
~
\thispagestyle{empty}

	\chapter{General Relativity in Three Dimensions and BF Theory}
\label{chap1}

Three-dimensional General Relativity is a topological field theory \cite{Witten:1988hc}. By topological, we mean that there are no local degrees of freedom. In particular, pure three-dimensional gravity is a theory of constant curvature directly proportional to the cosmological constant. In presence of matter fields, curvature is then concentrated at and only at the location of matter. This topological nature of three-dimensional gravity implies that physically distinct solutions must be parametrized by a finite set of {\it global} parameters. The number of parameters depends on the topology of the space-time manifold $\cM$ considered, and more particularly on the dimensions of the fundamental group of the space-time. More precisely, three-dimensional gravity is well-described by a very well-known class of theories, called $BF$ theories. They were first introduced by Horowitz in 1989 \cite{Horowitz:1989ng} and named a few years later in the seminal paper of Blau and Thompson \cite{BLAU1991130} where the study of non-abelian $BF$ theory started. Coming back to the case of gravity, it happens that only flat three-dimensional gravity, i.e. with vanishing cosmological constant is exactly described by a $BF$ action. The $BF$ action can be further modified to take into account the presence of a cosmological constant.

The first section of this chapter quickly introduces General Relativity in the Einstein-Hilbert formalism. The motivation behind this section is not to give a comprehensive presentation of General Relativity, but rather to show how the three-dimensional case is peculiar, and hence allows some interesting simplifications. We emphasize that this section is by no means a review of General Relativity. Providing such a review would be beyond the scope of this chapter and thesis. We refer the interested reader to the already existing excellent review on General Relativity in the literature. See for example the full course by Blau \cite{Blau2015} or the lecture by Carroll \cite{Carroll:1997ar}. Following this brief introduction of General Relativity, we will then focus on the general formulation of $BF$ theory and quickly point out some interesting properties and on the expression of gravity in this formalism. The last section of this chapter focuses on the geometry of asymptotic flat and AdS gravity: the goal being to point out the physics behind the torus topology. At the same time, we will briefly review the asymptotic symmetry group of asymptotically AdS and flat gravity, leading to the Virasoro and BMS group.
\section{The Einstein-Hilbert formulation of General Relativity}

The first formulation of General Relativity we learn is the one based on the Einstein-Hilbert action. Consider a $D$-dimensional manifold $\cM$ acting as the space-time, and a (Euclidean) metric $g$ on it. The Einstein-Hilbert action then reads, in absence of boundary $\pp \cM$
\begin{equation}
	S = \f{1}{16 \pi G} \int_{\cM} \text{d}^D x \sqrt{g} (R-2\Lambda) \; ,
\end{equation}
where $G$ is the Newton constant, $R$ the Ricci scalar and $\Lambda$ the cosmological constant. In case of a manifold with boundary, it is well-known that some terms must be added to the action for the action principle to still be well-defined. In the metric formalism, the natural quantity to be kept constant at the boundary is the induced metric $h$ on $\pp \cM$, that is, the pull-back of the metric $g$ to the boundary $\pp \cM$. The corresponding boundary term is the well-known Gibbons-Hawking-York (GHY) term \cite{Gibbons:1976ue,York:1986lje}
\begin{equation}
	S_{GHY} = \f{1}{8 \pi G} \int_{\pp \cM} \dd^{D-1} x \sqrt{h} K \; ,
	\label{chap1:eq:EH_action}
\end{equation} 
where $K$ is the extrinsic curvature and $h$ the induced boundary metric. While considering asymptotic boundary condition, it is necessary to consider a choice of fall-off conditions for the field. In that case, to ensure convergence at infinity, it might be necessary to add other terms to the boundary action. This is the case for example, considering an asymptotically AdS space with the so-called Brown-Henneaux boundary condition. To ensure the action principle, one must add a term depending on the induced metric $h$ only \cite{Carlip:2005tz}. Denoting $\Lambda = -1/l^2$ with $l \in \R$ the AdS radius, the final action is
\begin{equation}
	S = \f{1}{16 \pi G} \left[ \int \sqrt{g} (R-2\Lambda) + 2 \oint \sqrt{h} \left(K - \frac{1}{l}\right) \right].
	\label{eq:chap1:fullAdS_action}
\end{equation}
\medskip

In this section, we are not interested in such a modification of the boundary term. In fact, we are not really interested in any boundary term for now but rather in the equations of motion for General Relativity and their peculiarity in three dimensions. The equations of motions are easily found by applying the variational principle on the action \eqref{chap1:eq:EH_action} with respect to the metric. This returns the standard Einstein equations
\begin{equation}
	R_{\mu \nu} - \f{1}{2} g_{\mu \nu} R + \Lambda g_{\mu \nu} = 8 \pi G T_{\mu \nu} \; ,
	\label{chap1:eq:vaccum_einstein}
\end{equation}
with $R_{\mu \nu}$ being the Ricci tensor and $T_{\mu\nu}$ the energy-momentum tensor. Considering the case of pure gravity, i.e. without any matter fields, the tensor energy-momentum $T_{\mu \nu}$ vanishes everywhere. The right-hand side of the Einstein equations is then identically zero. Therefore, from the Einstein equations, it is straightforward to see that the Ricci tensor and scalar are deeply linked to the presence of matter. Taking the trace of \eqref{chap1:eq:vaccum_einstein}, we have that the Ricci scalar is proportional to the cosmological constant on-shell
\begin{equation*}
	R = \f{2 D \Lambda}{D-2}
\end{equation*}
while the Ricci tensor is basically the metric tensor up to some coefficients given by the cosmological constant and the space-time dimension
\begin{equation*}
	R_{\mu \nu} = \f{2}{D-2} \Lambda g_{\mu \nu} \;.
\end{equation*}

Consider now the flat case where the cosmological constant vanishes. Then, both the Ricci tensor and the Ricci scalar become identically zero on-shell. In some sense, they do not carry any information about gravity. It is only with matter that they carry non-trivial information.

Recall now that by its very definition, General Relativity is a theory about the geometry of space and time. What encodes intrinsically this geometry is the Riemann tensor $R_{\mu \nu \alpha \beta}$. It is well-known that this tensor can be expressed in terms of the Ricci scalar, the Ricci tensor and of what is known as the Weyl tensor $W_{\mu \nu \alpha \beta}$. The decomposition then reads
\begin{align*}
R_{\mu \nu \alpha \beta} = W_{\mu \nu \alpha \beta} + &\f{1}{D-2} \left( g_{\mu \alpha} R_{\nu \beta} + g_{\nu \beta} R_{\mu \alpha} - g_{\nu \alpha} R_{\mu \beta} - g_{\mu \beta} R_{\nu \alpha} \right) \\
& - \f{1}{(D-1)(D-2)} \left( g_{\mu \beta} g_{\nu \alpha} - g_{\mu \alpha}g_{\nu \beta}  \right)R \; .
\end{align*}
The Riemann tensor is a rank four tensor. Hence, it has by definition at most $D^4$ independent components. Actually, it only has $\f{D^2(D^2-1)}{12}$ independent ones because of its symmetry. Indeed, recall that the Riemann tensor is antisymmetric in $(\mu, \nu)$ since the metric is covariantly constant and in $(\alpha, \beta)$ by construction. It also follows the first Bianchi identity $R_{\mu \nu \alpha \beta} + R_{\mu \alpha \beta \nu} + R_{\mu \beta \nu \alpha} = 0$ and is symmetric under the exchange of the pair $(\mu\nu)$ and $(\alpha\beta)$. Applying the previous result, in the flat case and on-shell, only the Weyl tensor contributes to the Riemann tensor without matter. As such, it is the one carrying the information about pure gravity.

We can go even further by considering the three-dimensional case. Indeed, in three dimensions, the Ricci tensor and the Riemann tensor happen to have the {\it exact} same number of independent variables. The Riemann tensor has $\f{D^2(D^2-1)}{12}|_{D=3}= 6 $ independent components. The Ricci tensor, by virtue of being a symmetric tensor of rank 2, has $\f{D(D+1)}{2} |_{D=3}=6$ independent components. There is no room for the Weyl tensor in three dimensions, and so it must vanish by construction. The Riemann tensor is thus just a function of the Ricci tensor and scalar
\begin{equation}
R_{\mu \nu \alpha \beta} = g_{\mu \alpha} R_{\nu \beta} + g_{\nu \beta} R_{\mu \alpha} - g_{\mu \beta}R_{\nu \alpha} - g_{\nu \alpha}R_{\mu \beta} + \f{1}{2} \left( g_{\mu \beta} g_{\nu \alpha} - g_{\mu \alpha}g_{\nu \beta} \right)R \; ,
\end{equation}
and in absence of matter, the Riemann tensor identically vanishes. 

That is, we see that pure gravity in three dimensions is a really simple theory. The geometry of the theory is one of constant curvature given by the cosmological constant. And due to the fact that the Weyl tensor identically vanishes, there are no gravitational degrees of freedom. The only possible degrees of freedom are therefore topological ones.

It is, in fact, really easy to recover this result by a simple counting of the numbers of variables and constraints for three-dimensional gravity. In $D$ dimensions, the phase space of General Relativity is characterized by a spatial metric on a constant time hyper-surface, which has $\f{D(D-1)}{2}$ independent components since it is a symmetric tensor of rank 2. Adding the conjugate momenta to the counting, we arrive at $(D)(D-1)$ degrees of freedom for the full phase space. Now, to recover the actual degrees of freedom of the theory, we must subtract to this result the number of constraints. If $D \geq 3$, we have $D$ constraints coming from the choice of coordinates and $D$ coming from the Einstein equations which impose a constraint on the initial condition rather than on the dynamics of the system. Putting everything together, we are left with $D(D-1)-2D = D(D-3)$ degrees of freedom. If $D=3$, this clearly vanishes and we recover the previous claim of no gravitational degrees of freedom for pure gravity in three dimensions.

This is not the only peculiarity of the three-dimensional case. On top of the vanishing of the Weyl tensor, the Newtonian potential also vanishes. That is, the Newtonian limit of three-dimensional General Relativity is very peculiar; there is no force between masses. For more detail in general on three-dimensional gravity, see also \cite{CarlipBook:1998}.
\\

While the Einstein-Hilbert action is generally suitable to study the classical feature of the theory, it is quite difficult to quantize. In three dimensions, the two favourite formulations to quantize gravity are the Chern-Simons formulation for AdS and the $BF$ formulation for flat gravity. For a negative cosmological constant $\Lambda < 0 $, i.e. for an AdS space, it can be shown that the Einstein-Hilbert action is equivalent to the difference between two Chern-Simon actions for non-abelian gauge fields \cite{Achucarro:1987vz,Witten:1988hc}. The quantization procedure of a Chern-Simons theory is well understood, at least for compact groups and in the case of Euclidean gravity, the gauge group is $\SU(2)$, which is compact. The second formulation is the most interesting for our purpose and is based on a first order formulation of General Relativity. The next section is dedicated to a brief review of the $BF$ formalism and gravity as a $BF$ theory.
\section{General Relativity as a BF theory}

$BF$ theory is the name given by Blau and Thompson \cite{BLAU1991130} to the very general class of topological field theory first introduced\footnote{The Abelian case was first considered by Schwarz a few years earlier \cite{Schwarz:1978cn}.} by Horowitz \cite{Horowitz:1989ng}. It is a fairly simple class of topological field theories and can in a sense be considered as the simplest possible gauge theory. A $BF$ theory can be defined in any dimension and is naturally independent of any choice of background geometry. Indeed, these theories are defined without any mention of an existing metric or any other geometrical structure of the space-time manifold. It is therefore natural to consider this class of models for the study of a background independent theory. The first mention of this type of theory in General Relativity dates back to the work of Plebanski in 1977 and its first order formulation of gravity. In particular, in three dimensions, General Relativity is a special case of $BF$ theory, while in four dimensions, one must add extra constraints to the initial $BF$ action to recover General Relativity \cite{Baez:1995ph}. This was to be expected since gravity is not a topological field theory in four dimensions. See \cite{Celada:2016jdt} for a review of $BF$ gravity theories.

\subsection{Definition}

Consider a $D$-dimensional manifold $\cM$, acting as the space time, and a trivial $\cG$ bundle on top of it associated to the Lie group $G$. We denote its Lie algebra by $\g$. Consider now a connection $\omega$ on the trivial $\cG$ bundle. Locally, it can be seen as a $\g$-valued one-form. We also need a $(D-2)$ form $B$ valued in the adjoint bundle of $\cG$, i.e. the vector bundle associated to $\cG$ via the adjoint action of $G$ on its algebra. Locally, $B$ is a $\g$-valued $(D-2)$ form. The curvature of $\omega$ is defined as the $\g$-valued two-form $F[\omega] = \dd \omega + \omega \wedge \omega$. The $BF$ action is then defined by the integral
\begin{equation}
	S_{BF} = \int_{\cM} \tr(B \wedge F[\omega]) \; .
\end{equation}
The name\footnote{Note that it is rather more intuitive to rename the $B$ field $E$. At the phase space level, $B$ is understood as the conjugate variable to the connection $\omega$. When considering the group $U(1)$, this corresponds to the electric field of electromagnetism. } $BF$ is self-explained by the structure of the action. The notation trace refers to the inner product of the Lie algebra $\g$. 

Trivially, as an integral of a differential form, the $BF$ action is invariant under the diffeomorphisms of $\cM$. Forgetting for a moment about the diffeomorphisms, the $BF$ action exhibits two gauge symmetries. The first one is the usual gauge transformation defined by the gauge group $G$. It is parametrized by a $\g$-valued zero-form $\alpha$ and its infinitesimal action on the field is
\begin{equation}
	\delta_{\alpha} B = [B,\alpha] \; , \quad \delta_{\alpha} \omega = d_{\omega} \alpha \; .
\end{equation}

The second invariance is called the translational, or shift symmetry, due to its action on the field $B$. It is parametrized by a $\g$-valued $(D-3)$ form $\eta$. The reason for this symmetry is to be found in the Bianchi identity for the curvature $F$ stating that $\dd_{\omega} F[\omega] = 0$. Its infinitesimal action on the field is
\begin{equation}
	\delta_{\eta} B = d_{\omega} \eta \; , \quad \delta_{\eta} \omega = 0 \; . 
\end{equation}

The equations of motion of the theory are easily derived from the action by taking its variation with respect to the fields $B$ and $\omega$. We find
\begin{equation}
	F[\omega] = 0 \; , \qquad \dd_{\omega} B = 0 \; ,
\end{equation}
where $\dd_{\omega}$ stands for the exterior covariant derivative with respect to the connection $\omega$. These equations are really simple. The first one is telling us that the connection $\omega$ is flat, i.e. that basically, the $(D-2)$ form $B$ acts as a Lagrange multiplier enforcing the flatness of the connection everywhere. On the one hand, due to the flatness, the solutions for $\omega$ are locally all the same up to gauge transformations. On the other hand, the shift symmetry allows to always transform $B$ into its locally trivial solution $ B = 0$. Indeed, the second equation of motion tells us that $B$ is closed. Now, since the connection is also flat, we know that locally, all closed forms are also exact. That is, it exists $\eta$ a $(D-3)$-form such that $B = \dd_{\omega} \eta$. It is clear now that the shift symmetry always allows us to shift $B$ to zero.

Note that these equations admit a useful interpretation in the context of co-homology. Indeed, the square of the exterior covariant derivative can be expressed in terms of the curvature through the formula $d^{2}_{\omega} f = [F[\omega],f]$. Hence, when the connection is flat, we have $d^{2}_{\omega} f = 0 $. Looking at the equation of the $B$ field now, we see that the space of solution for $B$, taking into account the gauge symmetry, is the second DeRham co-homology group. This interpretation will be useful in the following when looking at the path integral of the $BF$ action, since it will allow us to relate it to the Ray-Singer torsion, a topological invariant \cite{BLAU1991130}. In the BRST formalism, the reducibility of the shift symmetry plays an important role and involves the introduction of the "ghost for ghost". 

To conclude this section, we will point out that the diffeomorphism invariance is encoded into the two gauge symmetries defined above \cite{Freidel:2002dw,Buffenoir:2004vx}. Indeed, consider a vector field $\xi$ on $\cM$. From this vector field, we can define a zero-form $\phi_{\xi}$ valued in the adjoint bundle of $G$ via the interior product of form $i_{\xi}$ by $\phi_{\xi} = i_{\xi} \omega$. Similarly, we define the $(D-3)$-form $\eta_{\xi}$ by $i_{\xi} B$. The action of the diffeomorphism $\xi$ on the connection and the field $B$ is then
\begin{equation}
	\delta_{\xi} \omega = d_{\omega} \phi_{\xi} + i_{\xi} F[\omega] \; , \qquad \delta_{\xi} B = \dd_{\omega} \eta_{\xi} + [\phi_{\xi},B] + i_{\xi} \dd_{\omega} B \;.
\end{equation}
On-shell, this is just a combination of the gauge and shift symmetry.

\subsection{Three-dimensional BF theory and gravity}

Obviously, the point of interest here is the application of the $BF$ formalism to the case of gravity in three dimensions. In that case, the gauge group is $\SU(2)$ and the $BF$ action corresponds to a first order formulation of gravity. The field $B$ is then called the tetrad field, and usually denoted by $e$. It is a $\su(2)$-valued one-form. The action of General Relativity reads in this formalism
\begin{equation}
	S_{BF} = \f{1}{16 \pi G} \int_{\cM} \dd^{3}x \; \left(  \tr(e \wedge F[\omega]) + \f{\Lambda}{6} \tr(e \wedge e \wedge e) \right)\; .
\end{equation}
Explicitly $e^a = e^a_\mu\dd x^\mu$,  $\omega^a = \tfrac12 \epsilon^a{}_{bc}\omega^{bc}_\mu \dd x^\mu$ and $F^a[\omega] = \dd\omega^a + \tfrac12 \epsilon^{a}_{bc}\omega^b\wedge \omega^c$ are the tetrad one-form, the spin-connection one-form, and the curvature two-form, respectively. The index $a$ is a tangent space index to the space-time manifold, in agreement with the standard relation between the tetrad and the metric, i.e. $g_{\mu\nu}=\eta_{ab}e^a_\mu e^b_\nu$, where $\eta_{ab}$ is the flat metric of signature equal to that of $g_{\mu\nu}$. For the sake of completeness, we consider the case with a possible non-vanishing cosmological constant. We have added to the $BF$ action a volume term coinciding with the addition of the cosmological constant.

Note that in order to recover the usual formulation of General Relativity, the tetrad field must be non-degenerated and associated to a positive volume form. That is, it requires the constraint $\det(e) > 0$ everywhere. In that case, the metric $g_{\mu \nu}$ can be defined as above, and we can recover the usual Einstein-Hilbert action from the first order formalism. Now, in $BF$ theory, no constraints are imposed on the tetrad field. Hence, the above action defines a more general theory compared to General Relativity. When some space configurations where $\det(e) = 0$ exist, the metric field is degenerated and no longer invertible anymore. These configurations are expected to be relevant to the discussion of topological change, see for example \cite{Hawking:1979zw,Horowitz:1990qb,Kaul:2016lhx}. Taking $e$ to its opposite, it is immediate to see that to each tetrad field with $\det(e)>0$ is also associated a tetrad field with negative volume form \cite{Freidel:1998ua}, which is related to the orientation of the space-time manifold. At the classical level, these two other field configurations might not matter, since we can always restrict the study to the space of non-degenerate metric field. However, at the quantum level, it is not possible to do so. Indeed, the path integral associated to the $BF$ action will naturally be over all tetrad fields. However, only those with $\det(e) > 0$ are directly related to General Relativity. Hence, we naturally obtain a bigger theory than General Relativity. See \cite{Freidel:1998ua} for a discussion about this point.

In the previous section, we did not discuss the possibility of having a manifold with boundaries. In that case, the variation of the action is not necessarily zero on-shell. The explicit computation of the variation of the $BF$ action with a volume term returns
\begin{align*}
	16 \pi G \delta S_{BF} 
	&=
	\int_{\cM} \dd^3 x \; \tr(\delta e \wedge (F[\omega]+ \f{\Lambda}{3} e\wedge e) + e \wedge \delta F[\omega])
	\\
	&=
	\int_{\cM} \dd^3 x \; \tr(\delta e \wedge (F[\omega]+ \f{\Lambda}{3} e\wedge e) + e \wedge \dd_{\omega} \delta \omega)
	\\
	&=
	\int_{\cM} \dd^3 x \; \tr(\delta e \wedge (F[\omega]+ \f{\Lambda}{3} e\wedge e)  + \dd_{\omega} e \wedge \delta \omega) + \int_{\pp\cM} \dd^2 x \; \tr(e \wedge \delta \omega) \; ,	
\end{align*}
where we first applied the identity $\delta F[\omega] = \dd_{\omega} \delta \omega$ and then integrated by part. On-shell of the equations of motion the variation of the action is
\begin{equation}
	\delta S_{BF} = \f{1}{16 \pi G} \int_{\pp\cM} \tr(e \wedge \delta \omega).
\end{equation}

This expression of the boundary term shows us that the natural quantity to be kept fixed at the boundary is the pull-back of the connection $\omega$, such that no boundary term is needed in the action. If it is the pull-back of the local tetrad field to be kept fixed at the boundary, a boundary term is then needed, and we get
\begin{equation}
S_{BF} = \f{1}{16 \pi G} \int_{\cM} \dd^{3}x \; \tr(e \wedge F[\omega]) - \int_{\pp\cM} \dd^{2}x \; \tr(e\wedge\omega) \; .
\end{equation}
This term is similar to the Gibbons-Hawking-York term needed in the metric formulation in presence of a boundary. Note that the presence of boundary makes the analysis from the gauge theory point of view more complicated. Indeed, it is well-known that gauge invariance might be broken at the boundary. In order to restore the gauge invariance at the boundary, one must consider new compensating fields living at the boundary. These are related to the soft modes for asymptotic boundary and edge modes for boundary at a finite distance \cite{Strominger:2017zoo}.
\medskip

The starting point for General Relativity in this thesis is the action
\begin{equation}
	S_{BF} = \int_{\cM} \dd^3 x \;  \tr(e \wedge F[\omega]) \;.
	\label{chap1:eq:BF_action_starting}
\end{equation}
We choose the unit such that $16\pi G = 1$ and we keep fixed the connection at the boundary such that no boundary term is needed in the action. 
\section{Asymptotically AdS and flat space: metric and geometry}

We finish this chapter with a brief presentation of the AdS and flat space geometry. The main point of this section is two-fold. First, we want to give another explanation for the use of the torus topology in our computation. This geometry naturally arises when looking at the asymptotic geometry of the AdS space, the BTZ black hole and flat space. Secondly, this thesis is centred on the study of quasi-local regions. It is however appealing to compare our results with the one done in the continuum. As such, we will briefly review these two cases. We begin by the presentation of the metric of the AdS case, then we look at the asymptotic symmetry group and charge algebras. We end the section with the same analysis for the flat case as a limit of the AdS space. Since most of the computations are involved, we mainly focus on the core ideas and results, and refer to the literature for more details, see \cite{Brown:1986nw,Oblak:2015sea,Strominger:2017zoo,Compere:2018aar}.

As we said in the introduction, the notion of asymptotic symmetry is one of the key elements of the study of dual theories. A field theory with boundary, as a Hamiltonian theory, is defined by its field content, Poisson structure and boundary conditions. For three-dimensional gravity, the last point is crucial. Indeed, as we have emphasized since the beginning, three-dimensional pure gravity is trivial in the sense that it has no local degrees of freedom. In presence of a boundary however, global degrees of freedom arise. Therefore, the choice of boundary conditions completely determines the behaviour of the theory. In the context of asymptotic holography, boundaries are at infinity, and fall-off conditions on the fields act as boundary conditions. Recall that the choice of fall-off conditions might affect the well-definiteness of action. The asymptotic symmetry group of such a theory is defined as the quotient between allowed and trivial gauge transformations. Allowed gauge transformations are the ones that preserve the fall-off conditions while trivial gauge transformations are the allowed transformations associated with a zero conserved charge. 
In short, the researches of asymptotic symmetry boils down to the definition of fall-off conditions and the study of gauge transformations with respect to these conditions.

\subsection{AdS space and Virasoro algebra}

We consider in this section a non-vanishing length scale $l \in \R^{+}$, called the $AdS$ radius, such that the cosmological constant is defined by $\Lambda = -\f{1}{l^2}$. In three dimensions, the unique maximally symmetric space-time with negative cosmological constant is the AdS space, and is described by the metric
\begin{equation}
	\dd s^2 = - \left( 1 + \f{r^2}{l^2} \right) \dd t^2 + \left( 1 + \f{r^2}{l^2} \right)^{-1} \dd r^2 + r^2 \dd \varphi^2 \; .
	\label{eq:chap1:AdS_metric}
\end{equation}
In these notations, $t$ is a time-like coordinate without any periodicity condition, whose range, as usual, is infinite to avoid closed time-like curves. The other coordinates are $r \in \R^{+}$, the luminosity distance, and $\varphi \in [0,2\pi]$, the angular coordinate, which are both space-like. The manifold spanned by such a metric is diffeomorphic to $S^1 \times \R^2$. At short distances, i.e for $r \ll l $, the metric is just the Minkowski metric, whereas at long distances, $r \gg l$, it is 
\begin{equation}
	\dd s^2 \sim \f{r^2}{l^2} \dd r^2 + \f{r^2}{l^2} \left(-\dd t^2 + l^2 \dd \varphi^2 \right) \; .
	\label{eq:chap1:AdS_metric_spatialinfinity}
\end{equation} 
From the form of the metric at long distances, we conclude that the boundary geometry is that of a cylinder.

It is interesting to rewrite AdS in a new set of coordinates $(\tau,\rho,\varphi)$ where $\tau = \f{t}{l}$ (the rescaling of the time coordinate is mainly done for convenience) and where $\rho$ is a dimensionless coordinate defined by
\begin{equation*}
	r = l \sinh(\rho) \; .
\end{equation*}
With this new set of coordinates, the AdS metric becomes
\begin{equation}
	\dd s^2 = l^2 \left( \dd \rho^2 - \cosh^2(\rho) \dd \tau^2 + \sinh^2(\rho) \dd \varphi \right) \; .
	\label{eq:chap1:AdS_metric_staticglobal}
\end{equation}

The easiest way to study the symmetries of AdS is to immerse the hyper-surface spanned by \eqref{eq:chap1:AdS_metric_staticglobal} into the four dimensional space $\R^{2,2}$ with the metric $\eta = (-1,-1,1,1)$. AdS can then be seen as the hyper-surface in $\R^{2,2}$ satisfying the constraint
\begin{equation}
	-u^2-v^2+x^2+y^2 = l^2 \;. 
	\label{eq:chap1:R22_AdS_Constraint}
\end{equation}

A natural parametrization of AdS in $\R^{2,2}$ is found by taking\footnote{This parametrization naturally leads to interpreting $\tau$ as an angle, thus closed time-like curves exist. Going back to \eqref{eq:chap1:AdS_metric} involves taking the universal cover after inverting all the relations, such that $t = l \tau$ has an infinite range}.
\begin{align*}
	u &= l \cosh(\rho) \cos(\tau) \\
	v &= l \cosh(\rho) \sin(\tau) \\
	x &= l \sinh(\rho) \cos(\varphi) \\
	y &= l \sinh(\rho) \sin(\varphi) \;. 
\end{align*}
Evaluating the metric of $\R^{2,2}$ on the hyper-surface parametrized by the previous set of coordinates yields back \eqref{eq:chap1:AdS_metric_staticglobal}.

The biggest advantage of such formulation is that it allows to easily study the isometries of an AdS space. By its very definition, $\R^{2,2}$ has $10$ isometries (as a space in 4 dimensions). There is one translation per direction, i.e. 4 translations total and 6 matrices transformations. These matrices transformations emerge from the invariance of the metric under the action of a $\SO(2,2)$ matrix by its very definition. The translation symmetries are not conserved on AdS. Indeed, \eqref{eq:chap1:R22_AdS_Constraint} is clearly not conserved under translations. The action of these symmetries is to map a given AdS space of radius $l$ to another AdS space of radius $l'$, with $l'$ a function of $l$ and the translation parameters. However, \eqref{eq:chap1:R22_AdS_Constraint} is clearly conserved by the action of an $\SO(2,2)$ matrix. Hence AdS admits $\SO(2,2)$ as symmetry group.
\footnote{Another really elegant proof is to parametrize the AdS space embedded in $\R^{2,2}$ by an $\SL(2,\R)$ matrix $g$
\begin{equation*}
	g = \f{1}{l^2}
	\begin{pmatrix}
	x-u & v+y \\
	v-y & x+u
	\end{pmatrix}
	\end{equation*}
	with the condition on $g$
	\begin{equation*}
	\det(g) = 1 \;
	\end{equation*}
	imposing \eqref{eq:chap1:R22_AdS_Constraint}. The AdS metric is then recovered using the definition of the Killing-Cartan metric on the group manifold $\SL(2,R)$
	\begin{equation*}
	\dd s^2 = \f{1}{l} \tr\left( (g^{-1} \dd g)(g^{-1} \dd g) \right) \;.
	\end{equation*}
	This definition is invariant by the action of two independent $\SL(2,R)$, acting on the right or on the left on $g$, modulo $\Z^2$. This gives back $\SO(2,2) \approx \f{\SL(2,R)_{L} \times \SL(2,R)_{R}}{\Z^2}$ as symmetry group.
} 

The generators of $\SO(2,2)$ are 6 independent Killing vectors, whose general form is (consider $x_\mu = (u,v,x,y)$)
\begin{align}
J_{\mu \nu} = x_\nu \pp_\mu - x_\mu \pp_\nu \; .
\end{align}
Note that the generator of time translations is $\pp_{\tau} = \f{1}{l}(u \pp_{v} - v \pp_{u})$ and the generator of rotations is $\pp_{\varphi} = x \pp_{y} - y\pp_x $.

The space described above is called globally AdS, in the sense that the metric is always the AdS metric. Such a space is really constrained as it is unique. In the following, we are interested in spaces that are only asymptotically AdS. That is, the metric of the considered space must take the form \eqref{eq:chap1:AdS_metric_spatialinfinity} at spatial infinity.

\subsubsection{Euclidean and thermal AdS}

In this thesis, we are mainly interested in the Euclidean space. It is straightforward to go from the Lorentzian to the Euclidean metric from \eqref{eq:chap1:AdS_metric}. This is done via Wick rotation. Consider $t_{E} = i t$. The AdS metric \eqref{eq:chap1:AdS_metric} then becomes
\begin{equation}
\dd s^2 = \left( 1 + \f{r^2}{l^2} \right) \dd t^2 + \left( 1 + \f{r^2}{l^2} \right)^{-1} \dd r^2 + r^2 \dd \varphi^2 \; .
\label{eq:chap1:AdS_metric_euclidean}
\end{equation}
The geometrical picture at "spatial" infinity, that is for $r \gg l$ is still that of a cylinder. It can still be viewed as a circle with a "time" evolution giving the cylinder. The top and the bottom of the cylinder are then seen as initial and final states for the evolution. The main difference between the Lorentzian case is that the symmetry group of globally Euclidean AdS space is $\SL(2,\C)$ instead of $\SO(2,2)$. There is nothing more to say at the classical level.

However, this is not the case from the quantum mechanical point of view. Indeed, in the Euclidean formalism obtained via Wick rotation, the partition function of the model is usually defined via the associated statistical partition function. Hence, the partition function takes a natural form of the trace over the energy level of the theory \cite{Zee:2003mt}. Geometrically, it means taking the "trace" over the cylinder. That corresponds to identifying the top and the bottom of the cylinder, hence spanning a torus. The identification map of the cylinder to a torus is endowed with a natural parameter, called the Dehn twist. Geometrically, we have the freedom of rotating, for example, the top of the cylinder by an angle $\gamma$ before doing the identification, see figure \ref{chap1:fig:dehn_twist_cylinder}. This freedom accounts for the large diffeomorphism on the torus, hence is of primordial importance in the context of gravity. After doing this identification, the space we obtain is called thermal AdS.
\begin{figure}[!htb]
	\centering
	\begin{tikzpicture}[scale=1.9]
		\draw (-3,1.8) node{$\beta$};

		\draw[thick] (-2,-1.4) arc (-70:275: 3cm and 1.5cm);
		\draw[thick] (-1,0.2) to [bend left=10] (-2.4,-0.35); 
		\draw (-2,-1.4) to[in =0, out = 45] (-2.4,-0.35);
		\draw (-2,-1.4) to[out=180,in=225] (-2.4,-0.35);
		\draw[->] (-2,-1.1) to[out=60, in=-10]node[pos=0.5,left]{$\gamma$} (-2.3,-0.55);
		
		\draw (-2.8,-1.49) to[in =0, out = 45] (-3.1,-0.44)  ;
		\draw[dashed] (-2.8,-1.49) to[out=180,in=225] (-3.1,-0.44)  ;
		
		\draw[thick] (-3.1,-0.44) to [bend left=10] (-5,0.2);
		\draw[thick] (-1.5,-0.06) to [bend right] (-4.5,-0.06);

%		\draw[thick] (-0.5,1) arc (-80:270: 2.5cm and 1cm); 

	\end{tikzpicture}
	\caption{Geometrical picture of the Dehn twist associated to the cylinder before gluing it to obtain a torus. One side is rotated by an angle $\gamma$ before being identified to the other side. The length of the would-be non contractible cycle of the torus is $\beta$.}
	\label{chap1:fig:dehn_twist_cylinder}
\end{figure}
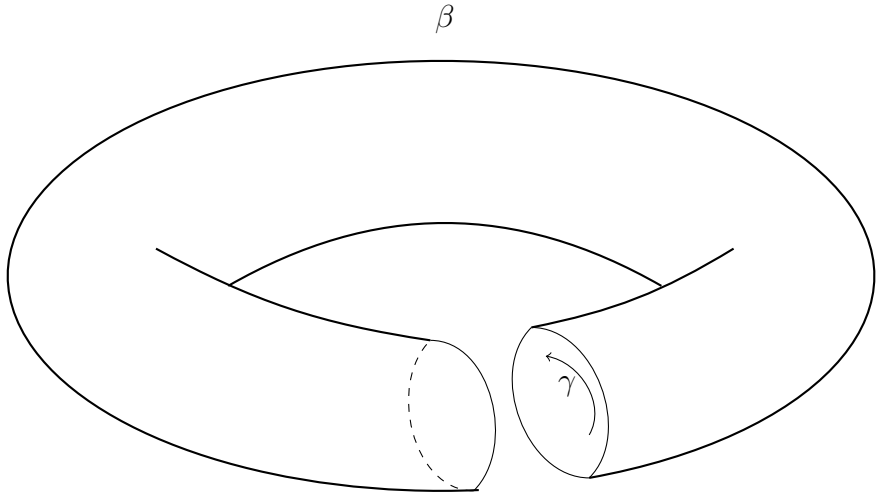

Explicitly, thermal AdS is represented by the metric \eqref{eq:chap1:AdS_metric_euclidean} with the following identification for the coordinates
\begin{equation}
	(t,r,\varphi) \sim (t+\f{\beta}{l},r,\varphi+\gamma)\;.
	\label{chap1:eq:periodic_condition_thermal_AdS}
\end{equation}
We recall that $\varphi$ is also defined in $[0,2\pi]$. The parameter $\beta$ represents the inverse temperature from the statistical system point of view. It is basically the height of the cylinder in the time direction, see figure \ref{chap1:fig:dehn_twist_cylinder}. This is the representation of a torus with modular parameter\footnote{This is obviously not the same $\tau$ as the rescaled time parameter. This notation is unfortunate, but both $\tau$ are standard notations. In the rest of this thesis, we will not use the rescaled time parameter anymore, and in the current section, the context is enough to differentiate between both $\tau$.}
\begin{equation}
	\tau_{\text{TAdS}} = \f{1}{2\pi}\left( \gamma + i \f{\beta}{l} \right)
\end{equation}
Another possible way to find this metric is to work again with an immersion into a bigger space. Here, instead of $\R^{2,2}$, we consider $\R^{1,3}$ as a straightforward generalization to the Euclidean case. Euclidean AdS is then described by the constraint
\begin{equation}
	-u^2 + v^2 +x^2 + y^2 = l^2 \;,
	\label{eq:chap1:euclidean_AdS_constraint}
\end{equation}
and the space can be parametrized in almost the same way as in the Lorentzian case, taking into account the Wick rotation. Note that thermal AdS is still a space that is globally AdS, and not only asymptotically AdS. Generally, any globally AdS space can be obtained as the quotient of the general AdS space introduced previously and some finite groups. For example, thermal AdS is the quotient between global AdS and $\Z$.

\subsubsection{BTZ black hole}

Before moving to asymptotic symmetry, let's briefly discuss the case of black holes in Euclidean gravity. Black hole solutions are a primordial feature of General Relativity. In fact, one of the reasons that made people believe that three-dimensional gravity was of no interest whatsoever was the fact that no black hole solutions were discovered. It was only in 1992 that Bañados, Teitelboim and Zanelli \cite{Banados:1992wn,Banados:1992gq} found black holes solution in three dimensions. The interesting fact about these black holes is that they are not the usual black holes as in four dimensions: they do not possess any curvature singularity. In fact, these solutions are still locally AdS everywhere. They are black holes because they possess an event horizon. Still being locally AdS, it is possible to relate the usual BTZ black hole solution to the thermal AdS space by considering a modular transformation \cite{Carlip:1994gc,Carlip:2005zn}. More precisely, the geometry of the BTZ black hole is again a torus, related to the thermal AdS torus by the modular transformation
\begin{equation}
	\tau_\text{TAdS} = - \frac{1}{\tau_\text{BTZ}}.
	\label{chap1:eq:modular_relation_BTZ_TAdS}
\end{equation}

The effect of this transformation is to interchange the temporal direction $t$ and the angular direction $\varphi$. Hence, while it is the temporal direction that is non contractible for thermal AdS, it is the angular direction for the Euclidean black hole, see figure \ref{chap1:fig:TAdS_BTZ}
\begin{figure}[!htb]
	\centering
	\begin{tikzpicture}[scale=1]
	%TAdS
	\draw[thick] (0,0) arc (0:360: 3cm and 1.5cm);
	\draw[thick] (-1,0.2) to [bend left] (-5,0.2); 
	\draw[thick] (-1.5,-0.06) to [bend right] (-4.5,-0.06);
	
	\draw[red] (-0.5,0) arc (0:360: 2.5cm and 1cm);
	\draw[red] (-0.5,0) node[left]{$t$};
	\draw[blue] (-4,-0.25) to [bend right] (-4.1,-1.4);
	\draw[blue,dotted] (-4,-0.25) to [bend left] (-4.1,-1.4);
	\draw[blue] (-4.2,-1) node[left]{$\varphi$};
	
	\draw (-3,-2) node{TAdS};

	%BTZ
	
	\draw[thick] (8,0) arc (0:360: 3cm and 1.5cm);
	\draw[thick] (7,0.2) to [bend left] (3,0.2); 
	\draw[thick] (6.5,-0.06) to [bend right] (3.5,-0.06);
	
	\draw[blue] (7.5,0) arc (0:360: 2.5cm and 1cm);
	\draw[blue] (7.5,0) node[left]{$\varphi$};
	\draw[red] (4,-0.25) to [bend right] (3.9,-1.4);
	\draw[red,dotted] (4,-0.25) to [bend left] (3.9,-1.4);
	\draw[red] (3.8,-1) node[left]{$t$};
	
	\draw (5,-2) node{BTZ};
	
	\end{tikzpicture}
	\caption{Torus associated to the thermal AdS space on the left and BTZ on the right. They are related to the modular transformation \eqref{chap1:eq:modular_relation_BTZ_TAdS} switching the non-contractible cycle. It is in the time direction (red) for thermal AdS and in the angular direction (blue) for the BTZ black hole.}
	\label{chap1:fig:TAdS_BTZ}
\end{figure}
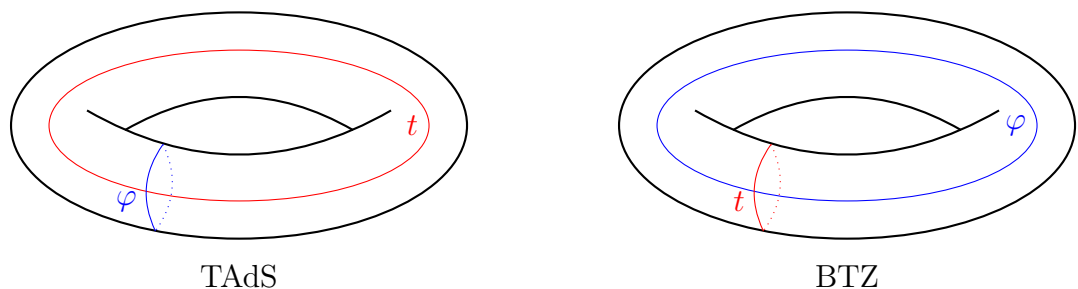

\subsubsection{Evaluation of the Einstein-Hilbert action on thermal AdS}

Finally, let us explicitly compute the on-shell Einstein-Hilbert action for thermal AdS with boundary at spatial infinity. The expression will come in handy in the following chapter when looking at the partition function of the model. The computation is not as easy as one might think, since, at first glance, the on-shell action diverges due to the infinity of the radial extension in space. To do the computation, we hence introduce a regularization by considering the boundary to be located at a distance $r=a$. The evaluation of the Einstein-Hilbert action on-shell is then proportional to the volume of the AdS space with the cut-off $a$
\begin{equation}
	\frac{1}{16\pi G}\int \sqrt{g}\; \left(R + \frac{2}{l^2} \right) 
	= - \frac{ \text{Vol}(a)}{8 \pi G l^2} 
	= - \frac{\beta}{8 G}\Big(\cosh(2a) - 1\Big).
\end{equation}
where $\text{Vol}(a)$ is the volume associated to the AdS space.

For the GHY boundary term, it is most easily calculated by introducing the normal at the boundary
\begin{equation}
	\frac{1}{8 \pi G} \oint \sqrt{h} \; K = \frac{1}{8 \pi G} \pp_n \oint \sqrt{h},
\end{equation}
where $\pp_n$ stands for the normal derivative. With the cut-off $a$, the area of the boundary is
\begin{equation}
	\text{Area}(a) = \pi \beta l \sinh(2a),
\end{equation}
and thus
\begin{equation}
	\frac{1}{8 \pi G} \oint \sqrt{h} \; K =  \frac{1}{8 \pi G l} \frac{\pp}{\pp r}_{|r=a} \text{Area}(r) = \frac{2 \beta}{8 G}\cosh(2a).
\end{equation}

We can now combine these results and take the limit $a\to\infty$. Recall however that the full AdS action is given by \eqref{eq:chap1:fullAdS_action} . Finally, the on-shell AdS action reads
\begin{equation}
	S_\text{TAdS} = -\frac{\beta}{8 G}\Big( 1 - \cosh(2a) + 2 \cosh(2a) - \sinh(2a)\Big) \xrightarrow{a\to\infty} -\frac{\beta}{8 G} \; .
	\label{eq:chap1:TAdS_onshell_action}
\end{equation}

\subsubsection{Asymptotic symmetry and charge algebra}

It is now time to focus on the asymptotic symmetry of an {\it asymptotically} AdS space. The asymptotic symmetry group and associated asymptotic conserved charge algebra were first derived in the seminal work of Brown and Henneaux \cite{Brown:1986nw}. A more recent derivation can also be found in the thesis of Oblak \cite{Oblak:2015sea} where the flat case is also explained. In both works, the fall-off conditions are first imposed on the metric before imposing any gauge fixing. It is possible to first get rid of the gauge redundancies to make the computation simpler, see \cite{Compere:2018aar}. In this section, we just state the results and refer the reader to the literature for the computational details.
\medskip

The asymptotic symmetry group of an asymptotically AdS space is the direct product of two commuting {\it Witt} algebras. Hence, it is spanned by two infinite families of generators. We call the Fourier modes of these two families
\begin{equation*}
	\xi_{n} \qquad \text{and} \qquad {\tl \xi}_{n} 
\end{equation*}
with $n \in \Z$. The Lie brackets of these generators are
\begin{align*}
	[\xi_{n},{\tl \xi}_{n}] &= 0 \\
	i[\xi_{n},\xi_{m}] &= (m-n)\xi_{n+m} \\
	i[{\tl \xi}_{n},{\tl \xi}_{m}] &= (m-n){\tl \xi}_{n+m} \;.
\end{align*}
The Witt algebra is the centreless algebra of circle diffeomorphism. Starting from the previous Lie brackets, it is immediate to see that the initial $\SO(2,2)$ symmetry group of a globally AdS space is a sub-set of the full symmetry group of asymptotically AdS space. Indeed, consider the generators for $n = -1,0,1$. They clearly form a closed subalgebra for both $\xi_n$ and $\bar{\xi}_n$ which is isomorphic to an $sl(2,\R)$ algebra and $\SO(2,2)$ is recovered as
\begin{equation*}
	\SO(2,2) \approx \f{\SL(2,R)_{L} \times \SL(2,R)_{R}}{\Z^2}	\;.
\end{equation*}

The final step is to look at the algebra of conserved charges arising from the previous symmetry group. These conserved charges span two independents Virasoro algebras. Calling $\cL_{n}$ and $\bar{\cL}_n$ the generators of the two independent families of charges (each coming from its own Witt algebra), the associated Poisson bracket reads
\begin{equation}
	i\left\lbrace\cL_m,\cL_n\right\rbrace = (m-n)\cL_{n+m} + \f{c}{12} m(m^2-1)\delta_{m+n,0}
\end{equation}
where $c$ is the Brown-Henneaux central charge
\begin{equation}
	c = \f{3l}{2G} \; .
\end{equation}

\subsection{Flat gravity and the BMS group}

To end this chapter, we focus on the more interesting case in the context of this thesis, the flat case. It is appealing to look at the flat case as the limit $l \to \infty$ of the AdS case. And indeed, most of the results derived for the flat case were obtained using this limit. Geometrically, the limit corresponds to pushing the boundary of the AdS cylinder to infinity. And as the length scale goes to infinity, the curvature vanishes, giving a completely flat theory. The general metrics that are asymptotically flat are usually parametrized in terms of the (retarded) Bondi coordinate $(u,r,\varphi)$. In these notations, $u$ is the retarded time, $r$ is still the luminosity distance, and $\varphi$ is the angle associated to the cirlce at infinity, acting as the basis for the infinite cylinder. In these coordinates, the general solution for the metric takes the form
\begin{equation}
\dd s^2 = \Theta(\varphi) \dd u^2 -2 \dd u \dd r + 2\left( \Xi(\varphi) + \f{u}{2}\pp_{\varphi}\Theta(\varphi) \right) \dd u \dd \varphi + r^2 \dd \varphi^2 \; .
\end{equation}
As in the AdS case, the phase space is characterized by two arbitrary functions on the boundary circle at infinity. Again, going to the Euclidean time, one should identify the top and the bottom of the cylinder to compute the associated statistical partition function.

Some work is necessary to derive the asymptotic symmetry and charges from the AdS case. Indeed, when $l$ goes to infinity, the Brown-Henneaux central charge diverges, and the previous brackets are hence ill-defined.

Consider the Fourier modes $P_m$ and $J_m$, $m \in \Z$ such that
\begin{equation}
\xi_{m} = \f{1}{2}\left(l P_m +  J_m \right) \quad \text{and} \quad {\tl \xi}_{m} = \f{1}{2}\left(l P_{-m} + J_{-m} \right) \;.
\end{equation}

The Lie brackets between the $P$'s and the $J's$ are
\begin{equation}
i[P_m,P_n] = \f{1}{l^2} (m-2) J_{m+n} \; , \; \; i[J_m,J_n] = (m-n)J_{m+n} \; , \; \; i[J_m,P_n] = (m-n) P_{m+n}
\end{equation}
such that in the limit where $l$ goes to infinity the $P_m$ commute, the $J_m$ form a Witt algebra and the $J_m$ act non-trivially on the $P_m$. We get the so-called $bms_3$ algebra
\begin{equation}
i[P_m,P_n] = 0 \; , \; \; i[J_m,J_n] = (m-n)J_{m+n} \; , \; \; i[J_m,P_n] = (m-n) P_{m+n} \; .
\end{equation}
In the language of $bms$, the $J_m$ span the infinite dimensional diffeomorphism algebra of the infinite circle, and are called the superrotations. Whereas the $P_m$ are called the supertranslations and correspond to translations of asymptotically null infinity in the time direction.

Similarly, we can obtain the charge algebra by taking the flat case limit of the AdS case. We define the supermomenta $\cP_m$ and the superrotation charges $\cJ_m$ by
\begin{equation}
\cL_{m} = \f{1}{2}\left(l \cP_m +  \cJ_m \right) \quad \text{and} \quad {\tl \cL}_{m} = \f{1}{2}\left(l  \cP_{-m} + \cJ_{-m} \right) \;.
\end{equation}
and the computation of the Poisson brackets gives
\begin{align}
i[\cP_m,\cP_n] &= \f{1}{l^2} (m-2) \cJ_{m+n} \; , \; \; i[\cJ_m,\cJ_n] = (m-n)\cJ_{m+n} \; , \nonumber \\ i[\cJ_m,\cP_n] &= (m-n) \cP_{m+n} + \f{1}{4G}m(m^2-1) \delta_{m+n,0} \;.
\end{align}

The central charge does not depend on $l$ anymore, and the flat limit can be safely taken to obtain
\begin{align}
i[\cP_m,\cP_n] &= 0 \; , \; \; i[\cJ_m,\cJ_n] = (m-n)\cJ_{m+n} \; , \nonumber \\ i[\cJ_m,\cP_n] &= (m-n) \cP_{m+n} + \f{1}{4G}m(m^2-1) \delta_{m+n,0} \;.
\end{align}
\medskip

Fundamentally, these two results tell us a lot about the theory. In particular, they entirely define the partition function, which should be a character of some representation of the asymptotic charge algebra. In practice, these results were confirmed by the computation of the partition function in the AdS case \cite{Giombi:2008vd} and in the flat case \cite{Barnich:2015mui} which we will quickly present in the next chapter.

In this chapter, we did not however address the question of finite boundary. The lookout of the symmetry group and algebra of conserved charges for finite boundary is still ongoing, and even though a lot of progress has been made in the last few years, see \cite{Freidel:2015gpa,Freidel:2016bxd,Geiller:2017xad,Speranza:2017gxd,Freidel:2019ees}, we are still far from having a clear understanding of the symmetry for finite boundary. In a sense, one of the aims of my work is to study this problem by directly computing the partition function for a quasi-local region, and see what can be learned from the structure of the results.

	\newpage
	~
	\thispagestyle{empty}
	
	\chapter{Quantum Gravity on the Torus}
\label{chap2}

While the previous chapter was dedicated to a short review of classical gravity and the asymptotic symmetry groups, this chapter focuses on the quantum counterpart. More particularly, we briefly introduced the work done in \cite{Giombi:2008vd} and \cite{Barnich:2015mui} on the computation of the partition function for three-dimensional quantum gravity in the AdS and flat case respectively. These computations and the methods are far from being related to the work of this thesis, and hence will not be presented in details. The results, however, are deeply connected. We are indeed interested in the more general case of finite boundary, and our results can easily be pushed at asymptotic infinity and compared with the results presented in this chapter. This chapter serves as a bridge between the asymptotic computation in the continuum and the quasi-local computation in the discrete.

On top of this utility purpose, it is a fact that the status of dual boundary theories for three-dimensional gravity is most thoroughly developed in the case of negative cosmological constant and for asymptotic boundaries. In the AdS case, this duality can obviously be understood as a particular representation of the AdS/CFT duality \cite{Witten:1988hc}. In this case, the emerging dual field theory is well-understood and corresponds to a really specific conformal field theory, known as Liouville theory \cite{Carlip:2005zn,Carlip:2005tz}. This Liouville theory can also be found starting from the Chern-Simon formulation of gravity. It is well-known that Chern-Simons theory is dual to a Wess-Zumino-Novikov-Witten model \cite{Wess:1971yu,Witten:1983tw,Knizhnik:1984nr,Gawedzki:1999bq}. Consequently, by imposing the right constraints on the fields of the Wess-Zumino-Novikov-Witten model to encode the Brown-Henneaux fall-off conditions, the model becomes a Liouville theory. Note that the associated Liouville potential might differ \cite{Carlip:2005zn} due to the choice of fall-off conditions. Once again, this emphasizes the importance of fall-off conditions and boundary conditions in general. For the flat case, recent progress was made \cite{Barnich:2012rz,Barnich:2013yka,Carlip:2016lnw}, see the main text for more details.

In the previous paragraph we discussed results at the classical level. At the quantum level, the first object of interest for holographic dualities is the partition function and understanding how the dual field theory is related to the understanding of the structure of the partition function. In this chapter, we first focus on the AdS computation presented in \cite{Giombi:2008vd} then on the analogous flat case \cite{Barnich:2015mui}. In both cases, the main point of the computation is to be found in the use of the heat kernel method\footnote{An introduction to heat kernel methods and one loop computation can be found in \cite{Vassilevich:2003xt}.}, involving lots of computations. We will not detail any of them here but rather focus on the results and their interpretations.

\section{Partition function of twisted thermal $\text{AdS}_3$}

In \cite{Giombi:2008vd}, the one-loop partition function of an asymptotically AdS space was computed starting from the path integral formulation of three-dimensional gravity with the Einstein-Hilbert action. To do so, the authors performed the Wick rotation of the path integral into the Euclidean space. As a consequence, they considered the thermal AdS space described in the previous chapter: they computed the partition function for gravity with boundary for a space-time described by a solid torus of modular parameter
\begin{equation}
	\tau = \frac{1}{2\pi}\left(\gamma +\I \frac{\beta}{l} \right) \; .
\end{equation}
We recall that $\beta$ is the temporal extension of the torus, i.e. its height in the non-contractible direction, $\gamma$ the Dehn twist and $l$ the AdS radius.

One of the key points of the computation is that it is claimed to be perturbatively exact \cite{Maloney:2007ud}. It means that all other loops corrections identically vanish. Of course, recall that three-dimensional General Relativity is perturbatively non-renormalizable, so the fact that the partition function is one-loop exact must be carefully understood. Recall that three-dimensional gravity is a theory without any local excitations, i.e. any local degrees of freedom. This suggests that, in some sense, it is an integrable system \cite{Witten:1988hc} and quantities in such a system are often one-loop exact. We will see in the following another viewpoint on this one-loop exact property of the computation. Note that this does not address the question of having non-perturbative corrections to the partition function. As mentionned in \cite{Maloney:2007ud}, the very definition of such non-perturbative corrections is hard to separate from the sum over all admissible manifolds allowed by the boundary conditions. One of the key results of this thesis is to show that non-perturbative corrections do arise in the quasi-local regime and might indeed be related to the admissible class of manifolds given the boundary data.

The actual computation involves a classical contribution: the on-shell action, together with a combination of functional determinants, among which there is a Faddeev-Popov determinant for the scalar and vector gauge modes. The ensuing result, after some non-trivial algebraic simplifications between scalar, vector and tensor (graviton) mode contributions, is found to be 
\begin{equation}
	Z_\text{TAdS}(\tau,\bar\tau) = \E^{ -S_\text{TAdS}} Z^\text{1-loop}_\text{TAdS}(\tau,\bar\tau)
	\qquad\text{with}\qquad
	Z^\text{1-loop}_\text{TAdS}(\tau,\bar\tau) = \prod_{k=2}^\infty \f1{\left|1 - \E^{2\pi\I \tau \cdot k} \right|^{2}}.
	\label{chap2:eq:AdS_one_loop_partition_function}
\end{equation}
We recall that $S_\text{TAdS}$ is the on-shell evaluation of the Einstein-Hilbert action on thermal $\text{AdS}_3$, given by \eqref{eq:chap1:TAdS_onshell_action}
\begin{equation}
	S_\text{TAdS}(\tau,\bar\tau) = -\f{\beta}{8 G} \; .
\end{equation}

One crucial and non-trivial fact about the resulting partition function is that the product over $k$ starts at $k=2$. Let us recall the interpretation of $k\in \mathbb Z$ in the calculation of \cite{Giombi:2008vd}. Note that the very same parameter will acquire a completely different interpretation in the various method of computations of the partition function of three-dimensional gravity. In the present case, recall that thermal $\text{AdS}_3$ space-time is obtained as a hyper-surface of the four-dimensional hyperbolic space. An equivalent construction is to consider the three-dimensional hyperbolic and its quotient by \eqref{chap1:eq:periodic_condition_thermal_AdS}. Under this construction, $k$ labels the copies of this space-time in the hyperbolic space. This identification naturally comes into the computation from the use of the method of images in the heat kernel. 

From the metric perspective, the reason why the product starts at $k=2$ is quite mysterious, and so is the fact that $Z(\tau,\bar \tau)$ factorizes in a holomorphic and an antiholomorphic contribution. However, both facts become transparent once the result is understood from the point of view of the boundary. From the boundary perspective, one has two uncoupled conformal field theories, each one coming with a Virasoro algebra, with Hamiltonians $\cL_0$ and $\tl \cL_0$, and Brown-Henneaux central charges $c = \tl c = \f{3 l }{2 G}$, as introduced in the previous chapter.

By definition, the total Hamiltonian and momentum operators for the full conformal field theory are
\begin{equation}
H =\frac{1}{l} ( \cL_0 + \tl \cL_0 )
\qquad\text{and}\qquad
P =\frac{1}{l} ( \cL_0 - \tl \cL_0 ) \; ,
\end{equation}
and the partition function can be expressed in the dual theory as a trace over them
\begin{equation}
	Z(\tau,\bar\tau) = \tr\left( \E^{ -\I \gamma P} \E^{ -\beta H } \right) = \tr_L\left( \E^{2\pi \I \tau  \cL_0}\right) \tr_R\left(   \E^{-2\pi\I \bar \tau \tl \cL_0}\right),
\end{equation}
where in the last equality we have highlighted the factorization of the two uncoupled theories. This encodes a sum over all the states which have been first evolved for an Euclidean time $\beta$, then translated by an amount $\gamma$ to the left, and finally re-identified with themselves. This corresponds to the view of having an initial and final state at the bottom and top of the cylinder respectively before identification.

Following Maloney and Witten \cite{Maloney:2007ud}, one can argue that the trace can be calculated as the sum of the contribution of the fundamental state $|\Omega\rangle$ of the conformal field theory and of its descendants. The descendants are obtained via the action of the two uncoupled Virasoro modes $\cL_{-k}$ and $\tl \cL_{-k}$ on $|\Omega\rangle$. Now, the fundamental state $|\Omega\rangle$ has vanishing momentum and energy $E_\Omega = -\f{c}{12 l} = -\f{1}{8G}$, which gives precisely the classical contribution to $Z$:
\begin{equation}
\langle \Omega| \E^{ -\I \gamma P} \E^{ -\beta H } |\Omega \rangle  = \E^{\frac{\beta}{8 G}} \equiv \E^{-S_\text{TAdS}}.
\end{equation}
The remaining contributions start at modes $k=2$ given that $\cL_{-1}$ and $\tl \cL_{-1}$ annihilate $|\Omega\rangle$ since $c = \tl c$. Therefore, from the dual conformal field theory viewpoint, the index $k$ in equation \eqref{chap2:eq:AdS_one_loop_partition_function} labels the contributions from each vacuum descendant. From this result, one can see the quality of being one-loop exact under a new light. Recall that there are no local degrees of freedom. Therefore, we do not expect other contributions except for the ones coming from the vacuum, which are already all taken into account.

An interesting generalization of this result consists in considering characters of the operators $\E^{2 \pi i \tau \cL_0}$ and $\E^{-2\pi i \bar \tau \tl \cL_0}$ in representations with highest weights $h$ and $\tl h$ respectively, different from the vacuum one. This would lead to
\begin{equation}
	Z_{h,\tl h} (\tau, \bar \tau) = \frac{\E^{2 \pi\I \tau (h - \tfrac{c}{24})} \E^{- 2\pi \I \bar \tau (\tl h - \tfrac{\tl c}{24})} }{\prod_{k=1}^\infty \left|1 - \E^{2\pi \I\tau \cdot p}\right|^2}.
\end{equation} 
Notice that in this case the product starts at $k=1$. This is, up to a phase, the inverse of the Dedekind $\eta$-function, which is a typical example of a modular form.  Modular forms are holomorphic functions defined on the upper-half part of the complex plane, which have extremely simple transformation properties under modular transformations.
\medskip

The computation done in \cite{Giombi:2008vd} therefore gives an explicit proof of having two copies of the Virasoro algebra as asymptotic algebra for the symmetry charges. Due to the computation being one-loop exact, it takes into account all but non-perturbative corrections to the partition function.

%--------------------------------------------------------------------------------------
\section{Flat space and the limit of vanishing cosmological constant}

From these results on thermal AdS, one can hope to extract meaningful predictions for flat Euclidean three-dimensional space.

Let us start from the on-shell action. Its value depends only on the quantity $\beta$, the length of the cylinder before identification, which is left untouched by the limit we are considering, and on the Newton constant. For this reason, one can expect the value of the on-shell action to be preserved by the limiting procedure $l \to\infty$. Therefore, we defined $S_{\text{flat}} = S_{\text{TAdS}}= \f{-\beta}{8G}$. Subtleties arise however for the limit of the one-loop contribution. In the flat limit, the torus modular parameter $\tau$ becomes effectively real, and the convergence properties of the partition function \eqref{chap2:eq:AdS_one_loop_partition_function} get spoiled. For this reason, it is convenient to keep track of a positive infinitesimal regulator $\epsilon^+$, as proposed in \cite{Barnich:2014kra,Barnich:2015uva,Oblak:2015sea}. Hence, we define
\begin{equation}
	\lim_{l \to \infty} \tau = \frac{1}{2\pi}( \gamma + \I \epsilon^+).
	\label{chap2:eq:tau_reg}
\end{equation}
Notice how this regularization keeps $\tau$ slightly within the upper-half part of the complex plane, where modular forms are defined.

In \cite{Barnich:2015mui}, it is shown, using techniques analogous to those used for the AdS case of the previous section, that the thermal partition function of three dimensional flat gravity naturally matches the limit $l \to \infty$ of the partition function defined in \eqref{chap2:eq:AdS_one_loop_partition_function} discussed above and we get
\begin{equation}
	Z_\text{flat}(\tau,\bar\tau) = \E^{ -S_\text{flat}} Z^\text{1-loop}_\text{flat}(\tau,\bar\tau)
	\qquad\text{with}\qquad
	Z^\text{1-loop}_\text{rlat}(\tau,\bar\tau) = \prod_{k=2}^\infty \f1{\left|1 - \E^{\I \gamma k + i \epsilon k} \right|^{2}}.
	\label{chap2:eq:flat_one_loop_partition_function}
\end{equation}

In \cite{Oblak:2015sea}, specific induced representations of the (centrally extended) BMS$_3$ group in three space-time dimensions, are studied and their characters are computed. The $\text{BMS}_3$ group \cite{Ashtekar:1996cd} is an infinite dimensional group, with the following semidirect product structure
\begin{equation}
	\text{BMS}_3 = \mathrm{Diff}^+(S_1) \ltimes_\mathrm{Ad} \mathrm{Vect}(S_1),
\end{equation}
where $ \mathrm{Diff}^+(S_1)$ denotes the group of orientation-preserving diffeomorphisms of the circle, $\mathrm{Vect}(S_1)\cong \mathrm{Lie}(\mathrm{Diff}^+(S_1))$ the Abelian additive group of vector fields on the circle, and $ \ltimes_\mathrm{Ad}$ the semidirect product of these two groups, with the first acting on elements of the second via the adjoint action. 

From a generic element $(f,\alpha)\in\text{BMS}_3$, one can extract a rigid super-translational part $\beta$ and a rotation $\gamma$. The BMS$_3$ characters of interest for the three-dimensional partition function turn out to depend only on these two properties of $(f,\alpha)\in\mathrm{BMS}_3$. They are related to induced representations of the {\it centrally extended} BMS$_3$ group \cite{Oblak:2015sea}. These representations are labelled by two real parameters, $(m,j)$. 
The parameter $m\geq0$ represents the space-time mass and characterizes the orbit of the representation while $j$ is related to the stabilizer \cite{Barnich:2014kra,Barnich:2015uva}.

The character for a non-vanishing value of the mass, regularized as in equation \eqref{chap2:eq:tau_reg} via $\tau = \frac{1}{2\pi}(\gamma + \I \epsilon^+)$, is then
\begin{equation}
	\chi^{m,j}( (f,\alpha) ) = \frac{ \E^{\I j \gamma} \E^{\I \beta (m- \f{1}{8 G})} }{\prod_{k=1}^\infty \left| 1 - \E^{2\pi \I \tau \cdot k}\right|^2} \qquad \text{if} \qquad m\neq0\,.
\end{equation}

If, however, the mass vanishes, we are then left with only
\begin{equation}
\chi^\text{vac}( (f,\alpha) ) = \frac{ \E^{- \I \frac{\beta}{8 G}}}{\prod_{k=2}^\infty \left| 1 - \E^{2\pi \I \tau \cdot k}\right|^2} \qquad \text{if} \qquad m=0,
\end{equation}
with the product starting at $k=2$ for reasons analogous to the Virasoro case. Following \cite{Oblak:2015sea} and \cite{Barnich:2015mui}, the interpretation attached to the $k$ label is that of higher Fourier modes in the super-momenta (Bondi mass aspect) associated to the various elements in the (coadjoint) orbit of the constant super-momentum. These are in turn closely related to the Fourier modes of the diffeomorphisms $f$.

Note that the partition function computed for Euclidean gravity with a Wick rotation corresponds to the character of the symmetry groups for a Lorentzian signature. It has been shown in \cite{Barnich:2012aw,Barnich:2012xq} that the result for BMS can be extended to the Euclidean case without loss. Of course, in that case, the symmetry group is not related to the isometries of the Euclidean case, but corresponds to the symmetry of the twisted torus in the Euclidean space.
\medskip

Recently, progresses in the identification of the field theory dual to three-dimensional gravity without cosmological constant have been made, see \cite{Barnich:2012rz,Barnich:2013yka,Carlip:2016lnw}. Interesting hints also emerged from the semi-classical discrete approach of \cite{Bonzom:2015ans}, see chapter \ref{chap4}.  We conclude this chapter noticing that, beyond convergence issues due to the infinite product, the partition function seen as a function of $\tau$ formally has poles at all real rational twist angles. Hence the necessity of keeping the regulator. This pole structure is deeply connected to the theory of modular forms, which are in turn a crucial ingredient of the AdS/CFT approach to quantum gravity. They are fundamental building blocks of two-dimensional conformal field theories \cite{Maloney:2007ud,Witten:2007kt}. Also, notice that, as in the AdS case, the computation is one-loop exact. In the last two chapters of this thesis, the partition function will be computed using the Ponzano-Regge model and linked to the results of this chapter.

	\newpage
	~
	\thispagestyle{empty}
	
	\chapter{The Ponzano-Regge Model}
\label{chap3}

First defined in 1968, the Ponzano-Regge model \cite{PR1968} is better seen today as an instantiation of $BF$ topological quantum field theory. It has been rigorously related to other approaches of three-dimensional quantum gravity, notably to the combinatorial quantization of Chern-Simons theory and to Loop Quantum Gravity.

Being formulated in terms of a local state-sum, the Ponzano-Regge model allows to compute the amplitude of quantum gravitational processes within {\it finite}, i.e. quasi-local, regions. This is to be contrasted with the AdS/CFT framework, which intrinsically refers to the asymptotic boundary of AdS.  It also implies, again differently from the AdS/CFT philosophy, that each amplitude is associated to one given spacetime topology, just as in the Chern-Simons theory. 

The proposal of Ponzano and Regge comes from their observations that the asymptotic, i.e. large spins limits of the $\{6j\}$ symbols reproduces the Regge action for General Relativity \cite{Regge:1961px,Hartle:1981}. This observation was proven explicitly by Schulten and Gordon in 1975 \cite{Schulten:1975yu}. See also \cite{Roberts:1998zka} for a proof via geometric quantization by Roberts or \cite{Freidel:2002mj} for a proof with group integral by Freidel and Louapre.

Soon after its discovery, Mizoguchi and Tada \cite{Mizoguchi:1991hk} suggested that, in an analogous asymptotic limit, the Turaev-Viro model was related to a version of the Regge action involving a cosmological term proportional to the tetrahedron's volume \cite{Bahr:2009qd}. The Turaev-Viro model is indeed the analogue of the Ponzano-Regge model based on the quantum group $\SU(2)_q$ instead of the group $\SU(2)$. This relation between the Turaev-Viro model and the Regge action was rigorously proven by Taylor and Woodward who showed that the asymptotic of the Turaev-Viro model involves homogeneously curved tetrahedra \cite{Taylor:2003}. In presence of a cosmological constant, the status of the network of correspondences with gravity is still a work in progress \cite{Freidel:1998ua,Noui:2011im,Pranzetti:2014xva,Dupuis:2013lka,Bonzom:2014wva,Bonzom:2014bua,Livine:2016vhl,Dittrich:2016typ}. That is, there is still no direct proof between the Turaev-Viro model and a discretized version of gravity, compared to the Ponzano-Regge model, which is, as we will see, the discrete version of $BF$ gravity in three dimensions.

This chapter is dedicated to a review of the Ponzano-Regge model. As a prelude, we start with a short presentation of the (quantum) Regge calculus, which is a first step between General Relativity and the Ponzano-Regge model. In Regge calculus, lengths are not discretized, but are still continuous parameters. We then proceed to the mathematical definition of the Ponzano-Regge model and the recovery of the historical formulation. This chapter ends with a brief presentation of earlier work on the Ponzano-Regge model.

\section{A prelude to the Ponzano-Regge model: the (quantum) Regge Calculus}

Regge calculus \cite{Regge:1961px} is a discrete approach to General Relativity based on a piecewise flat simplicial decomposition of the space-time manifold $\cM$. Such a decomposition for a $D$-dimensional manifold is built as the collection of $D$-dimensional flat simplexes glued together along their $(D-1)$-dimensional faces. The gluing must be done in such a way that the topological features of $\cM$ are conserved. That is, a piecewise flat simplicial decomposition is a topological representation of a manifold. Recall that a simplex is the higher-dimensional generalization of a triangle in two dimensions. For example, in three dimensions, the simplex is the tetrahedron. Mathematically, a $D$-dimensional simplex can be seen as the convex hull of $D+1$ affinely independent points, which means that no more than $m+1$ points are in the same $m$-dimensional plane. Denoting these points by $v_k$ for $k \in[1,D]$, the $D$-dimensional simplex $\sigma_D$ is defined by\footnote{The constraint that the sums over the $a_k$ is one comes from the affine part whereas the condition $a_k \geq 0$ comes from the convex part.}
\begin{equation}
\sigma_{D} = \left\lbrace \sum_{k=1}^{D} a_k v_k | \sum_{k=1}^d a_k = 1 \; \text{and} \; a_k \geq 0 \right\rbrace \; .
\end{equation}
The reason to consider simplicial decompositions is to construct an approximation of General Relativity independent of the coordinate system. Also, since a piecewise flat simplex is rigid, the knowledge of the length is enough to describe it entirely. For a curved $D$-dimensional manifold, the curvature is concentrated on simplexes of dimensions $(D-2)$, historically called hinges.

Consider now the three-dimensional case. A simplicial decomposition is then usually called a triangulation, and curvature is located at the edges of the triangulation. The idea behind Regge calculus is to fix the metric by fixing the length associated to the edges of the triangulation since it is enough to determine it entirely. On the triangulated space-time manifold, the Einstein-Hilbert action with the GHY boundary term is represented by the so-called Regge action\footnote{Generalization to arbitrary dimensions is straightforward: the $e$ {\it on the rhs} of this formula should be understood as a codimension 2 simplex, $l_e$ its volume, and $\theta_e^\sigma$ (see below) the internal hyper-dihedral angle at $e$. All these quantities must be understood as functions of the simplex {\it edge} lengths, in any dimension.}
\begin{equation}
S_\text{R}[l_e] = -\frac{1}{8 \pi G}\left[ \sum_{e \in \mathrm{int}(\cM)} l_e \epsilon_e(l_e) + \sum_{e \in \partial \cM} l_e \psi_e(l_e) \right],
\label{chap3:eq:Regge_action}
\end{equation}
where $l_e$ is the length of the edge $e$ of the triangulation, $\epsilon_e$ is the deficit dihedral angle at the edge $e$ measuring the curvature around it, while $\psi_e$ is the angle between the normal to the two boundary tetrahedra\footnote{i.e. tetrahedra having at least one face being part of the boundary triangulation.} hinging (possibly among other tetrahedra) around the boundary edge $e$. Both angles must be understood as functions of the triangulation's edge lengths and are related to the curvature of the manifold, intrinsic and extrinsic respectively.

In formulas, by introducing the internal dihedral angle at the edge $e$ within the tetrahedron $\sigma$, $\theta_e^\sigma$
\begin{subequations}
	\begin{align}
	\epsilon_e(l_{e'}) & = 2\pi - \sum_{\sigma \supset e} \theta_e^\sigma(l_{e'}) 
	\qquad\text{for}\qquad e\in\mathrm{int}(\cM),\\
	\label{chap3:eq:dihedral_angle_int}
	\psi_e(l_{e'}) &= \pi - \sum_{\sigma \supset e} \theta_e^\sigma(l_{e'})
	\;\;\qquad\text{for}\qquad e\in\partial \cM.
	\end{align}
\end{subequations}
The deficit angle $\epsilon_{e}$ is represented figure \ref{chap3:fig:dihedral_angle_regge} in the two-dimensional case, which is easier to draw. On the other hand, the deficit angle $\psi_{e}$ is the angle between the normal of two consecutive edges.
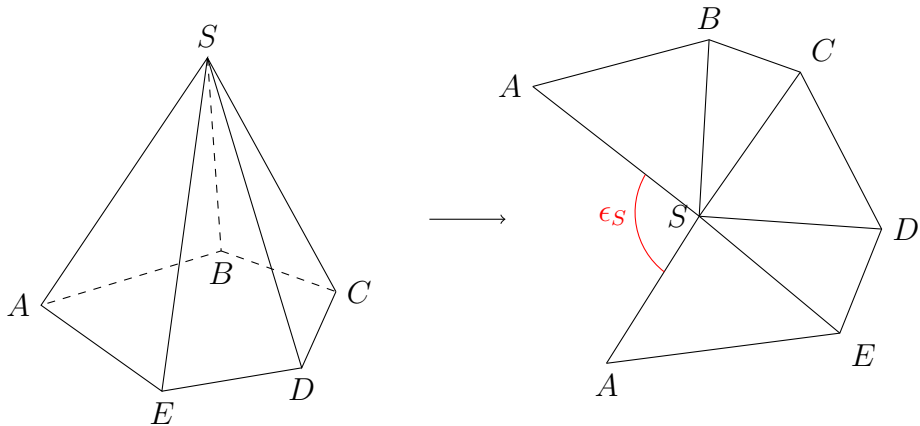
\begin{figure}[!htb]
	\centering
	\begin{tikzpicture}[scale=1]
	%1
	\coordinate (S) at (0.08,2.64);
	\coordinate (A1) at (-2.11,-0.64);
	\coordinate (A2) at (0.26,0.08);
	\coordinate (A3) at (1.77,-0.46);
	\coordinate (A4) at (1.32,-1.47);
	\coordinate (A5) at (-0.52,-1.78);
	
	\draw (S) node[above]{$S$};
	\draw (A1) node[left]{$A$}; \draw (A2) node[below]{$B$}; \draw (A3) node[right]{$C$} ; \draw (A4) node[below]{$D$}; \draw (A5) node[below]{$E$};
	
	\draw (S)--(A1)--(A5)--cycle;
	\draw[dashed] (S)--(A2)--(A1);
	\draw[dashed]  (A2)--(A3); \draw (S)--(A3);
	S		\draw (S)--(A4)--(A3);
	\draw (A4)--(A5);
	\coordinate (Arr) at (3,0.5);
	\draw[->] (Arr) --+(1,0);
	
	%2
	\coordinate (S) at (6.55,0.54);
	\coordinate (A1) at (4.36,2.26);
	\coordinate (A2) at (6.68,2.88);
	\coordinate (A3) at (7.88,2.45);
	\coordinate (A4) at (8.95,0.37);
	\coordinate (A5) at (8.4,-1.01);
	\coordinate (A1b) at (5.33,-1.41);
	
	\draw (S)--(A1); \draw(S)--(A2); \draw (S)--(A3); \draw (S)--(A4); \draw (S)--(A5); \draw (S)--(A1b);
	\draw (A1)--(A2)--(A3)--(A4)--(A5)--(A1b);
	
	\draw (S) node[left]{$S$};
	\draw (A1) node[left]{$A$}; \draw (A2) node[above]{$B$}; \draw (A3) node[above right]{$C$} ; \draw (A4) node[right]{$D$}; \draw (A5) node[below right]{$E$}; \draw (A1b) node[below]{$A$};
	\coordinate (Sangle) at (5.84,1.09);
	\draw[red] (Sangle) arc(150:232:1); \draw[red] (S) node[left=0.8]{$\epsilon_S$};
	\end{tikzpicture}
	\caption{Two-dimensional representation of the dihedral angle. In the two-dimensional case, curvature is located at the vertex. Considering the vertex $S$, it is shared by $5$ triangles. The easiest representation of the deficit angle is to project the two-dimensional structure into the flat plane, right figure. The dihedral angle $\epsilon_{S}$, in red, is then equal to $2\pi$ minus the sum of the angles of the triangle at the point $S$, as provided by the definition \eqref{chap3:eq:dihedral_angle_int}.}
	\label{chap3:fig:dihedral_angle_regge}
\end{figure}

The first term of the Regge action is the discretized Einstein-Hilbert action, whereas the second term is the Hartle-Sorkin \cite{Hartle:1981} action. It is the proper discretization of the Gibbons-Hawking-York boundary term. This action has the correct composition properties under the gluing of manifolds, and most importantly it implements boundary conditions where the induced metric, i.e. the boundary edge-lengths are kept fixed.
This can be easily checked. The Schläfli identity\footnote{The Schläfi identity is a differential relation connecting the volume of a polyhedron in  a $D$-dimensional space of constant curvature with the $(D-2)$ volumes and dihedral angles of its $(D-2)$ dimensional faces.} gives the relation
\begin{equation}
\sum_{e \in \sigma} l_e \delta \theta_{e^{^{\sigma}}} = 0 \;.
\end{equation}
We can then vary the Regge action with respect to the length $l_e$, and we find the expected result for three-dimensional gravity that the curvature defects vanish around all the edges of the triangulations
\begin{align}
0 = \delta_{l_e} S_{R} = - \f{1}{8 \pi G} \epsilon_e \; .
\end{align}
These equations are the discrete analogue of saying that three-dimensional General Relativity without cosmological constant is flat everywhere. One must be careful that the length of the bulk edges are not uniquely fixed by this equation. This is to be related to the diffeomorphism symmetry of the bulk \cite{Rocek:1982fr,Rocek:1982tj,Dittrich:2008pw,Bahr:2009ku,Bahr:2009mc}.

As a side note, let us say that Regge calculus admits a compelling generalization to the cosmological case, where $\Lambda\neq0$. In its most elegant version, one not only adds an obvious cosmological term $\Lambda\sum_{\sigma} V_\sigma$ to the action, but also makes use of homogeneously curved simplexes of constant curvature rather than flat ones \cite{Taylor:2003,Bahr:2009qd}. The main motivation for this modification is to obtain homogeneously curved solutions through a vanishing deficit-angle condition. However, the main reason why this works surprisingly well is the fundamental interplay between the Regge equations of motion and the generalization of the Schlaefli identities to curved simplices, see e.g. \cite{Bahr:2009qd, Haggard:2014gca} for classical applications of this idea, and \cite{Haggard:2015kew} for a quantum geometrical one involving Chern-Simons theory.

Quantum Regge calculus on a manifold $\cM$ can then be defined via the following finite-dimensional path integral \cite{Hamber:1985qj,Regge:2000wu}
\begin{equation}
Z_\text{R} = \int \mathcal D \mu (l) \;\E^{- S_{\text{R}}(l)} \; ,
\end{equation}
where $\mu(l)$ is the path integral measure associated to the length parameters. In three dimensions, the problem of fixing the quantum measure can be elegantly solved by requiring that the invariance of the theory under changes of the bulk triangulation holds at least at the linearised level around some background solution  $\{l_e^0\}_e$ \cite{Dittrich:2011vz}. Note that the resulting measure coincides with the measure one would deduce from the asymptotic limit of the Ponzano-Regge model that we will describe later on in this chapter \cite{Barrett:2008wh,Barrett:1993db,Roberts:1998zka}. Interestingly enough, it is not possible to find a measure with this property in four dimensions \cite{Dittrich:2014rha} and the definition of quantum Regge calculus is therefore more complicated.

Using this measure, a perturbative theory of three-dimensional quantum gravity can be defined. This theory is, by construction, diffeomorphism invariant, i.e. invariant under displacements of the triangulation's bulk vertices at one-loop \cite{Rocek:1982fr,Dittrich:2008pw, Dittrich:2012qb}. It is formally defined via the path integral 
\begin{equation}
	Z_\text{R}^{\text{1-loop}} = \E^{-S_\text{R}(l_e^0)}\int \mathcal D\mu_{l_e^0}(\lambda)\; \E^{- \frac{1}{16 \pi G} \sum_{\sigma,e,e'} H^{\sigma}_{e e'} \lambda_e \lambda_{e'}},
	\label{chap3:eq:on_loop_regge}
\end{equation}
where the one-loop Hessian is
\begin{equation}
	H^\sigma_{e e'} = \left.\frac{\partial \theta_e^\sigma}{\partial l_{e'}}\right|_{l_e = l^0_e}.
\end{equation}
and $\lambda_e\ll l^0_e$ are the small edge-length perturbations, $l_e = l_e^0 + \lambda_e$.

We say formally defined because the above formula hides two difficulties.
The first one is related to diffeomorphism invariance, which, if not gauge-fixed, induces a divergence. The gauge-fixing procedure amounts to the fixing of the position of the internal vertices of the triangulation. As we said before, diffeomorphisms are related to invariance under change of triangulation, and hence Pachner move. Note that only the  $(4-1)$ Pachner move induces a divergence. The $(3-2)$ move is exact. The same happens in the full Ponzano-Regge model as we will see in the following \cite{Barrett:1996gd}. The second problem is related to an unbounded-from-below mode which is analogous to the conformal mode of continuum gravity and can be dealt with by analytic continuation. For more details, see  \cite{Dittrich:2011vz,Bonzom:2015ans} and references therein. Once those issues have been dealt with and triangulation invariance has been established, the partition function $Z_\text{R}^{\text{1-loop}}$ can be calculated by choosing the most convenient bulk triangulation. More importantly, the result will still depend on the {\it boundary} triangulation, acting as boundary state, whose edge lengths are kept fixed in the process consistently with the chosen action principle. This means that, although this discrete theory captures all the symmetries of the continuum regime for what concerns the bulk of the spacetime, its boundary is {\it discrete and finite}.
The physical and conceptual role of dealing with finite boundaries can be physically justified in terms of the {\it general boundary} framework \cite{Oeckl:2003ws,Oeckl:2016tlj}, which focuses on the realistic operational structure of any intrinsically localized measurement. In the case of a single boundary, one can consider the resulting partition function as a generalized version of Hartle-Hawking state \cite{Hartle:1983ai}.

The discreteness of the boundary is also less severe than it looks at first sight. At the classical level, boundary discreteness can in fact be understood as the imposition of peculiar, i.e. piecewise linear, boundary conditions within the {\it continuum} theory. At the quantum level, this is reflected by the fact that spin network states can be embedded into a continuum Hilbert space. This is indeed a key achievement of loop quantum gravity \cite{Ashtekar:2004eh,Thiemann:2007zz}. There is a caveat to this statement though: a priori there are different possible embeddings leading to inequivalent Hilbert spaces. Accordingly, the quantum geometries encoded in the spin network states are completed to continuum quantum geometries in very different manners. 
Indeed, a choice of embedding into a continuum Hilbert space assigns to all degrees of freedom finer than the spin network scale a natural geometric vacuum state. In the case of three-dimensional gravity this is the $BF$ vacuum state. See \cite{Dittrich:2014wpa,Bahr:2015bra} for detailed discussions of these subtle points. Note that due to the presence of local degrees of freedom, the identification of a suitable vacuum state for four-dimensional gravity is a key open issue. See \cite{Dittrich:2012jq,Dittrich:2013xwa,Dittrich:2014ala} for a framework to address this problem.
\medskip

\section{The Ponzano-Regge model as a discretized BF-theory}

In this section, we will focus on the derivation of the Ponzano-Regge model. We will first give its mathematical definition, before looking at its link with General Relativity and BF theory. Recall that this section is not a comprehensive review of the Ponzano-Regge model. While we will try to mention most of the relevant details of the model, we refer the reader to the excellent reviews of the Ponzano-Regge model already existing in the literature, for example \cite{Freidel:2004vi,Freidel:2004nb,Freidel:2005bb,Barrett:2008wh}.

\subsection{Mathematical definition of model}

The Ponzano-Regge model is defined as a state sum model on a three-dimensional discretized manifold $\cM$. We emphasize the fact that there is no need to restrict the discretization to be a simplicial decomposition. Any cellular decomposition will do. Compared to the simplicial decomposition described above for the Regge calculus, the building blocks of a cellular decomposition are not restricted to be simplexes and can be any polytope. In the Regge calculus, the use of simplexes was necessary since the fundamental parameters of the theory were the length of the edges. Recall that only simplexes are entirely defined by the data of their edge lengths. We will see at the end of this chapter that restricting ourselves to a simplicial decomposition allows to recover the historical model as defined by Ponzano and Regge in 1968 \cite{PR1968}.Recall that this section is not a comprehensive review of the PR model. 

Before going deeper into the definition of the model, let us shortly introduce the notion of cellular decomposition in bit more details and the associated notation we will use throughout this thesis. For more details about cellular decomposition and CW-complex in general, which is a generalization of cellular decomposition, see \cite{BookMassey1991,BookHatcher:478079}. We denote a cellular decomposition of a three-dimensional manifold $\cM$ by $\cK$. Similarly to the simplicial decomposition described for the Regge calculus, a cellular decomposition is built with discrete fundamental elements called cells. A $n$-cell is a discrete structure homeomorphic to an open ball of dimension $n$. For example, a $0$-cell is a just a point while a $1$-cell is basically a segment. Cells of higher dimensions are constructed by "gluing" cells of one dimension lower together. A $D$-dimensional cellular decomposition of a manifold is built from $D+1$ type of cells, such that the gluing of a finite number of $(i-1)$-cells give a $i$-cells. The $D+1$ types of cells are of course $0$ to $D$-dimensional cells. Note that, since cells are closed, the gluing must follows some constraints in order to truly recover an $i$-cell from $(i-1)$-cells. 

In order to better understand this construction, we look at the simple two-dimensional case. In two dimensions, there are $2+1=3$ types of cells. The $0$-cells, the points, are denoted $v$ (in red figure \ref{chap4:fig:cell_plane_decomposition}), and are called the vertices of the discretization. The $1$-cells (in black figure \ref{chap4:fig:cell_plane_decomposition}) are then constructed by making a connection between two $0$-cells. They are thus edges of the discretization, denoted $e$. Finally, $2$-cells are constructed as a closed loop of edges. That is, they are polygons (in grey figure \ref{chap4:fig:cell_plane_decomposition}), denoted\footnote{The notation $t$ stands for triangle. If the cellular decomposition is simplicial, then all the $2$-cells of the decomposition are triangles.} $t$. In figure \ref{chap4:fig:cell_plane_decomposition}, we have represented a cellular decomposition of the two-dimensional flat plane. The decomposition is not simplicial since the polygons $t$ are not all triangles. It is immediate to see that this cellular decomposition can be made simplicial by adding edges such that every $2$-cells is a triangle. This property is easily generalizable to any dimensions.
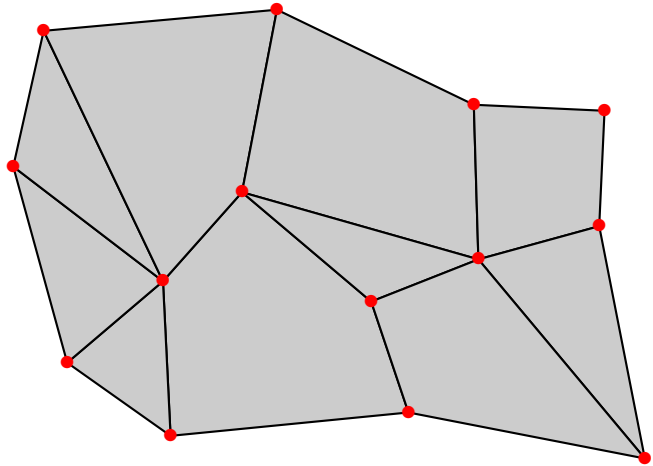
\begin{figure}[!htb]
	\begin{center}
	\begin{tikzpicture}[scale=1]
		\coordinate (A1) at (0,0); \coordinate (A1b) at (1.41,0.56);
		\coordinate (A2) at (-1.7,1.45); \coordinate (A2b) at (-1.24,3.87);
		\coordinate (A3) at (-2.74,0.27); \coordinate (A3b) at (-4,-0.81); \coordinate (A3bc) at (-4.71,1.79) ; \coordinate (A3c) at (-4.31,3.59);
		\coordinate (A4) at (-2.64,-1.78);
		\coordinate (A5) at (0.49,-1.47);
		\coordinate (B1) at (1.35,2.61);
		\coordinate (B2) at (3.07,2.53);
		\coordinate (B3) at (3,1); \coordinate (B3b) at (3.6,-2.08);

		\draw[fill=black!20,thick] (A1)--(A2)--(A3)--(A4)--(A5)--cycle;
		\draw[fill=black!20,thick] (A3)--(A4)--(A3b)--cycle;
		\draw[fill=black!20,thick] (A3)--(A3b)--(A3bc)--cycle;
		\draw[fill=black!20,thick] (A3)--(A3bc)--(A3c)--cycle;
		\draw[fill=black!20,thick] (A2)--(A3)--(A3c)--(A2b)--cycle;
		\draw[fill=black!20,thick] (A2)--(A2b)--(B1)--(A1b)--cycle;
		\draw[fill=black!20,thick] (B1)--(B2)--(B3)--(A1b)--cycle;
		\draw[fill=black!20,thick] (A1)--(A2)--(A1b)--cycle;
		\draw[fill=black!20,thick] (A1b)--(B3b)--(B3)--cycle;
		\draw[fill=black!20,thick] (B3b)--(A5)--(A1)--(A1b)--cycle;

		\draw[red] (A1) node{$\bullet$}; \draw[red] (A2) node{$\bullet$}; \draw[red] (A3) node{$\bullet$}; \draw[red] (A4) node{$\bullet$}; \draw[red] (A5) node{$\bullet$}; \draw[red] (B1) node{$\bullet$}; \draw[red] (B2) node{$\bullet$}; \draw[red] (B3) node{$\bullet$}; \draw[red] (A1b) node{$\bullet$}; \draw[red] (A2b) node{$\bullet$}; \draw[red] (A3b) node{$\bullet$}; \draw[red] (A3bc) node{$\bullet$}; \draw[red] (A3c) node{$\bullet$}; \draw[red] (B3b) node{$\bullet$}; 
	\end{tikzpicture}
	\end{center}
	\caption{Cellular decomposition of the two-dimensional plane with 14 vertices (in red), 23 edges (in black) and 10 polygons (in grey). The decomposition is not simplicial since the polygons are not all triangles.}
	\label{chap4:fig:cell_plane_decomposition}
\end{figure}

The three-dimensional case is not more complicated but it is harder to draw it clearly. In three dimensions, we have one more type of cells. It is intuitive to see that $3$-cells, denoted $\sigma$, are polyhedra. Indeed, they are constructed by the gluing of $2$-cells along their edges and, adding the constraint that the resulting object must be closed, this clearly defines a polyhedron Again, in the case of simplicial decomposition, all polyhedra must by simplexes, i.e. tetrahedra.

In the case where $\cM$ has a boundary, the cellular decomposition $\cK$ of $\cM$ naturally induces a cellular decomposition $\pp \cK$ of the boundary. Any $i$-cells of $\cK$ touching the boundary results in a $i-1$-cells for the boundary cellular decomposition $\pp \cK$. This can be easily understood looking at figure \ref{chap4:fig:cell_plane_decomposition} for the two-dimensional case. The polygons touching the boundary give rise to edges whereas the edges touching the boundary give rise to vertices.
\\

Finally, the last piece of information we need to introduce concerns the dual cellular decomposition. To any cellular decomposition $\cK$ we associate a dual cellular decomposition $\cK^{*}$. The cellular decomposition $\cK^{*}$ is such that a $i$-cell of $\cK$ is replaced by a dual $(D-i)$-cell in $\cK^{*}$. In order to differentiate direct objects from the dual ones, we introduce new notations for the dual cells. Namely, in three dimensions we call bubble $b$, face $f$, link $l$ and node $n$ the corresponding $3$, $2$, $1$ and $0$-cells of $\cK^{*}$. They are duals to vertexes, edges, polygons and polyhedra respectively. The notation are given in the following table for keepsake.
\begin{center}
	\begin{tabular}{|c|c|c|c|}
		\hline
		dimension in $\cK$ & dimension in $\cK^{*}$  & ~~~~~$\cK$~~~~~ & ~~~~$\cK^{*}$~~~~ \\ \hline
		0& 3  & $v$ & $b$ \\ \hline
		1& 2  & $e$ & $f$ \\ \hline
		2& 1  & $t$ & $l$ \\ \hline
		3& 0  & $\sigma$ & n \\
		\hline
	\end{tabular}
\end{center}

To keep the notation as simple as possible, we do not introduce special notations for boundary cells. That is, we are still calling the elements of $\pp \cK$ vertices and edges whereas the elements of $\pp \cK^{*}$ are still called nodes and links. To avoid any confusion, we might specify that an object belongs to the boundary with the use of the subscript $\pp$. Usually, if we are saying that an object belongs to $\cK$, we imply that it is a bulk object. If we want to consider both the bulk and boundary objects, we will explicitly say that it belongs to $\cK \cup \pp\cK $. The same holds for dual objects.
\bigskip

In order to define the Ponzano-Regge model, we need to dress the cellular decomposition with some more ingredients. First, we assign an orientation to the links and faces of $\cK^{*} \cup \pp \cK^{*}$. Note that, by duality, this corresponds to a choice of an orientation for $\cK \cup \pp \cK$. To do so, we assign to each link an arrow and to each face a clockwise or anticlockwise orientation. Note that a link usually belongs to two faces. However, it is assigned one single orientation. To keep track of the orientation of the link with respect to the face it belongs to, we define the relative orientation between a face $f$ and a link $l$ belonging to $f$ by
\begin{align*}
\eps_{l,f} &= +1 \quad \text{if $l$ and $f$ have the same orientation}  \\
\eps_{l,f} &= -1 \quad \text{if $l$ and $f$ have opposite orientation} \; .
\end{align*}
This definition can be extended to every face and link of the dual cellular decomposition saying that $\eps_{l,f} = 0$ if $l$ does not belong to the face $f$. Providing the orientation of the links, we denote the source node of the link $l$ by $s_l$ and its target node by $t_l$.

The remaining element necessary for the definition of the Ponzano-Regge model is the definition of a discrete connection to the manifold $\cM$. To every link of $\cK^{*} \cup \pp \cK^{*}$ we hence assign a $\SU(2)$ element $g_l$. This element represents the holonomy, i.e. the parallel transport between the source node of $l$ and its target node. The curvature of the connection is then naturally associated to the faces of the dual cellular decomposition. To properly define it, we need to make a choice of starting node for every face. We denote the starting node of the face $f$ by $st(f)$. The discrete curvature associated to the face $f$ is then given by the $\SU(2)$ element $G_f$ defined by
\begin{equation*}
G_f = \prod_{l \in f} g_l^{\eps(l,f)} \;.
\end{equation*}
The product is to be understood starting with the element $g_{l_0}$ where the source node of $l_0$ is $st(f)$. Actually, the choice of starting node does not matter for the Ponzano-Regge model. Indeed, we will see in the following that the model is, by definition, restricting the connections to the subspace of flat connection, that is when $G_f$ is fixed to be the identity by constraint.

Finally, as an $\SU(2)$ discrete connection, it is naturally associated to a gauge symmetry acting at every node of the cellular decomposition. Explicitly, we have 
\begin{equation*}
	g_{l} \rightarrow k_{t_l} g_l k_{s_l}^{-1} \; .
\end{equation*}
That is, the full gauge group is one copy of $\SU(2)$ per node of the decomposition. 
\bigskip

We now have all the necessary ingredients to define the Ponzano-Regge model on a manifold with boundary

\begin{Definition}
	Consider a three-dimensional manifold $\cM$ with boundary $\pp\cM$ and an orientated cellular decomposition $\cK$ and its dual $\cK^{*}$ of the $\cM$. This induces a boundary cellular decomposition $\pp \cK$ and a boundary dual cellular decomposition $\pp \cK^*$ on the boundary manifold $\pp \cM$. Let us assign to the links of $\cK^{*}$ and $\pp \cK^{*}$ a $SU(2)$ element $g_l$ and choose an arbitrary starting point $st(f)$ for all the faces.
	
	The Ponzano-Regge amplitude is formally defined by
	\begin{equation*}
	Z_{PR}\left(\left\lbrace g_{l \in \pp \cK^{*}} \right\rbrace\right) = \int_{SU(2)} \prod_{l \in \cK^{*}} \dd g_{l} \prod_{f \in \cK^{*}} \delta\left( G_{f} \right)
	\end{equation*}
	where $\dd g$ is the Haar measure on $\SU(2)$ and
	\begin{equation*}
	G_f = \prod_{l \in f} g_l^{\eps(l,f)} \; ,
	\end{equation*}
	where the product starts with $g_{l_0}$ where the source node of $l_0$ is $st(f)$.
	The Ponzano-Regge model is computing the volume associated to the moduli space of flat connection for the discretized manifold $\cM$.
	\label{chap3:def:PR_model_without_gauge_fixing}
\end{Definition}

In the definition of the Ponzano-Regge model, we use the $\SU(2)$ delta function. That is, the distribution given by
\begin{equation}
	\int_{\SU(2)}  \delta(g) f(g) \; \dd g = f(\id)
\end{equation}
for any function $f$ of $\SU(2)$. A useful formula for the $\SU(2)$ delta function is its Plancherel, i.e. spectral decomposition in terms of irreducible representation of $\SU(2)$. Explicitly we have
\begin{equation}
	\delta(g) = \sum_{j \in \f{\N}{2}} d_{j} \chi^{j} (g) \;.
	\label{chap3:eq:spectral_decompo}
\end{equation}
In this formula, $d_j = 2j+1$ is the dimension of the spin $j$ representation $V_j$, and $\chi^j(g) = \tr \left( D^{j}(g) \right)$ with $D^j$ the Wigner matrix, its character. This formula will come in handy when deriving the historical Ponzano-Regge model from the previous formula.

This amplitude is only formally defined since it might diverge. Note that the divergences do not come from the gauge symmetry on the connection since $\SU(2)$ is compact. The divergences are due to possible redundancies in the delta-functions distributions appearing in the definition of the amplitude. More precisely, for handlebodies, these divergences arise in the presence of bubbles in the bulk cellular decomposition \cite{Perez:2000fs}. Recall that bubbles are polyhedra in $\cK^{*}$. These polyhedra are naturally bounded by a set of faces $f \in b$. And to all of these faces the model assigns a delta function to ensure that the connection is flat everywhere. It happens that exactly one delta function is redundant per bubble, in the sense that its information is already encoded in the other delta functions and the fact that bubbles are cells. This can easily be seen by considering the simple case of a tetrahedron. The contribution of the tetrahedron to the Ponzano-Regge amplitude reads (see figure \ref{chap4:fig:tetra_not_divergence} for the notation and orientation)
\begin{equation}
	Z_{PR}^{tet} = \int_{\SU(2)} \left(\prod_{k=1}^{6} \dd g_k \right) \;  \delta(g_1 g_2 g_3) \delta(g_3^{-1} g_4 g_5) \delta(g_5 g_1 g_6^{-1}) \delta(g_2 g_4 g_6) \;.
\end{equation}
Since a tetrahedron has four faces, there are four delta functions. The integration over the $g$'s can be done successively. Starting with $g_1$, then $g_2$ ... we get
\begin{align*}
	Z_{PR}^{tet} 
	&= \int_{\SU(2)} \left(\prod_{k=2}^{6} \dd g_k \right) \;   \delta(g_3^{-1} g_4 g_5) \delta(g_5 g_3^{-1} g_2^{-1} g_6^{-1}) \delta(g_2 g_4 g_6) 
	\\
	&=
	\int_{\SU(2)} \left(\prod_{k=3}^{6} \dd g_k \right) \;   \delta(g_3^{-1} g_4 g_5) \delta(g_5 g_3^{-1} g_4 g_6 g_6^{-1})
	\\
	&=
	\int_{\SU(2)} \left(\prod_{k=4}^{6} \dd g_k \right) \; \delta(g_5 g_5^{-1} g_4^{-1} g_4 g_6 g_6^{-1}) \\
	&=
	\delta(\id) \; ,
\end{align*}
thus giving a divergence. Clearly, removing by hand one delta function is enough to obtain the convergence without any loss of information regarding the flatness of the connection. We will see in the following that this idea can be generalized and allows to obtain a regularization of the amplitude. Note, however, that this is not true for every manifold \cite{Barrett:2008wh}. Deep below, the convergence of the model is related to the Reidemeister torsion on the space of flat connection \cite{Barrett:2008wh,Dubois:torsion}.

\begin{figure}[htb!]
	\begin{center}
		\begin{tikzpicture}[scale = 1]
			\coordinate (A) at (-1.85,-1.31);
			\coordinate (B) at (3.75,-0.07);
			\coordinate (C) at (2.27,-1.89);
			\coordinate (D) at (0.31,2.15);
			
			\draw[->-=0.5,dashed] (B) --node[below]{$g_6$} (A) ;
			\draw[->-=0.5] (A) --node[above]{$g_2$} (C) ;
			\draw[->-=0.5] (C) --node[above]{$g_4$} (B) ;
			\draw[->-=0.5] (D) --node[above left]{$g_1$} (A);
			\draw[->-=0.5] (C) --node[above right]{$g_3$} (D);
			\draw[->-=0.5] (B) --node[above]{$g_5$} (D);
		\end{tikzpicture}
	\end{center}
	\caption{Tetrahedron with an $\SU(2)$ element associated to each of its links. The links are orientated as described in the picture, and each face is considered clockwise orientated.}
	\label{chap4:fig:tetra_not_divergence}
\end{figure}
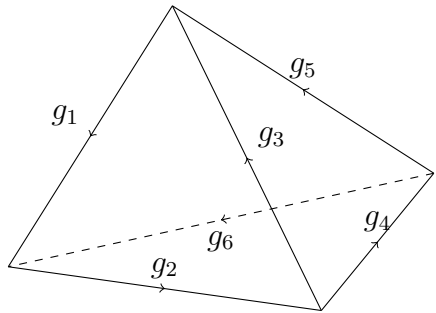

One way to quantify the divergences is to look at the collapsing move and the {\it tardis}\footnote{Doctor Who's time travelling machine!} of the cellular decomposition. The tardis is the set of edges of $\cK$ causing the divergences of the partition. Recall that the faces of the dual decomposition are dual to edges. Hence, we can see the delta functions as living on the edge of $\cK$. Since the only possible divergences come from the delta function, it is natural to define the set of edges whose associated delta functions are problematic. Another way to understand the divergences is to consider the spectral decomposition of the delta function as an infinite sum over the irreducible representation of $\SU(2)$ given previously. The tardis is then the set of edges where the previous sum is not constrained to be finite by the boundary data and spin relations. In the following, we shortly explain a way to obtain the tardis of a cellular decomposition. This gives us an interesting insight about the gauge fixing procedure of the next section. Finding the tardis of the cellular decomposition is done by looking at the collapsing move on $\cK^*$ \cite{Barrett:2008wh}. Consider a particular $3$-cell $b$ and $2$-cell $f$ such that $f$ appears only as the boundary of $b$. The collapsing operation removes both $b$ and $f$ from the dual cellular decomposition $\cK^*$. In the case where $\cM$ is connected with boundary, this operation can be repeated until there are no 3-cells left. That is, there are no more bubbles $b$, coinciding with the previous statement of divergences arising from bubbles. Basically, these collapsing operations happen either on a face $f$ already on the initial boundary $\pp \cM$ or on a face coming from a previous collapsing operation, which are dual to edges. The tardis is then the set of edges of $\cK$ dual to the face $f$ of the collapsing move. In the following, we will perform an operation where we remove a delta function for some edges of $\cK$, th edges of the tardis.

\subsubsection{Gauge-fixed Ponzano-Regge amplitude}

In order to have a properly defined model, it is therefore necessary to find a way to fix the divergences. The example with the tetrahedron hints that removing delta functions associated to faces of the dual cellular decomposition will do, whereas the tardis tells us that the problem arises from the edges of the cellular decomposition, which are indeed dual to the faces. In this section, we provide a way to remove the redundant delta functions from the definition of the model. At the same time, we will take care of the gauge symmetry of the $\SU(2)$ connection by providing a gauge fixing. We mainly follow the logic of \cite{Freidel:2004nb}.
\\

Before focusing on the deltas, let us consider the gauge fixing of the connection.
The gauge group $\SU(2)$ acts on the connection following
\begin{equation*}
	g_{l} \rightarrow k_{t_l} g_l k_{s_l}^{-1} \; ,
\end{equation*}
such that each node of the cellular decomposition contributes with a $\SU(2)$ group to the full gauge symmetry group. Denoting by $\#_{n}$ the number of nodes of $\cK^{*} \cup \pp \cK^{*}$, the full gauge group is therefore $\#_{n}$ copies of $\SU(2)$. To gauge fix this symmetry, we use a standard technique of lattice gauge theory. We consider a maximal tree $T^*$ of $\cK^{*} \cup \pp \cK^{*}$. A tree of a graph is a sub-graph which contains no loop. A tree is called maximal if it goes through every node of the graph. Providing such a tree, we can perform a partial gauge fixing of the $\SU(2)^{\#_n}$ gauge symmetry. This will allow us to fix $\SU(2)$ elements of the links belonging to the tree to whatever value we want. For simplicity, we will gauge-fixed them at the identity. As an example, consider again the two-dimensional plane and the discretization represented figure \ref{chap4:fig:tree_gauge_fixing}. Note that this does not change the logic of the action of the tree for the gauge fixing compared to the three-dimensional case. We represent the tree $T^{*}$ in blue figure \ref{chap4:fig:tree_gauge_fixing}. It is indeed maximal since it goes through every node. For simplicity, all the links belonging to the tree are considered with the same orientation, defining in turn the orientation of the whole tree.
\begin{figure}[!htb]
	\begin{center}
		\begin{tikzpicture}[scale=1]
		\coordinate (A1) at (0,0); \coordinate (A1b) at (1.41,0.56);
		\coordinate (A2) at (-1.7,1.45); \coordinate (A2b) at (-1.24,3.87);
		\coordinate (A3) at (-2.74,0.27); \coordinate (A3b) at (-4,-0.81); \coordinate (A3bc) at (-4.71,1.79) ; \coordinate (A3c) at (-4.31,3.59);
		\coordinate (A4) at (-2.64,-1.78);
		\coordinate (A5) at (0.49,-1.47);
		\coordinate (B1) at (1.35,2.61);
		\coordinate (B2) at (3.07,2.53);
		\coordinate (B3) at (3,1); \coordinate (B3b) at (3.6,-2.08);

		\draw (A1)--(A2)--(A3)--(A4)--(A5)--cycle;
		\draw (A3)--(A4)--(A3b)--cycle;
		\draw (A3)--(A3b)--(A3bc)--cycle;
		\draw (A3)--(A3bc)--(A3c)--cycle;
		\draw (A2)--(A3)--(A3c)--(A2b)--cycle;
		\draw (A2)--(A2b)--(B1)--(A1b)--cycle;
		\draw (B1)--(B2)--(B3)--(A1b)--cycle;
		\draw (A1)--(A2)--(A1b)--cycle;
		\draw (A1b)--(B3b)--(B3)--cycle;
		\draw (B3b)--(A5)--(A1)--(A1b)--cycle;
		
		\draw[thick,blue,->-=0.5] (A1) --node[pos=0.5,below]{$g_0$} (A2) ;
		\draw[thick,blue,->-=0.5] (A2) --node[pos=0.5,below]{$g_1$} (A3) ;
		\draw[thick,blue,->-=0.5] (A3) -- (A4) ;
		\draw[thick,blue,->-=0.5] (A4) -- (A3b) ;
		\draw[thick,blue,->-=0.5] (A3b) -- (A3bc) ;
		\draw[thick,blue,->-=0.5] (A3bc) -- (A3c) ;
		\draw[thick,blue,->-=0.5] (A3c) -- (A2b) ;
		\draw[thick,blue,->-=0.5] (A2b) -- (B1) ;
		\draw[thick,blue,->-=0.5] (B1) -- (B2) ;
		\draw[thick,blue,->-=0.5] (B2) -- (B3) ;
		\draw[thick,blue,->-=0.5] (B3) -- (A1b) ;
		\draw[thick,blue,->-=0.5] (A1b) -- (B3b) ;
		\draw[thick,blue,->-=0.5] (B3b) -- (A5) ;
		
		\draw[red] (A1) node{$\bullet$}; \draw[red] (A2) node{$\bullet$}; \draw[red] (A3) node{$\bullet$}; \draw[red] (A4) node{$\bullet$}; \draw[red] (A5) node{$\bullet$}; \draw[red] (B1) node{$\bullet$}; \draw[red] (B2) node{$\bullet$}; \draw[red] (B3) node{$\bullet$}; \draw[red] (A1b) node{$\bullet$}; \draw[red] (A2b) node{$\bullet$}; \draw[red] (A3b) node{$\bullet$}; \draw[red] (A3bc) node{$\bullet$}; \draw[red] (A3c) node{$\bullet$}; \draw[red] (B3b) node{$\bullet$}; 
		
		\draw[red] (A1) node[above]{$n_0$}; \draw[red] (A2) node[above left]{$n_1$}; \draw[red] (A3) node[left]{$n_3$};
		\end{tikzpicture}
	\end{center}
	\caption{Drawing of the tree $T^{*}$ for a two-dimensional planar graph. The tree is initialized at the node $n_0$ and links of the tree are represented in blue. For simplicity, we consider that the links belonging to $T^{*}$ all have the same orientation. }
	\label{chap4:fig:tree_gauge_fixing}
\end{figure}
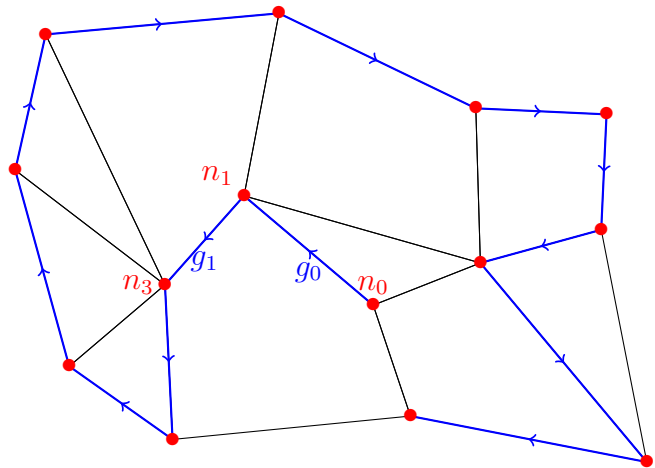

The tree is considered to be initialized at the node $n_0$. Note that it could have been the other extremity of the tree. Following the tree, we call the next node $n_1$ and so on. To perform the gauge fixing, we follow once again the links belonging to the tree. Arriving at the node $n_1$, we use the local $\SU(2)$ invariance to fix $g_0$ to be the identity. That is, we act on $n_1$, the target of the link $0$ with the element $g_0^{-1}$. By doing so, we are also modifying $g_1$. It becomes $g_{1} g_{0}$. Note that {\it all} the links touching $n_1$ are affected by this gauge transformation. Then, we proceed to the node $n_2$ and repeat the same kind of operation. At the end of the day, it is clear that all the group elements of the tree are fixed to the identity whereas we are left with only one node where the gauge invariance was not applied, $n_0$. That is, the gauge group goes from $\SU(2)^{\#_{n}}$ to only one remaining copy of $\SU(2)$. Hence the only partial gauge fixing. Note once again that we could have fixed the group elements belonging to the tree to whatever valued we wanted. It is not even necessary for the elements of the tree to be fixed at the same group element. In the spirit of the gauge fixing procedure, it does however make more sense to gauge-fix them at the identity.

In summary, the consequence of the gauge fixing is to fix all the group elements of the links belonging to the maximal tree $T^{*}$ of $\cK^* \cup \pp \cK^{*}$ to the identity
\begin{equation}
	g_{l \in T^{*}} = \id \; .
\end{equation}
Again, we emphasize that since $\SU(2)$ is compact, this gauge symmetry was not the cause of the divergences.
\bigskip

It is now time to focus on the delta functions causing the ill-definiteness of the Ponzano-Regge amplitude. Recall that there is one delta function per face of $\cK^{*}$ in the initial model and that faces are dual to edges of the direct decomposition. The example with the tetrahedron tells us that removing one delta is enough for convergence, at least in that case. This can be generalized to any $3$-cells using the property that any $D$-cells can be subdivided into a certain number of $D$-simplexes. For $D=3$, simplexes are tetrahedra. Hence, the idea is to consider all the bubbles of $\cK^{*}$, and remove one delta function associated to one of its faces. The easiest way to do this in a consistent way is to go back to $\cK$, where the bubbles are the vertices. Now, pick an internal maximal tree $T$ touching the boundary only one time. That is, $T$ goes through all the vertices of $\cK$ by only one vertex of $\pp \cK$. This coincides with the logic behind the collapsing move saying that we consider a $3$-cell and a $2$-cell such that the $2$-cell only belongs to the $3$-cell and no other. Basically, this is a boundary $2$-cell.

The final operation is to remove the delta function associated to all the dual faces of the edges belonging to $T$. The reason that $T$ must be maximal is quite obvious considering the example of the tetrahedron. The real question is: do we have a sufficient number of delta functions remaining to impose the flatness of the connection everywhere? This is the key point. Indeed, remember that the Ponzano-Regge model is basically evaluating the volume of flat connection. It happens to be enough thanks to the following lemma
\begin{lemma}
	\begin{equation*}
	G_f = \id \quad \forall f \; \text{dual to} \; e \in \cK \setminus T \quad \implies G_f = \id \; \forall f \in \cK^{*} \;.
	\end{equation*}
	\label{chap4:lemma:treeT_sufficiant_condition}
\end{lemma}
This lemma tells us that imposing the flatness of all the faces of $\cK^{*}$ dual to edges of $\cK$ {\it not} belonging to the tree $T$ is enough to impose the flatness of all the faces of $\cK^{*}$. See \cite{Freidel:2004nb,Barrett:2008wh} for the proof. The basic idea was already explained above however. The choice of tree $T$ basically removes one face per polyhedra and the flatness of the other faces imply the flatness on the face where the delta function was removed. Hence, we are not losing any information about having a flat theory.
\bigskip

With the addition of the trees $T$ and $T^*$, we can define a finite gauge-fixed version of the Ponzano-Regge model
\begin{Definition}
	Consider a three-dimensional manifold $\cM$ with boundary $\pp\cM$. Consider a cellular decomposition $\cK$ and its dual $\cK^{*}$ of the manifold $\cM$. This induces a boundary cellular decomposition $\pp \cK$ and boundary dual cellular decomposition $\pp \cK$. Assign to the edges of $\cK^{*}$ and $\pp \cK^{*}$ an $SU(2)$ element $g_l$. Pick an internal maximal tree $T$ of edges in $\cK$ touching the boundary only once and a maximal tree $T^{*}$ of edges in $\cK^{*} \cup \pp \cK^{*}$. 
	The gauge-fixed Ponzano-Regge amplitude is defined by
	\begin{equation}
	Z_{PR}\left(\left\lbrace g_{l \in \pp \cK^{*}} \right\rbrace\right) = \int_{SU(2)} \prod_{l \in \cK^{*} } \dd g_{l} \prod_{l \in T^{*}}\left( \delta(g_l) \right) \prod_{f \in \cK^{*} \setminus T_f} \delta\left( G_{f} \right)
	\label{chap3:eq:PR_model}
	\end{equation}
	where $\dd g$ is the Haar measure on $\SU(2)$ and
	\begin{equation}
	G_f = \prod_{l \in f} g_l^{\eps(l,f)} \; .
	\end{equation}
	There is still a global $\SU(2)$ invariance
	\begin{equation}
	g_l \rightarrow G g_l G^{-1} \quad \forall G \; \in \SU(2) \; ,
	\end{equation}
	due to the existence of the remaining nodes where the gauge symmetry was not applied.
	\label{chap3:def:PR_amplitude}
\end{Definition}

There are two points that need to be made explicit. First, the gauge-fixing procedure also affects the boundary. On the right hand side of the Ponzano-Regge amplitude, the delta functions over the links of $T^{*}$ involve links of the boundary. Hence, this defines an amplitude which is a functional of the gauge-fixed boundary data. Secondly, the way we handle the divergence of the Ponzano-Regge model is not complete, and some divergences might appear even without the presence of bubbles \cite{Barrett:2008wh,Bonzom:2010ar,Bonzom:2012mb}. We will discuss a bit more this point later while dealing with topological invariant and torsions.

Note that it is interesting that the divergences coming from the bubbles are handled in a gauge-fixing like way for a lattice theory, that is with the help of a tree. This hints that behind these delta functions and the divergences might exist a gauge invariance too. In the next section, we will present another viewpoint on the Ponzano-Regge model, obtained from the discretization of a $\SU(2)$ $BF$ theory. In that context, the cause of the divergences will be clearer, and related to another gauge symmetry of the $BF$ action. 

\subsubsection{Completely local definition of the Ponzano-Regge amplitude}

Before focusing on the topological invariance of the Ponzano-Regge model, we come back to the initial definition \ref{chap3:def:PR_model_without_gauge_fixing}. Intrinsically, the Ponzano-Regge model is a completely local model. The model computes the volume of flat connection given the discretization of a manifold. It is true that the flatness is defined while looking at closed loop via the curvature, and hence implies that it is rather natural to express the model on object living on the face. However, while doing so, we seem to lose a part of the locality behind the model. As it turns out, this can easily be recovered by expressing the delta functions in terms of local variables instead of variables on a face. We provide in this sub-section a new way to look at the Ponzano-Regge model, making its locality explicit. 

To do so, focus first on a given face $f$, with $\#$ nodes and links. Recall that this face is dual to an edge of $\cK$. In turn, this edge naturally belongs to $\#$ 3-cells, duals to the nodes. Now recall that $g_l$ corresponds to the parallel transports between the sources and target nodes, that is, between two polyhedra. The idea is to associate to each polyhedra a $\SU(2)$ group elements representing its state. Now, the key points are that nodes belong to more than one face. In order to be as local as possible, it seems natural to introduce one state element for the node $n$ per face it belongs to. Explicitly, we associate a set of $\SU(2)$ elements $\{h_{n,f} \}_{f \ni n}$ to the node $n$. See figure \ref{chap3:fig:PR_ultra_local}.

Since $g_l$ is just the parallel transport from $h_{s(l),f}$ to $h_{t(l),f}$ in $f$, it is natural to consider the contribution
\begin{equation}
	\delta\left(h_{t(l),f}^{-1} g_l h_{s(l),f}\right)
\end{equation}
per link $l$ and per face $f$. With these news variables, the Ponzano-Regge amplitude reads
\begin{equation}
	Z_{PR}\left(\left\lbrace g_{l \in \pp \cK^{*}} \right\rbrace\right) = \int_{SU(2)} \prod_{l \in \cK^{*}} \dd g_{l} \left( \prod_{f \in \cK^{*}} \prod_{n \in f}  \dd h_{n,f} \right)\; \delta(h_{t_l,f}^{-1} g_l h_{s_l,f}) \;.
\end{equation}
That is, we explode the face contribution $G_{f}$ into its local counterpart per link.
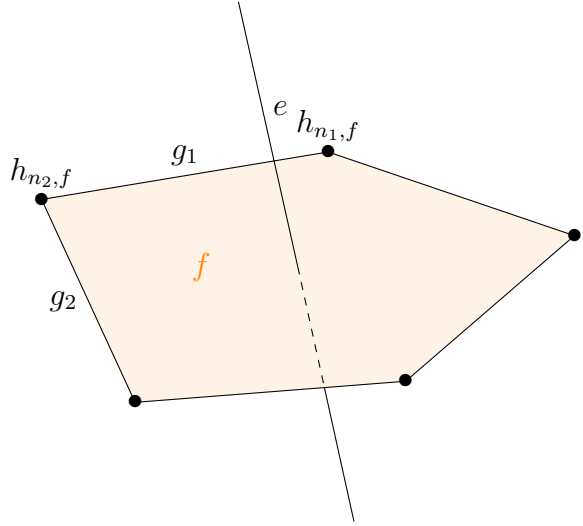
\begin{figure}[!htb]
	\begin{center}
		\begin{tikzpicture}[scale=1]
		\coordinate (A) at (0.76,1.11);
		\coordinate (B) at (-3.01,0.48);
		\coordinate (C) at (-1.78,-2.21);
		\coordinate (D) at (1.78,-1.93);
		\coordinate (E) at (4,0);
		\coordinate (U1) at (-0.42,3.1); \coordinate (U2) at (0.36,-0.41); \coordinate (U3) at (0.71,-2.01); \coordinate (U4) at (1.1,-3.79);
		
		\draw[fill=orange!10] (A)--node[above]{$g_1$}(B)--node[left]{$g_2$}(C)--(D)--(E)--cycle;
		\draw (U1)--node[right,pos=0.4]{$e$}(U2); \draw[dashed] (U2)--(U3); \draw (U3)--(U4);
		
		\draw (A) node{$\bullet$}; \draw (B) node{$\bullet$}; \draw (C) node{$\bullet$}; \draw (D) node{$\bullet$}; \draw (E) node{$\bullet$};
		\draw (A) node[above]{$h_{n_1,f}$}; \draw (B) node[above]{$h_{n_2,f}$};
		
		\draw[orange] (U2) node[left=1cm]{$f$};
		\end{tikzpicture}
	\end{center}
	\caption{Face $f$ of $\cK^{*}$, in orange, and its associated dual edge $e$. The face $f$ has five nodes and links. To each of the nodes, we associate a $\SU(2)$ element $h_{n,f}$ representing the state of the 3-cell dual to the node $n$, and to each link we associate the discrete connection $g_l$ representing the parallel transport from $h_{s_l,f}$ to $h_{t_l,f}$. }
	\label{chap3:fig:PR_ultra_local}
\end{figure}

Note the key difference between $g_l$ and $h_{n,f}$: while there is only one group element per link of $\cK^{*}$, there are as many group elements associated to the nodes that the number of faces it belongs to. In a sense, one can view the group elements associated to the nodes to be truly a representation of the state associated to the edge $e$. And the edge $e$ carries as many representations as the number of 3-cells sharing it. In the next section, while looking at the relation between the Ponzano-Regge model and the $BF$ action, we will see that $h_{n,f}$ is related to the discretization of the $B$ field, whose natural discretization is located at the edge $e$ of the cellular decomposition. 

Finally, we point out that this formulation of the Ponzano-Regge model is suitable to see its link with the $GFT$ formulation of the spin-foam model \cite{Boulatov:1992vp,Ooguri:1992eb,Freidel:2005qe}. Indeed, integrating out the $g_l$ instead of the $h_{n,f}$ allows to recover the $GFT$ formulation of the Ponzano-Regge model. With the above formulation, the delta function represents the propagator associated to the $GFT$.

\subsubsection{Topological invariance of the Ponzano-Regge amplitude}

In this section, we consider a simplicial decomposition and show that the definition of the Ponzano-Regge model provided by \ref{chap3:def:PR_amplitude} does not depend on the choice of decomposition. We will only consider the case of simplicial decomposition for simplicity, and refer the reader to \cite{Girelli:2001wr} for the general case. Since the computations to address these questions are a bit cumbersome, we will not develop them here, and just give a heuristic approach. See \cite{Girelli:2001wr,Freidel:2004nb,Barrett:2008wh} for more details. Note that the questions about the invariance under the choice of trees, orientation, starting position for the faces ... also need to be addressed. We will not develop these specific points here, and refer the reader to \cite{Freidel:2004nb}. We point out however that, for the choice of tree $T$, the lemma \ref{chap4:lemma:treeT_sufficiant_condition} is (almost) enough to prove the independence on the choice of $T$. Indeed, the lemma tells us that, once the tree $T$ is chosen, we are always projected back to the subspace of flat connection. Changing the tree does not change the result of the lemma, which is the important part. For the tree $T^{*}$ it is immediate to see that changing the tree changes which elements are fixed to identity by gauge-fixing, but not the number of them. Intuitively, we do not change the dimension of the space of freedom for the connection elements.
\medskip

Consider now a simplicial decomposition of the manifold and its dual decomposition. The main reason behind the choice of a simplicial decomposition rather than staying more general is to be found in the fact that we know how two different simplicial decompositions are related. It always exists a finite number of moves, called Pachner move, to go from one simplicial decomposition of a manifold to another simplicial decomposition of the same manifold. The fact that the amplitude defined in \ref{chap3:def:PR_amplitude} is independent from the choice of decomposition can therefore be done by proving its independence under Pachner move. In three dimensions, there are four moves, namely the $(1-4)$ move, the $(2-3)$ move and their inverses. The number in the name of the move refers to the number of tetrahedra. That is, the $(1-4)$ move is creating four tetrahedra from only one whereas the $(2-3)$ move is transforming two tetrahedra into three. A $(1-4)$ move is represented figure \ref{chap3:fig:1_4_move}.
\begin{figure}[!htb]
	\begin{center}
	\begin{tikzpicture}[scale=2]
	\coordinate (A) at (-1.36,-0.76);
	\coordinate (B) at (0.71,-0.69);
	\coordinate (C) at (1.03,0.3);
	\coordinate (D) at (-0.51,1.04);
	
	% 1 tetra
	\draw (A)--(B) ; \draw (B)--(C) ; \draw[dashed] (A)--(C); \draw (A)--(D) ; \draw (D)--(B) ; \draw (D)--(C);
	
	%arrow
	\coordinate (S) at (1.5,0.2) ; 
	\draw[->-=1,thick] (S) --+(1,0);
	
	% 4 tetra
	\coordinate (A1) at (3,-0.76);
	\coordinate (B1) at (5.07,-0.69);
	\coordinate (C1) at (5.39,0.3);
	\coordinate (D1) at (3.85,1.04);
	\coordinate (N) at (4.22,0.12);
	
	\draw[red] (N) node{$\bullet$} ;
	
	\draw (A1)--(B1); \draw (B1)--(C1); \draw[dashed] (A1)--(C1); \draw (A1)--(D1); \draw (D1)--(B1); \draw (D1)--(C1);
	\draw[red,dotted] (N)--(A1); \draw[red,dotted] (N)--(B1);  \draw[red,dotted] (N)--(C1);  \draw[red,dotted] (N)--(D1); 
	
	\end{tikzpicture}
	\end{center}
	\caption{$(1-4)$ Pachner move. On the left a single tetrahedron. On the right, we add a single node (in red) inside the tetrahedron, hence producing four tetrahedra from the single one. The reverse move corresponding to removing the single node.}
	\label{chap3:fig:1_4_move}
\end{figure}
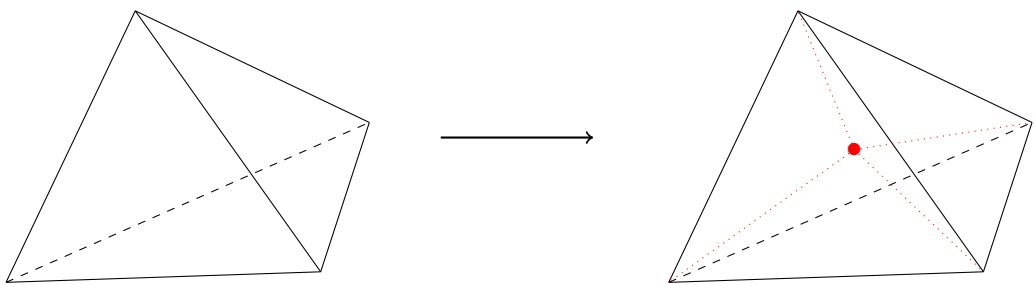

It is clear from the figure that such an operation adds edges, faces and tetrahedra to the decomposition. For the amplitude to still be well-defined according to the definition \ref{chap3:def:PR_amplitude}, we need to modify the tree $T$ and $T^{*}$ accordingly. By construction, the $(1-4)$ move adds one vertex, therefore one bubble. However, we saw previously that the bubbles were taken care of with $T$. Basically, the choice of $T$ amounts to removing all the bubbles from the decomposition. This can easily be understood considering the gauge-fixing as removing the edges of the tardis. We will see later, while working on recovering the historical Ponzano-Regge model, an explicit proof of this fact.

On the other hand, the $(2-3)$ move splits two tetrahedra into three. This can be done by connecting with an edge the opposite vertices of the two tetrahedra with respect to the shared basis, see figure \ref{chap3:fig:2_3_move}
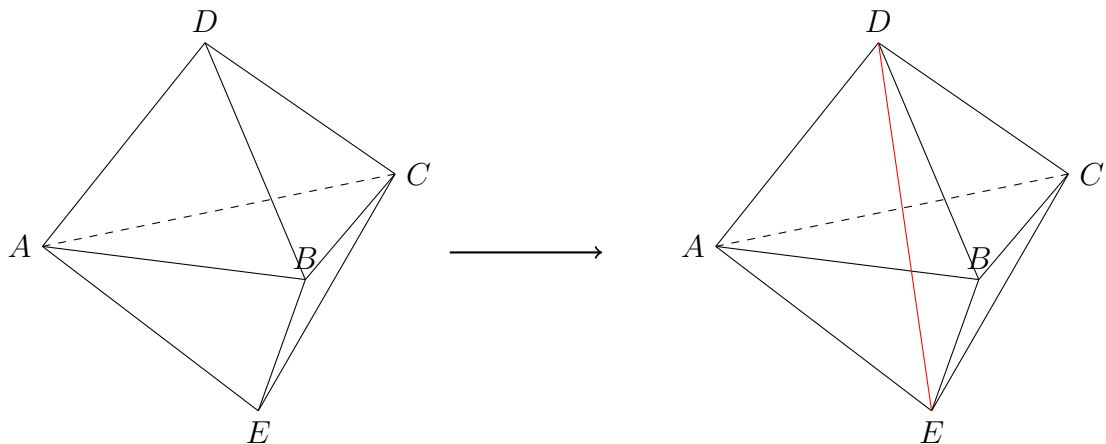
\begin{figure}[!htb]
	\begin{center}
	\begin{tikzpicture}[scale=2]
			% 2 tetra
		\coordinate (A) at (-1.43,0.24);
		\coordinate (B) at (0.3,0.02);
		\coordinate (C) at (0.89,0.72);
		\coordinate (D) at (-0.36,1.59);
		\coordinate (E) at (-0.01,-0.85);
		
		\draw (A) node[left]{$A$}; \draw (B) node[above]{$B$}; \draw (C) node[right]{$C$}; \draw (D) node[above]{$D$}; \draw (E) node[below]{$E$};
		
		\draw (A)--(B) ; \draw (B)--(C) ; \draw[dashed] (A)--(C); \draw (A)--(D) ; \draw (D)--(B) ; \draw (D)--(C);
		\draw (A)--(E) ; \draw (B)--(E) ; \draw (C)--(E);
		
			% arrow
		\coordinate (S) at (1.25,0.2);
		\draw[->-=1,thick] (S) --+ (1,0);
		
			% 3 tetra
		\coordinate (A1) at (3,0.24);
		\coordinate (B1) at (4.73,0.02);
		\coordinate (C1) at (5.32,0.72);
		\coordinate (D1) at (4.07,1.59);
		\coordinate (E1) at (4.42,-0.85);
		
		\draw (A1)--(B1) ; \draw (B1)--(C1) ; \draw[dashed] (A1)--(C1); \draw (A1)--(D1) ; \draw (D1)--(B1) ; \draw (D1)--(C1);
		\draw (A1)--(E1) ; \draw (B1)--(E1) ; \draw (C1)--(E1);
		\draw[red] (E1)--(D1);
		
		\draw (A1) node[left]{$A$}; \draw (B1) node[above]{$B$}; \draw (C1) node[right]{$C$}; \draw (D1) node[above]{$D$}; \draw (E1) node[below]{$E$};
				
	\end{tikzpicture}
	\end{center}
	\caption{$(2-3)$ Pachner move. On the left two tetrahedra. On the right three new tetrahedra constructed by adding a link between the node $D$ and $E$, in red in the picture. The new tetrahedra are $(ABDE)$, $(BCDE)$ and $(CADE)$.}
	\label{chap3:fig:2_3_move}
\end{figure}
While this move does add a bubble to the simplicial decomposition, we will see later that invariance under the $(2-3)$ Pachner move does not (really) depend on the tree $T$. Indeed, in the case of a simplicial decomposition, invariance under the $(2-3)$ Pachner move is always exact thanks to the Biedenhart-Elliot identity.

The explicit proof of the invariance under Pachner move can either be done via direct computation \cite{Freidel:2004nb}  or using a more elegant formulation of diagrammatic identities \cite{Girelli:2001wr}. In \cite{Girelli:2001wr}, the generalization of the invariance for any cellular decomposition is also provided.

In the next section, we will finally focus on the link between the Ponzano-Regge model and quantum gravity in three dimensions by showing that the discrete $BF$ formulation of gravity is related to our model.

\subsection{BF theory, $SO(3)$ and $SU(2)$ Ponzano-Regge model}

Now that we have a correct interpretation and definition of the Ponzano-Regge model from a mathematical point of view, it is time to show how it is related to quantum gravity in three dimensions \cite{Rovelli:1993kc,Perez:2003vx}.

As we already said, the starting point for this relation is the $BF$ formulation\footnote{Note that we call $e$ both the tetrad field and the edges of the cellular decomposition. It does not, however, cause any misunderstanding since it is always clear via the context to which one we are referring to.}  of General Relativity\eqref{chap1:eq:BF_action_starting}:
\begin{equation*}
S_{BF} = \int_{\cM} \dd^3x \; \tr\left( e \wedge F[\omega] \right) \; .
\end{equation*}
The quantization of the theory is straightforwardly done via the path integral formulation. We introduce the partition function of the $BF$ action
\begin{equation}
Z_{BF} = \int \cD e \cD \omega e^{-i S_{BF}(e,\omega)} \; . 
\end{equation}
Introducing a boundary for the manifold $\cM$ does not change the action and thus the partition function if we consider the connection to be kept fixed at the boundary as seen previously. The partition function becomes in that case a functional of the boundary data
\begin{equation}
Z_{BF}(\omega|_{\pp\cM}) = \int_{\pp \omega = \omega_{\pp \cM}} \cD e \cD \omega e^{-i S_{BF}(e,\omega)} \; 
\end{equation}
where the path integral is now understood with fixed boundary connection. As usual, this path integral is only formally defined since we did not take care of the divergences arising from the gauge symmetry of the $BF$ action. 

We emphasized however that this is a {\it true quantum amplitude}, and not a statistical amplitude: we do not consider any Wick rotation to go from the Lorentizan to the Euclidean signature. As such, we are studying a true quantum theory of gravity in Euclidean signature. This is a well-defined and well-posed problem. Indeed, Euclidean gravity has a well-defined quantization scheme, as shown by the Chern-Simons formulation of the theory. This is one of the most important differences with the works we explained previously in chapter \ref{chap2}. Recall that in this chapter, the partition function for quantum gravity was computed via the Wick rotation trick. Thus, it was truly a quantum {\it statistical} amplitude and not a true quantum amplitude for a quantum theory of gravity.

Still working at the formal level, we can make use of the linearity of the $BF$ action in the tetrad field $e$ to formally integrate over it. This corresponds to interpreting the tetrad field as a Lagrange multiplier enforcing the flatness of the connection everywhere. This directly comes from the equations of motion of the $BF$ theory. In formula, this integration formally leads to 
\begin{equation}
	Z_{BF}=\int \mathcal D\omega\; \delta(F[\omega]) \; ,
	\label{chap3:eq:formal_flat_connection}
\end{equation}
where $\delta$ is the Dirac delta function over the group $\SU(2)$, by the very definition of $\omega$ being $\SU(2)$ valued. Therefore, we see that the partition function associated to the $BF$ action computes the volume of the moduli space of flat spin connection on the manifold $\cM$. It is immediate to see the link with the Ponzano-Regge model from this relation. It computes the volume of the moduli space of flat connection for the discretized manifold $\cM$.
\\

Now, we still need to make sense of the path integral. Using the property that $BF$ is a topological field theory, we can make use of lots of methods developed in this context to evaluate the partition function, such as the heat kernel method. However, here, we want to make the link with the Ponzano-Regge model explicit. Therefore, we choose to introduce a cellular decomposition of the $\cM$. We will use the same notation as developed in the previous section. 

Introducing the orientated cellular decomposition $\cK$ and its dual $\cK^{*}$, it is then necessary to find a discrete analogue for the connection $\omega$ and the tetrad field $e$. For the connection, we consider, as previously, a discrete connection on the cellular decomposition. That is, we associate to each link of $\cK^{*}$ a $SU(2)$ element $g_l$. In terms of the connection one-form $\omega$, $g_l$ is just its holonomy along the link $l$
\begin{equation}
	g_l = \cP e^{\int \omega} \;.
\end{equation}
The curvature, as a $2$-form, is then located at the faces of $\cK^{*}$. We define at each face
\begin{equation}
	G_{f} = \prod_{l \in f} g_l^{\eps_{l,f}}
\end{equation}
where we recall that $\eps(l,f)$ is the relative orientation of the link $l$ with respect to the orientation of the face $f$. At this point, it is clear why the Ponzano-Regge model and the $BF$ theory are related. However, the $BF$ action also depends on the tetrad field. From the viewpoint of $BF$ theory, the tetrad field and the connection are dual variables. Hence, keeping this duality in mind it is natural to discretize the tetrad field on object of the direct cellular decomposition since the connection is discretized on the dual decomposition. Since the $B$ field is also a one-form, it is then natural to discretize it along the edge of the cellular decomposition. In the same way the holonomies are defined, we discretize the tetrad field $B$ by integrating it over the edges of $\cK \cup \pp \cK$. This defines a $\su(2)$-valued element $X_f$, where $f$ is the dual cell in $\cK^{*} \cup \pp \cK^{*}$ of the edge $e$.

In terms of the discrete variables, the continuous symmetries of the $BF$-action are broken, but there are residual discrete symmetries. The local gauge transformation becomes parametrized by $\SU(2)$ elements $k_l$ living on the links whose action is
\begin{align}
	g_l &\rightarrow k_{t_l} g_l k_{s_l}^{-1} \\
	X_f &\rightarrow k_{st(f)} X_f k_{st(f)}^{-1} \; .
\end{align}
We recall that $st(f)$ holds for the starting node of the face $f$.

The translational symmetries on another hand becomes \cite{Freidel:2004vi} parametrized by $\su(2)$ elements $\Omega_e^n$ living at the node $n$ 
\begin{equation}
	g_l \rightarrow g_l \qquad
	X_f \rightarrow X_f + U_e^{t_e} \phi_{t_e} - [\Omega_e^{t_e},\phi_{t_e}] - U_e^{s_e} \phi_{s_e} + [\Omega_e^{s_e},\phi_{s_e}]
	\label{chap3:eq:translational_sym_discrete}
\end{equation}
where $U_e^n$ are some given scalar. See \cite{Freidel:2004vi} for more detail. In these notations, $e$ is the dual edge of the face $f$ and $s_e$, $t_e$ are still the source and target vertex of the edge $e$.
\bigskip

Fundamentally, we are left with two choices. First, we can directly start with equation \eqref{chap3:eq:formal_flat_connection} after having formally integrated out the tetrad field. The discretization of the $BF$ partition function is then straightforward given the previous discrete analogue of the connection. We obtain
\begin{equation}
	Z_{BF}^{\cK} = \int_{SU(2)} \prod_{l \in \cK^{*}} \dd g_{l} \prod_{f \in \cK^{*}} \delta\left( G_{f} \right) \; .
\end{equation}
Since we already integrated out the tetrad field, it does not play a role anymore. We may notice that the discretized expression of the $BF$ partition function exactly matches the definition of the Ponzano-Regge model before taking care of the divergences \ref{chap3:def:PR_model_without_gauge_fixing}. This shows two things. First that the Ponzano-Regge model is the discretized version of the $BF$ action, and hence a discretized realisation of three-dimensional gravity without cosmological constant in a first order formalism. Secondly, and more importantly for our understanding, the formal integration over the tetrad field is the one producing divergences! Indeed, we explained in the last section why the model as defined in \ref{chap3:def:PR_model_without_gauge_fixing} diverges. However, from the $BF$ theory point of view, these delta functions are the consequences of the integration over the tetrad field, which is hence the cause of the divergences. In fact, with the knowledge of the way we took care of the divergences previously, this result was expected. The reason is quite simple when looking at where the discrete tetra field lives. It is living at the faces of the dual cellular decomposition, i.e. at the edge of the direct decomposition, which are the roots of the divergences as seen previously. This provides us with a clearer picture of the divergences in the Ponzano-Regge model. Looking at the discrete translational symmetry \eqref{chap3:eq:translational_sym_discrete}, we see that it has an action on every vertices of the cellular decomposition with a $\su(2)$ parameter. Recall that $\su(2)$ is isomorphic to $\R^3$, which is clearly not compact. Hence, this gauge symmetry must be gauge-fixed in order not to produce any divergences. It is clear now that the choice of tree $T$ does the job of gauge-fixing the translational symmetry for the discrete tetrad field.
\medskip

Let us start again with our discretized approach of $BF$ theory but starting from the initial partition function before formal integration. The question that still needs to be answer is about the discretized form of the $BF$ action. A straightforward and naive discretization of the $BF$ action is, in terms of $X_f$ and $G_f$
\begin{equation}
S_{BF}^{\cK} = \sum_{f \in \cK^{*}} \tr\left( X_f G_f\right) \; . 
\end{equation} 
That is we straightforwardly replace the tetrad by the $X_f$ and the curvature by $G_f$. Then, the partition function takes the form
\begin{equation}
Z_{BF}^{\cK} = \int_{\SU(2)} \prod_{l \in \cK^{*}} \dd g_l \; \int_{\su(2)} \dd X_f \; e^{i \sum_{f \in \cK^{*}} \tr\left( X_f G_f \right)} \;.
\end{equation}
The integration over the manifold is replaced by the finite sum over the faces of the dual decomposition whereas the path integral is replaced by its equivalent in terms of integration over $\SU(2)$ and $\su(2)$. The partition is still formal since we need to gauge fix the translational symmetry to avoid divergence arising from the integration over $\su(2)$. Forgetting about this particular point for now, the integration over the discrete tetrad field $X_f$ can be exactly done using the usual material of defining the delta function of a group as the Fourier Transform of a distribution. Doing the integrations over the Lie algebra elements $X_f$ carefully return \cite{Freidel:2004vi,Livine:2008sw}
\begin{equation}
Z_{BF}^{\cK} = \int_{\SU(2)} \prod_{l \in \cK^{*}} \dd g_l \; \prod_{f \in \cK^{*}} \delta_{SO(3)}\left( G_f \right) \; 
\end{equation}
where this time the delta is over the group $\SO(3)$ instead of $\SU(2)$. This is the only difference compared to the previous formula. That is, the partition function we obtain compute the volume of the moduli space of flat $\SO(3)$ connection instead of $\SU(2)$. This correponds to the so-called $\SO(3)$ Ponzano-Regge model.

We do not expect, however, the result to differ compared to the formal integration of the tetrad field in the continuum.  The problem lies in the fact that our choice of discretized action for $BF$ is rather naive, and completely loses track of the $\SU(2)$ nature of the model. An easy way to see this problem is to look at the way the Fourier Transform works. Since the $X$ and the $g$ are dual variables, they can naturally be mapped by a Fourier transform and a natural candidate is
\begin{equation}
\hat{f}(X) = \int_{\SU(2)} \dd g f(g) e^{\tr(X g)} \;,
\end{equation}
where $\hat{f}$ is the Fourier component of the $\SU(2)$ function $f$. However, such a Fourier transform has a non-trivial kernel. If $f(-g^{-1}) = -f(g)$, then $\hat{f}(X) = 0$ without the necessity of having $X$ identically zero. Hence, the inverse Fourier Transform is not well defined. This can easily be seen by looking at the inverse Fourier Transform for the delta function. It reads
\begin{equation}
\int_{\su(2)} \dd^3 X e^{\tr(X g)} \propto \delta(g) + \delta(-g) \propto \delta_{SO(3)}(g) \; .
\end{equation}
That is, we do not recover the $\SU(2)$ delta function, but the $\SO(3)$ one. This result could have been expected looking at the symmetry of the kernel previously introduced which identified the identity and its opposite.

A more refined Fourier Transform is needed to completely capture the $\SU(2)$ behaviour. Its expression takes a natural form using spinors \cite{Dupuis:2011fx} (see \ref{chap5} for more details on the spinors in general). At the end of the day, the $\SU(2)$ delta function takes the form
\begin{equation}
	\delta(g) = \int_{\C^{2}} \f{\dd^4 z}{\pi^2}(|z|^2-1)e^{-|z|^2}e^{\la z | g | z \ra} \;.
\end{equation}

Note that this writing in terms of spinor is also in correspondence with the fully local expression of the Ponzano-Regge model. Indeed, while writing the Ponzano-Regge model in a fully local form, we introduced $\SU(2)$ elements with respect to the direct edge of the cellular decomposition. Here, we are saying that to capture the $\SU(2)$ aspect of the theory, it is necessary to discretize the tetra field not by a vector, but by a spinor. We quickly recall that a spinor is given by the data of two complex numbers (see chapter \ref{chap5} for more details), which is the same as a $\SU(2)$ element. At the end of the day, it was shown in \cite{Dupuis:2011fx} that by keeping carefully track of the $\SU(2)$ structure in the discretization of the tetrad field of the $BF$ action, the full $\SU(2)$ Ponzano-Regge model can be recovered.

\subsection{Ponzano-Regge model, $BF$ theory and torsions}

We shortly come back to the explanation about the divergences of the Ponzano-Regge model and its topological invariance. First, it is now clear that the divergences emerge because of gauge symmetry. Secondly, the way we have handled the divergences of the model, even though rather intuitive, is still more heuristic than mathematical. That is, we provide a gauge-fixing in terms of tree. This works quite well for handle-bodies and in three dimensions, but it quickly becomes more complicated in dimension bigger than three and for more complicated manifolds.

In that case, the full study of the partition function of both $BF$ theory and the Ponzano-Regge model is better understood in the context of twisted co-homology and the analytic Ray-Singer torsion \cite{BLAU1991130} and discrete equivalent Reidemeister torsion \cite{Dubois:torsion} respectively. Torsions are, in a sense, a generalization of the notion of determinant. They are homomorphism equivalent, allowing us to probe deeply the topological structure of the object they characterize. Hence, the important point here is that both the Ray-Singer torsion and the Reidemeister are topological invariants. That is, the study of the convergence of the partition function can be summarized into the well-definiteness of these invariants. We will not dwell upon the details here, and refer the reader to \cite{BLAU1991130} for the link between the $BF$ action and the Ray-Singer torsion and to \cite{Barrett:2008wh,Bonzom:2012mb} for the case of the Ponzano-Regge model and the Reidemeister torsion.
\medskip

Let us just remark that it is quite easy to see how the $BF$ theories are, in general related to the co-homology theory. Indeed, consider the equations of motion for the connection coming from the $BF$ action without cosmological constant
\begin{equation}
\dd F[\omega] = 0 \; .
\end{equation}
For simplicity, consider now that the gauge group is not $\SU(2)$, but an abelien group, for example $U(1)$. The previous equation of motion becomes simply $\dd A = 0$ where $A$ is the connection for the abelien group. In the context of De Rham co-homology, this equation has a simple meaning. By definition, it says that $A$ is a closed form. In the case where $A$ is exact, i.e. that it exists a zero form $\alpha$ such that $A = \dd \alpha$, the equation of motion is trivially satisfied. However, that corresponds to a trivial pure gauge connection. Hence, we want to consider exact form which are not close. This basically defines the first De Rham co-homology group.

In case of non-abelian $\SU(2)$ group, we equivalently obtain the twisted co-homology theory. It is a well-known fact since their introduction by Blau and Thompson \cite{BLAU1991130} that the partition function of the $BF$ theory is better understood in terms of the integration on the moduli space of flat connection with measure provided by the Ray-Singer torsion which in turn is related to the volume of the twisted co-homology group due to the $BF$ equations of motion and gauge invariance. And the partition function is finite when the Ray-Singer torsion is well-defined. Equivalently, the discrete counterpart of the $BF$ theory, that is the Ponzano-Regge model, is related to the discrete counterpart of the Ray-Singer torsion. This is the Reidemeister torsion \cite{Dubois:torsion}. It was shown in \cite{Barrett:2008wh,Bonzom:2010zh} that the Ponzano-Regge amplitude can indeed be evaluated as the integration over the moduli of discrete flat connection with measure provided by the Reidemeister torsion. 

Note that these studies were done in absence of a boundary term. It is unclear how much the presence of the boundary affects the statement of this section.

\section{Recovering the historical Ponzano-Regge model}

In the previous sections, we focused on a formulation of the Ponzano-Regge model that fits our needs for the computation presented in the next chapters. This is not, however, the model as derived by Ponzano and Regge in 1968 \cite{PR1968}. The purpose of this section is not to redo the historical derivation of the model but rather to show how to recover it from the definition of the Ponzano-Regge model given in \ref{chap3:def:PR_model_without_gauge_fixing}. First of all, it is necessary to restrict the choice of cellular decomposition to only a simplicial one. Indeed, what made Ponzano and Regge interested in building the model was their observation that $\Big\{ 6j \Big\}$ symbol, describing a (quantum) tetrahedron, has a peculiar asymptotic limit. Considering a tetrahedron $\sigma$, they remarked that in the limit where the spins are large, the associated $\Big\{ 6j_{e \in \sigma} \Big\}$ symbol reads \cite{PR1968,Roberts:1998zka}
\begin{equation}
	\Big\{ 6j_{e \in \sigma} \Big\} \xrightarrow{j\gg1} \frac{1}{\sqrt{12 \pi V_\sigma}} \cos\left(S_{R}[l_e = j_{e}+\f{1}{2}] +  \frac{\pi}{4}\right)
	\qquad\text{if}\qquad
	V^2_\sigma>0 \;
	\label{chap3:eq:asymptotique_6j}
\end{equation} 
where $S_{R}$ is the Regge action \ref{chap3:eq:Regge_action}, without boundary term, and where $V_{\sigma}$ is the volume of the tetrahedron $\sigma$ whose edges length are given by $j+\f{1}{2}$. The interpretation of spins being related to the lengths of the edges of the triangulation naturally comes from the relation with the Regge action which takes for parameter the edges lengths of the triangulation. In the case where the volume square of the tetrahedron is negative, the asymptotic limit is exponentially suppressed. These negative square volume configurations can be interpreted as Lorentzian geometries, see \cite{Barrett:1993db}. It is interesting therefore to note that the Euclidean Ponzano-Regge model already contains Lorentzian configuration, although suppressed. Note also that the corresponding tetrahedron might not exist. This happens if and only if the relevant Caley-Menger determinant, which is the generalization of the Heron formula to arbitrary dimensions, $-6V^2_\sigma(j_{e\in\sigma} )$ is positive. The presence of the cosine in the asymptotic formula is best understood as being a remanent of the two possible orientations for the tetrahedron. This interpretation is reinforced by the presence of the $\pi/4$ phase shift, which can in turn be understood as due to the two possible signs in $\pm\sqrt{V^2_\sigma}$.
\medskip

We now focus on the derivation of the original model. A full derivation, taking care of all the subtleties, can be found in \cite{Barrett:2008wh}. From the definition \ref{chap3:def:PR_model_without_gauge_fixing}, we will, in a sense, go to the dual representation using the spectral decomposition introduced in \eqref{chap3:eq:spectral_decompo}. We replaced all the delta functions by infinite sums over the spins. All the edges of the simplicial decomposition contribute then with a factor
\begin{equation}
\delta(G) = \sum_{j\in\frac12\mathbb N} d_j \chi^j(G) \; ,
\end{equation}
where we recall that $\chi^j(g) = \tr(D^j(h))$ is the character of the spin $j$ representation $V_j$ of dimension $d_j = 2j+1$, with  $D^j(g)$ the Wigner matrices. Now, the faces of $\cK^{*}$ do not carry a delta function, but a spin label, over which we are doing an infinite sum. In a way, this operation can be interpreted as a Fourier decomposition on the group manifold. Geometrically, the delta function associated to a face imposes flatness of the spin connection holonomy around the corresponding triangulation edge. This corresponds to the zero deficit angle imposed by the equations of motion in Regge calculus. A Fourier transform trades two conjugated variables, in this case holonomies for spins, i.e. connections for tetrad. Since the summation is infinite, it is natural that it diverges in general.

The next step to recover the historical formulation is to explicitly expand the characters associated to each face. By doing so, we make the dependencies over the links explicit for each group element. 
\begin{equation}
	\chi^j(G_f)\,=\, \chi^j \left( \prod_{l \in f} g_{l}^{\epsilon(l,f)} \right)\,=\, \prod_{l \in f}  D^{j_l}(g_{l}^{\epsilon(l,f)} ) ^{m'_l}{}_{m_l} \; .
\end{equation}
The summation over the magnetic indices of the Wigner matrices are implicit. Their contractions follow the connectivity given by the dual cellular decomposition $\cK^{*}$. The group integrations in the Ponzano-Regge amplitude can then be done explicitly by means of a standard identity for the Clebsch-Gordan coefficients $C_{j_1,j_2,j_3}\in V_{j_1} \otimes V_{j_2} \otimes V^*_{j_3}$
\begin{equation}\label{CGD1}
\int \dd g \, D^{j_1}(g) \otimes D^{j_2}(g) \otimes \overline{D^{j_3}(g)} = \frac{1}{d_{j_3}}  C_{j_1 j_2 j_3} \otimes \overline{C_{j_1 j_2 j_3}} \; ,
\end{equation}
where we identified $D^j(g): V_j \to V_j$ with an element of $V_j\otimes V_j^*$, and where we omitted the six magnetic indices associated to each copy of $V_j$ or $V_j^\ast$ for simplicity. The reason why each group element appears exactly three times is because we consider a simplicial decomposition in this section. Hence, every link is dual to a triangle. The group element $g_l$ appears for each edge of the triangle to which $l$ is dual.

After performing all the group integrations, we are left with two Clebsch-Gordan coefficients per link. It is rather natural to associate one of the two elements to the source node while the second is associated to the target node. This association is supported by the fact that the Clebsch-Gordan coefficients come with a natural opposite orientation from the integration. Equivalently in the direct representation, we can see that this corresponds to two Clebsch-Gordan coefficients per triangle, i.e. one for each tetrahedron the triangle is shared by.

At the end of the day, since we have considered a simplicial decomposition, all the nodes of $\cK^*$ are 4-valent, since they are dual to tetrahedra in $\cK$. From the previous construction, we have associated a Clebsch-Gordan coefficient to each of the links of the node $n$. Hence, the node $n$ has four coefficients, which naturally contracted into a given $\{6j\}$ symbol by carefully keeping track of the magnetic indices. Hence, we have finally associated one symbol per tetrahedron of the simplicial decomposition. There are a few key subtleties we did not mention in keeping track of the magnetic indices and the signs. Details can be found in \cite{Barrett:2008wh} and are linked to the choice of a spherical category for the model. Finally, putting everything together, we are left with the original Ponzano-Regge partition function
\begin{equation}
	Z_{PR}^{\text{original}}  = \sum_{\{j_e\}} \prod_{e} v^2_{j_e} \prod_t (-1)^{\sum_{e'\in t} j_e'} \prod_\sigma \Big\{ 6 j_{e''\in\sigma} \Big\}
	\label{chap4:eq:PR_original}
\end{equation}
with 
\begin{equation}
v^2_j = (-1)^{2j} d_{j} = (-1)^{2j}(2 j +1) \; 
\end{equation}
for the signed dimension of the $j$-th representation of $\SU(2)$. This is a local state sum model: local weights are associated to edges, triangles, and tetrahedra of the triangulation. Again, recall that the sums are infinite, and that the model is generally divergent. We sum over the spins associated to all the edges of the triangulations with a weight $v_{j_e}$ associated to the edges, a sign weight associated to the triangles, and a $\{6j\}$ symbol associated to the tetrahedra.

From the original partition function, the behaviour over Pachner move of the partition function is clear to see. Recall that the $(2-3)$ move corresponds to transforming two tetrahedra into three. From the $\{6j\}$ symbol perspective, this is just the Biedenharn-Elliot identity. Calling the spins of the shared triangles of the two tetrahedra $(j_1,j_2,j_3)$, see figure \ref{chap4:fig::2_3_move_spin}, the identity reads, calling $s_i = j+\sum_{i}j_i+k_i+h_i$
\begin{equation*}
	\begin{Bmatrix}
		j_1	&	j_2	&	j_3			\\
		k_1	&	k_2	&	k_3
	\end{Bmatrix}
	\begin{Bmatrix}
		j_1	&	j_2	&	j_3			\\
		h_1	&	h_2	&	h_3
	\end{Bmatrix}	
	=
	\sum_{j} (-1)^{s_i} 
	\begin{Bmatrix}
		k_1	&	h_1	&	j			\\
		h_2	&	k_2	&	j_3	
	\end{Bmatrix}
	\begin{Bmatrix}
		k_2	&	h_2	&	j			\\
		h_3	&	k_3	&	j_1	
	\end{Bmatrix}
	\begin{Bmatrix}
		k_3	&	h_3	&	j			\\
		h_2	&	k_2	&	j_2	
	\end{Bmatrix} \; .
	\label{chap3:eq:Biedenharn_elliot_identity}
\end{equation*}
The signs coming from the identity happen to match the one needed for the Ponzano-Regge model. Therefore, the $(2-3)$ Pachner move holds exactly. This is the point we mentioned previously. Even though this Pachner move does add a bubble to the triangulation, it is still an exact transformation even without considering gauge-fixing due to the symmetries of the $\{6j\}$ symbols.

\begin{figure}[!htb]
	\begin{center}
		\begin{tikzpicture}[scale=2]
		% 2 tetra
		\coordinate (A) at (-1.43,0.24);
		\coordinate (B) at (0.3,0.02);
		\coordinate (C) at (0.89,0.72);
		\coordinate (D) at (-0.36,1.59);
		\coordinate (E) at (-0.01,-0.85);

		\draw (A)--node[above,pos=0.6]{$j_1$} (B) ; \draw (B)--node[left,pos=0.6]{$j_2$} (C) ; \draw[dashed] (A)--node[above]{$j_3$} (C); \draw (A)--node[left,pos=0.6]{$k_2$} (D) ; \draw (D)--node[right]{$k_3$} (B) ; \draw (D)--node[right]{$k_1$} (C);
		\draw (A)--node[left,pos=0.6]{$h_2$} (E) ; \draw (B)--node[left]{$h_3$} (E) ; \draw (C)--node[right]{$h_1$} (E);
		
		% arrow
		\coordinate (S) at (1.25,0.2);
		\draw[->-=1,thick] (S) --+ (1,0);
		
		% 3 tetra
		\coordinate (A) at (3,0.24);
		\coordinate (B) at (4.73,0.02);
		\coordinate (C) at (5.32,0.72);
		\coordinate (D) at (4.07,1.59);
		\coordinate (E) at (4.42,-0.85);
		
		\draw (A)--node[above,pos=0.6]{$j_1$} (B1) ; \draw (B)--node[left,pos=0.6]{$j_2$} (C) ; \draw[dashed] (A)--node[above,pos=0.4]{$j_3$} (C); \draw (A)--node[left,pos=0.6]{$k_2$} (D) ; \draw (D)--node[right]{$k_3$} (B) ; \draw (D)--node[right]{$k_1$} (C);
		\draw (A)--node[left,pos=0.6]{$h_2$} (E) ; \draw (B)--node[left,pos=0.3]{$h_3$} (E) ; \draw (C)--node[right]{$h_1$} (E);
		\draw[red] (E)--node[right]{$j$} (D);
		
		\end{tikzpicture}
	\end{center}
	\caption{$(2-3)$ Pachner move corresponding to the Biedenharn-Elliot identity.}
	\label{chap4:fig::2_3_move_spin}
\end{figure}
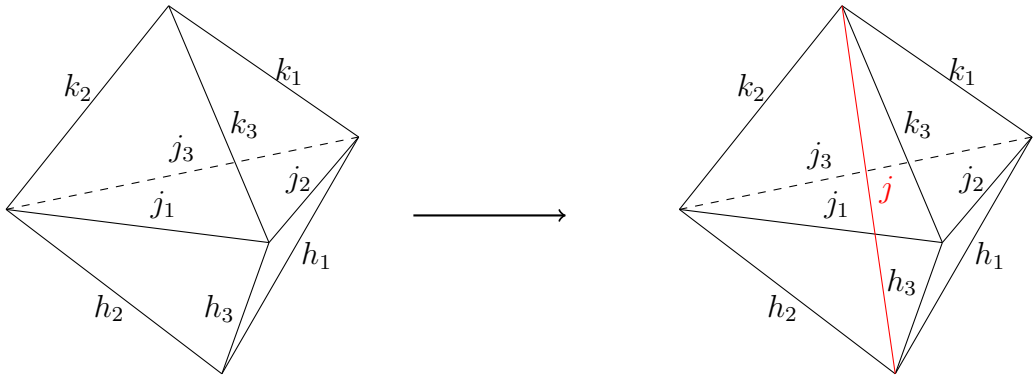

Similarly, we can look at the $(1-4)$ move. As expected, this move does produce an explicit divergence. See figure figure \ref{chap4:fig:1_4_move_spin} for the notations. There are two ways to derive this move on the historical Ponzano-Regge model. The first one is to act with a holonomy operator Wilson loop \cite{Bonzom:2009zd} on the initial tetrahedron, and then to evaluate this action on the identity. The second way is to work directly on the orthogonality relation of the $\{6j\}$ symbols and on the Biedenharn-Elliot identity. The advantage of the Wilson loop approach is that the divergence problem results in an interesting interpretation. The action of the holonomy operator of spin $j$ evaluated at the identity returns
\begin{align*}
	d_j
	\begin{Bmatrix}
		j_1	&	j_2	&	j_3			\\
		j_4	&	j_5	&	j_6
	\end{Bmatrix}
	=
	\sum_{k_1,k_2,k_3} (-1)&^{j + \sum_{i=1}^{6} j_i + \sum_{i=1}^{3} k_i} d_{k_1} d_{k_2} d_{k_3} \\
	&
	\begin{Bmatrix}
		j_1	&	j_2	&	j_3			\\
		k_1	&	k_2	&	k_3
	\end{Bmatrix}
	\begin{Bmatrix}
		j_1	&	j_6	&	j_5			\\
		j	&	k_2	&	k_3
	\end{Bmatrix}
	\begin{Bmatrix}
		j_6	&	j_2	&	j_4			\\
		k_1	&	j	&	k_3
	\end{Bmatrix}
	\begin{Bmatrix}
		j_3	&	j_4	&	j_5			\\
		j	&	k_2	&	k_1
	\end{Bmatrix}	
\end{align*}

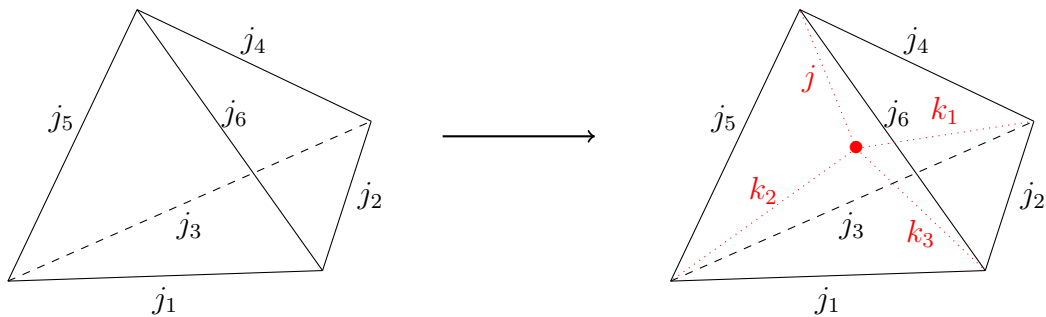
\begin{figure}[!htb]
	\begin{center}
		\begin{tikzpicture}[scale=2]
		\coordinate (A) at (-1.36,-0.76);
		\coordinate (B) at (0.71,-0.69);
		\coordinate (C) at (1.03,0.3);
		\coordinate (D) at (-0.51,1.04);
		
		% 1 tetra
		\draw (A)--node[below]{$j_1$} (B) ; \draw (B)--node[right]{$j_2$} (C) ; \draw[dashed] (A)--node[below,pos=0.5]{$j_3$} (C); \draw (A)--node[left,pos=0.6]{$j_5$} (D) ; \draw (D)--node[right,pos=0.4]{$j_6$} (B) ; \draw (D)--node[above]{$j_4$} (C);
		
		%arrow
		\coordinate (S) at (1.5,0.2) ; 
		\draw[->-=1,thick] (S) --+(1,0);
		
		% 4 tetra
		\coordinate (A) at (3,-0.76);
		\coordinate (B) at (5.07,-0.69);
		\coordinate (C) at (5.39,0.3);
		\coordinate (D) at (3.85,1.04);
		\coordinate (N) at (4.22,0.12);
		
		\draw[red] (N) node{$\bullet$} ;
		
		\draw (A)--node[below]{$j_1$} (B) ; \draw (B)--node[right]{$j_2$} (C) ; \draw[dashed] (A)--node[below,pos=0.5]{$j_3$} (C); \draw (A)--node[left,pos=0.6]{$j_5$} (D) ; \draw (D)--node[right,pos=0.4]{$j_6$} (B) ; \draw (D)--node[above]{$j_4$} (C);
		\draw[red,dotted] (N)--node[above]{$k_2$} (A); \draw[red,dotted] (N)--node[below]{$k_3$} (B);  \draw[red,dotted] (N)--node[above]{$k_1$} (C);  \draw[red,dotted] (N)--node[left]{$j$} (D); 
		
		\end{tikzpicture}
	\end{center}
	\caption{$(1-4)$ Pachner move. On the left a single tetrahedron. On the right, we add a single node (in red) inside the tetrahedron. The spin $j$ is fixed in the holonomy operator action.}
	\label{chap4:fig:1_4_move_spin}
\end{figure}

This holonomy operator action produces four $\{6j\}$ symbols from only one with the correct sign from the viewpoint of the Ponzano-Regge model. However, one of the newly introduced spins, $j$ in figure \ref{chap4:fig:1_4_move_spin}, is fixed and not summed over. Hence, we do not recover the full Ponzano-Regge model. To recover it, it is necessary to sum over this last spin $j$ since all the spins are summed over in the Ponzano-Regge model. Such a summation obviously returns a divergent factor of the form  $\sum\limits_{j} d_j^2$. Hence, from the viewpoint of the present section, the edge carrying the spin $j$ belongs to the tardis of the decomposition and the vertex we add inside the tetrahedron corresponds to a bubble. Geometrically speaking, the divergence is better understood noticing that the vertex inside the tetrahedron can be moved outside the region made of the tetrahedra sharing it. Indeed, recall that the spin $j$ is understood at the length of the edge, and here, is totally unconstrained. The amplitude remains constant in this case. This is due to the sum over orientations that is implemented in the Ponzano-Regge model \cite{Christodoulou:2012af}. The fact that we have this freedom is explained by the residual diffeomorphism symmetry, which in turn also explains the invariance under triangulation changes of the partition function \cite{Bahr:2011uj, Dittrich:2012qb}.

Hence, to recover the Ponzano-Regge model, one needs to sum over $j$. However, by doing so, one creates a divergence. We already studied the gauge fixing of this divergence in the group picture. In the dual picture, with the spin, the gauge fixing basically tells us that we should not sum over $j$, but rather fix $j$ to a particular value. It is interesting to note that in this dual picture case, it is rather non intuitive to fix the spin to zero, since we want the spin to encode the length of the edge.

At the end of the day, one can understand the holonomy operator action evaluated at the identity as implementing a gauge-fixed version of the $(1-4)$ Pachner move.

\section{Ponzano-Regge model and Wheeler-DeWitt equation}

Before moving on to the Ponzano-Regge model on the torus, we shortly review the link between the model and canonical quantization and loop quantum gravity. 

We claim that the Ponzano-Regge model is a good discrete model to describe quantum gravity in three dimensions. The goal of this section is to relate the model with other approaches, namely canonical quantization and loop quantum gravity. To do so, we will first show that the Ponzano-Regge model can be seen as a projector \cite{Ooguri:1991ni,Noui:2004iy,Noui:2004iz} into the physical Hilbert space construct from standard loop quantum gravity technique.

To do so, let us first consider the model for a three-dimensional space-time manifold of the form $\Sigma \times [0,1]$ where $\Sigma$ is a closed orientable surface and a simplicial decomposition of the manifold. This is basically the topology of a cylinder. We denote by $\Delta_t$ and $\Delta_{b}$ the triangulation of the top (resp. bottom) of the boundary cylinder. The Ponzano-Regge partition function defined on such a manifold takes the form $Z_{\Delta_t,\Delta_b}(\{j_t\},\{j_b\})$. The set $\{j_t\}$ (resp. $\{j_b\}$) corresponds to the spins associated to the edges of $\Delta_t$ (resp. $\Delta_b$). The partition function is a functional of the boundary data. In the historical formulation of the model, it corresponds to the spins of the boundary triangulation.

One of the most natural operations to consider in this context is the gluing operation, i.e. the gluing of manifolds together along their boundaries. More specifically, consider two manifolds with the cylinder topology provided above. We denote the boundary triangulation of the second manifold by $\Delta^1_t$ and $\Delta^1_{b}$. To glue the top of the first manifold with the bottom of the second, it is necessary that the triangulations $\Delta_{t}$ and $\Delta_{b}^{1}$ match. That is, we need to have $\Delta_{t} = \Delta_{b}^{1}$. If this condition is true, then we can define the partition function of the glued manifold via the Ponzano-Regge model. It is straightforward to show that this partition function $Z_{\Delta_t^1,\Delta_b}(\{j^1_t\},\{j_b\})$ can be expressed in terms of the partition function of the two initial manifolds. Explicitly, we have
\begin{equation}
	Z_{\Delta_t^1,\Delta_b}(\{j^1_t\},\{j_b\}) \propto \sum_{\{j_t\}} Z_{\Delta_t,\Delta_b}(\{j_t\},\{j_b\}) Z_{\Delta_t^1,\Delta_t}(\{j^1_t\},\{j_t\}) \; .
\end{equation}
This formula is readily interpretable. It is clear that outside the shared boundary, the contribution of each partition function is the same as compared to the initial two manifolds. At the glued boundaries however, what were initially boundary edges become bulk edges. Hence, we need to sum over them to obtain the Ponzano-Regge amplitude on the full manifold.

On a different ground, it feels natural to interpret the direction encoded in the interval $[0,1]$ as being the time direction. Hence, the previous topology is understood as having an initial state at $0$ which evolves to the final state as $1$. Then, every surface of the manifold at fixed $t$ can be interpreted as a kinematic state at fixed time. The loop quantum gravity formalism associates a Hilbert space structure to the family of kinematic state as the space of gauge-invariant wave-functions on the configuration space of $\SU(2)$ connection \cite{Baez:1993id,Ashtekar:1994mh,Ashtekar:1994wa}. See the beginning of chapter \ref{chap5} for more details on this kinematic Hilbert space. For now, we denote by $\Psi$ such a state and note that it is supported by the triangulation of the surface at constant $t$ it characterizes. It naturally depends on the spins of this surface. Physical states are then obtained from kinematical states by restraining them to flat connections only. This is the natural role of the Ponzano-Regge model which encodes a theory of flat connections. We consider a kinematical state $\Psi$ of $\cH_{kin}(\Delta)$, i.e. with support on a triangulation $\Delta$. We define its projected counterpart into the space of physical solutions by 
\begin{equation}
	P(\Psi_{\Delta})(\{j_{e}\}) = \sum_{k_{e}\; e \in \Delta} Z_{\Delta,\Delta}(\{j_e\},\{k_e\}) \psi_{\Delta}(\{k_e\}) \;.
\end{equation}
It is immediate to check that this is indeed a projector using the relation provided earlier on the gluing of the partition functions. This projector can then be used to define a scalar product between states \cite{Noui:2004iy}. A state is then called physical if it is an eigenvector of the projector operator
\begin{equation}
	\Psi_{\Delta} \in \cH_{phys}(\Delta) \quad \implies \quad P(\Psi)(\{j_{e}\}) = \Psi_{\Delta}(\{j_{e}\}) \;.
\end{equation}

Hence, this condition can also be interpreted as implementing a discrete Wheeler-DeWitt condition. On a different ground, an explicit Wheeler-DeWitt equation can be recovered by looking at the recurrence relation of the $\{6j\}$ symbols \cite{Bonzom:2009zd,Dittrich:2008pw}.

From the Biedenharn-Elliot identity \eqref{chap3:eq:Biedenharn_elliot_identity}, one can find recurrence relations for the $\{6j\}$ symbol \cite{Schulten:1975yu,Bonzom:2009zd}. For example, the $\{6j\}$ symbol follows the identity \cite{Varshalovich:1988ye}
\begin{equation}
	A_{-1}(j_i)
	\begin{Bmatrix}
		j_1-1	&	j_2	&	j_3			\\
		j_4	&	j_5	&	j_6
	\end{Bmatrix}
	+
	A_{0}(j_i)
	\begin{Bmatrix}
		j_1	&	j_2	&	j_3			\\
		j_4	&	j_5	&	j_6
	\end{Bmatrix}
	+
	A_{+1}(j_i)
	\begin{Bmatrix}
		j_1+1	&	j_2	&	j_3			\\
		j_4	&	j_5	&	j_6
	\end{Bmatrix}
	=
	0
\end{equation}
This relation, with the data of an initial condition, is enough to entirely determine the $\{6j\}$ symbol. The explicit form of the constants can be found in \cite{Schulten:1975yu}. Going to the asymptotic limit, i.e. scaling the spins with a common factor $c$ going to infinity, it is possible to rewrite the recurrence relation in terms of a second order discrete differential operator. This recurrence relation happens to also be interpretable as coming from the constraint of vanishing curvature from the $BF$ action, and hence implementing a discrete Wheeler DeWitt condition \cite{Bonzom:2009zd}.

\section{Ponzano-Regge model on the sphere}

Finally, before studying the torus case, it seems natural to look at what happens on the trivial topology, that is on a 3-ball with boundary the 2-sphere. We start this section with a short presentation of the gauge-fixing procedure in this case. We then discuss two results obtained in this particular topology. First, the computation of n-point correlation function, allowing a link between the spin-foam approach and the usual quantum field theory. Secondly, a link between three-dimensional gravity on this particular topology and statistical model, namely the Ising model.
\\

Let us begin with the gauge-fixing of the Ponzano-Regge model in that particular topology. The simplest triangulation of a 3-ball is one tetrahedron. It is a discretization without any bulk vertices. Hence, there are no bubbles and the model should not diverge since the tree $T$ is empty. The dual triangulation is represented figure \ref{chap3:fig:tetra_dual}. It has four boundary nodes, six boundary edges and four boundary faces. However, it only has one bulk node, four bulk edges, six bulk faces and four bulk 3-cells. The modal has thus six delta functions for 10 links. However, four group elements are fixed at the identity by the choice of tree $T^{*}$.
\begin{figure}[htb!]
	\begin{center}
		\begin{tikzpicture}[scale = 3]
			\coordinate (T1) at (-1.34,0);
			\coordinate (T2) at (0.66,-0.79);
			\coordinate (T3) at (0.71,0.45);
			\coordinate (T4) at (-0.08,1.53);
			
			%tetra
			\draw[black=20!,dotted]  (T1)--(T2)--(T3)--cycle; 
			\draw[black=20!,dotted]  (T1)--(T2)--(T4)--cycle;
			\draw[black=20!,dotted]  (T1)--(T4)--(T3)--cycle;
			
			%dual
			\coordinate (B) at (-0.04,0.5);
			\coordinate (F1) at (-0.26,0.7); \coordinate (F2) at (-0.17,0.45); \coordinate (F3) at (0.36,0.56); \coordinate (F4) at (0,0);
			\coordinate (E1) at (-0.74,0.73); \coordinate (E4) at (0.34,0.97); \coordinate (E3) at (0.27,0.44); \coordinate (E2) at (-0.37,-0.38); \coordinate (E5) at (-0.29,0.23); \coordinate (E6) at (0.69,-0.11);
			
			\draw[red] (B)--(F1); \draw (B)--(F2); \draw (B)--(F3); \draw (B)--(F4);
			\draw[dashed,red] (F1)--(E1); \draw[dashed] (F1)--(E4); \draw[dashed] (F1)--(E5);
			\draw[red] (E1)--(F2); \draw[red] (F2)--(E3); \draw (F2)--(E2);
			\draw[red] (F3)--(E3); \draw (F3)--(E4); \draw[red] (F3)--(E6);
			\draw (F4)--(E5); \draw[red] (F4)--(E6); \draw (F4)--(E2);
			
			\draw[red] (-0.36,0.62) node[scale = 0.5]{$1$};
		\end{tikzpicture}
	\end{center}
	\caption{Drawing of a tetrahedron, in black dots, and its dual (in black, dashed black and red). We draw in red the maximal tree $T^{*}$. It has 4 links, and only one is a bulk link.}
	\label{chap3:fig:tetra_dual}
\end{figure}
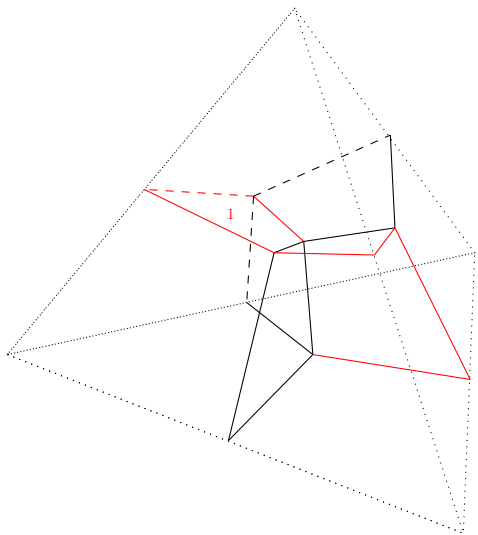
The remaining six can also be fixed to the identity with the six deltas. Starting at the only bulk face where two of the three links are fixed by the choice of tree, denoted $1$ in \ref{chap3:fig:tetra_dual}, the last one is also the identity using the delta function on this face. This generates another face where only one element is left unfixed. This logic can be followed until all the elements are the identity and no delta functions remain. Thus the Ponzano-Regge model on this simple triangulation clearly returns a well-defined partition function. More precisely, {\it all} the boundary group elements are fixed to be the identity after gauge-fixing! Hence, in the case of the 3-ball, evaluating the Ponzano-Regge model on a boundary state $\Psi(g_{\pp})$ just returns the evaluation of the boundary state at identity
\begin{equation}
	\la Z_{PR}^{\cK} | \Psi(\{ g_{\pp} \}) \ra|_{3-ball} = \Psi(\{ g_{\pp} = \id \}) \; .
\end{equation}
In the context of loop quantum gravity, where boundary state are spin networks, this is called the spin network evaluation.

The result can easily be understood remembering that the 3-ball has trivial topology. Hence, all the cycles, even at its boundary, are contractible to a point. This can easily be generalized to any handle-bodies up to the fixing of the group elements. In the following, we will study the case of the three-dimensional cylinder with boundary the two-dimensional torus, i.e. for a handle-body of genus one. With this manifold, one $\SU(2)$ element is not fixed by gauge-fixing, and we still need to integrate over it when evaluating the Ponzano-Regge model amplitude given the boundary data. This remaining integration will represent the remaining information about the bulk.
\\

Formulated as a discrete model, the Ponzano-Regge model is well suited for analytical and numerical computation. This thesis is focused on analytical computation in the case of a non-trivial topology. For the sphere, both the analytical and numerical aspects were already studied. In \cite{Dittrich:2013jxa,Bonzom:2015ova} a link between the Ponzano-Regge model and statistical models was proposed. More specifically, it was shown that quantum gravity on the 3-ball is related to the Ising model. This result comes, in fact, from Westbury theorem \cite{Westbury1998} on explicit expression for the generating function of spin network. In \cite{Bonzom:2015ova} the generating function of spin network as a boundary data for the Ponzano-Regge model was considered. Using some supersymmetric dualities, it was shown that quantum gravity and the Ising model are dual. More specifically, they showed that
\begin{equation}
	Z_{PR}(Y_l) \propto \f{1}{Z_{Ising}(y_l)^2} \; ,
\end{equation}
where $Y_l = \tanh(y_l)$. The set of parameters $\{y_l\}$ are the coupling constants of the Ising partition function while the set $\{Y_l\}$ are the coupling constants of the spin networks generating function. See chapter \ref{chap6} for more detail. Recall that the Ising partition function takes the form
\begin{equation*}
Z_{Ising}(y_l) = \sum_{S_{s_l,t_l} = \pm 1} e^{\sum_{l} y_l S_{s_l}S_{t_l}} \;
\end{equation*}
where a spin variable $S = \pm 1$ is associated to every node of the graph and the $y_l$ are the coupling constants associated to the links of the graph. 	

In the last chapter of this thesis, we will also consider the generating function of spin network as boundary state. However, no duality with existing statistical model will (yet) be proposed in the case of the torus topology.
\\

Even though the Ponzano-Regge model is formulated as a completely background independent model, we expect to recover the conventional perturbative expansion at low-energy described in terms of graviton. The study of this emergence was started in \cite{Speziale:2005ma,Livine:2006ab,Bonzom:2008xd} where analytic and numerical computation of the two-point correlations function were performed. In the usual quantum field theory approach, correlation functions are computed from the path integral formulation by introducing fields to the path integral at some given position of space-time. Then they can  be computed either by Feynman perturbative approach or by introducing a lattice. The problematic point of this definition is that space-time is then equipped with a background metric. The computation of the equivalent correlation functions is then not straightforward in the spin foam formalism, which is a background independent theory. To solve this problem, in the spin foam model, correlations functions are constructed by means of the propagation kernel \cite{Modesto:2005sj}. This provides an amplitude for the value of the field on a boundary of the space-time.

In \cite{Speziale:2005ma,Livine:2006ab,Bonzom:2008xd}, the computation of the graviton of the theory, the two-point correlation function was performed in a toy model consisting of a single tetrahedron and the quantum amplitude defined by the Ponzano-Regge model. Even though the theory is topological, it is possible to compute the correlation function in a gauge where it does not vanish. Of course, gravitons are then pure gauge effect. From these computations, it is already possible to recover the expecting behaviour of the correlation function: it decreases as one over the distance. The advantage of this approach is that it allows to see the implication of the non-perturbative correction naturally arising from the definition of a non-perturbative model. In this thesis, we did not look at such computations. We restricted ourselves to the computation of the partition function. It is, however, a natural follow-up of the work presented here.
\\

It is now time to look at the Ponzano-Regge model with a more complicated topology, the three-dimensional cylinder with boundary the two-dimensional torus.

	\newpage
	~
	\thispagestyle{empty}
\renewcommand{\afterpartskip}{}
\part*{Part II\\[.3cm]
Discrete Quantum Gravity and Partition Function} 
\addcontentsline{toc}{part}{II\ From Discrete Quantum Gravity to the BMS character and Beyond} \label{part:partII}
\newpage
~
\thispagestyle{empty}

	\chapter{The Ponzano-Regge model on a torus}
\label{chap4}

In the previous chapter, we presented the Ponzano-Regge model in its most general form. The present chapter serves as a prelude to the explicit computation of the amplitude of the last two chapters. Similarly to the previous chapter, we start with the case of quantum Regge calculus following \cite{Bonzom:2015ans}, where the partition function for three-dimensional discrete gravity in flat space was computed. Then, we focus on introducing the necessary knowledge for the definition of our boundary states, before describing the cellular decomposition of the torus considered in this thesis. Finally, in the last section, we perform a first computation of the Ponzano-Regge amplitude in a simple case. As simple as this case is, it will allow us to have a basic understanding of the property of the partition function on top of showing a link with a statistical model, the 6-vertex model. The end of this chapter is mainly based on \cite{Dittrich:2017hnl}.

\section{Quantum Regge calculus and the BMS character}

In \cite{Bonzom:2015ans}, Dittrich and Bonzom applied perturbative quantum Regge calculus to compute the thermal partition function of flat gravity. It is a first step toward the direction of the exact quasi-local computation. Indeed, while it allows to work on a finite region of space-time, it is still a perturbative approach. It is also worth noting that the amplitude computed here is again a statistical one obtained after a Wick rotation. The beauty of this computation is to be found in the fact that it allows to recover a BMS-like structure for the amplitude of three-dimensional gravity on a torus, however for a finite region! We will not dwell upon the presentation of the computation and refer the reader to \cite{Bonzom:2015ans} for the details.

The setup of the computation is as follows: consider a solid flat cylinder of height (time extension) $\beta$, and radius $a$, regularly divided into $N_t$ time slices each in turn subdivided into $N_x$ cake-slice-like prisms. This gives a discretization of the cylinder, see the left figure of \ref{chap4:fig:clyndier_decompo}. Recall however that Regge calculus is only defined on a simplicial decomposition. To make the discretization of the cylinder simplicial, it is enough to subsequently divide each prism into three tetrahedra (see the right figure of \ref{chap4:fig:clyndier_decompo}). Consequently, the boundary triangulation consists of a regular rectangular lattice subdivided into triangles along the rectangle's diagonals. To obtain the torus, the top and the bottom of the cylinder are identified, up to a twist parametrized by $N_\gamma \in \N$. The lengths of the background edges are then fixed in terms of $a$ and $\beta$ by the flatness requirement and the fact that the cylinder is embeddable into $\mathbb R^3$. We will detail more the chose of discretization in the next section while discussing the Ponzano-Regge case.

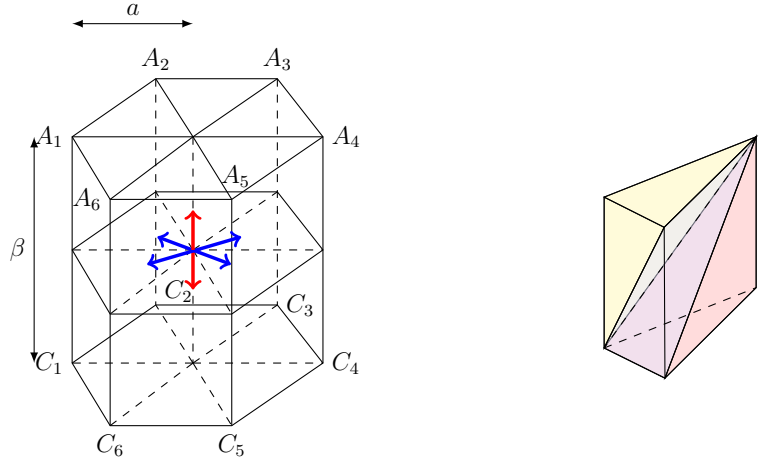
\begin{figure}[h!]
	\begin{center}
		\begin{tikzpicture}[scale=1]
		%figure 1 : cylinder
		%coordinate
		%hexagone A
		\coordinate (OA) at (1.59,0);
		\coordinate (A1) at (0,0);
		\coordinate (A2) at (1.1,0.77);
		\coordinate (A3) at (2.7,0.77);
		\coordinate (A4) at (3.3,0);
		\coordinate (A5) at (2.1,-0.83);
		\coordinate (A6) at (0.5,-0.83);
		
		%hexagone B
		\coordinate (OB) at (1.59,-1.5);
		\coordinate (B1) at (0,-1.5);
		\coordinate (B2) at (1.1,-0.73);
		\coordinate (B3) at (2.7,-0.73);
		\coordinate (B4) at (3.3,-1.5);
		\coordinate (B5) at (2.1,-2.35);
		\coordinate (B6) at (0.5,-2.35);
		
		%hexagone A
		\coordinate (OC) at (1.59,-3);
		\coordinate (C1) at (0,-3);
		\coordinate (C2) at (1.1,-2.23);
		\coordinate (C3) at (2.7,-2.23);
		\coordinate (C4) at (3.3,-3);
		\coordinate (C5) at (2.1,-3.83);
		\coordinate (C6) at (0.5,-3.83);
		
		%drawing
		%hexagone A	 
		\draw (A1) -- (A2) -- (A3) -- (A4) -- (A5) -- (A6) -- cycle; 
		\draw (OA) -- (A1) ; \draw (OA) -- (A2); \draw (OA) -- (A3); \draw (OA)--(A4); \draw (OA) --(A5); \draw (OA)-- (A6); 
		\draw (A1) node[scale=0.8,left] {$A_1$}; \draw (A2) node[scale=0.8,above] {$A_2$}; \draw (A3) node[scale=0.8,above] {$A_3$}; \draw (A4) node[scale=0.8,right] {$A_4$}; \draw (2.15,-0.78) node[scale=0.8,above]{$A_5$}; \draw (A6) node[scale=0.8,left]{$A_6$};
		
		%hexagone B
		\draw (B1) -- (B2) -- (B3) -- (B4) -- (B5) -- (B6) -- cycle;
		\draw [dashed] (OB) -- (B1); \draw [dashed] (OB) -- (B2); \draw [dashed] (OB) -- (B3); \draw [dashed] (OB)--(B4); \draw [dashed] (OB) --(B5); \draw [dashed] (OB)-- (B6);
		
		%hexagone C
		\draw (C1) -- (C2) -- (C3) -- (C4) -- (C5) -- (C6) -- cycle;
		\draw [dashed] (OC) -- (C1); \draw [dashed] (OC) -- (C2); \draw [dashed] (OC) -- (C3); \draw [dashed] (OC)--(C4); \draw [dashed] (OC) --(C5); \draw [dashed] (OC)-- (C6);
		\draw (C1) node[scale=0.8,left] {$C_1$}; \draw (1.1,-2.05) node[scale=0.8,right] {$C_2$}; \draw (C3) node[scale=0.8,right] {$C_3$}; \draw (C4) node[scale=0.8,right] {$C_4$}; \draw (C5) node[scale=0.8,below]{$C_5$}; \draw (C6) node[scale=0.8,below]{$C_6$};
		
		%link between hexagone
		\draw (A1)--(B1)--(C1); \draw [dashed] (A2)--(B2)--(C2); \draw [dashed] (A3)--(B3)--(C3); \draw (A4)--(B4)--(C4); \draw (A5)--(B5)--(C5); \draw (A6)--(B6)--(C6);
		\draw [dashed] (OA)--(OB)--(OC); 
		
		%annotation
		\draw [<->,>=latex] (-0.5,0) -- (-0.5,-3); \draw (-0.5,-1.5) node[scale=0.8,left]{$\beta$};
		\draw [<->,>=latex] (0,1.5) -- (1.59,1.5); \draw (0.8,1.5) node[scale=0.8,above]{$a$};
		
		%figure 2
		\coordinate (T1u) at (9,0);
		\coordinate (T1d) at (9,-2);
		\coordinate (T2d) at (7.8,-3.2);
		\coordinate (T3d) at (7,-2.8);
		\coordinate (T2u) at (7.79,-1.2);
		\coordinate (T3u) at (7,-0.8);
		
		\draw[fill=yellow!20,opacity=.9] (T3d)--(T1u)--(T3u)--cycle;			
		\draw[fill=red!15,opacity=.9] (T3d)--(T1u)--(T1d)--(T2d)--cycle;
		\draw[fill=blue!10,opacity=.5] (T3d)--(T2u)--(T1u)--(T2d)--cycle;
		\draw[black] (T3u)--(T2u)--(T3d)--cycle; \draw[black] (T3u)--(T1u)--(T2u)--cycle; \draw[black] (T3d)--(T2u)--(T1u);
		
		\draw[black] (T3d)--(T2d)--(T1u); \draw[black] (T1u)--(T2d)--(T1d)--cycle;

		\draw[dashed] (T3d)--(T1d);  \draw[dashed] (T1u)--(T3d); \draw (T2u)--(T2d);
		
		%arrows
		\draw [->,red, line width = 1.3 pt] (OB) -- ( $ (OB)!.35!(OC) $ ); 	
		\draw [<->,blue, line width = 1.3 pt] ( $ (OB)!.45!($(B3)!.5!(B4)$)$ ) -- ( $ (OB)!.45!($(B1)!.5!(B6)$) $ ); 	
		\draw [<->,blue, line width = 1.3 pt] ( $ (OB)!.45!($(B1)!.5!(B2)$)$ ) -- ( $ (OB)!.45!($(B4)!.5!(B5)$) $ ); 	
		\draw [->,red, line width = 1.3 pt] (OB) -- ( $ (OB)!.35!(OA) $ );		
		\end{tikzpicture}
	\end{center}
	\caption{Example of the background triangulation with $N_x=6$ and $N_t=2$. The effect of the twist $N_\gamma$ appears when we identified $A_i$ and $C_i$ through $A_{i}=C_{i+N_\gamma}$. Each prism is triangulated with three tetrahedra, that can be constructed by considering a diagonal per vertical face of the prism. In the right panel we draw a prism triangulated with three tetrahedra, draw in red, blue and white.}
	\label{chap4:fig:clyndier_decompo}
\end{figure}

Now that the background structure for the length parameters is fixed, we can focus on the partition function. First, the classical contribution of the Regge action happens to match the classical contribution from the continuum
\begin{equation}
S^{cl}_\text{R} = -\frac{ \beta}{8 G} \; .
\end{equation}
This is not a surprising result since the on-shell action is defined in terms of $\beta$ and the Newton constant, which are left untouched by the discretization process.
\\

The one-loop contribution is computed by performing a Gaussian integral over the fluctuating bulk edges. This is done via the formula \eqref{chap3:eq:on_loop_regge} at fixed boundary edges. For $N_x$ odd, the result of this computation is
\begin{equation}
	Z_\text{R}^{one-loop}(\beta, \gamma) = \mathcal{N}\,\E^{\frac{\beta}{8 G}} \;\prod_{k=2}^{\tfrac12({N_x-1})} \frac{1}{\left| 1 - \E^{ i \gamma \cdot k} \right|^2} ,
	\label{chap4:eq:one_loop_regge_result}
\end{equation}
where
\begin{equation}
\gamma = 2\pi \frac{N_\gamma}{N_x}
\end{equation}
is the discrete twist angle. The factor $\mathcal N = \mathcal N(\beta, a, N_x, N_t) $ is a complicated normalization factor. Contrary to the classical on-shell contribution it does not contain any exponential dependence on neither $\beta$ nor $a$. More importantly, it also features no dependence on the twist $N_\gamma$. Such dependence is limited to the now familiar product of \eqref{chap4:eq:one_loop_regge_result}.

This result shows that the perturbative Regge calculation displays a discrete regularization of the one presented in chapter \ref{chap2} and formula \eqref{chap2:eq:AdS_one_loop_partition_function}. The regularization is however different compared to the flat case limit leading to the formula \eqref{chap2:eq:flat_one_loop_partition_function}. In the continuum case, it was necessary to keep track of a complex part to the modular parameter due to the product over the modes going to infinity. Here, on the other hand, it is the discrete lattice used to describe the finite-resolution boundary state which naturally provides a cut-off for the product.

The important point of the formula \eqref{chap4:eq:one_loop_regge_result} is that the product starts again at the mode $k=2$. The computation of the partition function is done via Fourier Transform, and the modes $k$ labels the Fourier modes in the spatial direction for the fluctuations of the bulk radial edges. At the end of the day, the disappearance of the first mode is again related to the diffeomorphism invariance also in perturbative quantum Regge calculus. To understand this, we focus on the geometrical meaning of the missing modes $k=0,1$. Given a constant time hyper-surface, theses modes correspond to the rigid translation of the unique internal vertex of the slice. There are two contributions to this translation. It can either be in the same hyper-surface, or in the orthogonal direction. These directions are represented by blue and red arrows respectively in figure \ref{chap4:fig:clyndier_decompo}. This interpretation readily comes from the Fourier transform. Recall that in order to obtain a finite result, most of the gauge was fixed. The symmetry described here is the residual symmetry after gauge fixing and hence must be dropped from the mode expansion to reach convergence.

All the other modes however involve a change of the boundary shape, and contribute to the one-loop amplitude. This geometrical interpretation of the one-loop contribution is in agreement with Carlip picture of boundary modes as would be normal to the boundary diffeomorphisms whose action is broken by the presence of the boundary itself \cite{Carlip:2005zn}.

Another interesting feature of the one-loop result is that it diverges whenever there is a $k\in\{0,\dots,\frac12(N_x-1)\}$ such that $\gamma k \in 2\pi \mathbb Z$. This can only happen when the greatest common divisor between $N_x$ and $N_\gamma$ is strictly larger than one
\begin{equation}
K=\mathrm{GCD}(N_\gamma, N_x) > 1 \; .
\end{equation}
This condition is equivalent to $\gamma \in 2\pi \mathbb Q$ in the continuum. This condition will again arise in the Ponzano-Regge computation describes in the next chapters. For the moment, we can say that, geometrically, if $K>1$ then the homogeneous boundary structure is not rigid enough to provide a unique solution for the lengths of the bulk edges of the {\it linearized} equation of motions. These ambiguities show up as null modes of the (bulk) Hessian, leading to poles for the inverse of its determinant, which determines the one-loop correction. On the other hand, one finds that for a certain class of inhomogeneous perturbations of the boundary data one does not find any solution to the linearized equations of motion. Therefore, this situation rather describes the emergence of an accidental symmetry of the linearized theory due to the (homogeneous) boundary conditions, rather than the emergence of a new gauge symmetry. 
This indicates a breakdown of the linear approximation and thus the cases with $K>1$ need in principle a more refined analysis.

\section{Boundary Hilbert space and spin network}

We consider a manifold $\cM$ with boundary, and the group formulation of the Ponzano-Regge model introduced previously, encoding the space of flat connection. The space of boundary connection can be endowed with a Hilbert space structure \cite{Baez:1993id,Ashtekar:1994mh,Ashtekar:1994wa} using standard loop quantum gravity techniques. Consider now a particular orientated graph $\Gamma$ on the boundary\footnote{When working with the Ponzano-Regge model, we will require that the graph $\Gamma$ is compatible with the dual cellular decomposition of the boundary manifold, namely, $\pp \cK^{*} = \Gamma$.}. The Hilbert space is simply the space of gauge-invariant square integrable with respect to the Ashtekar-Lewandowski measure \cite{Ashtekar:1994wa}
\begin{equation}
	\cH = L^{2}(\SU(2)^{|\Gamma_1|})/ \SU(2)^{|\Gamma_0|} \;,
\end{equation}
where $|\Gamma_i|$ is the number of $i$-cells in $\Gamma$. That is, elements of this Hilbert space are functions whose parameters are $\SU(2)$ elements living on the links, and with a $\SU(2)$ gauge invariance at each node. Considering a element $\Psi$ of $\cH$, it respects the property
\begin{equation}
	\Psi(\{ g_{l} \}) = \Psi\left(G_{t_l} g_{l} G_{s_l}^{-1}\right)\in L^2(\SU(2)^{|\Gamma_1|})\qquad \forall G_n\in \SU(2),
\end{equation}
with, as previously $t_l$ and $s_l$ hold for the target and source nodes of $l$ respectively. The inner product associated to the Hilbert space is simply
\begin{equation}
	\la \Psi | \Phi \ra = \left[\prod_{l_\in\Gamma}\int_{\SU(2)} \dd g_{l}\right] \overline{\Psi(g_{l})} \; \Phi(g_{l}).
\end{equation}

A basis of such Hilbert space is given by the spin network functions \cite{Rovelli:1995ac}. Spin-networks were first introduced by Penrose \cite{penrose1971angular} as abstract graphs to describe discrete space-time. Considering a group $G$, they are defined as such \cite{Baez:1997zt}
\begin{Definition} 
	A spin network $\Psi$ is given by the data of a triplet $(\Gamma,j,\iota)$. $\Gamma$ is a two-dimensional orientated graph. We called l the links of $\Gamma$ and $n$ its nodes. To each link $l$ of $\Gamma$, we associate an irreducible representation of $G$, $V_{l}$ and to each vertex we associate an intertwiner $\iota_v$, which is a map 
	\begin{equation*}
		\iota_v \; : V_{l_1} \otimes V_{l_2} \otimes \; ... \; \otimes V_{l_n} \rightarrow V_{l'_1} \otimes V_{l'_2} \otimes \; ... \; \otimes V_{l'_m}
	\end{equation*}
	where $n$ (resp. $m$) is the number of incoming (resp. outgoing) links to (resp. from) $v$
\end{Definition}
Using the Peter-Weyl theorem, one can understood spin network define in the previous definition as $\SU(2)$ function as introduced in the beginning of this section.
\\

In the case of interest here, the gauge group is $\SU(2)$, and the irreducible representations are labelled by a spin $j \in \f{\N}{2}$. A $\SU(2)$ spin network function based on a graph $\Gamma$ takes the form
\begin{equation}
	\Psi_{j,\iota}( g_l) \,=\, \left(\bigotimes_{v\in \Gamma} \iota_v  \right) \bullet_\Gamma \left(\bigotimes_{l \in \Gamma}   \sqrt{d_{j_l}}  D^{j_l}(g_l)\right) \; ,
	\label{chap4:eq:SN_general}
\end{equation}
where $\bullet_\Gamma$ stands for the contraction of all the magnetic indices as prescribed by the graph $\Gamma$. Note again that a spin network function naturally comes with a choice of orientation. In fact, changing the orientation of one given link $l$ contributes with a factor $(-1)^{2 j_l}$ to the spin network function. Such a property is the reason why spin networks are often said to be unitary up to a sign. Intertwiners implement the state's gauge invariance thanks to their defining property
\begin{equation*}
	\Big(D^{j_1}(G)\otimes\dots \otimes D^{j_k}(G) \otimes D^{j_{k+1}}(G^{-1}) \otimes \dots \otimes D^{j_m}(G^{-1})\Big) \triangleright \iota_n = \iota_n \quad \forall G\in\SU(2).
\end{equation*}

Considering a manifold $\cM$ with a boundary, the spin network function provides us with an admissible boundary state, that is a kinematical state. The transition amplitude for the Ponzano-Regge model given the boundary state $\Psi$ supported on $\Gamma = \pp \cK^{*}$ is just its evaluation on the partition function of the model as defined previously
\begin{equation}
	\la Z_{PR}^{\cK}|\Psi\ra = \left[\prod_{l_\pp\in\Gamma }\int_{\SU(2)} \dd g_{l_\pp}\right]  Z_{PR}^{\cK}(g_{l \in \cK^{*}}) \Psi(g_{l \in \cK^{*}}).
	\label{chap4:eq:PR_amp_with_bdr_initial}
\end{equation}
This amplitude is the integral of the boundary state $\Psi$ projected by the Ponzano-Regge amplitude, in accordance to the previous interpretation of the Ponzano-Regge model being a projector on the space of physical solution. The interpretation of this amplitude is as follows: if the boundary of $\cK$ has two disconnected components corresponding to $\pp \cK_1$ and $\pp \cK_1$, then the previous formula calculates the transition amplitude between two states across the history represented by $\cK$ with weight provided by the Ponzano-Regge partition function. This is exactly the case described previously with the cylinder topology. On the other hand, if $\cK$ has a single boundary component, the amplitude can be interpreted as in the Hartle-Hawking no-boundary proposal \cite{Hartle:1983ai}: as the probability of nucleation of a given state from nothing.

\subsubsection{Geometric interpretation of  higher valent spin network states}

For a three-valent node, the intertwiner is uniquely given by the Clebsch-Gordan coefficient associated to the three adjacent representations $j_1,j_2,j_3$.  They are non-vanishing if and only if the three representations satisfy the triangle inequalities
\begin{equation}
	j_1 \leq j_2+j_3 \quad \text{and cyclic permutations},
\end{equation}
together with the extra integrability condition 
\begin{equation}
	j_1+j_2+j_3 \in \N\; .
\end{equation}
This hints at the fact that the spins $j_l$ can be interpreted as the lengths of the edges $e$ dual to the links $l$. This is confirmed by the construction of length-measuring operators associated to the edges of the boundary triangulation, which is indeed diagonalized by the spin network basis \cite{Rovelli:1995ac}.  

This interpretation can be generalized to higher valent nodes: an $m$-valent intertwiner\footnote{Spin-networks are also the boundary states of four-dimensional quantum gravity and can thus  describe quantum states of 3d geometry. From this perspective, intertwiners are naturally interpreted as {\it polyhedra} embedded in the flat 3d Euclidean space $\mathbb R^3$ \cite{Barbieri:1997ks,Freidel:2009ck,Freidel:2009nu,Bianchi:2010gc,Freidel:2010tt,Livine:2013tsa}. Quantum deforming $\su(2)$ allows to extend this geometrical interpretation to polyhedra in homogeneous curvature \cite{Dupuis:2013lka,Bonzom:2014wva,Haggard:2015ima}.} can be used to define a quantum $m$-sided polygon with fixed edge lengths determined by the spins $(j_1,\ldots,j_m)$ \cite{Livine:2013tsa}. The intertwiner space is not unique anymore, but still finite dimensional. Let us underline an issue with this polygonal interpretation of intertwiners, which is due to the possibility of different orderings of the edges around the polygon. As explained in \cite{Livine:2013tsa}, there are two possibilities. If we do not specify any ordering for the legs of the intertwiner, we can reconstruct multiple possible polygons. It is possible to recover a unique convex polygon, at least in the planar case, in which case the intertwiner contains enough data to deduce an ordering.  The problem is however automatically cured by considering graphs embedded in a surface, as we do here.

The non-uniqueness of the higher-valent intertwiners fits nicely with the fact that the geometry of non simplexes polygons with fixed edge length is also not unique. The intertwiner space describes also a possible bending of polygons: that is if we introduce diagonals in the polygon, there might be non-trivial dihedral angle hinging on such diagonals, see figure \ref{fig:four_intertwiners_to_two_three}.

This space of (possibly bent) polygons admits a canonical symplectic structure named after Kapovich and Millson \cite{Kapovich:1995,Kapovich:1996} which allows its quantization \cite{Livine:2013tsa,Conrady:2009px}. The Kapovich-Millson symplectic structure sets the length of each diagonal and the corresponding dihedral angle to be canonically conjugate variables. Thus, these two variables cannot be determined at the same time by a given intertwiner. This is compatible with the fact that in a boundary polygon the dihedral angle associated to a diagonal encodes some extrinsic curvature of the manifold, dual to the intrinsic metric determined by the length of the diagonal itself. We see that higher-valent nodes can encode quantum geometry, which features non-commutative aspects \cite{Freidel:2005bb,Freidel:2005me,Baratin:2010nn}.

In this thesis, we are mainly interested in boundary states which are four-valent. A 4-valent intertwiner can be decomposed into two 3-valent ones glued by a recoupling spin. Interpreting it as a quadrilateral, the decomposition into 3-valent intertwiners corresponds to cutting the quadrilateral into two triangles along one of its diagonals, see figure \ref{fig:four_intertwiners_to_two_three}. 
\begin{figure}[h!]
	\begin{center}
		\begin{tikzpicture}[scale=0.8]
		
		\coordinate (C) at (0.09,0.21);
		\coordinate (A) at (-1.68,0.32);
		\coordinate (B) at (0.17,1.81);
		\coordinate (D) at (1.88,0.11);
		\coordinate (E) at (-0.14,-1.74);
		
		\draw (-1.91,1.92) -- (1.97,1.72)--(1.75,-1.97)--(-1.41,-1.59)--cycle;
		
		\draw (C) node[above left]{$\iota_4$};
		\draw (C)--(A) node[midway,below]{$j_3$}; \draw (C)--(B) node[midway,left]{$j_1$}; \draw (C)--(D) node[midway,above]{$j_4$}; \draw (C)--(E) node[midway,right]{$j_2$}; 
		
		\draw (3,0) node{$\rightarrow$};
		
		\coordinate (I1) at (5.25,-0.40);
		\coordinate (I2) at (8.7,0.3); 
		\draw (7.75,-1.97)--(4.59,-1.59)--(4.09,1.92)--cycle; \draw (I1) node{$\bullet$}; \draw (5.30,0) node{$\iota^{1}_3$};
		\draw (6.09,1.92) -- (9.97,1.72)--(9.75,-1.97)--cycle; \draw (I2) node{$\bullet$}; \draw (I2) node[below right]{$\iota^{2}_3$};
		
		\draw (I1) -- (4.4,-0.35) node[midway,above]{$j_3$} ; \draw (I1) -- (6,-1.75) node[midway,right]{$j_2$};
		\draw (I1) to[bend left] node[pos=0.4,above]{$j$} node[pos=0.54,below]{$\theta$} (I2);
		\draw (I2)--(8.1,1.82) node[midway, right]{$j_1$}; \draw (I2)--(9.87,0.1) node[midway,above]{$j_4$};
		
		\end{tikzpicture}
	\end{center}
	\caption{Decomposition of a 4-valent intertwiner into two 3-valent ones along one of the diagonals. The recoupling spin associated to the length of the diagonal is denoted by $j$.  $\theta$ is the dihedral angle between the two triangles.}
	\label{fig:four_intertwiners_to_two_three}
\end{figure}
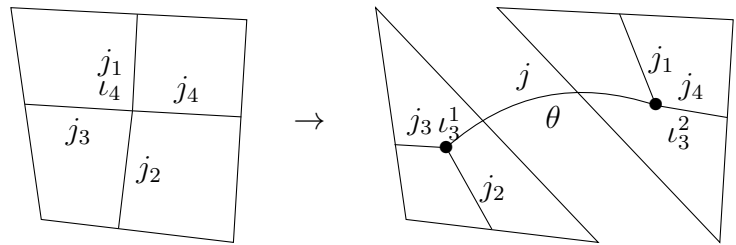
The recoupling spin associated to this diagonal corresponds to the length of the diagonal. The conjugate variable to this length according to the Kapovich-Millson symplectic structure is the dihedral angle between the two triangles hinged by the diagonal itself, as illustrated in figure \ref{fig:four_intertwiners_to_two_three}.  We refer to this angle as the extrinsic curvature of the quadrilateral. Since the length and the angle are canonically conjugate variables, a 4-valent intertwiner with fixed recoupling spin is totally spread in the extrinsic curvature. Once the edge lengths of the quadrilateral are fixed, notice that the length of one diagonal together with the related dihedral angle determine the length of the other diagonal. It follows that the two diagonals do have a non-trivial commutation relation, both in phase space and as quantum observables. Emerging from this arises the question of defining coherent intertwiners, that is intertwiners peaked on a given polygons. We will develop this notion later in this chapter.

\section{Coherent spin network}

Coherent spin network states are a particular representation of spin networks where the intertwiners are chosen to be the so-called coherent intertwiners. These coherent intertwiners describe, in the semi-classical limit, states that are peaked on a particular geometry. They were introduced in \cite{Livine:2007vk} in a three-dimensional context for loop quantum gravity. They are interpreted as quantum polyhedra \cite{Barbieri:1997ks,Freidel:2009ck,Bianchi:2010gc,Livine:2013tsa}. In particular, a four-valent coherent intertwiner is understood as a quantum tetrahedron\footnote{See also \cite{Haggard:2015ima,Charles:2016xzi} for a discussion on polyhedra in homogeneously curved space.}. In the context of three-dimensional gravity, where the boundary is two-dimensional, coherent intertwiners are understood as polygons. Three-valent intertwiners are thus triangles and four-valent, the case of interest here, quadrilateral. Coherent intertwiners are defined using $\SU(2)$ coherent states, which in turn have a nice formulation in terms of spinors.

We will start this section by a short presentation of the spinor formalism before focusing on the construction of coherent spin network. The coherent spin network will be used as boundary state in chapter \ref{chap6}.

\subsection{Spinors and coherent state}

Consider the representation space $V_j$ of a given spin $j$, and choose a basis $\{|j,m\ra\,,\; m=-j,\ldots,+j\}$, which diagonalizes the angular momentum operator $J_z$, in addition to the $\SU(2)$ Casimir $J^2=J_x^2+J_y^2+J_z^2$. As usual, we can think of a basis element $|j,m\ra \in V_j$ as a quantum vector of length $j$ and $z$-projection $m$, the eigenvalue of $J_z$. On the other hand, the azimuthal direction of these quantum vectors is totally uncertain.

The case of $|j,j\ra$ is particularly interesting: it has maximal $z$-projection and is therefore peaked along this direction. In other words, $|j,j\ra$ is a coherent state in $V_j$ representing a vector of length $j$ pointing in the $z$-direction. This is the starting point of the Perelomov coherent states. The (relative) uncertainty about the direction in which the vector is pointing decreases with $j$. Indeed, computing the expectation values of the $\su(2)$ generators on the state $|j,j\ra$, we get $\la \vec{J}\ra=(0,0,j)$. We can also compute the variance $\la \vec{J}^{2}\ra=j(j+1)$, which is simply given by the $\su(2)$ Casimir. The state $|j,j\ra$ corresponds to a semi-classical vector of length $j$ in the $z$-direction,  peaked on $(0,0,j)$ with spread 
\begin{equation}
\f1{\la \vec{J}\ra}\sqrt{\la \vec{J}^{2}\ra-\la \vec{J}\ra^{2}}\,\sim\f1{\sqrt{j}} \; .
\end{equation}
The corresponding polar angle $\theta$ can be estimated to be $\theta \approx \frac{1}{\sqrt j}$ which goes to $0$ with the spin increasing.

To obtain coherent states \`a la Perelomov representing unit three-vectors pointing in an arbitrary direction $\hat n\in S_2$, we rotate $|j,j\ra$ appropriately. Consider the state
\begin{equation}
	|j, \hat n\ra = D^j(G_{\hat n})|j,j\ra
	\label{chap5:eq:perelemov_coherent}
\end{equation}
where $G_{\hat n}$ is some element of $\SU(2)$ which in the vectorial (spin 1) representation corresponds to a rotation $R_{\hat n}$, taking the $z$-axis to the $\hat n$ direction. Of course, the element $G_{\hat n}$ is not uniquely defined, and admissible group elements only differ by a rotation around $\hat z$. This freedom on the choice of element $G_{\hat n}$ translates into a choice of phase  for the vectors $|j,\hat n\ra$. We can keep track of this phase by going to the spinorial representation. First focus on the case $j = \f{1}{2}$.

Since $V_{\frac12} \cong \mathbb C^2$, it is natural to introduce spinors
\begin{equation}
	|w\ra = \mat{c}{w^0 \\ w^1} = w^0 |\uparrow\ra + w^1 |\downarrow\ra \in V_{\frac12},
\end{equation}
where $|\uparrow\ra$ ($|\downarrow\ra$) is the $V_\frac{1}{2}$ basis element with $m=+\tfrac12$ $\left(-\tfrac12\right)$. Denote by
\begin{equation}
	\la w | = \mat{cc}{\bar w^0 & \bar w^1} = \mat{c}{w^0 \\ w^1}^\dagger
	\qquad\text{and}\qquad
	| w ] = \varsigma |w\ra = \mat{c}{- \bar w^1 \\ \bar w^0} \;
\end{equation}
where $\varsigma$ is the $\SU(2)$ structure map, which is anti-unitary 
\begin{equation}
	\varsigma^2 = -1 
	\qquad\text{and}\qquad
	G \varsigma |w\ra = \varsigma G|w\ra \quad \forall G\in\SU(2),\; w\in \mathbb C^2.
	\label{chap5:eq:structure_map_prop}
\end{equation}

We can associate a vector $\vec{n} \in \R^{3}$ to a spinor $| w \ra$ by evaluating it on the Pauli matrices
\begin{equation}
	\vec{n} = \la w | \vec{\sigma} | w \ra \; ,
\end{equation}
where $\vec{\sigma} = (\sigma_1,\sigma_2,\sigma_3)$. It is immediate to see that if $|w \ra$ is associated to the vector $\vec{n}$ then $| w ]$ is associated to $-\vec{n}$ using the property of anti-linearity for the structure map. If the spinor is normalized, then the associated vector $\hat{n}$ is also normalized and we denote normalized spinor by $|\xi\ra$. The spinor $|\xi\ra$ corresponds to a coherent state along the direction $\hat{n}$. This time, the definition is not up to a phase, since the phase is entirely captured by the spinorial representation.

Generalization to higher spins is immediate using the property that $V_j$ can be decomposed into the tensor product of $2j$ times the fundamental representation $V_{\f{1}{2}}$
\begin{equation*}
	|j,j \ra = |\up\ra^{\otimes 2j} \; .
\end{equation*}
This leads to the definition of a coherent state in the spin $j$ representation
\begin{equation}
	|j,\xi\ra = |\xi\ra^{\otimes 2j} \; .
	\label{chap4:eq:spinj_to_spin12}
\end{equation}

These states are the building blocks of the coherent intertwiners.

\subsection{From coherent state to coherent spin network}

We can now proceed to the construction of the coherent intertwiner which will encode polygons embedded in $\mathbb{R}^3$.  Polygons with $m$ edges can be described by a set of $m\geq3$ vectors $\{\vec{v}_i\}_{i=1}^m$ under the closure constraint
\begin{equation*}
	\sum_{i=1}^{m} \vec{v}_i = \vec{0} \; .
\end{equation*}
These vectors are interpreted as edge vectors. Assuming the flatness of the polygon, these vectors span in general the plane $\R^{2}$ and their set is of course an over-complete basis of the plane. Note that this construction allows the polygon to be degenerate. In that case, the edge vectors only span $\R$. This construction also naturally comes with a choice of orientation. It is immediate to see that taking the set of vectors to be defined modulo global rotations is enough to lose track of the orientation. Similarly, a polygon can be defined by the data of its normal vector up to the same closure constraint.

This classical construction is really similar to the construction of the coherent state previously presented. In fact, coherent intertwiners are defined in a similar way \cite{Livine:2007vk}. Consider a set of $m$ coherent states described by a spin $j_i$ and a spinor $\xi_{i}$. We define the $m$-valent coherent intertwiner with only outgoing links by
\begin{equation}
	\iota_{\text{coh}} = \int_{\SU(2)} \dd G \; G \act \left( |j_1,\xi_1 \ra \otimes \; ... \; \otimes | j_m,\xi_m \ra \right)
\end{equation}
The $\SU(2)$ integration is similar in spirit to taking the set of classical vectors up to rotations. Interestingly enough, it turns out that coherent intertwiners built out of a set of closing vectors are enough to provide an (over-)complete basis of the intertwiner space \cite{Conrady:2009px,Freidel:2009nu,Livine:2013tsa}. Hence, we can safely restrict the definition of the spin network state to coherent spin network state without loss of generality.

They described a quantum polygon with {\it normal} edge vectors $\hat{n}_i = \la \xi_{i}| \vec{\sigma}|\xi_{i} \ra$ and edge length $j_i$ under the closure constraint
\begin{equation}
	\sum_{i=1}^{m} j_i \hat{n}_i = 0
\end{equation}
which is equivalent to the classical condition for a polygon defined by its normal vectors. This closure constraint is necessary in the semi-classical limit for the intertwiners not to be exponentially suppressed \cite{Livine:2007vk}.
\\

Having introduced coherent intertwiners, we can now construct the associated coherent spin network state \cite{Dupuis:2011fz}, which is a spin network where all the intertwiners are coherent. When doing so, the question of orientation of the links arises. Indeed, spin networks are defined on an orientated graph, and the definition of a coherent intertwiner given above is for a node with only outgoing links. Clearly, if the link is outgoing at a node, namely its source node, it will be ingoing for the node at its other extremity, namely the target node. We again denote by $s_l$ (resp. $t_l$) the ingoing (resp. outgoing) node $n$ with respect to the link $l$. For every ingoing link $l$ to the node $n$, the associated coherent state of the intertwiner at the node $n$ is
\begin{equation*}
	(\varsigma |\xi_{t_l} \rangle )^\dagger = [ \xi_{t_l}| \; .
\end{equation*}
The action of the structure map $\varsigma$ on the coherent state produces a rotation of $\pi$ with direction orthogonal to the vector associated to the spinor. This corresponds to the geometrical picture of seeing the link from the opposite direction. Finally, we point out that the action of $G \in \SU(2)$ on $[ \xi_{t_l}|$ is
\begin{equation*}
	G\act [ \xi_{t(l)}| = [ \xi_{t_l}| G^{-1} \;.
\end{equation*}
\\

We are now ready to write equation \eqref{chap4:eq:SN_general} in an explicit form given that it is a coherent spin network function. Consider a particular link $l$, with spin $j_l$ and $\SU(2)$ element $g_l$. The target and source nodes of $l$ are described by coherent intertwiners. Called $G_{s_l}, \xi_{s_l}$ and $G_{t_l}, \xi_{t_l}$ the associated $\SU(2)$ element and spinor of the intertwiners. The contribution of the source node to the spin network function is (forgetting about the $\SU(2)$ integration temporarily)
\begin{equation*}
	G_{s_l} |j_l,\xi_{s_{l}} \ra
\end{equation*}
whereas it is 
\begin{equation*}
	[j_l,\xi_{t_l} |G_{t_l}^{-1}
\end{equation*}
for the target node. Gluing everything together taken into account the parallel transport between the nodes encoding by the group element $g_l$, we obtain the link contribution to the spin network function
\begin{equation}
	[\xi_{t_l} |G_{t_l}^{-1} g_l G_{s_l} |\xi_{s_{l}} \ra^{2 j_l} \; ,
\end{equation}
where we used property \eqref{chap4:eq:spinj_to_spin12} to express it in terms of spinor in the fundamental representation. 

That is, a coherent spin network takes the form\footnote{These states are not normalized. This could be corrected by inserting factors of $\sqrt{d_j}$ for every link as well as a normalization factor for the coherent intertwiners $\iota_{(j_a,\xi_a)}$. Such normalization is not needed for our purpose, and we will ignore it in the following.}
\begin{equation}
	\Psi_{(j,\xi)} (g_l) = \left[ \prod_n \int_{\SU(2)} \dd G_n\right] \prod_l [\xi_{t_l} |G_{t_l}^{-1} g_l G_{s_l} |\xi_{s_{l}} \ra^{2 j_l} \; .
	\label{chap4:eq:general_coherent_SN}
\end{equation}

For the record, coherent states play an important role in four-dimensional spin foam models, where they are known to impose boundary conditions reproducing the GHY boundary term in a large spin limit \cite{Conrady:2008mk,Barrett:2009gg,Dowdall:2009eg,Han:2011re,Han:2011rf,Haggard:2014xoa}.

In this form, it is immediate to check the previous claim about the dependency on the orientation of the spin network. Indeed, consider the switch of the orientation of one link $l$. This modification is local, and only affects the link $l$. In that case, the role of the source and target node is reverse and the link contribution is
\begin{align*}
	[\xi_{s_l} |G_{s_l}^{-1} g_l G_{t_l} |\xi_{t_{l}} \ra^{2 j_l} 
	&= [\xi_{s_l} | G_{s_l}^{-1} g_l G_{t_l} \; \xi_{t_{l}} \ra^{2 j_l} 
	\\
	&= (-1)^{2j} [\xi_{t_{l}}| G_{t_l}^{-1} g_l^{-1} G_{s_l} |\xi_{s_l}  \ra^{2 j_l} 
\end{align*}
where we used the property $[w,\eta\ra = - [\eta,w\ra$. That is, changing the orientation of a given link $l$ does not leave the spin network invariant, but contributes with a factor $(-1)^{2 j_l}$ as previously claimed. This is the reason why spin network functions are sometime called unitary "up to a sign". We will see in the following that for the boundary state we consider, this sign uncertainty does not produce any complication in the well-definiteness of our amplitude. 

\subsubsection{Geometrical meaning of the link contribution}

Let us focus more on the contribution of each link for the coherent spin network
\begin{equation}
	[\xi_{t_l} |G_{t_l}^{-1} g_l G_{s_l} |\xi_{s_{l}} \ra^{2 j_l}  \equiv [j_l,\xi_{t_l} |G_{t_l}^{-1} g_l G_{s_l} |j_l,\xi_{s_{l}} \ra \; .
	\label{chap4:eq:link_contribution}
\end{equation}

As we explained previously, the spinors encode the normal vectors and the length of the quantum polygon. The action of the $\SU(2)$ element $G_n$ at each node is to rotate these vectors in space. This rotation has two important features: it is common to all the spinors at one node, thus corresponding to a rotation of the whole quantum polygon, and it is finally integrated over in the state. The integration over $G_n$ has the role of implementing gauge invariance of $\Psi$ by removing any reference to the standard frame mentioned above. 

The remaining element we need to discuss is the $SU(2)$ element $g_l$ associated to the link. This is given by the definition of the boundary state in the Ponzano-Regge model. Recall that $g_l$ encodes the discrete connection between two nodes of the graph. That is, it represents the parallel transport from one node to another. To be more precise, in the coherent spin network state, $g_l$ parallel transports spinors, i.e. edges of the quantum polyhedron, to another spinor. The link contribution \eqref{chap4:eq:link_contribution} calculates how much ``superposition'' there is between the quantum edges $G_{s_l}|j_l,\xi_{s_l}\ra$ and $G_{t_l}|j_l, \xi_{t_l}\ra$ once parallel-transported by $g_l$.

From the point of view of the Ponzano-Regge model, recall that both $G_n$ and $g_l$ are integrated over under the flatness constraint. One can then ask at which value these integrals happen to concentrate. Intuitively, one expects these integrals to concentrate precisely where the superposition is maximal, that is at the values of the group variables which would induce (provided it exists) a consistent gluing among {\it all} the edges of {\it all} the polygons in the discretization. Here, maximal does not mean a perfect superposition. We will see in the next chapter that indeed, this superposition is not perfect and involves a phase. Geoemtrically, this phase encodes the dihedral angle between the quantum polyhedron. In other words, the left-over phase encodes the extrinsic curvature.

\subsubsection{Symmetry of the integrand}

The integrand of \eqref{chap4:eq:general_coherent_SN} has a interesting symmetry with respect to the $G_n$. Consider the local transformation at the node $n$
\begin{equation}
	G_{n} \rightarrow - G_{n} \; .
	\label{chap6:eq:intertwiner_coherent_Z2_sym}
\end{equation}
Under this transformation, the integrand becomes
\begin{equation}
	\prod_l [ \xi^{t(l)}_l| G_{t(l)}^{-1} g_l G_{s(l)} | \xi^{s(l)}_l\ra^{2j_l} \rightarrow (-1)^{2\sum j_{l}} \prod_l [ \xi^{t(l)}_l| G_{t(l)}^{-1} g_l G_{s(l)} | \xi^{s(l)}_l\ra^{2j_l}
\end{equation}
where the sums are over all the links connecting to the node $n$. Indeed, the $\SU(2)$ element $G_n$ appears as many times as the number of links connected to the node $n$. Now, recall that a necessary condition for a non-trivial intertwiner to exist at a given node $n$ is that the sums of the spins of the links around it is an integer. That implies that the integrand is actually invariant under the transformation given by \eqref{chap6:eq:intertwiner_coherent_Z2_sym}. This symmetry will be of the utmost importance for the well-definiteness of our boundary state in the following.

\section{Torus topology and discretization}

From the first chapter, we have learned that the torus topology is of interest for gravity, since it corresponds to the geometry of Euclidean AdS and BTZ. Moreover, we already discussed the case of the ball in chapter \ref{chap3} for the Ponzano-Regge model. In this section, we start considering the next topology after the trivial one: the torus topology. In this section, we will introduce the considered discretization of the torus and perform a first computation of the partition function in a very simple case. The starting point of the computation is given by the definition of the Ponzano-Regge amplitude for a boundary state $\Psi$ \eqref{chap4:eq:PR_amp_with_bdr_initial}
\begin{equation}
	\la Z_{PR}^{\cK}|\Psi\ra = \left[\prod_{l_\pp\in\Gamma }\int_{\SU(2)} \dd g_{l_\pp}\right]  Z_{PR}^{\cK}(g_{l \in \cK^{*}}) \Psi(g_{l \in \cK^{*}}).
\end{equation}
We recall that $\cK$ is the cellular decomposition of the manifold, and $Z_{PR}^{\cK}$ the partition function of the Ponzano-Regge model, imposing the flatness of the bulk, given by the definition \ref{chap3:def:PR_amplitude}.

\subsection{Torus, cellular decomposition and gauge fixing}

We start our investigation of the torus topology by considering the three-dimensional cylinder with a two-dimensional torus as boundary. To apply the Ponzano-Regge model on such manifold, we need to introduce a cellular decomposition. Consider a cylinder with $N_t$ verticals slices and $N_x$ horizontal pie slices, see figure \ref{chap4:fig:cylinder_discretization}. Comparing to the triangulation introduced for Regge calculus from the previous chapter, the only difference comes from the fact that we did not sub-divide the resulting prism into three tetrahedra.
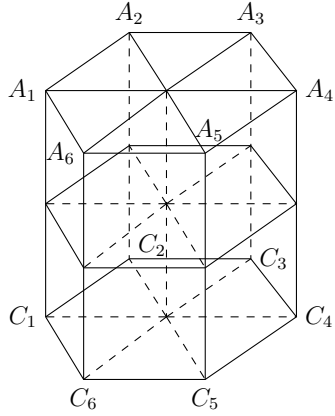
\begin{figure}[!htb]
	\begin{center}
		\begin{tikzpicture}[scale=1]
		%figure 1 : cylinder
		%coordinate
		%hexagone A
		\coordinate (OA) at (1.59,0);
		\coordinate (A1) at (0,0);
		\coordinate (A2) at (1.1,0.77);
		\coordinate (A3) at (2.7,0.77);
		\coordinate (A4) at (3.3,0);
		\coordinate (A5) at (2.1,-0.83);
		\coordinate (A6) at (0.5,-0.83);
		
		%hexagone B
		\coordinate (OB) at (1.59,-1.5);
		\coordinate (B1) at (0,-1.5);
		\coordinate (B2) at (1.1,-0.73);
		\coordinate (B3) at (2.7,-0.73);
		\coordinate (B4) at (3.3,-1.5);
		\coordinate (B5) at (2.1,-2.35);
		\coordinate (B6) at (0.5,-2.35);
		
		%hexagone A
		\coordinate (OC) at (1.59,-3);
		\coordinate (C1) at (0,-3);
		\coordinate (C2) at (1.1,-2.23);
		\coordinate (C3) at (2.7,-2.23);
		\coordinate (C4) at (3.3,-3);
		\coordinate (C5) at (2.1,-3.83);
		\coordinate (C6) at (0.5,-3.83);
		
		%drawing
		%hexagone A	 
		\draw (A1) -- (A2) -- (A3) -- (A4) -- (A5) -- (A6) -- cycle; 
		\draw (OA) -- (A1) ; \draw (OA) -- (A2); \draw (OA) -- (A3); \draw (OA)--(A4); \draw (OA) --(A5); \draw (OA)-- (A6); 
		\draw (A1) node[scale=0.8,left] {$A_1$}; \draw (A2) node[scale=0.8,above] {$A_2$}; \draw (A3) node[scale=0.8,above] {$A_3$}; \draw (A4) node[scale=0.8,right] {$A_4$}; \draw (2.15,-0.78) node[scale=0.8,above]{$A_5$}; \draw (A6) node[scale=0.8,left]{$A_6$};
		
		%hexagone B
		\draw (B1) -- (B2) -- (B3) -- (B4) -- (B5) -- (B6) -- cycle;
		\draw [dashed] (OB) -- (B1); \draw [dashed] (OB) -- (B2); \draw [dashed] (OB) -- (B3); \draw [dashed] (OB)--(B4); \draw [dashed] (OB) --(B5); \draw [dashed] (OB)-- (B6);
		
		%hexagone C
		\draw (C1) -- (C2) -- (C3) -- (C4) -- (C5) -- (C6) -- cycle;
		\draw [dashed] (OC) -- (C1); \draw [dashed] (OC) -- (C2); \draw [dashed] (OC) -- (C3); \draw [dashed] (OC)--(C4); \draw [dashed] (OC) --(C5); \draw [dashed] (OC)-- (C6);
		\draw (C1) node[scale=0.8,left] {$C_1$}; \draw (1.1,-2.05) node[scale=0.8,right] {$C_2$}; \draw (C3) node[scale=0.8,right] {$C_3$}; \draw (C4) node[scale=0.8,right] {$C_4$}; \draw (C5) node[scale=0.8,below]{$C_5$}; \draw (C6) node[scale=0.8,below]{$C_6$};
		
		%link between hexagone
		\draw (A1)--(B1)--(C1); \draw [dashed] (A2)--(B2)--(C2); \draw [dashed] (A3)--(B3)--(C3); \draw (A4)--(B4)--(C4); \draw (A5)--(B5)--(C5); \draw (A6)--(B6)--(C6);
		\draw [dashed] (OA)--(OB)--(OC); 
		
		%figure 2: dual

		\end{tikzpicture}
	\end{center}
	\caption{Example of discretization of the torus with $N_x=6$ and $N_t=2$. The twist parameter $N_\gamma$ is recovered by identifying the $A$'s with the $B$'s  to obtain the torus through $A_{i}=C_{i+N_\gamma}$.}
	\label{chap4:fig:cylinder_discretization}
\end{figure}

To recover the boundary torus, we need to identify the top and the bottom of the cylinder. Before such identification, as we previously mentioned, we have the freedom to consider a twist. For example, in figure \ref{chap4:fig:cylinder_discretization}, we identify $A_1$ to $C_1+2$, for a shift of $2$. In general, this shift is parametrized by a parameter $N_\gamma$. From the continuum point of view, i.e. in thermal AdS, Euclidean BTZ or Euclidean flat space, this corresponds to a twist $\gamma$ in the identification of the Euclidean time parametrized by
\begin{equation}
	\gamma = 2 \pi \f{N_\gamma}{N_x} \; .
\end{equation}
For future utility, we introduced the parameter $K$ to be the greatest common divisor of $N_\gamma$ and $N_x$ and $\W$
\begin{equation}
	K = GCD(N_\gamma,N_x) \; , \quad W = \f{N_x}{K} \; .
\end{equation}
The parameter $K$ basically counts the number of independent vertical loops. Indeed, if $K=1$ it is immediate to see that, for a loop to close, it must go through every vertical link, giving only one vertical loop. This logic can easily be applied for any $K$ (bounded by $N_x$). The case $K=0$ is equivalent to the case $K=N_x$, giving as expected $N_x$ closed vertical loops. On the other hand, $\W$ is the winding number of the closed loop, that is the number of times each closed loop winds around the torus before closing.

Such a choice of cellular decomposition naturally induces a particular boundary cellular decomposition $\pp \cK$. It is immediate to see that the boundary discretization is that of a square lattice with periodic conditions. Denoting the vertices of the discretization by their positions $(t,x)$ with $t \in [0,N_t-1]$ and $x \in [0,N_x-1]$ the periodic conditions are given by
\begin{equation}
	(t, x+N_x) \sim (t,x) \sim (t+N_t , x +N_\gamma).
	\label{chap4:eq:periodic_condition}
\end{equation}

The spin network graph $\Gamma$ dual to the previous square lattice is also a square lattice. Since we are mainly working on $\Gamma$, we also denote the nodes of the dual boundary graph by $(t,x)$. The vertical and horizontal links of the dual graph, labelled by subscripts $v$ and $h$, are dual to the {\it space} and {\it time} edges respectively. To fix the notation for the discretization and the dual graph, we consider the top-left node of the boundary square lattice as the origin $(t,x)=(0,0)$. Vertical links between nodes  $(t,x)$ and $(t+1,x)$ are labelled by $(t,x)$ and similarly for the horizontal ones. Vertical and horizontal links are oriented in the direction of growing time and space coordinates. See figure \ref{chap4:fig:induced_bdry_square} for the details.
\begin{figure}[h!]
	\begin{center}
		\begin{tikzpicture}[scale=1.3]
		\foreach \i in {0,...,6}{
			\foreach \j in {-1,...,3}{
				\draw (\i,\j) node {$\bullet$};
				\draw[<-] (\i,\j-.5) --(\i,\j+.5);
				\draw[->] (\i-.5,\j) --(\i+.5,\j);
			}
		}
		
		\foreach \i in {0,...,4}{
			\draw[rounded corners=3 pt,->] (\i,4.5+0.3) --(\i,4+0.1+0.3)-- (\i+2,4-0.1+0.3)--(\i+2,3.5+0.3)   ;
		}
		
		\draw (0,-1.7) node{$0$};
		\draw (2,-1.72) node{$N_\gamma$};
		\draw (3,-1.7) node{$x$};
		\draw (6,-1.7) node{$N_x$-1};
		
		\draw(-.9,-1) node{$N_t$-1};
		\draw(-0.7,1) node{$t$};
		\draw(-0.7,3) node{$0$};
		
		\draw(3.5,1) node[below]{$g^{h}_{t,x}$}; \draw (2.5,1) node[above]{$T_{t,x-1}$};
		\draw(3,0.5) node[left]{$g^{v}_{t,x}$}; \draw (3,1.5) node[right]{$L_{t-1,x}$};
		
		\end{tikzpicture}
	\end{center}
	\caption{Oriented square lattice on the twisted torus. The twist angle is $\gamma=2\pi \f{N_{\gamma}}{N_{x}}$. Starting from the vertex $(t,x)$, the edge on the right is associated to $g^{h}_{t,x}$ and the edge below to $g^{v}_{t,x}$. The horizontal periodic condition is without twist.}
	\label{chap4:fig:induced_bdry_square}  
\end{figure}
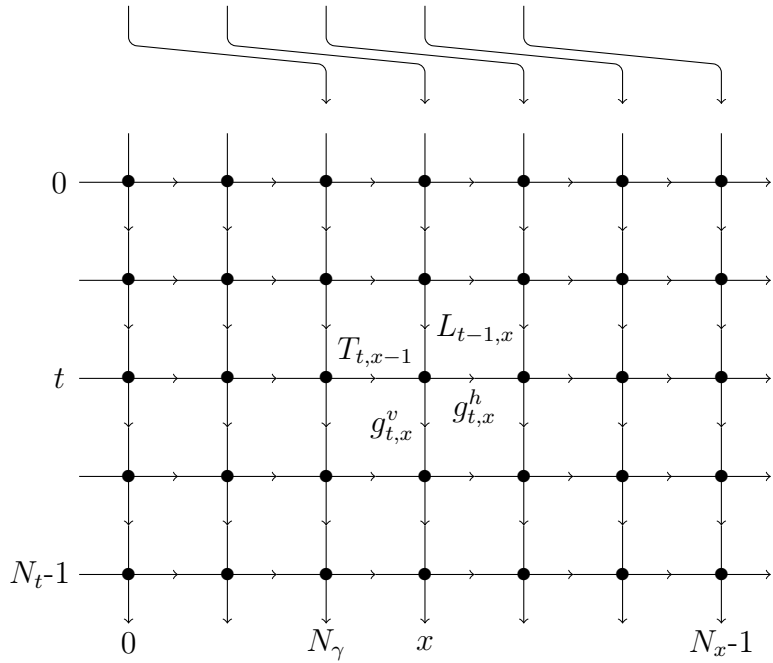

To obtain a spin network state on this graph, we associate to each link a $\SU(2)$ element $g^{h,v}_{t,x}$ and a spin $j_{t,x}^{h,v}$ where $(h,v)$ holds for either a horizontal or a vertical link. To each vertex we associate a coherent intertwiner. A generic spin network on such a discretization is denoted by
\begin{equation}
	\Psi = \Psi( g^h_{t,x}, g^v_{t,x} ).
\end{equation}
To make the notation lighter, we keep silent about the presence of the spins and intertwiners as parameters in the definition of the spin network state.

In this whole thesis, we are interested in a homogeneous choice of intertwiners on the whole boundary lattice. That is, we assign to each node the exact same intertwiner. To construct a generic intertwiner respecting this property, the choice of orientation of the boundary graph is really constrained: we need all the vertical and horizontal links to have the same orientation. This constraint implies that there are only four possible orientations for the boundary graph, see figure \ref{chap5:fig:four_possible_orientation}. The key point is that the intertwiners defined on these four different orientations are in fact the same, hence the spin networks on these four orientations are the same. This can be seen by applying the $\SU(2)$ invariance of the integrand under $G_{n} \rightarrow -G_{n}$ of the coherent spin network to absorb all the $(-1)^{2j_l}$ coming from the switch of the orientation.
\begin{figure}[!htb]
	\centering
	\begin{tikzpicture}[scale=1.5]
		\coordinate (A) at (-4.5,0);
		\coordinate (B) at (-1.5,0);
		\coordinate (C) at (1.5,0);
		\coordinate (D) at (4.5,0);
		
		\draw[-<-=0.5,red] (A)--+(-1,0); \draw[->-=0.5,blue] (A)--+(1,0); \draw[->-=0.5,blue] (A)--+(0,1); \draw[-<-=0.5,red] (A)--+(0,-1);
		\draw[->-=0.5,blue] (B)--+(-1,0); \draw[-<-=0.5,red] (B)--+(1,0); \draw[->-=0.5,blue] (B)--+(0,1); \draw[-<-=0.5,red] (B)--+(0,-1);
		\draw[-<-=0.5,red] (C)--+(-1,0); \draw[->-=0.5,blue] (C)--+(1,0); \draw[-<-=0.5,red] (C)--+(0,1); \draw[->-=0.5,blue] (C)--+(0,-1);
		\draw[->-=0.5,blue] (D)--+(-1,0); \draw[-<-=0.5,red] (D)--+(1,0); \draw[-<-=0.5,red] (D)--+(0,1); \draw[->-=0.5,blue] (D)--+(0,-1);
	\end{tikzpicture}
	\caption{Four possible orientations to have a homogeneous choice of intertwiners on the whole boundary square lattice. One vertical (resp. horizontal) link must be ingoing (red) while the other one must be outgoing (blue).}
	\label{chap5:fig:four_possible_orientation}
\end{figure}
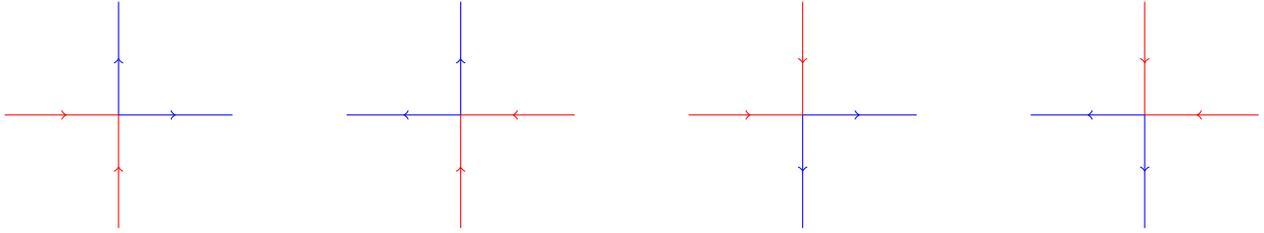
We will not detail the actual construction of such an intertwiner here, and leave this task for the next chapter. 

In fact, the rational reason behind the use of a quadrangulation rather than a triangulation (as in \cite{Dowdall:2009eg}) is that it allows us to perform calculations explicitly. We will be able to explicitly evaluate the one loop amplitude considering coherent spin network state on the boundary (see next section for a definition). Moreover, thanks to the use of a quadrangulation, considering a slightly  more general boundary state, constructed as the superposition of coherent spin network state, we will be able to exactly compute, for the first time, a quasi-local amplitude for three-dimensional gravity on the torus.

\subsubsection{Gauge fixing of the Ponzano-Regge amplitude on the torus}

From this cellular decomposition, it is straightforward to write the formal Ponzano-Regge amplitude in the group representation as given by the definition \ref{chap3:def:PR_model_without_gauge_fixing}. As we explained previously however, it is necessary to gauge-fix to obtain a finite expression. This is the object of this paragraph.

Recall that the gauge-fixing procedure require two trees. The first one is an internal maximal tree $T$ of the cellular decomposition $\cK$ touching the boundary at one vertex. Say, start at $(t,x)=(0,0)$ then follow a radial vertical edge toward the bulk. By doing so, we are at the "centre" of the cylinder. Then follows the $(N_t-1)$ central edges in the horizontal direction, see figure \ref{chap4:gauge_fixing_T}. The regularization proceeds by removing for every edge $e\in T$ the corresponding delta function on the dual face $f$.
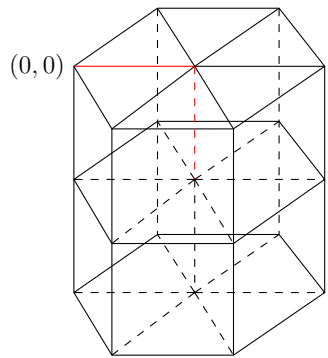
\begin{figure}[!htb]
	\begin{center}
		\begin{tikzpicture}[scale=1]
		%figure 1 : cylinder
		%coordinate
		%hexagone A
		\coordinate (OA) at (1.59,0);
		\coordinate (A1) at (0,0);
		\coordinate (A2) at (1.1,0.77);
		\coordinate (A3) at (2.7,0.77);
		\coordinate (A4) at (3.3,0);
		\coordinate (A5) at (2.1,-0.83);
		\coordinate (A6) at (0.5,-0.83);
		
		%hexagone B
		\coordinate (OB) at (1.59,-1.5);
		\coordinate (B1) at (0,-1.5);
		\coordinate (B2) at (1.1,-0.73);
		\coordinate (B3) at (2.7,-0.73);
		\coordinate (B4) at (3.3,-1.5);
		\coordinate (B5) at (2.1,-2.35);
		\coordinate (B6) at (0.5,-2.35);
		
		%hexagone A
		\coordinate (OC) at (1.59,-3);
		\coordinate (C1) at (0,-3);
		\coordinate (C2) at (1.1,-2.23);
		\coordinate (C3) at (2.7,-2.23);
		\coordinate (C4) at (3.3,-3);
		\coordinate (C5) at (2.1,-3.83);
		\coordinate (C6) at (0.5,-3.83);
		
		%drawing
		%hexagone A	 
		\draw (A1) -- (A2) -- (A3) -- (A4) -- (A5) -- (A6) -- cycle; 
		\draw[red] (OA) -- (A1) ; \draw (OA) -- (A2); \draw (OA) -- (A3); \draw (OA)--(A4); \draw (OA) --(A5); \draw (OA)-- (A6); 
		\draw (A1) node[scale=0.8,left] {$(0,0)$};
		
		%hexagone B
		\draw (B1) -- (B2) -- (B3) -- (B4) -- (B5) -- (B6) -- cycle;
		\draw [dashed] (OB) -- (B1); \draw [dashed] (OB) -- (B2); \draw [dashed] (OB) -- (B3); \draw [dashed] (OB)--(B4); \draw [dashed] (OB) --(B5); \draw [dashed] (OB)-- (B6);
		
		%hexagone C
		\draw (C1) -- (C2) -- (C3) -- (C4) -- (C5) -- (C6) -- cycle;
		\draw [dashed] (OC) -- (C1); \draw [dashed] (OC) -- (C2); \draw [dashed] (OC) -- (C3); \draw [dashed] (OC)--(C4); \draw [dashed] (OC) --(C5); \draw [dashed] (OC)-- (C6);
		
		%link between hexagone
		\draw (A1)--(B1)--(C1); \draw [dashed] (A2)--(B2)--(C2); \draw [dashed] (A3)--(B3)--(C3); \draw (A4)--(B4)--(C4); \draw (A5)--(B5)--(C5); \draw (A6)--(B6)--(C6);
		\draw [dashed,red] (OA)--(OB); \draw[dashed] (OB)--(OC);

		\end{tikzpicture}
	\end{center}
	\caption{Internal maximal tree $T$ in red for $N_x=6$ and $N_t=2$. It starts at the boundary vertex $(0,0)$ then follows the unique radial edge at $(0,0)$ toward the bulk. It then moves to the next temporal slice $N_t=1$ through the dashed red edge. The tree $T$ stops here in this example since it will produce a loop otherwise.}
	\label{chap4:gauge_fixing_T}
\end{figure}

For the maximal tree $T^{*}$ of $\cK^{*} \cup \pp \cK^{*}$ we also start at the initial node $(t,x)=(0,0)$. Then, we go along all the links of the boundary discretization on the slice $t=0$, except the one between the nodes $(0,0)$ and $(0,N_x-1)$. Then, we move to the slice $t=1$, and do the same. Once arriving at the slice $(t=N_t-1)$, we move along the radial link in the bulk, and repeat the same operation in the reverse order, ensuring that $T^{*}$ is maximal.

Considering these two trees, it is immediate to see that the Ponzano-Regge amplitude defined in \ref{chap3:def:PR_amplitude} becomes
\begin{equation}
	\la Z_{PR}^{\cK}| \Psi \ra = \int_{\SU(2)} \dd g \; \Phi(g^{h}_{t,x} = \id,g^{v}_{t \neq N_t-1,x} = \id, g^{v}_{N_t-1,x} = g) \; .
\end{equation}
There is still one non trivial integration over $\SU(2)$. This was to be expected. Contrary to the ball, the case of the torus is not topological trivial, and it exists one class of non-contractible cycle, represented by this remaining integration. This is basically the last, and only information coming from the bulk. This expression can be further simplified by using the remaining global $\SU(2)$ invariance to fix $g$ along a specific direction. Consider the following parametrization for $g$ in $\SU(2)$
\begin{equation}
	g = \E^{\I \varphi \hat{u}.\vec{\sigma}} \; ,
\end{equation}
with $\varphi \in [0,2 \pi]$ and $\vec{\sigma}$ the Pauli matrices vector. The corresponding $\SO(3)$ angle is $2\varphi$. As usual, the previous integral can be expressed in terms of the class angle of $g$ using the remaining global $\SU(2)$ invariance
\begin{equation}
	\la Z_{PR}^{\cK}| \Psi \ra = \f{1}{\pi} \int_{0}^{2\pi} \sin^{2}(\varphi) \dd \varphi \; \Phi(g^{h}_{t,x} = \id,g^{v}_{t \neq N_t-1,x} = \id, g^{v}_{N_t-1},x = e^{i \varphi \sigma_3}) \; \;
\end{equation}
where $\sigma_3$ is the Pauli matrix along the $z$ direction. Even then, there is still a global $U(1)$ gauge invariance corresponding to a global $\SU(2)$ transformation along the $z$ direction.

This is not necessarily the end of the story. For computational reason\footnote{We are going to make use of Fourier transform to compute the amplitude.}, it is needed to recast the amplitude in a more symmetrical form, where the dependency on $\varphi$ is homogeneous on the boundary lattice. Recall that the Haar measure is left and right invariant by definition. First, consider the change of variable $\varphi \rightarrow N_t \varphi$. The amplitude becomes
\begin{equation}
	\la Z_{PR}^{\cK}| \Psi \ra = \f{N_t}{\pi}\int_{0}^{\f{2\pi}{N_t}} \dd\varphi \; \sin^2(N_t \varphi) \; \Phi(\{ g_{t,x}^{h} = \id, g_{t \neq N_t-1,x = \id}, g_{N_{t}-1,x} = e^{i N_t \varphi \sigma_z} = g^{N_t}\}) \; .
\end{equation}
Now, apply this series of gauge transformation
\begin{equation}
	g^{v,h}_{t,x} \rightarrow g^{-t} g^{v,h}_{t,x} \; .
\end{equation}
starting from $t = N_t-1$ to $t = 1$. These transformations are constant in each time slice, so the condition $g^{h}_{t,x} = \id$ is left untouched. After this chain of transformations, the amplitude becomes
\begin{equation}
	\la Z_{PR}^{\cK}| \Psi \ra = \f{N_t}{\pi}\int_{0}^{\f{2\pi}{N_t}} \dd\varphi \; \sin^2(N_t \varphi) \; \Phi(\{ g_{t,x}^{h} = \id, g_{t,x} = e^{i \varphi \sigma_z} \}) \; .
\end{equation}

Finally, we can express the integration over the usual $[0,2\pi]$ interval again
\begin{equation}
	\la Z_{PR}^{\cK}| \Psi \ra = \f{1}{\pi}\int_{0}^{2\pi} \text{d}\varphi \sin^2(\varphi) \; \Psi(\{ g_{t,x}^{h} = \id, g_{t,x} = e^{i \f{\varphi}{N_t} \sigma_z} \}) \; .
	\label{chap4:eq:formula_to_compute}
\end{equation}
Note that doing the transformation $\varphi \rightarrow N_t \varphi$ seems rather unnecessary. Clearly, we would have obtained the same result without doing this change of variable. The only use of this change of variable is to avoid the discussion about taking the roots of unity for $e^{i\varphi \sigma_{z}}$, since they are not unique. The result of this tedious procedure could have been guessed on the basis of triangulation invariance as well as from the fact that the flatness of the model requires all contractible loops to be flat, and integrates over all possible values for the non-contractible cycle. 

The main goal of this thesis is to compute the amplitude \eqref{chap4:eq:formula_to_compute} given some class of boundary state $\Psi$. By looking at two classes of boundary state, we will see that the Ponzano-Regge model allows to recover the BMS character as the partition function in the asymptotic limit on top of providing quantum correction and regularization for quasi-local region.

\section{A first computation of the Ponzano-Regge amplitude}

This section is dedicated to a first computation of \ref{chap4:eq:formula_to_compute} with a very simple choice of boundary states. Going through the standard intertwiners expressions, we focus on the spin-0 recoupling intertwiner in the s-channel. The advantage of this simple choice is that it allows for an exact computation of the Ponzano-Regge amplitude. This simple and even naive case will allow us to illustrate the dependence of the asymptotic partition function on the twist angle $\gamma$, and especially the distinction that arises between its rational vs. irrational values: in the appropriate limit, we will get poles for all rational angles. Hence, we recover the basic feature of the BMS character.

\subsection{Intertwiner basis and spin network evaluation}

Before focusing on the actual computation, it is necessary to introduce the boundary state we consider. In this section, we decide to fix all the spins of the boundary lattice at some given values. Then we need to choose a corresponding for 4-valent intertwiner per node. We already explained previously that 4-valent intertwiners are not unique. A basis of the intertwiner Hilbert (vector) space can however be easily constructed. See appendix \ref{App:4valent_intertwiner} for the details of the construction.

In short, we need to choose a pairing of the spins,  $(12)-(34)$, or $(13)-(24)$, or $(14)-(23)$, and to split the 4-valent intertwiner into two 3-valent ones linked by an intermediate spin. Basis states are then defined by the spin $J$ carried by that intermediate link, as shown on figure \ref{chap4:fig:split_intertwiner}. Similarly to the terminology used in particle scattering, we refer to the three possible pairings as the channels $s$, $t$ or $u$. The $t$ and $u$ choices correspond geometrically to fixing the lengths of a diagonal in the quadrilateral picture we discussed with figure \ref{fig:four_intertwiners_to_two_three}.

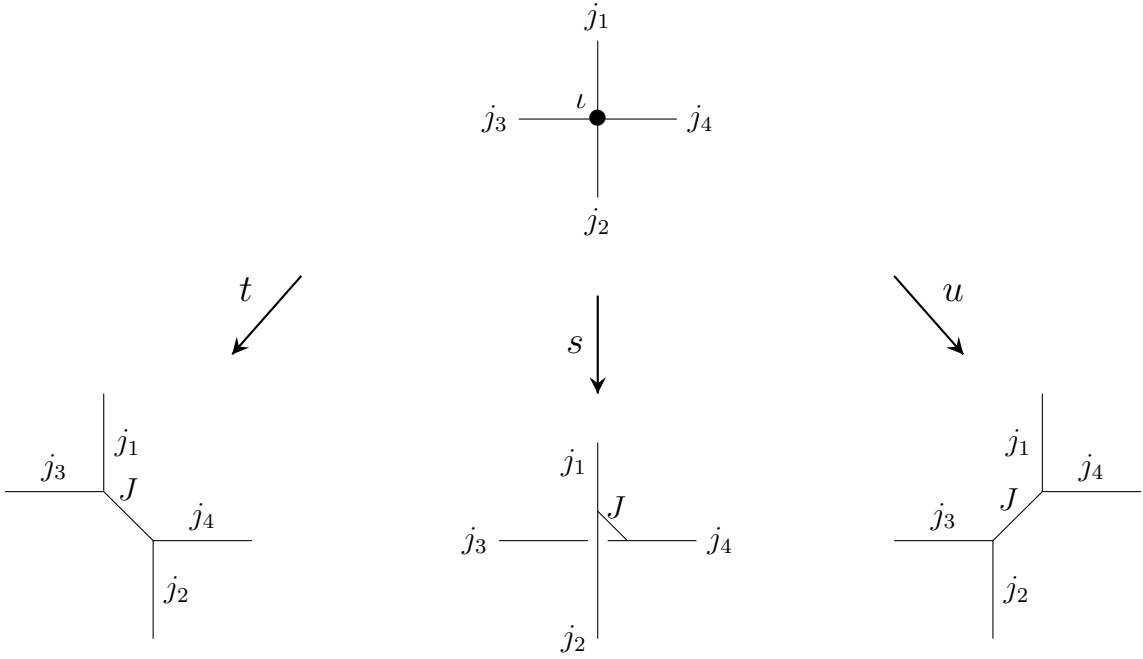
\begin{figure}[t]
	\begin{center}
		\begin{tikzpicture}[scale=1.3]
		
		\draw (-.8,0) -- node[pos=0, left]{$j_3$} node[pos=1, right]{$j_4$} (.8,0) ; \draw (0,-.8)-- node[pos=0, below]{$j_2$} node[pos=1, above]{$j_1$}(0,.8); \draw (0,0) node[scale=1.3]{$\bullet$} node[above left]{$\iota$}; 
		
		\draw[thick,decoration={markings,mark=at position 1 with {\arrow[scale=1.5,>=stealth]{>}}},postaction={decorate}]
		(0,-2+0.2) -- (0,-3+0.2) node[midway,left,scale=1.2]{$s$};
		\draw[thick,decoration={markings,mark=at position 1 with {\arrow[scale=1.5,>=stealth]{>}}},postaction={decorate}]
		(-3,-2+0.4) -- (-3.7,-2.8+0.4) node[midway,above left,scale=1.2]{$t$};
		\draw[thick,decoration={markings,mark=at position 1 with {\arrow[scale=1.5,>=stealth]{>}}},postaction={decorate}]
		(3,-2+0.4) -- (3.7,-2.8+0.4) node[midway,above right,scale=1.2]{$u$};
		
		\draw (-6,-4+0.2)--node[pos=0.5, above]{$j_3$} (-5,-4+0.2)-- node[pos=0.5, right]{$j_1$} (-5,-3+0.2); \draw (-5,-4+0.2) -- node[pos=0.4, above]{\;$J$} (-4.5,-4.5+0.2); \draw (-4.5,-5.5+0.2) -- node[pos=0.5, right]{$j_2$} (-4.5,-4.5+0.2) -- node[pos=0.5, above]{$j_4$} (-3.5,-4.5+0.2);
		
		\draw (-1,-4.5+0.2)--node[pos=0, left]{$j_3$}(-0.1,-4.5+0.2); \draw (0.1,-4.5+0.2)--node[pos=1, right]{$j_4$}(1,-4.5+0.2); \draw (0,-5.5+0.2)--node[pos=0, left]{$j_2$} node[pos=0.9, left]{$j_1$} (0,-3.5+0.2); \draw (0,-4.2+0.2)--node[pos=0.6, above]{$J$} (0.3,-4.5+0.2);
		
		\draw (3,-5+0.7)--node[pos=0.5, above]{$j_3$} (4,-5+0.7)-- node[pos=0.5, right]{$j_2$} (4,-6+0.7); \draw (4,-5+0.7)-- node[pos=0.4, above]{$J$\;} (4.5,-4.5+0.7); \draw (4.5,-3.5+0.7)-- node[pos=0.5, left]{$j_1$} (4.5,-4.5+0.7)-- node[pos=0.5, above]{$j_4$} (5.5,-4.5+0.7);
		
		\end{tikzpicture}
	\end{center}
	\caption{The three channels for splitting a 4-valent intertwiner into two 3-valents ones linked by an intermediate link carrying a spin $J$.}
	\label{chap4:fig:split_intertwiner}
\end{figure}

Let us explicitly consider the $s$-channel, corresponding to the $(12)-(34)$ pairing. In this case, the 4-valent intertwiner basis state with intermediate spin $J$ reads 
\begin{equation*}
	|\iota^{s|J}\ra
	=
	\f1{d_{J}}
	\sum_{\{m_{i\}},M} (-1)^{J+M}
	|\otimes_{i=1}^4(j_{i}m_{i})\ra \,
	\la (j_{1}m_{1})(j_{2}m_{2})| J,M \ra
	\la (j_{3}m_{3})(j_{4}m_{4})| J,-M \ra
	\,.
\end{equation*}
Here, we have taken the two sets of Clebsh-Gordan coefficients, recoupling $j_{1}$ and $j_{2}$ into $J$ on one side and recoupling $j_{3}$ and $j_{4}$ into $J$ on the other, and glued them using the $\su(2)$ structure map $\varsigma$ along the intermediate link.
This map identifies the spin $j$ representation with its conjugate, according to
\begin{equation}
D^j(\varsigma ) \, |j,m\ra=(-1)^{(j+m)}\,|j,-m\ra
\nn
\end{equation}
(see appendix \ref{App:4valent_intertwiner} for further details). 

By convention, the above formula describes an intertwiner for four links {\it outgoing} from a node $n$. Thus for gluing the intertwiners along auxiliary two-valent nodes positioned between two half-links, we need to insert again the $\su(2)$ structure map at each such node, possibly together with the insertion of a group element associated to this link.  

We can know evaluate the Ponzano-Regge amplitude specializing equation \eqref{chap4:eq:formula_to_compute} to our choice of boundary state. In other words, we glue and contract the intertwiners together along the lattice links, with a group element insertion along the links of the last slice of the torus: 
\begin{equation}
	\la Z_{PR}^{\cK}| \Psi_{j_l,\iota^J_n}  \ra
	=
	\int_{\SU(2)} d g \,
	\bigotimes_{n}\iota_n  \bullet_{\pp\cK^{*}} \Big{[}\bigotimes_{l\ne(N_{t},x)} D^{j_{l}}(\varsigma)  \bigotimes_{l=(N_{t},x)} D^{j_{l}}(g\varsigma)\Big{]}
\end{equation}
where $\bullet_{\pp\cK^{*}}$ again stands for a trace over the magnetic indices following the connectivity of the boundary square lattice. For fixed spins $j_{l}$ and intertwiners $\iota_{n}$, this is the integral of a polynomial over $\SU(2)$. Indeed the Wigner matrix elements $D^j(g)^{m'}{}_m$ are polynomials of degree $2j$ in the $\SU(2)$ group element $g$ (defined as a 2$\times$2 matrix). We recover the fact that the Ponzano-Regge amplitude always gives a finite result. In the rest of this section, we focus on the computation of the amplitude for a simple case.

\subsection{The homogeneous $J=0$ s-channel intertwiner case}

Consider the spin network basis state on the square lattice defined by a homogeneous assignation of spins $j_l=j$ to all links and a homogeneous 4-valent intertwiner with spin $J=0$ in the $s$-channel at all nodes. This choice of intertwiner not only allows to completely decouple the horizontal and vertical links, allowing us to explicitly evaluate the partition function and its asymptotic limit, but it also comes with a clear geometrical interpretation.

\subsubsection{Geometric interpretation of the $J=0$ $s$-channel intertwiner}

This intertwiner is given by
\begin{align*}
	|\iota^{s|0}\ra
	& =
	\f1{d_{j}}
	\sum_{\{m_a\}_{a=1,\dots,4}} (-1)^{2j+m_1+m_3} \delta_{m_1+m_2,0} \delta_{m_3+m_4,0}
	| (j,m_1)(j,m_2)(j,m_3)(j,m_4)\ra\nn\\
	&= \f1{d_{j}}
	\sum_{m,\tilde{m}} (-1)^{2j+m+\tilde{m}}
	| (j,m)(j,-m)\ra_{(12)}\otimes| (j,\tilde{m})(j,-\tilde{m})\ra_{(34)}\,.
\end{align*}
Notice that all four edge vectors of the quadrilateral plaquette have equal norm, $\la \vJ_{a}^{2} \ra=j(j+1)$, for $a=1,2,3,4$. Moreover, the $J=0$ on the intermediate link implies that the edge 2 is exactly opposite to the edge 1, while the edge 3 is exactly opposite to  the edge 4. Indeed, one can check the following expectation values from the intertwiner formula above:
\begin{equation}
	\la \vJ_{1}\cdot\vJ_{2} \ra=\la \vJ_{3}\cdot\vJ_{4} \ra=-j(j+1)
	\,,
\end{equation}
Moreover, one can further compute the expectation values of the angles between the other pairs of vectors:
\begin{equation}
	\la \vJ_{1}\cdot\vJ_{3} \ra=\la \vJ_{1}\cdot\vJ_{4} \ra=\la \vJ_{2}\cdot\vJ_{3} \ra=\la \vJ_{2}\cdot\vJ_{4} \ra=0
	\,.
	\label{chap4:eq:mean}
\end{equation}
This means that, in average, the edge 1 (and thus also the edge 2) is orthogonal to the edges 3 and 4, and therefore the intertwiner seems to be dual to a geometrical square. However, this is only true in average. Computing the variance of the scalar products between the would-be orthogonal edges,
\begin{equation}
	\la (\vJ_{1}\cdot\vJ_{3})^{2} \ra=\f13j^{2}(j+1)^{2}\; ,
	\label{chap4:eq:variance}
\end{equation}
one realizes that it takes its maximal value.

Thus, the $J=0$ $s$-channel intertwiner is the furthest possible from a semi-classical state with a good geometrical interpretation. More precisely, while the spin-0 intertwiner in the $s$ channel defines a maximal entanglement between the opposite edges (1 and 2 on the one hand, and 3 and 4 on the other), the other pairings are left to be independent random vectors.\footnote{In fact, considering two random vectors $\hat{u}$ and $\hat{v}$ on the 2-sphere of radius $r$, $\mathbb S^2_r$, one can easily compute:
\begin{equation*}
	\int_{(\mathbb{S}^{2}_{r})^{\times 2}} \f{d^{2}\hat{u}}{4\pi}\f{d^{2}\hat{v}}{4\pi}\,
	\big(
	\hat{u}\cdot\hat{v}
	\big{)}
	=
	0
	\,,\qquad
	\int_{(\mathbb{S}^{2}_{r})^{\times 2}} \f{d^{2}\hat{u}}{4\pi}\f{d^{2}\hat{v}}{4\pi}\,
	\big(
	\hat{u}\cdot\hat{v}
	\big{)}^{2}
	=
	\f13 r^{4}
	\; ,
\end{equation*}
which are the corresponding classical calculations for the expectation values \eqref{chap4:eq:mean} and for the variances \eqref{chap4:eq:variance}.} More correctly, this intertwiner can be said to represent a superposition of all possible parallelograms centred on the rectangular one.

Even if the $J=0$ $s$-channel intertwiner is not fully peaked on the intrinsic geometric data encoded at a node, it nevertheless defines a legitimate quantum state, yielding a well-defined Ponzano-Regge amplitude and offering interesting insight in its structure and potential asymptotic limit. In the next chapter, we will consider coherent intertwiner states which have semi-classical properties in the intertwiner degree of freedom, defining coherent rectangular plaquettes from coherent intertwiners.

\subsubsection{Evaluation of the Ponzano-Regge amplitude}

Inserting the $J=0$ $s$-channel intertwiner in the Ponzano-Regge amplitude results in the complete decoupling of the horizontal from the vertical links. Absorbing the $\su(2)$ structure maps $\varsigma$ in the intertwiners, they become
\begin{equation}
\left( \iota^{s|0} \otimes D^j(\varsigma) \otimes D^j(\varsigma)\right)^{m',\tl m'}{}_{m,\tl m} \,\propto\, \delta^{m'}_m \delta^{\tl m'}_{\tl m} ,
\end{equation}
where $(m',m)$ are the magnetic indices associated to the two vertical links, and $(\tl m',\tl m)$ the magnetic indices associated to the two horizontal ones.  
This formula is evident from the graphical representation of the $s$-channel given in figure \ref{chap4:fig:split_intertwiner}: since a 0-spin link is mathematically the same as no link at all, the 0-spin intertwiner means that the two links go through the node without interacting with each other. 
As illustrated in figure \ref{chap4:fig:graph_integration}, this result is $N_{t}$ horizontal loops of spin $j$, completely decoupled from a number of vertical loops carrying the group element $g$ and winding around the torus with twist $N_{\gamma}$.

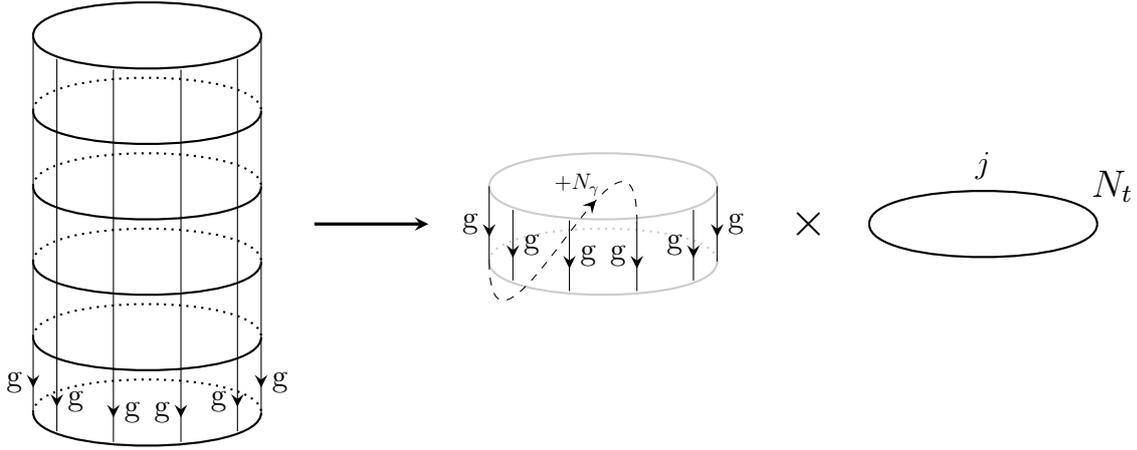
\begin{figure}[t]
	
	\begin{tikzpicture}
	\coordinate(ai) at (0,0);
	\coordinate(bi) at (0,5);
	\draw[dotted,thick,in=90,out=90,looseness=.5] (0,0) to (3,0);
	\draw[thick,in=-90,out=-90,looseness=.5] (0,0) to node[pos=0.2,inner sep=0pt](a1){}  node[pos=0.4,inner sep=0pt](a2){}  node[pos=0.6,inner sep=0pt](a3){}  node[pos=0.8,inner sep=0pt](a4){} (3,0);
	\draw[in=90,out=90,looseness=.5,thick] (0,5) to  (3,5);
	\draw[thick,in=-90,out=-90,looseness=.5] (0,5) to node[pos=0.2,inner sep=0pt](b1){}  node[pos=0.4,inner sep=0pt](b2){}  node[pos=0.6,inner sep=0pt](b3){}  node[pos=0.8,inner sep=0pt](b4){} (3,5);
	\coordinate(af) at (3,0);
	\coordinate(bf) at (3,5);
	
	\foreach \i in {1,...,4}{
		\draw[dotted,thick,in=90,out=90,looseness=.5] (0,\i) to (3,\i);
		\draw[thick,in=-90,out=-90,looseness=.5] (0,\i) to (3,\i);
	}
	
	\draw[decoration={markings,mark=at position 0.1 with {\arrow[scale=1.5,>=stealth]{<}}},postaction={decorate}] (ai) -- node[left,pos=0.08]{g}(bi);
	\draw[decoration={markings,mark=at position 0.1 with {\arrow[scale=1.5,>=stealth]{<}}},postaction={decorate}] (a1) --node[right,pos=0.08]{g} (b1);
	\draw[decoration={markings,mark=at position 0.1 with {\arrow[scale=1.5,>=stealth]{<}}},postaction={decorate}] (a2) -- node[right,pos=0.08]{g} (b2);
	\draw[decoration={markings,mark=at position 0.1 with {\arrow[scale=1.5,>=stealth]{<}}},postaction={decorate}] (a3) -- node[left,pos=0.08]{g} (b3);
	\draw[decoration={markings,mark=at position 0.1 with {\arrow[scale=1.5,>=stealth]{<}}},postaction={decorate}] (a4) --node[left,pos=0.08]{g} (b4);
	\draw[decoration={markings,mark=at position 0.1 with {\arrow[scale=1.5,>=stealth]{<}}},postaction={decorate}] (af) -- node[right,pos=0.08]{g} (bf);
	
	\draw[very thick,->,>=stealth] (3.7,2.5) --(5.2,2.5);
	
	\coordinate(ci) at (6,2); 
	\coordinate(di) at (6,3);
	\coordinate(cf) at (9,2);
	\coordinate(df) at (9,3);
	
	\draw[decoration={markings,mark=at position 0.7 with {\arrow[scale=1.5,>=stealth]{>}}},postaction={decorate}, dashed](6,2) to[bend right,in=75,out=-105,looseness=2.5] node[scale=0.9,pos=0, left]{%$0$
	} node[scale=0.7,pos=0.75,left]{$+N_\gamma\;$} node[pos=1,above,scale=0.9]{%
		%$N_\gamma$
	} (7.945,2.55) ;
	
	\draw[opacity=0.2,dotted,thick,in=90,out=90,looseness=.5] (ci) to (cf);
	\draw[opacity=0.2,thick,in=-90,out=-90,looseness=.5] (ci) to node[pos=0.2,inner sep=0pt](c1){}  node[pos=0.4,inner sep=0pt](c2){}  node[pos=0.6,inner sep=0pt](c3){}  node[pos=0.8,inner sep=0pt](c4){} (cf);
	\draw[opacity=0.2,in=90,out=90,looseness=.5,thick] (di) to  (df);
	\draw[opacity=0.2,thick,in=-90,out=-90,looseness=.5] (di) to node[pos=0.2,inner sep=0pt](d1){}  node[pos=0.4,inner sep=0pt](d2){}  node[pos=0.6,inner sep=0pt](d3){}  node[pos=0.8,inner sep=0pt](d4){} (df);
	
	\draw[decoration={markings,mark=at position 0.5 with {\arrow[scale=1.5,>=stealth]{<}}},postaction={decorate}] (ci) -- node[left,pos=0.5]{g}(di);
	\draw[decoration={markings,mark=at position 0.5 with {\arrow[scale=1.5,>=stealth]{<}}},postaction={decorate}] (c1) --node[right,pos=0.5]{g} (d1);
	\draw[decoration={markings,mark=at position 0.5 with {\arrow[scale=1.5,>=stealth]{<}}},postaction={decorate}] (c2) -- node[right,pos=0.5]{g} (d2);
	\draw[decoration={markings,mark=at position 0.5 with {\arrow[scale=1.5,>=stealth]{<}}},postaction={decorate}] (c3) -- node[left,pos=0.5]{g} (d3);
	\draw[decoration={markings,mark=at position 0.5 with {\arrow[scale=1.5,>=stealth]{<}}},postaction={decorate}] (c4) --node[left,pos=0.5]{g} (d4);
	\draw[decoration={markings,mark=at position 0.5 with {\arrow[scale=1.5,>=stealth]{<}}},postaction={decorate}] (cf) -- node[right,pos=0.5]{g} (df);
	
	\node[scale=1.5] at (10.2,2.5){$\times$};
	
	\coordinate(ci) at (11,2.5);
	\coordinate(cf) at (14,2.5);
	\draw[thick,in=90,out=90,looseness=.5] (ci) to node[pos=0.5,above]{$j$} (cf);
	\draw[thick,in=-90,out=-90,looseness=.5] (ci) to (cf);
	
	\node[scale=1.2] at (14.2,3){$N_{t}$};		
	
	\end{tikzpicture}
	
	\caption{The figure indicates the loops that arise with the choice of s-channel $J=0$ intertwiner. Due to this choice of intertwiner, the loops decouple into vertical and $N_t$ horizontal ones. The winding number $\W$  and total number $K$ of the vertical loops on the cylinder is determined by the  shift $N_\gamma$ in the periodic identification of the cylinder and the spatial size $N_x$ of the cylinder.}
	\label{chap4:fig:graph_integration}
\end{figure}

Whereas horizontal loops simply factor out, with the number of time slices $N_{t}$ only contributing with an overall volume factor, vertical loops acquire a non-trivial structure due to the interplay between the spatial size $N_{x}$ and the shift number $N_{\gamma}$. Putting these two ingredients together, we obtain the Ponzano-Regge amplitude in the form 
\begin{equation}
	\la Z_{PR}^{\cK}| \Psi_{j_l,\iota^{s|0}}  \ra
	=
	\frac{d_{j}^{N_{t}}}{d_j^{N_t N_x}}
	\int_{\SU(2)} dg \; \chi_{j}(g^{\W})^{K}
	= 
	\frac{1}{d_j^{N_t (N_x-1)}}
	\frac{2}{\pi} \int_{0}^{\pi} \text{d}\theta \;\sin^2(\theta) \; \chi_{j}(\W \theta)^K \; ,
	\label{chap4:eq:amplitude_J_0_before}
\end{equation}
where we recall that $\chi_{j}$ is the character in the spin-$j$ representation, which in terms of the  (half) class angle $\theta$ of $g$ reads
\begin{equation}
\chi_{j}(g)\equiv\chi_j(\theta) = \f{\sin d_{j}\theta}{\sin\theta} \; ,
\end{equation}
and where we recall that $K = GCD(N_x,N_\gamma)$ is the number of independent close vertical loops and $W =N_x/K $ their winding numbers. In \eqref{chap4:eq:amplitude_J_0_before}, the volume factor $d_{j}^{-N_{t}N_{x}}$ comes from the normalization of the intertwiner $\iota^{s|0}$. The factor $d_{j}^{N_{t}}$, on the other hand, comes from the contribution of $\chi_{j}(\id)$ given by each time slice.

We have two ways to evaluate this integral. We can either express it in terms of random walks and compute it exactly, or we can extract its asymptotic behaviour at large $K$ by a saddle point approximation.
Before proceeding, it is useful to notice that the integral vanishes for odd values of $2jN_x$, since $\chi_j(W(\pi-\theta))^K = \chi_j(W(\theta-\pi))^K =  (-1)^{2j K W} \chi_j(W\theta)^K$ and $N_x \equiv W K$. For the same reason, whenever the integral does not vanish, the integrand is periodic of period $\pi$, and the integration domain can be compactified to a circle $\cong\mathbb S_1$. 

\medskip
{\bf Exact evaluation}
First aiming for an exact evaluation, we expand the character into a sum over exponentials by expressing the trace of the group element $g$ in the $|j,m\ra$ basis of the Hilbert space of the spin-$j$ representation,
\begin{equation}
	\chi_{j}(\theta)
	=
	\f{\sin d_{j}\theta}{\sin\theta}
	=\sum_{m=-j}^{+j} \E^{2\I m\theta}
	\nn\,,
\end{equation}
from which
\begin{equation}
	\la Z_{PR}^{\cK}| \Psi_{j_l,\iota^{s|0}}  \ra
	=
	\frac{1}{4\pi d_j^{N_t (N_x-1)}}
	\int_{0}^{2\pi} {\dd}\theta\,
	(2-\E^{2\I \theta}-\E^{-2\I \theta})
	\sum_{m_{1},..,m_{K}} \E^{2\I \sum_{k=1}^{K}m_{k}\W\theta}\,.
	\label{chap4:eq:recouplingJ0_channel_general_amp}
\end{equation}
This integral has a straightforward combinatorial interpretation: we are counting the number of returns after $K$ steps to either the origin or to the positions $\pm 2$, of a random walk characterized by steps of arbitrary size between $-2\W j$ and $+2\W j$. One must always distinguish the case of a half-integer spin $j\in(\mathbb N+\frac12)$ for which each step is an odd multiple of $\W$, from the case of integer spin $j\in\mathbb N$ for which each step is an even multiple of $\W$. Due to the measure factor $(2-\E^{2\I \theta}-\E^{-2\I \theta})$, one must also distinguish the special cases $\W=1$ and $\W=2$ from the generic case $\W\ge 3$.  We also notice that the case $\W=1$ corresponds to computing the dimension of the intertwiner space between $K=N_x$ copies of the spin $j$. 

The computation is simple when considering the case $K=1$. For $\W=N_{x}\ge 3$, we do not have to take into account the terms in $\E^{\pm2\I \theta}$, and
\begin{equation}
	\la Z_{PR}^{\cK}| \Psi_{j_l,\iota^{s|0}}  \ra_{ K=1 \atop N_x\ge3}
	=
	\frac{1}{2\pi d_j^{N_t (N_x-1)}}
	\int_{0}^{2\pi} {\dd}\theta\,
	\sum_{m=-j}^{+j} \E^{2\I m N_x\theta}
	=
	\begin{cases}
	1/d_j^{N_t (N_x-1)} \quad &\text{if} \quad j \in \mathbb{N} \\
	0 \quad &\text{if} \quad j \in {\N+ \f12}
	\end{cases}.
	\label{chap4:eq:irra_case}
\end{equation}
When $\W=N_x=1$, the integral always vanishes (as soon as the spin is non-zero, $j\ne 0$). When $\W=N_x=2$, one gets half of the volume factor $d_j^{-N_t (N_x-1)}$ when the spin $j$ is an integer and minus half of this volume factor when $j$ is a half-integer. %
The exact expression for arbitrary $K$ as a rational function in the spin $j$ is given in appendix \ref{App:exact_computation_recouplingJ0} and is related to the Fourier series expansion of the cardinal sine function $\mathrm{sinc}\,\theta\equiv\sin \theta/\theta$ and to the Duflo map coefficients for $\SU(2)$.

\medskip

{\bf Asymptotic limit} 
Here, we are mainly interested in the asymptotic limit. We define it as a double scaling limit where both $N_{x}$ and $N_{\gamma}$ are sent to infinity while their ratio $2 \pi \frac{N_\gamma}{N_x} \rightarrow \gamma \in \mathbb{R}$ is kept finite. 
Although similar in spirit to a lattice refinement limit, this should be more correctly considered as an asymptotic limit: since the spin $j$ is fixed, the spatial size $\sim j N_{x}$ diverges.
At this point, it becomes clear that we need to distinguish the cases where $\gamma$ is rational or not:

\begin{itemize}
	
	\item {\it Irrational twist angle} $\gamma\in 2\pi(\R\setminus\Q)$
	
	In this case we approximate $\gamma$ via a sequence of pairs of integers $\left(N_{\gamma}^{(n)},N_{x}^{(n)}\right)_{n\in\N}$ which are always prime with each other (e.g. taking the continued fraction approximation):
	\begin{eqnarray}
	2\pi \frac{N_\gamma^{(n)}}{N_x^{(n)}} \underset{n\rightarrow\infty}\longrightarrow {\gamma}
	\,,\quad
	&&N_{\gamma,x}^{(n)}\underset{n\rightarrow\infty}\longrightarrow\infty
	\,,\quad
	K^{(n)}:=\mathrm{GCD}\left(N_{\gamma}^{(n)},N_{x}^{(n)}\right)=1
	\,,\nn\\
	&&\quad
	\text{and}\quad
	\W^{(n)}=N_x^{(n)}\to\infty.
	\end{eqnarray}
	This is exactly the case computed above in equation \eqref{chap4:eq:irra_case}. For half-integer spins, the amplitude always vanishes. For integer spins, putting aside the volume factor $d_{j}^{-N_{t}(N_{x}-1)}$, the amplitude  remains finite, always equal to 1.

	\item {\it Rational twist angle} $\gamma\in 2\pi \Q$
	
	In this case, the twist angle can be implemented exactly on a whole sequence of discrete lattices, at least provided one chooses $N_x^{(n)}$ appropriately. To do this, one first identifies the corresponding minimal fraction and then  considers multiples of its numerator and denominator:
	\begin{equation*}
	\gamma={2\pi}\f{P}{Q}\quad\text{with}\quad \mathrm{GCD}(P,Q)=1
	\quad\text{and hence}\quad
	\left(N_{\gamma}^{(n)},N_{x}^{(n)}\right)=(nP,nQ).
	\end{equation*}
	This case is combinatorially the reverse situation compared to the case of an irrational angle, since the number of loops $K$ grows to infinity while the winding number remains constant:
	\begin{equation*}
	N_{\gamma,x}^{(n)}\underset{n\rightarrow\infty}\longrightarrow\infty
	\,,\qquad
	K^{(n)}:=\mathrm{GCD}\left(N_{\gamma}^{(n)},N_{x}^{(n)}\right)=n\rightarrow\infty
	\,,\quad\text{and}\quad
	\W^{(n)}:=\f{N_{x}^{(n)}}{K^{(n)}}\equiv Q
	\,.
	\end{equation*}
	In this case, the simplest way to evaluate the Ponzano--Regge partition function is to compute its saddle point approximation at large $n$. We rewrite the integral as 
	\begin{equation}
	\la Z_{PR}^{\cK}| \Psi_{j_l,\iota^{s|0}} \ra
	=
	\frac{2}{\pi d_j^{N_t (N_x-1)}} \int_{0}^{\pi} \text{d}\theta \sin^2(\theta) \; \E^{\f{N_x}{\W} \ln(\chi_{j}(\W\theta))} .
	\end{equation}
	As noticed above, whenever non-vanishing, the integral can be considered defined on a circle with the points $\theta=0$ and $\theta=\pi$ identified.
	On this circle, the exponent $\ln(\chi_{j}(\W\theta))$ reaches its maximal value of $\ln(d_{j})$ exactly $\W$ times at the locations $\theta_{l}= \f{\pi l}{\W}$, with $l=0,..,(\W-1)$. The second derivative at those points is given by the $\SU(2)$ Casimir:
	\begin{equation}
	\f12 \left.\f{\pp^{2}\ln\chi_{j}(\W\theta)}{\pp\theta^{2}}\right|_{\theta=\theta_{l}}=-\f16(d_j^2-1)\W^2 = -\f{4}{6}\W^{2}j(j+1)\,.
	\nn
	\end{equation}
	In  the special case $\W=1$, we have a unique stationary point at $\theta=0$, which gives the asymptotics (recall that for $\W=1$, $N_x=K$):
	\begin{equation}
		\la Z_{PR}^{\cK}| \Psi_{j_l,\iota^{s|0}} \ra_{\W=1}
		\underset{ N_x\rightarrow\infty}\sim
		\sqrt{\f{3}{2\pi}}\,\f{12 d_{j}^{N_x} {N_x}^{-\f32} }{d_j^{N_t (N_x-1)} (d_j^2-1)^{\f32}}
		\, 
	\end{equation}
	Here, the $N_x^{-\f32}$ decrease is due to the measure factor, $\sin^{2}\theta \sim\theta^2$ around the saddle\footnote{This is analogous to a standard log-correction to the black hole entropy computed from the dimensions of the intertwiner spaces	 \cite{Livine:2005mw,Livine:2012cv,Kaul:2012pf}.}.
	In the generic case $\W\ge 2$, we sum over all the maxima and get (recall $N_x=\W K$): 
	\begin{equation*}
		\sum_{l=1}^{\W-1}\sin^{2}\f{\pi l}{\W}=\f \W2
	\end{equation*}
	which implies that
	\begin{equation}
		\la Z_{PR}^{\cK}| \Psi_{j_l,\iota^{s|0}}  \ra_{\W\geq 2}
		\underset{ N_x\rightarrow\infty}\sim
		{
			=\sqrt{\frac{3 }{2\pi }}\,\f{2 d_{j}^{\f{N_x}{\W}} \left(\f{N_x}{\W}\right)^{-\f12}}{d_j^{\f{N_t (N_x-\W)}{\W}}\left(d_j^2-1\right)^{\f12}} 
		}.
		\label{chap4:eq:easy_case_asympt}
	\end{equation}
	Notice the usual decrease in $(N_x/W)^{-\f12}\equiv K^{-\f12}$ expected for a random walk. The same can be said for volume factor $d_j^K$, since it corresponds to $K$ steps with $d_j$ possibilities. 
	Figure \ref{chap4:fig:exact_to_asympt_plot} compares this asymptotics to the exact value of the partition function, and shows that a good approximation is already obtained for $K=N_x/\W\sim 10$.
	
	We conclude the discussion of a rational twist angle, by observing that in this case, aside for the volume factor $d_j^{-N_t (N_x-1)}$, the partition function {exponentially {\it diverges} in the asymptotic limit $n\to\infty$, where $N_x^{(n)}=\W K^{(n)}\rightarrow\infty$. }
	
\end{itemize}

\begin{figure}[!]
	\centering
	\includegraphics[scale=1]{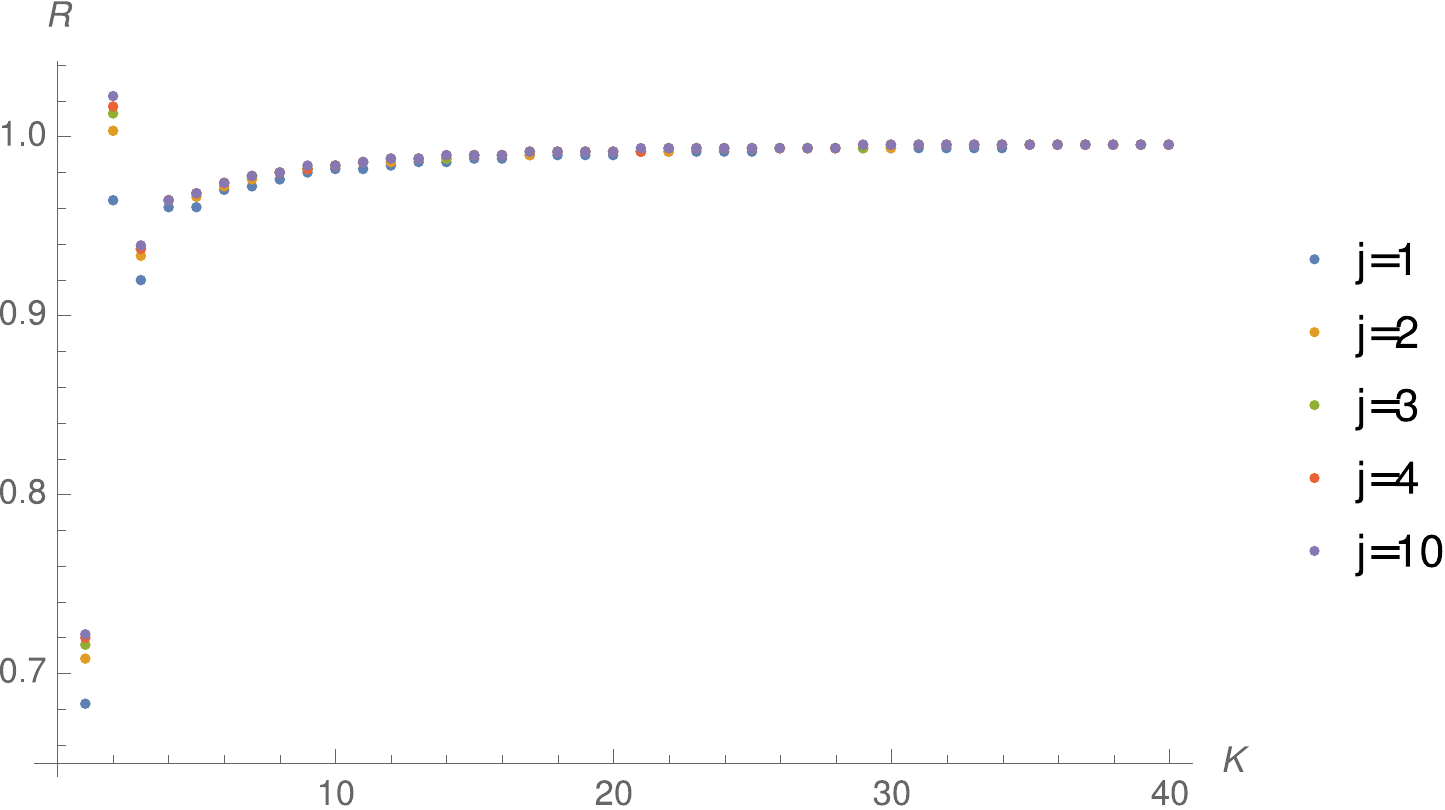}	
	\caption{ This shows a plot of the ratio $R$ of the exact partition function divided by its asymptotic form (\ref{chap4:eq:easy_case_asympt}) as a function of $K$ for several choices of integer spins. The asymptotic gives a good approximation already for $K\sim 10$. The case of half spins shows a similar structure, the only difference is that the partition function vanishes for all odd $K$.)}
	\label{chap4:fig:exact_to_asympt_plot}
\end{figure}

To summarize, we focussed on the asymptotics of the renormalized Ponzano-Regge amplitude $d_j^{N_t (N_x-1)}\,\la \text{PR} | \Psi_{j,\iota^{s|0}}\ra$, where the volume factor $d_j^{N_t (N_x-1)}$ stabilizes the limiting process (similarly to a wave-function renormalization in quantum field theory). In the continuum limit, irrational angles correspond to a trivial renormalized Ponzano-Regge amplitude, always equal to 1, while rational twists lead to divergent amplitudes and thus signify a pole in the asymptotic partition function.

This difference between rational and irrational twist angles is a crucial feature of the BMS character formula for the three-dimensional quantum gravity partition function as we explained previously.  We do not however obtain the exact formula. This is due to the fact that we are not working with a semi-classical boundary state (even in the asymptotical limit) since, as previously discussed, the $J=0$ $s$-channel intertwiner is as far as classicality as possible. 

Nonetheless, we see that even for such a deeply quantum spin network state, we obtain the correct pole structure for the {\it asymptotic} partition function. And this happens despite the fact that the partition function is finite for a finite-sized boundary.  There is, however, a key remark that must be done here. It is not possible to look at a true continuum limit in this context. We are considering a limit where the spin $j$ is fixed. Hence, intrinsically, the boundary is still discrete. Even considering the deep quantum case where the spin is in its fundamental representation, the continuum limit is not reached. Thus the question remains: how much can the previous result can be compared to the continuum result? In the last chapter, we will provide a formulation of the Ponzano-Regge amplitude with a boundary state that will allow to take a true continuum limit.

\subsection{Spin $\f{1}{2}$ chain and integrable model}

In the previous section, we have explicitly computed the Ponzano-Regge partition function for a boundary spin network state with fixed (but arbitrary) spin $j$ and the special choice of a $J=0$ $s$-channel intertwiner. Considering small spins, the Hilbert spaces have low dimensionality and the characterization of the intertwiners is therefore simple. In this section, we will consider such a regime, where the spin network state features a uniform choice of the smallest possible spin $j=\f12$. The low dimensionality of the space on intertwiners allows us to consider an arbitrary, but homogeneous on the lattice, intertwiner.

Geometrically, this boundary state corresponds to a lattice with all edge lengths set at the shortest possible allowed distance. In a sense, we are probing the deep quantum regime of the boundary geometry. For what concerns the dual theory, we will find that this boundary state maps exactly onto the 6-vertex (or ``ice-type'') model of statistical physics, with couplings defined by the choice of intertwiner. This will provide the archetype of the mapping of spin network evaluations and quantum gravity amplitudes onto condensed matter models as it was started for the trivial topology in \cite{Dittrich:2013jxa,Bonzom:2015ova}.
\\

Henceforth, we will suppose that all spins have been set to $j=\f12$. The space of 4-valent intertwiners with all spins $1/2$ has dimension two. Choosing a channel, $s$, $t$ or $u$, an orthonormal basis is provided by the two states with intermediate spins $J=0$ and $J=1$. Explicit formulas for those intertwiner states are given in details in appendix \ref{App:4valent_intertwiner}. Instead of considering such a basis $|0\ra_{s},|1\ra_{s}$, it is more convenient to consider the over-complete basis span by the 0-spin intertwiners in each of the possible three channels. Namely, the elements of the over-complete basis are $\{|0\ra_s , |0 \ra_t ,|0 \ra_u\}$, corresponding to the three pairings $(12)-(34)$, $(13)-(24)$ and $(14)-(23)$, see figure \ref{chap4:fig:split_intertwiner}. By making use of the low dimensionality of the space of intertwiner, we still get a basis by restricting us to the $0$ recoupling spins.

An arbitrary interwiner in this basis is defined considering three couplings $(\lambda,\mu,\rho)$ by
\begin{equation}
	|\iota\ra=\lambda|0\ra_s +\mu |0 \ra_t +\rho|0 \ra_u	\; .
\end{equation}
Since the basis is over-complete, the three elements are not independent. We have the relation $|0\ra_{s}-|0\ra_{t}+|0\ra_{u}=0$. Using this equality, we get that the intertwiner
\begin{equation}
	|\iota\ra=(\lambda+\eta)|0\ra_s +(\mu-\eta) |0 \ra_t +(\rho+\eta)|0 \ra_u
\end{equation}
does not actually depend on $\eta$ and any choice can be done. In the following, we consider the case where $\eta = \mu$, such that the intertwiner reads $|\iota \ra = \lambda |0\ra_s + \rho |0 \ra_u$. The intertwiner can be graphically represented, see figure \ref{chap5:fig:spin12:intetwiner}, allowing us to get a nice interpretation of its action on the lattice. Recall that, geometrically, the $s$-channel intertwiner represents a (maximally fuzzy) square. Mixing it with a $u$-channel intertwiner corresponds to turning the square into (maximally fuzzy) parallelograms with the angle between adjacent edges depending on the ratio $\rho/\lambda$.
\begin{figure}[h!]
	\begin{center}
	\begin{tikzpicture}
	\draw (-0.6,0) node {$|\iota[\lambda,\rho]\ra$ :=}; \draw (0.5,0) node{$\lambda$};
	
	\draw (1.5,-0.5) -- node[pos=0,scale=0.7,below]{$2$} node[pos=1,scale=0.7,above]{$1$} (1.5,0.5); \draw (1,0)-- node[pos=0,scale=0.7,left]{$3$} (1.4,0); \draw (1.6,0)-- node[pos=1,scale=0.7,right]{$4$} (2,0);
	
	\draw (2.8,0) node{$+ \, \rho$};
	
	\draw[rounded corners=5] (4,.5) --node[pos=0,scale=0.7,above]{$1$} (4,0)--(4.5,0)node[pos=1,scale=0.7,right]{$4$} ;
	\draw[rounded corners=5] (4,-0.5)-- node[pos=0,scale=0.7,below]{$2$} (4,0)--(3.5,0) node[pos=1,scale=0.7,left]{$3$};	
	\end{tikzpicture}
	\end{center}
	\caption{An arbitrary 4-valent intertwiner between four spins $\f12$ decomposes onto the non-orthogonal basis of 0-spin intertwiners in the $s$ and $u$ channel,  $|0\ra_s$ and  $|0 \ra_u$, which can be represented as the lines crossing or bending by the vertex  without interacting.}
	\label{chap5:fig:spin12:intetwiner}
\end{figure}

We focus now on a given time slice. We can expand the product of intertwiners over all nodes into a sum of configurations with each node coming equipped with either a $|0\ra_s$ or a $|0 \ra_u$ intertwiner. The $s$-channel goes straight through the time slice, while the $u$ channel bends the line which propagates along the time slice until it reaches the next $u$ channel intertwiner to exit the time slice. This is illustrated in figure \ref{chap5:fig:time_slice}.
\begin{figure}[!htb]
	\begin{center}
	\begin{tikzpicture}[scale=1]
	\draw (0,.7)--(0,.1); \draw[->,>=stealth] (0,-.1)--(0,-.7);
	\draw (0+.5,.7)--(0+.5,.1); \draw[->,>=stealth] (0+.5,-.1)--(0+.5,-.7);
	\draw (0+1.5,.7)--(0+1.5,.1); \draw[->,>=stealth] (0+1.5,-.1)--(0+1.5,-.7);
	\draw[very thick,rounded corners=5,->,>=stealth] (0+1,.7)--(0+1,0) -- (0+2,0)--(0+2,-.7);
	\draw (0+2.5,.7)--(0+2.5,.1); \draw[->,>=stealth] (0+2.5,-.1)--(0+2.5,-.7);
	\draw (0+3,.7)--(0+3,.1); \draw[->,>=stealth] (0+3,-.1)--(0+3,-.7);
	\draw[very thick,rounded corners=5,->,>=stealth] (0+2,.7)--(0+2,0) -- (0+3.5,0)--(0+3.5,-.7);
	\draw[very thick,rounded corners=5,->,>=stealth] (0+3.5,.7)--(0+3.5,0) -- (0+4,0)--(0+4,-.7);
	\draw (0+4.5,.7)--(0+4.5,.1); \draw[->,>=stealth] (0+4.5,-.1)--(0+4.5,-.7);
	\draw[very thick,rounded corners=5] (0+4,.7)--(0+4,0) -- (0+5,0);
	\draw[very thick,rounded corners=5,->,>=stealth] (0-.5,0)--(0+1,0) -- (0+1,-0.7);
	
	\draw[thick,dotted] (0+5,0)--(0+5.4,0);
	\draw[thick,dotted] (-.5-.4,0)--(0-.5,0);
	\end{tikzpicture}
	\end{center}
	\caption{One time slice: at every node, we insert either a $s$-channel intertwiner, in which case the horizontal and vertical links decouple, or a $u$-channel intertwiner in which case the incoming link bends and propagates along the time slice until it reaches the next $u$-channel intertwiner, where it bends again and propagates to the next time slice. To build the whole partition function, we need to compose such time slices together, insert the group element $g\in\SU(2)$ on all the lines going through the last time slice, take into account the twist when gluing back the final time slice with the initial one, and finally sum over all possible assignments of $s-$ and $u$-channel intertwiners at all the nodes.}
	\label{chap5:fig:time_slice}
\end{figure}
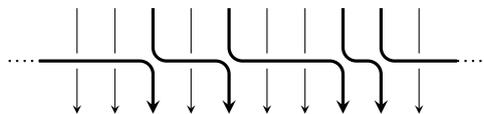

To compute the full partition function, we have to stack time slices together. For the last time slice, we need to insert the remaining $\SU(2)$ holonomy and then perform the gluing while taking into account the twist parameter. At this point, pretty much as in the previous section, we have to follow the lines across the nodes and time slices to see the loops that they form. The computation largely reduces to a purely combinatorial problem.
Notice that, whereas the choice of a purely $J=0$ $s$-channel intertwiner ($\rho=0$) lead to a factorization of the time slices, this is not the case anymore for an arbitrary intertwiner. The time slices are now non-trivially coupled to each other.

\subsubsection{The simplest case: $N_\gamma = 0$ and $N_t=1$}
We solve this combinatorial problem in the simple case of a vanishing twist, $N_{\gamma}=0$  and $N_t=1$, that is only one time slice. Then, with a look at figure \ref{chap5:fig:time_slice}, the Ponzano-Regge partition function is readily decomposed according to the number $p$ of $u$-channel intertwiners on the slice. The partition function sums over all possible $u$-channel intertwiner insertions along the time slice. For $p$ insertions, we are left with $N_x-p$ closed loop containing only one link, hence one $g$ element. The $p$ other links are all on the same closed loop, therefore containing $p$ times the element $g$. That is, the contribution of the $p$ insertion to the partition function is 
\begin{equation*}
	\int_{\SU(2)} dg \; \chi_{\f12}(g)^{N_{x}-p}\chi_{\f12}(g^{p}) \; .
\end{equation*}
Of course, there is more than one choice of $p$ nodes. At the end of the day, the full partition function reads
\begin{equation*}
	\la Z_{PR}^{\cK}| \Psi_{j = \f12,\iota[\lambda,\rho]}  \ra
	=
	\sum_{p=0}^{N_x} \binom{N_{x}}{k} \lambda^{N_x-p}\rho^k \int_{\SU(2)} dg \; \chi_{\f12}(g)^{N_{x}-p}\chi_{\f12}(g^{p}) \; ,
\end{equation*}
where the sum is over all possible choices of $p$ and the binomial factor counts the possibility of $p$ nodes. The integral can again be exaclty computed by expanding the character into its exponential form. Since the amplitude clearly vanishes for an odd number of nodes in the spatial direction, we fix
\begin{equation}
	N_x=2M\qquad  \text{with} \quad M\in \mathbb{N} \; .
\end{equation}

Focusing first on the integral, we get
\begin{align*}
	\int_{\SU(2)} \dd g \; &\chi_{\f12}(g)^{N_{x}-k}\chi_{\f12}(g^{k})
	\\
	&=
	\f{1}{2\pi}\int_{0}^{2\pi}\dd\theta \;  \left(1-\f{\E^{2\I\theta}+\E^{-2\I\theta}}2\right) \; (\E^{\I k\theta}+\E^{-\I k\theta}) \sum_{n=0}^{2M-k}\binom{2M-k}{n}\E^{\I(2M-k-2n)\theta}
	\\
	&=
	\binom{2M-k}{M} \; \f{\big{[}2(M+1)-k(k+1)\big{]}}{(M+1)(M-k+1)} \; .
\end{align*}

Plugging this back into the sum, we get the explicit expression for the Ponzano-Regge amplitude
\begin{eqnarray}
	\la Z_{PR}^{\cK}| \Psi_{j = \f12,\iota[\lambda,\rho]}  \ra
	&=&
	\sum_{k=0}^{M+1}
	\lambda^{2M-k}\rho^k 
	\,
	\f{(2M)!}{k!(M+1)!(M-k+1)!}
	\,
	\bigg{[}2(M+1)-k(k+1)\bigg{]}
	\nn\\
	&=&
	\f2{M+1}\binom{2M}{M}
	\lambda^{M-1}(\lambda+\rho)^{M-1}
	\bigg{[}\lambda^{2}+\lambda\rho-\f M2\rho^{2}\bigg{]}
	\,.
	\label{chap4:eq:no_twist_simple_onehalf}
\end{eqnarray}

\subsubsection{The general case}

The case of a non-vanishing twist $N_\gamma \neq 0$ and $N_t>1$ leads to a considerable combinatorial problem and is best formalized using transfer matrix techniques.

On one time slice, each set of $u$-channel insertions, at the positions  $0\le x_{1}<\dots<x_{k}\le N_{x}-1$, defines a permutation of $N_{x}$ elements given by the cycle $C_{\{x_{n}\}}\equiv(x_{1},x_{2},..,x_{k})$. The sum over all such cyclic permutations over subsets of nodes defines our transfer matrix.
Now, we have to choose arbitrary $u$-channel insertions on each time slice $\{x_{n}^{(t)}\}$ for $t=0,\dots,(N_{t}-1)$ and compose with a twist, i.e. the cyclic permutation $\cC_{N_{\gamma}}$ sending every position $i$ to $(i+N_{\gamma}) \,\mathrm{mod}\,N_{x}$:
\begin{equation}
	\cC_{N_{\gamma}}\circ C_{\{x_{n}^{(N_{t}-1)}\}}\circ\dots\circ C_{\{x_{n}^{(0)}\}}
	\,.
\nn
\end{equation}

The last ingredient before the identification of the first and last time slices, is the introduction of the twist and the integration over the holonomy $g\in\SU(2)$. This sits on all the vertical edges belonging to the final time slice. This means that we have to evaluate the integral
\begin{equation}
\int_{\SU(2)} dg\,
\prod_{i=0}^{N_{x}-1}D^{\f12}_{a_{i}b_{i}}(g)\,.
\end{equation}
This integral, also known as the Haar intertwiner, is non-vanishing if and only if $N_{x}$ is even, that is $N_{x}=2M$ as above. It can also be expressed purely in terms of permutations $\omega$'s, which match the incoming magnetic indices $a_{i}$ with permuted outgoing magnetic indices $b_{\omega(i)}$. The Haar intertwiner is given in terms of the characters $s_{[M,M]}$ of the symmetric group of $N_{x}$ elements $S_{N_{x}}$ in the representation associated to the partition of the integer $N_{x}=2M$ as $M+M$. That is
\begin{equation}
	\int_{\SU(2)} dg\,
	\prod_{i=0}^{2M-1}D^{\f12}_{a_{i}b_{i}}(g)
	=
	\f{M!}{(2M)!}
	\sum_{\omega\in S_{2M}} s_{[M,M]}[\omega] \prod_{i=0}^{2M-1}\delta_{a_{i}b_{\omega(i)}}\,.
\end{equation}
Taking the trace of this expression, we recover the dimension of the intertwiner spaces between $2M$ spins $\f12$, given by the Catalan numbers:
\begin{equation}
	\int_{\SU(2)} dg\,\big{(}\chi_{\f12}(g)\big{)}^{2M}
	=
	\f1{M+1}\binom{2M}{M}
	\,.
\end{equation}
The character for arbitrary permutations $s_{[M,M]}(\omega)$ can be computed using Young tableaux. Putting all these ingredients together, the Ponzano-Regge amplitude for $N_{t}$ time slices and a twist $N_{\gamma}$ is expressed as:
\begin{equation}
	\la \text{PR} | \Psi_{j = \f12,\iota[\lambda,\rho]}  \ra
	\,=\,
	\sum_{\{x_{n}^{(t)}\}}
	\lambda^{N_{x}N_{t}-\# x}
	\rho^{\# x}
	s_{[M,M]}
	\Bigg{[}
	\Big{(}
	\cC_{N_{\gamma}}\circ C_{\{x_{n}^{(N_{t}-1)}\}}\circ\dots\circ C_{\{x_{n}^{(0)}\}}
	\Big{)}^{-1}\Bigg{]}
	\,,
	\label{eqn:permut}
\end{equation}
where $\# x$ is the total number of $u$-channel intertwiner insertions on the whole lattice, i.e. the sum over all time slices of the cardinal of the sets $\{x_{n}^{(t)}\}$.

Studying the statistics of the composition of cycles is definitely a non-trivial combinatorial problem. Since we are mostly interested in the thermodynamical limit $N_{x},N_{t}\rightarrow\infty$, the most efficient approach is to look for a mapping of our spin evaluation onto known statistical models. This spin network evaluation for spin $\f12$ maps onto the 6-vertex model, which is integrable and for which the transfer matrix is known (and actually expressed and solved in terms of sums over permutations, see e.g. \cite{2016arXiv161109909D}). 
\\

\subsubsection{Mapping onto the 6-vertex model}
The 6-vertex model is defined on a regular square lattice. The variables of this model are a sign assigned to each edge of the lattice, which can be thought of as the orientation of the edges as represented on figure \ref{fig:6_vertex}. Each node is required to have the same number of ingoing and outgoing edges. In other words, at each node there are two incoming and two outgoing arrows. The partition function is defined as a sum over all admissible arrow configurations. 
As drawn on figure \ref{fig:6_vertex}, this gives six allowed vertex configurations. At each node, the simultaneous reversal of all four arrows is considered to be a symmetry of the model. This leaves us with three pairs of node configurations to which one associates three weights $a$, $b$ and $c$. Finally the partition function is given as
\begin{equation}
Z_\text{6-vertex} \,=\, \sum_{\text{arrows}} a^{(\#_{\rm I}+\#_{\rm II})} \,b^{(\#_{\rm III}+\#_{\rm IV})}\, c^{(\#_{\rm V}+\#_{\rm VI})}
\end{equation}
where $\#_i$ with $i={\rm I},\ldots,{\rm VI}$ represents the number of vertices in configuration $i$.
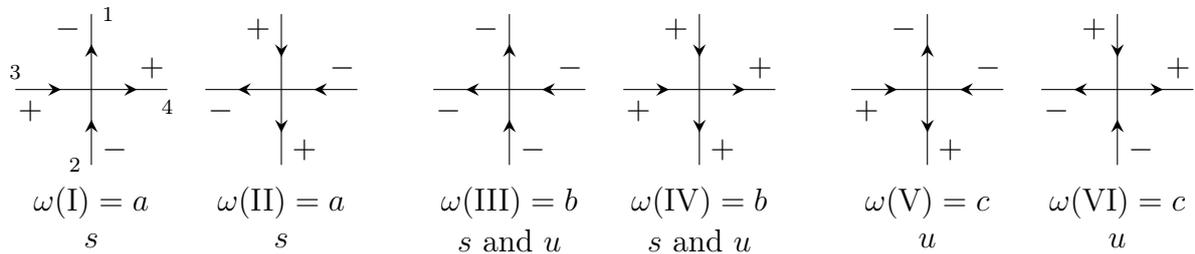
\begin{figure}[htb!]
	\begin{center}
		\begin{tikzpicture}[scale=1]
		%a-value
		\draw[decoration={markings,mark=at position 0.3 with {\arrow[scale=1.5,>=stealth]{>}}},decoration={markings,mark=at position 0.8 with {\arrow[scale=1.5,>=stealth]{>}}},postaction={decorate}]  (0,-1) node[left ]{\tiny{2}}-- node[pos=0.1,right]{$-$} node[pos=0.9,left]{$-$} (0,1)node[right ]{\tiny{1}}; \draw[decoration={markings,mark=at position 0.3 with {\arrow[scale=1.5,>=stealth]{>}}},decoration={markings,mark=at position 0.8 with {\arrow[scale=1.5,>=stealth]{>}}},postaction={decorate}] (-1,0)node[above ]{\tiny{3}}-- node[pos=0.1,below]{$+$} node[pos=0.9,above]{$+$} (1,0)node[below ]{\tiny{4}};
		\draw (0,-1.5) node[scale=1]{$\omega(\mathrm I) =a$}; \draw (0,-2) node[scale =1]{$s$};
		
		\draw[decoration={markings,mark=at position 0.3 with {\arrow[scale=1.5,>=stealth]{<}}},decoration={markings,mark=at position 0.8 with {\arrow[scale=1.5,>=stealth]{<}}},postaction={decorate}]  (2.5,-1) -- node[pos=0.1,right]{$+$} node[pos=0.9,left]{$+$} (2.5,1); 
		\draw[decoration={markings,mark=at position 0.3 with {\arrow[scale=1.5,>=stealth]{<}}},decoration={markings,mark=at position 0.8 with {\arrow[scale=1.5,>=stealth]{<}}},postaction={decorate}] (1.5,0)-- node[pos=0.1,below]{$-$} node[pos=0.9,above]{$-$} (3.5,0);
		\draw (2.5,-1.5) node[scale=1]{$\omega(\mathrm{II}) =a$}; \draw (2.5,-2) node[scale =1]{$s$};

		%b-value
		\draw[decoration={markings,mark=at position 0.3 with {\arrow[scale=1.5,>=stealth]{>}}},decoration={markings,mark=at position 0.8 with {\arrow[scale=1.5,>=stealth]{>}}},postaction={decorate}]  (5.5,-1) -- node[pos=0.1,right]{$-$} node[pos=0.9,left]{$-$} (5.5,1); 
		\draw[decoration={markings,mark=at position 0.3 with {\arrow[scale=1.5,>=stealth]{<}}},decoration={markings,mark=at position 0.8 with {\arrow[scale=1.5,>=stealth]{<}}},postaction={decorate}] (4.5,0)-- node[pos=0.1,below]{$-$} node[pos=0.9,above]{$-$} (6.5,0);
		\draw (5.5,-1.5) node[scale=1]{$\omega(\mathrm{III}) =b$}; \draw (5.5,-2) node[scale =1]{$s$ and $u$};
		
		\draw[decoration={markings,mark=at position 0.3 with {\arrow[scale=1.5,>=stealth]{<}}},decoration={markings,mark=at position 0.8 with {\arrow[scale=1.5,>=stealth]{<}}},postaction={decorate}]  (8,-1) -- node[pos=0.1,right]{$+$} node[pos=0.9,left]{$+$} (8,1); 
		\draw[decoration={markings,mark=at position 0.3 with {\arrow[scale=1.5,>=stealth]{>}}},decoration={markings,mark=at position 0.8 with {\arrow[scale=1.5,>=stealth]{>}}},postaction={decorate}] (7,0)-- node[pos=0.1,below]{$+$} node[pos=0.9,above]{$+$} (9,0);
		\draw (8,-1.5) node[scale=1]{$\omega(\mathrm{IV}) =b$}; \draw (8,-2) node[scale =1]{$s$ and $u$};
		
		%c-value
		\draw[decoration={markings,mark=at position 0.3 with {\arrow[scale=1.5,>=stealth]{<}}},decoration={markings,mark=at position 0.8 with {\arrow[scale=1.5,>=stealth]{>}}},postaction={decorate}]  (11,-1) -- node[pos=0.1,right]{$+$} node[pos=0.9,left]{$-$} (11,1); 
		\draw[decoration={markings,mark=at position 0.3 with {\arrow[scale=1.5,>=stealth]{>}}},decoration={markings,mark=at position 0.8 with {\arrow[scale=1.5,>=stealth]{<}}},postaction={decorate}] (10,0)-- node[pos=0.1,below]{$+$} node[pos=0.9,above]{$-$} (12,0);
		\draw (11,-1.5) node[scale=1]{$\omega(\mathrm{V}) =c$}; \draw (11,-2) node[scale =1]{$u$};
		
		\draw[decoration={markings,mark=at position 0.3 with {\arrow[scale=1.5,>=stealth]{>}}},decoration={markings,mark=at position 0.8 with {\arrow[scale=1.5,>=stealth]{<}}},postaction={decorate}]  (13.5,-1) -- node[pos=0.1,right]{$-$} node[pos=0.9,left]{$+$} (13.5,1); 
		\draw[decoration={markings,mark=at position 0.3 with {\arrow[scale=1.5,>=stealth]{<}}},decoration={markings,mark=at position 0.8 with {\arrow[scale=1.5,>=stealth]{>}}},postaction={decorate}] (12.5,0)-- node[pos=0.1,below]{$-$} node[pos=0.9,above]{$+$} (14.5,0);
		\draw (13.5,-1.5) node[scale=1]{$\omega(\mathrm{VI}) =c$}; \draw (13.5,-2) node[scale =1]{$u$};
		
		\end{tikzpicture}
	\end{center}
	\caption{The 6 arrow configurations around a node in the 6-vertex model and the 0-spin intertwiner channel that they correspond to, with the magnetic moment $m$ on each link.}
	\label{fig:6_vertex}
\end{figure} 

Let us compare this to the Ponzano-Regge partition function. Ignoring for the moment the group element $g$ attached to the last time slice, in the Ponzano-Regge partition function one sums over magnetic indices $m=\pm\f12$ on each link and the weights are determined by the choice of intertwiners. To map the spin-$\f12$ Ponzano-Regge model to the 6-vertex model, we match magnetic index configurations with arrow configurations. Specifically, we identify $m=+\f12$ with arrows that point to the right or downwards (along the time direction) and $m=-\f12$ to arrows that point to the left or upwards (opposite to the time direction).

We furthermore choose to expand the general spin $\f12$ intertwiner on the spin-0 intertwiners in the $s$-channel and $u$-channel, leaving the $t$-channel aside, as explained above.\footnote{
In order to include the $t$-channel in the mapping by considering a generic intertwiner as:
$$
\iota_{m_1m_2m_2m_4} \,=\,
\lambda \,\delta_{m_1 m_2} \delta_{m_3 m_4}
+ \mu \,\delta_{m_1 m_3} \delta_{m_2 m_4} 
+ \rho \,\delta_{m_1 m_4} \delta_{m_2 m_3}
\,, 
$$
we would need to introduce a new pair of node configurations corresponding to four incoming (or outgoing) arrows. This readily leads to a mapping onto an 8-vertex model. This is however not necessary in our spin-$\f12$ case, but might turn out to be useful for higher spin evaluations.}
This general intertwiner can be parametrized in such a way that the Ponzano-Regge partition function amounts to contracting the following four-valent tensors associated to the nodes of the graph ( we keep the notation $\iota$, although it already accounts for the presence of $\su(2)$ structure map on each link):

\begin{equation}
\iota{[\lambda,\rho]}_{m_1m_2m_2m_4} \,=\, \lambda \,\delta_{m_1 m_2} \delta_{m_3 m_4} + \rho \,\delta_{m_1 m_4} \delta_{m_2 m_3} \, . 
\end{equation}
Notice that this tensor only gives non-vanishing weights for the six configurations allowed by the 6-vertex model (after translating arrows into magnetic indices in the way just described). Evaluating this tensor for each of the six allowed node configurations, we obtain the following weights for the 6-vertex models in terms of our intertwiner parametrization
\begin{equation}
a=\lambda \,,
\quad
b=\lambda+\rho \,,
\quad
c=\rho \, .
\end{equation}
The expression of the $b$-coupling is actually reminiscent of the factors $(\lambda+\rho)$ of the no-twist formula \ref{chap4:eq:no_twist_simple_onehalf} derived earlier.

We can then use all the results obtained for the 6-vertex model to study our spin-$\f12$ Ponzano-Regge amplitude, especially the diagonalization and thermodynamic limit of its transfer matrix \cite{2016arXiv161109909D}. Of particular interest will be the gravitational interpretation of phase transitions in the model. However, we defer its study to future work. Before moving on, let us stress a key point. The Ponzano--Regge partition function on the twisted solid torus does {\it not} get mapped onto the 6-vertex partition function on a twisted torus, for it requires the insertion of the Haar intertwiner on the last time slice. This should correspond to the insertion of a specific  (non-local) operator in the 6-vertex model. If we call $\cF$ the 6-vertex transfer matrix, this means that we should not only look at the  6-vertex partition function $\mathrm{Tr}\left( \cF^{N_{t}}\right)$, but rather study $\mathrm{Tr}\left( \cF^{N_{t}}\cT_\gamma\cG \right)$ where $\cT_\gamma$ and $\cG$  are the operators implementing the torus twist and the Haar intertwiner, respectively.
This step is crucial, since it is precisely how the integration over all possible non-trivial (bulk) monodromies is taken into account. 
\newpage

In this chapter we started a systematic investigation of quasi-local three-dimensional non-perturbative gravity by looking at the partition function of the Ponzano-Regge model for a simple class of boundary state. Contrary to the quantum Regge calculus approach, the bulk theory is exactly solved, and we are left with a simple theory living at the boundary of the space-time. As such, it is a perfect ground to start looking at quasi-local holography.

As a first test, we considered the question of whether our setting can reproduce the structure of the partition function on the twisted solid torus as computed by a range of other methods, see the first section of this chapter and chapter \ref{chap2}. We found that already a very simple, and somewhat non-geometric, choice of boundary state can, in an appropriate limit, reproduce the characteristic pole structure of the one-loop partition function of three-dimensional gravity, seen as a function of the (Dehn) twist angle. The state involves a specific choice of intertwiner, representing a discretization of the toroidal boundary by ``fuzzy squares'', and of a homogeneous spin. The appropriate limit, on the other hand, consists in taking an infinitely refine lattice. As a consequence of keeping a fixed value of the spin, this limit corresponds to an asymptotic, infinite radius, limit but not to a continuum limit. The result is that in the limit, the renormalized amplitude, develops poles at every rational value of the twist angle.

We also analysed in some detail the partition function for a spin network state featuring only minimal spins $j=\f12$. In this case we keep the choice of intertwiner completely arbitrary, albeit uniform. What we showed is that the resulting partition function can be mapped to an interesting combinatorial problem, which can also be mapped onto a version of the 6-vertex model. To take into account the possibility of monodromies around the single (bulk) non-contractible cycle, the model has to be augmented by the insertion of specific non-local operators, which winds around the opposite cycle of the twisted boundary torus. This operator is essentially a rewriting of the so-called Haar intertwiner, and can be expressed in terms of combinatorial objects. The details of such insertions is still something that need to be work out. 

In the next chapter, we will focus on a more geometrical approach, and we will relate the Ponzano-Regge amplitude with coherent boundary state with the exact formula of the BMS character.

	\newpage
	~
	\thispagestyle{empty}
	
	\chapter{One-loop evaluation of the Ponzano-Regge Amplitude for Coherent Boundary State}
\label{chap5}

In this chapter, we focus on the computation of the Ponzano-Regge amplitude given a coherent spin network state as boundary state. The advantage of such a state is that it has a good behaviour in the asymptotic limit, since it is then peaked on a particular polygon as explained in the previous chapter. Even though we are not able to perform an exact computation of the amplitude in that case, we can easily do a saddle point approximation. This computation allows to recover the BMS character, on top of non-perturbative quantum corrections. The starting point of the computation is given by equation \eqref{chap4:eq:formula_to_compute}. Note that it is not one-loop in the sense of quantum field theory. It corresponds to a WKB approximation of the path integral, leading to the first correction at the classical action contribution for the evaluation of the path integral. This result can, and will, however, be compared with the quantum field theory like one-loop computation, as those detailed in the chapter \ref{chap2} and to the quantum Regge calculus detailed in the chapter \ref{chap4}. This chapter is based on \cite{Dittrich:2017rvb,Dittrich:2018xuk}

\section{Coherent spin network boundary state on a square lattice}

The discretization of the torus considered here is the one introduced in the last chapter. Recall that the boundary is therefore described by a square lattice, as depicted in figure \ref{chap4:fig:induced_bdry_square}. As we introduced in the last chapter, to define a coherent spin network in the boundary, we assign to each node of the dual discretization a coherent intertwiner. Instead of the general framework presented previously, we want to keep in mind now a particular geometrical picture: the square lattice is built as the superposition of a given rectangle. Hence, the coherent intertwiners is given by only two spinors and two spins, encoding the quantum rectangle, instead of four of each. The reason for this choice will make more sense later when we will make use of a Fourier Transform to compute the one-loop amplitude. If the spins are kept depending on the lattice position, the Fourier Transform will not be adapted for the computation. In the next chapter however, we will look at a boundary state built as a spin superposition, i.e. we will sum over the spins associated to the links. 

More particularly, we consider the vertical edges of the rectangle to be along the direction $\hat z$ with spin $T$ and the horizontal edges to be along the direction $\hat x$ with spin $L$. Thus, the coherent intertwiner is
\begin{equation}
	\iota_{t,x} = \int_{\SU(2)} \dd G_{t,x} \: G_{t,x} \act \left( |L, +\ra \otimes  |T,\up ]\otimes | L , +] \otimes | T, \up\ra \right) \;.
	\label{chap5:eq:coherent_intertwiner_square_lattice}
\end{equation}
In these notations, $| \uparrow \ra = (1,0)^{t} $  and $| + \ra = \f{1}{\sqrt{2}}(1,1)^{t} $ encode the direction $\hat z$ and $\hat x$ respectively. Due to the orientation of the rectangle, it is evident to see that the opposite edges are then associated to the spinors $| \uparrow ] = (0,1)^{t} $, $| + ] = \f{1}{\sqrt{2}}(-1,1)^{t}$ associated to $-\hat z$, $- \hat x$, see figure \ref{chap5:fig:coherent_intertwiner_rectangle}. Note that the previous intertwiner is defined for all links outgoing. On the other hand, the spinors are defined with respect to the orientation of the edges of the rectangle.
\begin{figure}[!htb]
	\centering
	\begin{tikzpicture}[scale=1.3]		
		% system coordinate
		\coordinate (O) at (-5,1); \coordinate (Oz) at (-5,0.5); \coordinate (Ox) at (-4.5,1);
		
		\draw[->] (O) --node[pos=1,below]{$\hat{z}$} (Oz); \draw[->] (O) --node[pos=1,right]{$\hat{x}$} (Ox);
		
		% rectangle
		\coordinate (R1) at (-1.5,1); \coordinate (R2) at (-1.5,-1); \coordinate (R3) at (1.5,-1); \coordinate (R4) at (1.5,1);
		
		\draw[-<-=0.75] (R1)--(R2); \draw[-<-=0.75] (R2)--(R3); \draw[-<-=0.75] (R3)--(R4); \draw[-<-=0.75] (R4)--(R1);
		
		%intertwiner
		\coordinate (I) at (0,0);
		\coordinate (Iup) at (0,2); \coordinate (Idown) at (0,-2);
		\coordinate (Ir) at (2.5,0); \coordinate (Il) at (-2.5,0);
		
		\draw[red] (I) node{$\bullet$};
		\draw[->-=0.25,red] (I)-- node[pos=0.75,left]{$L$}(Iup); \draw[->-=0.25,red] (I)--node[pos=0.75,right]{$L$}(Idown); \draw[->-=0.5,red] (I)--node[pos=0.75,above]{$T$}(Ir); \draw[->-=0.5,red] (I)--node[pos=0.75,below]{$T$}(Il); 
		
		\draw[red] (Iup) node[above]{$| + \ra$}; \draw[red] (Idown) node[below]{$| + ]$}; \draw[red] (Ir) node[right]{$| \uparrow \ra$}; \draw[red] (Il) node[left]{$| \uparrow ]$};
		
	\end{tikzpicture}
	\caption{Representation of a rectangle, in black, acting as the building block of the discretization of the boundary of the torus. The vertical edges are associated to the spin $T$ and are in the direction $\pm \hat{z}$ while the horizontal edges are associated to the spin $L$ and the direction $\pm \hat{x}$. In red is the dual of the rectangle. At the bullet lives an intertwiner associated with the spinor $|+\ra$, $|+]$, $|\uparrow\ra$ and $|\uparrow]$ and the spins $L$ and $T$ respectively.}
	\label{chap5:fig:coherent_intertwiner_rectangle}
\end{figure}
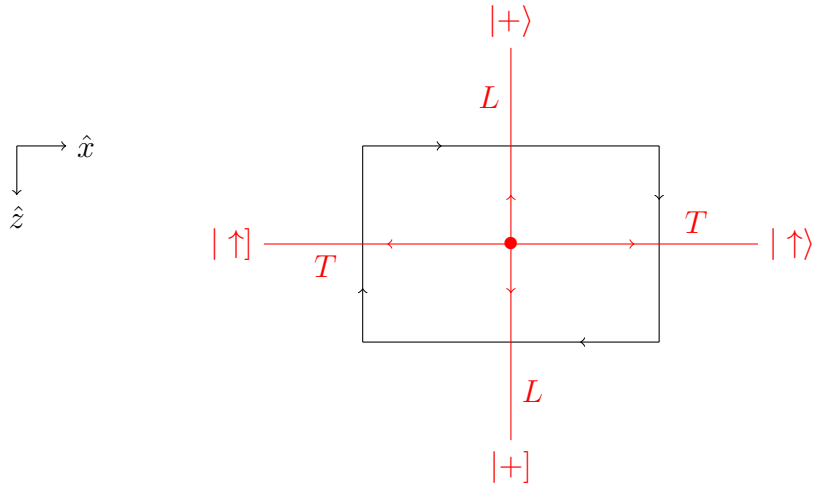
The main difference with the dressed boundary discretization depicted figure \ref{chap4:fig:induced_bdry_square} is that the spins do not depend on the lattice position anymore. We consider an anisotropic distribution for the spins, with the spin $L$ associated to the vertical links, dual to spatial edges and the spin $T$ associated to horizontal links, dual to temporal edges.

The boundary state then associated to such a choice of coherent intertwiners for each node of the dual discretization coming from the general definition \ref{chap4:eq:general_coherent_SN} is then
\begin{Definition}
	The coherent spin network state of an anisotropic spins square lattice with homogeneous coherent intertwiners defined by \eqref{chap5:eq:coherent_intertwiner_square_lattice} is\footnote{For simplicity of the notation, we do not make explicit the dependency on $\iota$ and the spins of the boundary state.}
	\begin{equation*}
	\Psi_{coh}(g^h_{t,x},g_{t,x}^v) =  \left[ \prod_{(t,x)} \int_{\SU(2)} \dd G_{t,x} \right] \, 
	\prod_{t,x} \la \up| G_{t,x+1}^{-1} g^h_{t,x} G_{t,x} | \up\ra^{2T}  \la +| G_{t+1,x}^{-1} g^v_{t,x} G_{t,x} | +\ra^{2L} \; .
	\end{equation*}
	\label{chap5:def:coherent_bdry_state}
\end{Definition}

Recall that the $g_{t,x}^{h,v}$ correspond to the $\SU(2)$ holonomies associated to the horizontal (resp. vertical) edges of the square lattice, see figure \ref{chap4:fig:induced_bdry_square}. One might note that the state does not depend on the $\SU(2)$ structure map anymore. The reason is to be found in its use for ingoing links as previously explained. Taking into account the orientation of the boundary graph, the spinors $|\uparrow ]$ and $| + ]$ are then associated to ingoing links. In the definition of the intertwiner \eqref{chap5:eq:coherent_intertwiner_square_lattice}, they must be replaced by
\begin{equation}
	[\uparrow| \varsigma = \la \uparrow | \varsigma^2 = \la \uparrow | (- \id) \quad \text{and} \quad \la + | (-\id) \;.
\end{equation}

After this replacement, the intertwiner \eqref{chap5:eq:coherent_intertwiner_square_lattice} becomes
\begin{equation*}
	\iota_{t,x} =
	\int_{\SU(2)} \dd G_{t,x} \: G_{t,x} \act \left( |L, +\ra \otimes  \la T,\up |(-\id) \otimes \la L , + | (-\id) \otimes | T, \up\ra \right) \; .
\end{equation*}

Therefore, the contribution per horizontal link is
\begin{equation}
	(-1)^{2 T} \la \uparrow| G_{t,x+1}^{-1} g_{t,x}^{h} G_{t,x} | \uparrow \ra
\end{equation}
and per vertical link
\begin{equation}
	(-1)^{2 L} \la + | G_{t+1,x}^{-1} g_{t,x}^{v} G_{t,x} | + \ra \; .
\end{equation}

The sign contribution that arises from the structure map then corresponds to a switch of the link orientation. However, as we previously explained, due to the invariance of the integrand of the spin network state, all four possible orientation for a homogeneous definition of the intertwiner on the lattice are equivalents. Hence, the signs can be safely removed without loss of generality and the boundary state we consider is truly given by the definition \ref{chap5:def:coherent_bdry_state}.

\subsection{Ponzano-Regge amplitude on a homogeneous square lattice}

Now that we have the definition of the boundary state, we can focus on the computation of the Ponzano-Regge amplitude for this class of boundary condition. Inserting the definition of the boundary state \ref{chap5:def:coherent_bdry_state} into the definition of the gauge-fixed Ponzano-Regge amplitude for a given boundary state \eqref{chap4:eq:formula_to_compute} on the torus, we get
\begin{equation}
\begin{split}
	\la Z_{PR}^{\cK}| \Psi_{coh}  \ra = 
	\f{1}{\pi}\int_{0}^{2\pi} \text{d}\varphi \; \sin^2(\varphi) \;  \prod_{(t,x)} \int_{\SU(2)} \dd G_{t,x} \; &\la \up| G_{t,x+1}^{-1} G_{t,x} | \up\ra^{2T} \\ &  \la +| G_{t+1,x}^{-1} e^{i \varphi \sigma_{z}} G_{t,x} | +\ra^{2L} \; .
\end{split}
\end{equation}

For the purpose of the computation, it is convenient to rewrite the amplitude making the appearance of an "action" explicit.
\begin{subequations}
	\begin{equation}
	\la Z_{PR}^{\cK}| \Psi_{coh}  \ra 
	=  \left[\frac{1}{\pi}\int_{0}^{2\pi} \dd \varphi \,\sin^2\left(\varphi\right)  \prod_{(t,x)} \int_{\SU(2)} \dd G_{t,x} \right]  \, \E^{-S (G_{t,x},\varphi)},
	\end{equation} 
	where the action is defined by
	\begin{equation}
	S(G_{t,x},\varphi)  = - \sum_{t,x}
	2T \ln \la \up| G_{t,x+1}^{-1}  G_{t,x} | \up\ra + 2L \ln  \la +| G_{t+1,x}^{-1} \E^{\frac{\varphi}{N_t}\sigma_z} G_{t,x} | +\ra.
	\label{chap5:eq:action} 
	\end{equation}
	\label{chap5:eq:amp_to_compute}
\end{subequations}
The intervention of the logarithm function is a mathematical shortcut and an abuse of notation. We do not actually require a choice of branch cuts for the complex logarithm. The exponential $\exp(-S)$ is well-defined and that is all that actually matters. Looking for stationary points of the action $S$ is exactly equivalent to looking for the saddle points of $\exp(-S)$. This notation however clarifies the role of the spins $T$ and $L$ as parameters assumed to be large in the logic of a saddle point approximation of the integral.

We recall that the amplitude is different compared to the case of the three dimensional ball in the sense that it is not a spin network evaluation. Due to the non-trivial topology of the bulk manifold, it remains one integration over $\SU(2)$, containing all the bulk information.

\section{Computation of the Ponzano-Regge amplitude}

We have everything we need to actually compute the amplitude now. This section is sub-divided six sub-section. The two first focus on the saddle points approximation equation of motions. The next one focuses on the geometrical reconstruction of the torus from the equations of motions while the next two sections are about the one-loop contribution. Finally, we end this section and the chapter by the discussion of the one-loop expansion of the Ponzano-Regge amplitude.

\subsection{Saddle point approximation: critical point equations}

The dominant classical contribution is given by a critical configuration $o$ where the real part of the action is an absolute minimum and its first derivative vanishes
\begin{equation}
	\Re(S)|_o \leq \Re(S) 
	\qquad\text{and}\qquad
	S' |_o = 0.
	\label{chap5:eq:critical_pts}
\end{equation}

To resolve the first condition of the critical configuration, we make use of the Cauchy-Schwarz inequality. First, recall that the real part of the complex logarithm is the logarithm of the module of the argument $\Re(\ln(z)) = \ln(|z|)$ for all complex $z$. In our case, we consider two normalized spinors $|\xi \ra$ and $|\eta \ra$. The real part of such logarithm is
\begin{equation}
	\Re(\ln\la \xi | \eta \ra) = \ln |\la \xi | \eta \ra| \;.
\end{equation}
Considering now the Cauchy-Schwarz inequality and the monotony of the logarithm, we get the relation
\begin{equation}
	\Re(\ln\la \xi | \eta \ra) = \ln |\la \xi | \eta \ra| \leq 0
\end{equation}
since $|\la \xi | \eta \ra| \leq 1$. In particular, the inequality can only be saturated if the spinors are proportional. That is, if it exist some phase $\psi\in [0,2\pi]$ such that
\begin{equation*}
	\la \xi | \eta \ra = e^{i \psi} \; .
\end{equation*}

Applying this result on the action \eqref{chap5:eq:action} returns that the critical action as vanishing real part and is determined by the gluing equation
\begin{subequations}
	\begin{align}
	\la \up| G_{t,x+1}^{-1}  G_{t,x} | \up\ra & = \E^{-\I \psi^T_{t,x}} \\
	\la +| G_{t+1,x}^{-1} \E^{\frac{\varphi}{N_t}\sigma_z} G_{t,x} | +\ra & \,=\, \E^{-\I \psi^L_{t,x}}
	\end{align}
	\label{chap5:eq:gluing_equation}
\end{subequations}
for some phases $\psi^{T,L}_{t,x}\in [0,2\pi]$ representing the proportionality conditions between spinors. This is, however, not the end of the story in our case. Indeed, recall that the integrand is symmetric under the transformation $G_{n} \rightarrow -G_{n}$. It is immediate to see that this corresponds to the freedom of adding $\pi$ to the phases $\psi^{T,L}_{t,x}$. Hence, taking into account this symmetry, we can safely defined the phases in a range of $\pi$: $\psi^{T,L}_{t,x}\in [0,\pi]$.

On-shell of these equations, the action $S$ takes the particular form
\begin{equation}
	S|_o = - \I \sum_{t,x} T (2\psi^T_{t,x}) + L (2\psi^L_{t,x}).
	\label{chap5:eq:classical_action}
\end{equation}
We see that from the gravitational perspective it is appealing to interpret these angles as (demi) dihedral angles so that the on-shell action above reproduces a discrete version of the Gibbons-Hawking-York boundary term, see figure \ref{chap6:fig:dihedral_angle}. This interpretation will be confirmed in the following when we will focus on the geometrical analysis and reconstruction of the critical point equations. Note also that, due to our parametrization of $\SU(2)$, the associated true $\SO(3)$ angles are $2\psi^{L,T}_{t,x}$, corresponding to the dihedral angles. Thanks to the symmetry of the integrand, these angles are defined in $[0,2\pi]$ as expected of dihedral angles. 

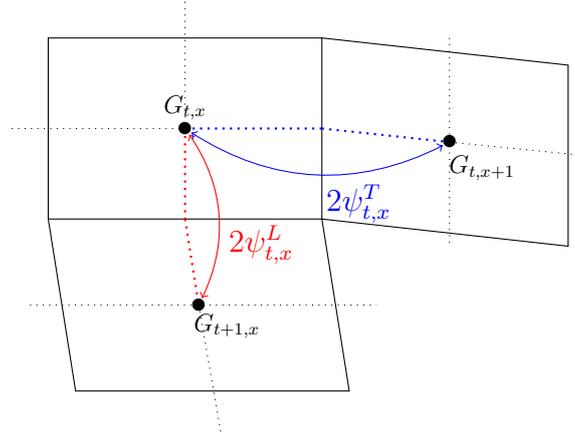
\begin{figure}[!htb]
	\begin{center}
		\begin{tikzpicture}[scale=1.2,line/.style={<->,shorten >=0.1cm,shorten <=0.1cm}]
		\coordinate(a1) at (0,0); \coordinate(b1) at (0.3,-1.9);
		\coordinate(a2) at (3,0); \coordinate(b2) at (3.3,-1.9); \coordinate(d1) at (5.7,-0.3);
		\coordinate(h1) at (0,2); 
		\coordinate(h2) at (3,2); \coordinate(d2) at (5.7,1.7);

		\draw (a1)--(a2)--(h2)--(h1)--cycle;
		\draw (a1)--(b1)--(b2)--(a2);
		\draw (a2)--(d1)--(d2)--(h2);
		
		\draw[dotted] (1.5,2.4)--(1.5,1);
		
		\draw[dotted] (-0.4,1)--(1.5,1); 
		
		\draw[dotted,blue,thick] (1.5,1)--(3,1); \draw[dotted, blue, thick] (3,1) -- (4.4,0.85);
		\draw[dotted,thick,red] (1.5,1)--(1.5,0); \draw[dotted,red,thick] (1.5,0)--(1.65,-0.95);

		\draw[dotted] (1.65,-0.95)--(1.9,-2.4);
		\draw[dotted] (-0.2,-0.95)--(3.6,-0.95);
		
		\draw[dotted] (4.4,0.85)--(5.9,0.7); \draw[dotted] (4.4,2)--(4.4,-0.3);
		
		\draw (1.5,1) node[above, scale=0.8]{$G_{t,x}$}; \draw (1.5,1) node{$\bullet$};
		\draw (1.5,-0.95) node[below right,scale=0.8]{$G_{t+1,x}$};\draw (1.65,-0.95) node{$\bullet$};
		\draw (4.3,0.8) node[below right,scale=0.8]{$G_{t,x+1}$};\draw (4.4,0.85) node{$\bullet$};
		
		\path [red,line,bend left] (1.5,1) edge node[midway,below right]{$2\psi_{t,x}^{L}$} (1.65,-0.95);
		\path [blue,line,bend right] (1.5,1) edge node[midway,below right]{$2\psi_{t,x}^{T}$} (4.4,0.85);
		\end{tikzpicture}
	\end{center}
	\caption{Three rectangles of the boundary discretization dual to 3 nodes. The dotted lines are the links of the dual lattice. In red, the dihedral angle $2 \psi_{t,x}^{L}$ between the rectangle dual to the vertices $(t,x)$ and $(t+1,x)$ and in blue the dihedral angle $2 \psi_{t,x}^{T}$ between the rectangle dual to the vertices $(t,x)$ and $(t,x+1)$.}
	\label{chap6:fig:dihedral_angle}
\end{figure}

It is useful to make explicit the parallel transport between the spinors in the gluing equations. We have
\begin{subequations}
	\begin{align}
		G_{t,x+1}^{-1}  G_{t,x} | \up\ra & = \E^{-\I \psi^T_{t,x}} |\up \ra \\
		G_{t+1,x}^{-1} \E^{\frac{\varphi}{N_t}\sigma_z} G_{t,x} | +\ra & \,=\, \E^{-\I \psi^L_{t,x}} |+ \ra
	\end{align}
\end{subequations}
It is clear then that the gluing equations tell us that the spinor $|\up\ra$ (resp. $|+\ra$) is left invariant, up to a phase, by its parallel transport by $G_{t,x+1}^{-1}  G_{t,x}$ (resp. $G_{t+1,x}^{-1} \E^{\frac{\varphi}{N_t}\sigma_z} G_{t,x}$). In the following, we will see that the phases are constrained due to the periodicity conditions on the lattice.
\\

The stationary condition $S'|_o=0$ is most easily studied introducing right derivatives (left-invariant vector fields) of functions on $\SU(2)$. Schematically, for a function on $\SU(2)$ $f$, we have
\begin{equation*}
	\nabla_k f(G) = \frac{\pp}{\pp a_k}\Big|_{\vec a =0 } f(G \E^{\I \vec a . \vec\sigma}) =\frac{d}{d t}\Big|_{t=0} f\left(G\E^{\I t\sigma_k}\right) =  \frac{d}{d t}_{|t=0} f\Big(G(\id + \I t\sigma_k)\Big).
\end{equation*}
Denoting by $\nabla_k^{t,x}$ the derivative with respect to $G_{t,x}$ in the direction $k$, the first derivatives are then
\begin{subequations}
	\begin{align}
	\nabla_k^{t,x} S
	&=    
	2\I T \frac{\la \up | \sigma_k G_{t,x}^{-1}G_{t,x-1} |\up\ra }{\la \up | G_{t,x}^{-1}G_{t,x-1} |\up\ra} 
	- 2\I T \frac{\la \up | G_{t,x+1}^{-1}G_{t,x} \sigma_k |\up\ra }{\la \up | G_{t,x+1}^{-1}G_{t,x}  |\up\ra} \notag\\
	& ~~ 2\I L \frac{\la + | \sigma_k G_{t,x}^{-1}\E^{\I \frac{\varphi}{N_t} \sigma_3} G_{t-1,x}  |+\ra }{\la + | G_{t,x}^{-1}\E^{\I \frac{\varphi}{N_t} \sigma_3} G_{t-1,x}  |+ \ra}
	- 2\I L \frac{\la + | G_{t+1,x}^{-1}\E^{\frac{\varphi}{N_t} \sigma_3} G_{t,x}   \sigma_k |+\ra }{\la + | G_{t+1,x}^{-1}\E^{\frac{\varphi}{N_t} \sigma_3} G_{t,x}  |+ \ra},\\
	%%%%%%
	\pp_\varphi S
	= & 
	-\frac{2\I L}{N_t} \sum_{t,x}  \frac{\la +| G_{t+1,x}^{-1}  \E^{\frac{\varphi}{N_t}\sigma_3} \sigma_3 G_{t,x} | +\ra}{ \la +| G_{t+1,x}^{-1} \E^{\frac{\varphi}{N_t}\sigma_3} G_{t,x} | +\ra}.
	\end{align}
	\label{chap5:eq:first_derivatives}
\end{subequations}
Evaluated on-shell of the gluing equations \eqref{chap5:eq:gluing_equation}, the stationary conditions simplify greatly and we obtain
\begin{subequations}
	\begin{align}
	\nabla^k_{t,x} S |_o
	= &   -2 \I\Big( T \la \up | \sigma_k |\up\ra 
	- T \la \up | \sigma_k |\up\ra + L {\la + | \sigma_k  |+\ra }
	- L \la + |   \sigma_k |+\ra\Big)
	\equiv 0 \; ,\\
	\pp_\varphi  S|_o
	= & \frac{2 \I L}{N_t} \sum_{t,x}  \la +| G_{t,x}^{-1}  \sigma_3 G_{t,x} | +\ra %
	= \frac{2 \I L}{N_t} \hat z. \sum_{t,x} G_{t,x}\triangleright \hat x = 0 \; .
	\end{align} 
\end{subequations}
On the one hand, the first equation just gives us the closure for the intertwiner at the node $(t,x)$. This corresponds to the closure constraint we consider previously when defining quantum polygons. The fact that we recover this result in the saddle approximation is more of a safety check than anything else. It only confirms that the geometrical conditions in the saddle points correspond geometrically to a quantum polygons.

The second equation, on the other hand, gives a global constraint on the solutions for the critical points. In order to obtain this expression, we used the fact that the projection of a spinor along the Pauli matrices $\sigma_{i}$ returns its component along the direction $i$.

To ease the following analysis, let us define the rotate element
\begin{equation*}
	\tl G_{t,x} = \E^{- \frac{t}{N_t} \varphi \sigma_3} G_{t,x}.
\end{equation*}
In terms of these new variables, the gluing and saddle point equations read:
\begin{subequations}
	\begin{align}
	& \tl G_{t,x+1}^{-1} \tl  G_{t,x} |\up\ra  = \E^{-\I \psi^T_{t,x}} |\up\ra \; , \label{chap5:eq:gluing1}	\\
	& \tl G_{t+1,x}^{-1} \tl G_{t,x} |+\ra = \E^{-\I \psi^L_{t,x}}  |+\ra  \; , \label{chap5:eq:gluing2} \\
	& \hat z. \sum_{t,x} \tl G_{t,x} \triangleright \hat x = 0 \; . \label{chap5:eq:constitancy_cycle}
	\end{align}
\end{subequations}
It is immediate to see that, thanks to the change of variables, the gluing condition in the temporal and spatial direction are now symmetric. Hence, a unique treatment will take care of both equations. The key to solve the gluing equations is to be found in the periodic conditions of the square lattice given by \eqref{chap4:eq:periodic_condition}
\begin{equation*}
(t, x+N_x) \sim (t,x) \sim (t+N_t , x +N_\gamma) \; .
\end{equation*}

For the $G_{t,x}$ they imply
\begin{equation}
	G_{t+N_t,x} = \pm G_{t,x+N_\gamma} \qquad \text{and} \qquad G_{t,x+N_x} = \pm G_{t,x} \; .,
\end{equation}
while, for the $\tl G_{t,x}$, they return
\begin{equation}
	\tl G_{t+N_t,x} = \pm e^{- \I \varphi \sigma_{3}} \tl G_{t,x+N_\gamma} \qquad \text{and} \qquad \tl G_{t,x+N_x} = \pm \tl G_{t,x} \; .
\end{equation}
The presence of the sign $\pm$ in the periodicity condition for the $G_{t,x}$ and $\tl G_{t,x}$ is of course due to the invariance of the action under the local transformation $G_{t,x} \rightarrow - G_{t,x}$.

\subsection{Saddle point approximation: the equations of motion}

We now focus on solving the gluing equations to find the critical solutions. As we said before, the gluing equations \eqref{chap5:eq:gluing1} and \eqref{chap5:eq:gluing2} tell us that $\tl G_{t,x+1}^{-1} \tl  G_{t,x}$ is leaving the spinor $|\up\ra$ invariant whereas $\tl G_{t+1,x}^{-1} \tl G_{t,x}$ is leaving $|+\ra$ invariant. Therefore, these $\SU(2)$ elements must be along the direction given by the spinor, namely along $\hat z$ and $\hat x$ respectively. The gluing equations also give us the value of the class angle. In formula, we have
\begin{equation}
	\tl G^{-1}_{t,x}\tl G_{t,x+1} = \E^{\I \psi^T_{t,x} \sigma_3}
	\qquad\text{and}\qquad
	\tl G^{-1}_{t,x}  \tl G_{t+1,x} = \E^{\I \psi^L_{t,x} \sigma_1} \; .
	\label{chap5:eq:recursion_relation_G}
\end{equation}

The key to solve the equations consistently is to focus on a given dual rectangle. Indeed, the previous constraints linked the four phases together in an interesting way. The basic idea is to construct the link between $G_{t+1,x+1}$ and $G_{t,x}$ either going through $G_{t,x+1}$ or $G_{t+1,x}$, see figure \ref{chap5:fig:two_way_link_GG}
\begin{figure}[htb!]
	\begin{center}
		\begin{tikzpicture}[scale = 1]
			\coordinate (A) at (3,-2); \coordinate (B) at (3,2);
			\coordinate (C) at (-3,2); \coordinate (D) at (-3,-2);

			\draw[blue] (A)--node[right,pos=0.5]{$\psi^{L}_{t,x+1}$} (B)--node[above,pos=0.5]{$\psi^{T}_{t,x}$} (C);
			\draw[red] (A)--node[below,pos=0.5]{$\psi^{T}_{t+1,x}$} (D)--node[left,pos=0.5]{$\psi^{L}_{t,x}$}(C);
		
			\draw (A) node{$\bullet$}; \draw (A)node[below right]{$(t+1,x+1)$};
			\draw (B) node{$\bullet$}; \draw (B)node[above right]{$(t,x+1)$};
			\draw (C) node{$\bullet$}; \draw (C)node[above left]{$(t,x)$};
			\draw (D) node{$\bullet$}; \draw (D)node[below left]{$(t+1,x)$};
		\end{tikzpicture}
	\end{center}
	\caption{The two ways, in red and in blue, between $(t,x)$ and $(t+1,x+1)$. The consistency between these two ways impose a constraint on the dihedral angle $\psi_{t,x}^{T}$,$\psi_{t,x+1}^{L}$, $\psi_{t,x}^{L}$ and $\psi_{t,x+1}^{T}$.}
	\label{chap5:fig:two_way_link_GG}
\end{figure}
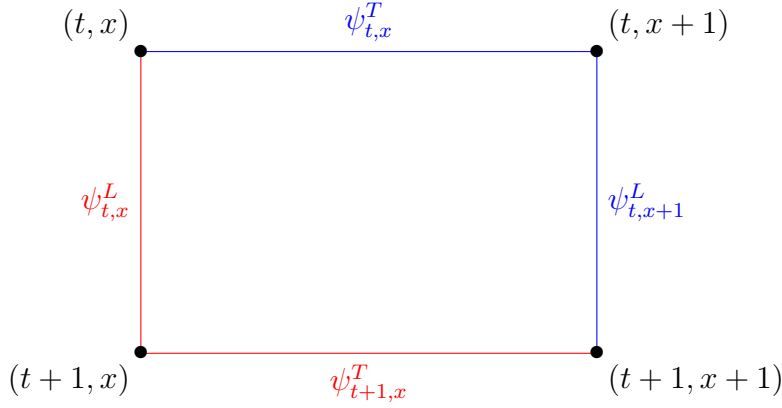

Indeed, the relation between $(t+1,x)$ and $(t+1,x+1)$ returns
\begin{equation*}
	\tl G_{t+1,x+1} = \tl G_{t+1,x} \E^{\I \psi^T_{t+1,x} \sigma_3}  \; 
\end{equation*}
while the relation between $(t,x+1)$ and $(t+1,x+1)$ is
\begin{equation*}
	\tl G_{t+1,x+1} = \tl G_{t,x+1} \E^{\I \psi^L_{t,x+1} \sigma_1} \; .
\end{equation*}
In turn, the group elements on the right hand side can both be expressed with $\tl G_{t,x}$. At the end of the day, we get
\begin{subequations}
	\begin{align}
		\tl G_{t+1,x+1} &= \tl G_{t,x} \E^{\I \psi^L_{t,x} \sigma_1} \E^{\I \psi^T_{t+1,x} \sigma_3}  
		\\
		\tl G_{t+1,x+1} &= \tl G_{t,x} \E^{\I \psi^T_{t,x} \sigma_3} \E^{\I \psi^L_{t,x+1} \sigma_1} \; .
	\end{align}
\end{subequations}

Comparing these two relations between $\tl G_{t+1,x+1}$ and $\tl G_{t,x}$, we get the constraint
\begin{equation}
	\E^{\I \psi^T_{t+1,x}\sigma_3} = \E^{-\psi^L_{t,x}\sigma_1} \E^{\I \psi^T_{t,x} \sigma_3} \E^{\I \psi^L_{t,x+1} \sigma_1}\;.
\end{equation}
This is the Euler decomposition for the $\SU(2)$ group element $\E^{\I \psi^T_{t+1,x}\sigma_3}$. Thanks to the uniqueness of the Euler decomposition, we are able to explicitly write the three different families of solutions. We name the first two the $X$ and $Z$ families and they are parametrized by
\begin{equation}
	X:\;\Big( \psi^L_{t,x} = \psi^L_t , \psi^T_{t,x} = 0\Big)
	\qquad\text{and}\qquad
	Z:\;\Big( \psi^L_{t,x} = 0 , \psi^T_{t,x} = \psi^T_x \Big) \; .
\end{equation}
Recall that the angles are defined modulo $\pi$ by the symmetry on $G_{t,x}$. The remaining family is the folding family At a given time slice $(t,x)$, this family is parametrized by
\begin{equation}
	\psi^{L}_{t,x} = \psi^L_{t,x+1} = \f{\pi}{2} \quad \text{and} \quad \psi^T_{t+1,x} = - \psi_{t,x}^{T} \; .
\end{equation}
The analogue solution $\psi^{T}_{t,x} = \psi^{T}_{t,x+1}= \f{\pi}{2}$ is however contained in the $Z$ family solution. For now, we will ignore this particular solution and come back to it later on, while discussing the geometry of the boundary.
\\

Focusing on the $X$ and $Z$ families, it is straightforward to see that \eqref{chap5:eq:recursion_relation_G} can be solved recursively. Once this is done, we find for solutions
\begin{equation}
	X:\;\tl G_{t,x} = \tl G \,\E^{\I \sum\limits_{k=0}^{t-1} \psi^L_{k} \sigma_1} 
\qquad\text{and}\qquad
	Z:\;\tl G_{t,x} = \tl G\, \E^{\I \sum\limits_{k=0}^{x-1} \psi^T_{k} \sigma_3} \; ,
\end{equation}
where $\tl G = \tl G_{0,0}$ is the starting point of the recursion process.

When expressed in terms of the variables $\tl G_{t,x}$, the gluing equations are totally symmetric in the directions $\hat x$ and  $\hat z$. The asymmetry between the two directions arises in the boundary periodic conditions and is related to the presence of  the angle $\varphi$. This makes sense considering the very origin of the variable $\varphi$ as encoding the holonomy around the non-trivial cycle of the solid torus.

The next step to the study of the solutions is to make use of the periodicity condition of the lattice. Indeed, a "true" solution must respect the periodicity conditions. We analyse the two families with respect to this aspect, starting with the $X$ family.
\\

\paragraph*{\bf $X$-family}

The $X$-family is a family of solutions that does not depend on the spatial lattice position $x$. Therefore, the periodic condition in the $x$ direction is trivially satisfied, since
\begin{equation*}
	\tl G_{t,x} = \tl G_{t,x+1} \quad \forall \, x \; .
\end{equation*}

The periodic condition in the temporal direction however returns the constraint
\begin{equation}
	\tl G_{t+N_t,x} = \pm \E^{- \I \varphi \sigma_{3}} \tl G_{t,x} \implies \E^{\I \varphi \sigma_{3}} = \pm \tl G^{-1} \, e^{-\I \sum\limits_{k=0}^{N_t-1} \psi^{L}_k \;  \sigma_1} \, \tl G\; .
\end{equation}

From this constraint, we deduce that the value of $\varphi$ is determined by the phase and that the action of $\tl G$ on the vector $\hat x$ is constrained to be
\begin{equation}
	\varphi = -\sum_{k=0}^{N_t-1} \psi^{L}_k  \; [\pi]  \qquad \text{and} \quad \tl G \act \hat x = \hat z \; .
\end{equation}
The constraint on $\tl G$ only exists however if $\varphi \neq 0 \; [\pi]$. If the sum vanishes, the periodic condition is trivially satisfied for any $\tl G$ and it is therefore unconstrained. 

The last point is to check that this family satisfies the condition \eqref{chap5:eq:constitancy_cycle}. However, it is i immediate to see that if $\sum_{t,x} \tl G_{t,x} \act \hat x$ is not the zero vector on its own, this condition can never be satisfied for $\tl G_{t,x}$ along $\hat x$. Indeed, the action of $\tl G_{t,x}$ on $\hat{x}$ is the action of $\tl G$ on $\hat{x}$, which is constrained to return $\hat{z}$. However, if the sum is the zero vector on its own, it implies that the $\psi^T_t$ must sum to $0$ modulo $\pi$. Therefore, the angle $\varphi$ is also $0$ modulo $\pi$, and we recover the unconstrained case. 

The unconstrained case corresponds to a continuous space of solutions for the saddle point equation parametrized by the unconstrained group element $\tl G$. We will get back to this case later on while studying the Hessian of the action. Here, we just point out an argument explaining why it does not contribute to the saddle point. The reason is to be found in the measure factor $\sin^{2}(\varphi)$. It is immediate to see that the case $\varphi = 0$ corresponds to the conjugacy classes of vanishing volume, and hence is killed by the measure factor.

\paragraph*{\bf $Z$-family}

For the $Z$-family, both periodicity conditions give rise to a constraint. The periodicity in the $x$ direction gives
\begin{equation}
	\tl G_{t,x+N_x} = \pm \tl G_{t,x} \implies \tl G \E^{\I \sum\limits_{k=0}^{N_x-1} \psi^{T}_k \sigma_3} \tl G^{-1} = \pm \id \;,
\end{equation}
while the periodicity condition in the $t$ direction returns
\begin{equation}
	\tl G_{t+N_t,x} = \pm \E^{-\I \varphi \sigma_3} \tl G_{t,x+N_\gamma} \implies \E^{\I \varphi \sigma_3} = \pm \tl G \E^{\I \sum\limits_{k=x}^{N_\gamma-1+x} \psi^T_{k} \sigma_3} \tl G \quad \forall \; x
\end{equation}

From these two constraints, we deduce that
\begin{subequations}
	\begin{equation}
		\sum_{k=0}^{N_x-1} \psi^{T}_{k} = 0 \; [\pi] \quad \text{and} \quad \varphi = \sum_{k=0}^{N_\gamma-1} \psi^{T}_k \; [\pi]
		\label{chap5:eq:sum_phase_solution}
	\end{equation}
	with the periodicity condition on $\psi^{T}_x$
	\begin{equation}
		\psi_{x+N_\gamma}^T = \psi_{x}^T \quad \forall \; x \; ,
		\label{chap5:eq:periodicity_closed_cycle}
	\end{equation}
	and that $\tl G$ is such that it leaves the direction $\hat z$ invariant. That is, it exists $\bar{\varphi}$ such that
	\begin{equation}
		\tl G = \E^{\I \bar{\varphi} \sigma_3} \; .
	\end{equation}
\end{subequations}

The explicit form of the $Z$ family solution is therefore
\begin{equation}
	\tl G_{t,x} = \E^{\I \left(\bar\varphi + \sum\limits_{k=0}^{x-1} \psi^T_{k}   \right) \sigma_3}
\qquad\text{or equivalently}\qquad
	G_{t,x} = \E^{\I \left(\bar\varphi + \frac{t}{N_t}\varphi + \sum\limits_{k=0}^{x-1} \psi^T_{k}   \right) \sigma_3} \; ,
\end{equation}
parametrized by the $\psi^{T}_{x}$ and $\bar \varphi$.

This solution, being along the direction $\hat z$, trivially satisfies the critical equation \eqref{chap5:eq:constitancy_cycle}. Indeed, the rotation of $\hat x$ along $\hat{z}$ necessary produces a vector orthogonal to $\hat z$.  Also, the constraint on the initial element $\tl G$ could have been easily guessed. Indeed, recall that, at the end of the day, there is still a remaining global gauge symmetry along $\hat z$, which exactly coincide with the choice of parameter $\bar\varphi$. 

Another remark can be made here. It is the constraint on $\varphi$ that select the solutions. And it is also $\varphi$ that tells us which cycle of the torus is non contractible. The two families of solutions $X$ and $Z$ corresponds to the two possible directions of the non-contractible cycle. That fact that it is the $Z$ family that matter at the end is in accordance with our choice of periodic condition. From the torus perspective, it corresponds to a choice of a given modular parameter. Therefore, this selection of families of solutions corresponds to the choice of either the Euclidean AdS or the BTZ torus as we explained previously.
\\

It is necessary to make further assumption to actually compute the saddle. Recall that $K$ is the greatest common divisor between $N_x$ and $N_\gamma$. The periodicity condition \eqref{chap5:eq:periodicity_closed_cycle} tells us that there are as many different angles that the number of close vertical loops. Considering $K=1$, hence one close loop, we have
\begin{equation}
	K= 1 \implies \psi_{x}^{T} = \psi_{x+1}^{T} = \psi^{T} \; [\pi] \; \forall \, x \; .
\end{equation}
Using equation \eqref{chap5:eq:sum_phase_solution}, we deduce that in that case
\begin{equation}
	\psi^T = \frac{\pi}{N_x} n \, [\pi]
	\quad\text{and}\quad
	\varphi = - \f{\gamma}{2} n  \, [\pi]\;
	\label{chap6:eq:solution_phase}
\end{equation}
for some $n\in\mathbb Z $, $|n| \leq \lfloor\frac{N_x}{2}\rfloor$, where $\lfloor . \rfloor$ is the floor function.

There are two cases which stand out, the case $n=0$ and, if $N_x$ is even, the case $n=\frac{N_x}{2}$. It is not complicated to see that the status of the $n=0$ solution is somewhat different. It imposes $\psi^T = 0 $ and $\varphi = 0$. Thus, it coincides to the allowed $\varphi=0$ case of the $X$-family solution and hence it belongs to a continuum set of solutions. We already argued why this configuration is suppressed in the saddle approximation. The case of $n=\frac{N_x}{2}$ corresponds however to $\psi^T_{t,x} = \pi$. For such value, we also have a continuum set of solutions for which $\psi^L_{t,x+1}=-\psi^L_{t,x}$ with the constraint $\sum\limits_{k}\psi^L_{k,x}=0\,[\pi]$. The origin of this continuum set of solutions is similar to the $n=0$ case. But the situation is much more complicated, since this time this contribution is not killed by any measure terms. We will come back to this particular case in the following. For now on, unless explicitly stated otherwise, we will restrict to the case where $N_x$ is odd. Some arguments about this situation will be given while discussing the Hessian of the action.

If $K>1$ however, the space of solutions is again a continuum one. This is easily understood by looking at the number of different $\psi_{x}^{T}$ unconstrained. There are $K-1$ unconstrained angles. For example, for $K=2$, it exists two angles, related by 
\begin{equation}
	\psi_{2}^{T} = \pi - \psi_{1}^{T} \; [\pi] \; .
\end{equation}
Therefore, the space of solutions is one-dimensional. This is easily generalizable for every $K$.
\medskip

In summary, {\it if $N_x$ is odd and $K=\mathrm{GCD}(N_\gamma, N_x)=1$, there is a finite number of (relevant) solutions labelled by parameter $n \in \N$, such that $1\leq|n|\leq \frac{N_x-1}{2}$. If $K>1$, each of the solutions above is part of a continuum $(K-1)$-dimensional family of solutions.}

\subsection{Geometry reconstruction}

The solutions to the saddle point equations encode a twisted torus {\it locally} embedded in $\mathbb R^3$ as a quadrangulated cylinder of height $\beta = N_t T $ and ``circumference'' $2\pi a = N_x L$. Recall that we refer to the horizontal direction of the torus as its spatial direction, and to the vertical one as its time direction. The details of the geometry can be read from the data above by juxtaposing neighbouring quadrilateral cells identifying their respective sides according to the gluing equations and orienting them in the embedding space according to the action of the $G_{t,x}$. 

For $n=1$, this allows to build a prism whose base is a $N_x$-sided polygon embedded in $\mathbb R^3$. The twisted torus is finally obtained by identifying the first and the last time slice after application of the twist encoded in the periodicity condition \eqref{chap4:eq:periodic_condition}. The resulting spatial cycle is contractible in the bulk, while the time cycle is not due to the topological identification. This is in agreement with the non-triviality of the holonomy $g=\E^{i \frac{\varphi}{N_t} \sigma_3}$ along the time cycle of the $Z$ family solution. Between two spatially neighbouring rectangular cells, there is a dihedral angle equal to $2\psi^T=2\pi/N_x$, while the dihedral angle $\psi^L$ between two temporally neighbouring cells vanishes, implying the flatness of the time direction, see figure \ref{chap5:fig:reconstruction_n_1}.
\begin{figure}
	\begin{center}
		\begin{tikzpicture}[scale = 3,line/.style={<->,shorten >=0.1cm,shorten <=0.1cm}]
			\coordinate (A1) at (-0.65,0); \coordinate (A2) at (-0.5,-0.1); \coordinate (A3) at (-0.3,-0.15); \coordinate (A4) at (0,-0.18); \coordinate (A5) at (0.3,-0.15); \coordinate (A6) at (0.5,-0.1); \coordinate (A7) at (0.65,0); \coordinate (A8) at (0.5,0.1); \coordinate (A9) at (0.3,0.15); \coordinate (A10) at (0,0.18); \coordinate (A11) at (-0.3,0.15); \coordinate (A12) at (-0.5,0.1);
			
			\coordinate (B1) at (-0.65,0.5); \coordinate (B2) at (-0.5,0.4); \coordinate (B3) at (-0.3,0.35); \coordinate (B4) at (0,0.32); \coordinate (B5) at (0.3,0.35); \coordinate (B6) at (0.5,0.4); \coordinate (B7) at (0.65,0.5); \coordinate (B8) at (0.5,0.6); \coordinate (B9) at (0.3,0.65); \coordinate (B10) at (0,0.68); \coordinate (B11) at (-0.3,0.65); \coordinate (B12) at (-0.5,0.6);
			
			\coordinate (C1) at (-0.65,1); \coordinate (C2) at (-0.5,0.9); \coordinate (C3) at (-0.3,0.85); \coordinate (C4) at (0,0.82); \coordinate (C5) at (0.3,0.85); \coordinate (C6) at (0.5,0.9); \coordinate (C7) at (0.65,1); \coordinate (C8) at (0.5,1.1); \coordinate (C9) at (0.3,1.15); \coordinate (C10) at (0,1.18); \coordinate (C11) at (-0.3,1.15); \coordinate (C12) at (-0.5,1.1);
			
			\coordinate (D1) at (-0.65,1.5); \coordinate (D2) at (-0.5,1.4); \coordinate (D3) at (-0.3,1.35); \coordinate (D4) at (0,1.32); \coordinate (D5) at (0.3,1.35); \coordinate (D6) at (0.5,1.4); \coordinate (D7) at (0.65,1.5); \coordinate (D8) at (0.5,1.6); \coordinate (D9) at (0.3,1.65); \coordinate (D10) at (0,1.68); \coordinate (D11) at (-0.3,1.65); \coordinate (D12) at (-0.5,1.6);
			
			\coordinate (Od) at (0,-0.3); \coordinate (OA) at (0,0); \draw coordinate (OD) at (0,1.5); \coordinate (Ou) at (0,1.8);
			
			\draw[thick] (Od)--(OA); \draw[dashed,thick] (OA)--(OD); \draw[thick,->] (OD)--(Ou); \draw (Ou) node[above right]{$g=e^{i \sigma_{3} \f{\varphi}{N_t}}$};
			
			\draw (A1)--(A2)--(A3)--(A4)--(A5)--(A6)--(A7); 
			\draw[dashed] (A7)--(A8)--(A9)--(A10)--(A11)--(A12)--(A1);
			
			\draw (B1)--(B2)--(B3)--(B4)--(B5)--(B6)--(B7); 
			\draw[dashed] (B7)--(B8)--(B9)--(B10)--(B11)--(B12)--(B1);

			\draw (C1)--(C2)--(C3)--(C4)--(C5)--(C6)--(C7); 
			\draw[dashed] (C7)--(C8)--(C9)--(C10)--(C11)--(C12)--(C1);

			\draw (D1)--(D2)--(D3)--(D4)--(D5)--(D6)--(D7); 
			\draw[dashed] (D7)--(D8)--(D9)--(D10)--(D11)--(D12)--(D1);
			
			\draw (A1)--(D1); \draw (A2)--(D2); \draw (A3)--(D3); \draw (A4)--(D4); \draw (A5)--(D5); \draw (A6)--(D6); \draw (A7)--(D7);
			
			% angle varphi
			
			\coordinate (Angle1) at (0.56,-0.16);
			\coordinate (Angle2) at (0,-0.43);
			
			\draw[thick,->] (Angle1) arc (0:-180: 0.6cm and 0.2cm);
			\draw (Angle2) node{$\varphi$};
			
			%dihedral angle
			\coordinate (L1) at (-0.58,0.7); \coordinate (L2) at (-0.57,0.18);
			\coordinate (L1t) at (-0.86,0.58); \coordinate (L2t) at (-0.82,0.07);
			\coordinate (L1tt) at (-0.87,0.53); \coordinate (L2tt) at (-0.83,0.09);
			
			\draw[thick, ->] (L1)--(L1t); \draw[thick, ->] (L2)--(L2t); \draw[thick,->] (L1tt) to[bend right] node[left]{$2 \psi_{t,x}^{L}$} (L2tt);
			
			\coordinate (T1) at (0.4,0.6); \coordinate (T2) at (0.57,0.69);
			\coordinate (T1t) at (0.57,0.34); \coordinate (T2t) at (0.83,0.6);
			\coordinate (T1tt) at (0.62,0.37); \coordinate (T2tt) at (0.8,0.57);
			
			\draw[thick, ->] (T1)--(T1t); \draw[thick, ->] (T2)--(T2t); \draw[thick,->] (T1tt) to[bend right] node[below right]{$2 \psi_{t,x}^{T}$} (T2tt);

			\coordinate (O) at (1.13,1.32);
			\draw (O) node{$n=1$};
		\end{tikzpicture}
		\caption{The reconstructed toroidal geometry, for $n=1$, represented as a cylinder with ends identified up to a twist of an angle $\varphi$. The definition of the dihedral angles $\psi^{T,L}_{t,x}$ have also been highlighted. On each vertical links of the boundary lives a group element $g=e^{i \sigma_{3} \f{\varphi}{N_t}}$.}
		\label{chap5:fig:reconstruction_n_1}
	\end{center}
\end{figure}
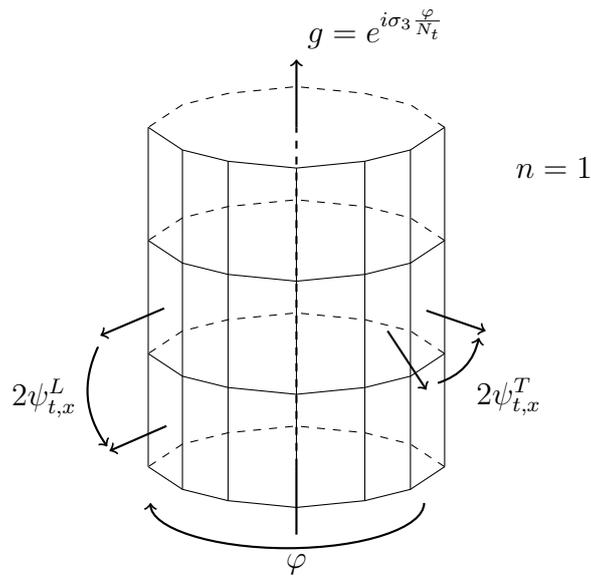

For a generic $n\neq 0$, the surface of the cylinder wraps around itself exactly $n~$ times before closing. Indeed, in that case, the sum over the dihedral angles is not $2\pi$ but $2 \pi n$ and each local dihedral angle is n time bigger that in the case $n=1$. Schematically, for $n=2$, this corresponds to a cylinder wrapping one time around itself before closing, see figure \ref{chap5:fig:wrapping_n_2}. This surface cannot be embedded in $\mathbb R^3$. It can, however, be immersed, see \cite{Dowdall:2009eg}. 
\begin{figure}[!htb]
	\begin{center}
		\includegraphics[width=.3\textwidth]{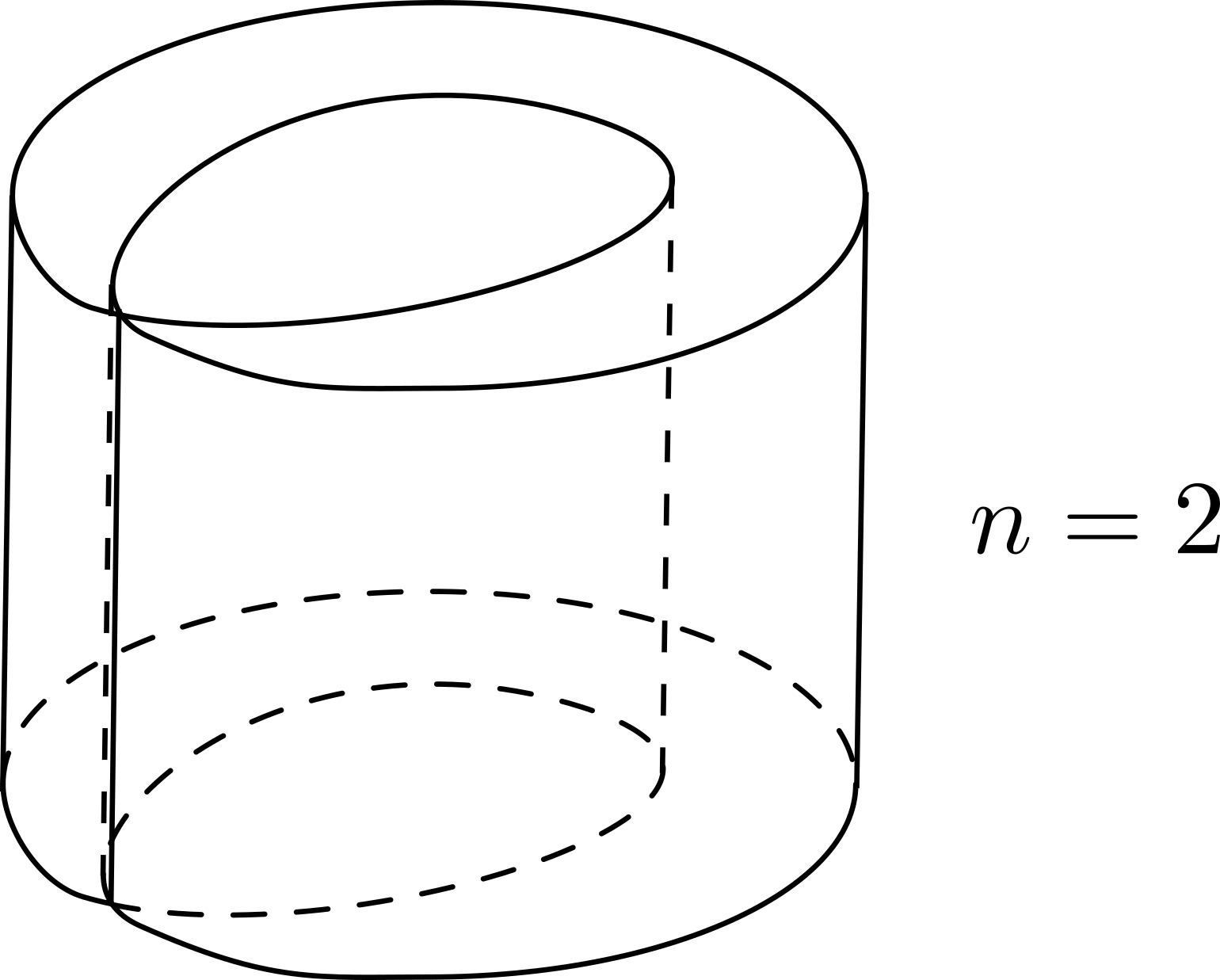}
		\caption{A sketch of the surface reconstructed for $n=2$.}
		\label{chap5:fig:wrapping_n_2}
	\end{center}
\end{figure}
As we anticipated, the case $n=0$ is peculiar. In that case, the torus is completely flat in the $x$ direction, while the curvature in the time direction switch from one angle to its opposite, see figure \ref{chap5:fig:case_n_0}. The identification in the $x$ direction is then not recovered from the geometrical perspective, but is more a compactification of a two-dimensional plane.
\begin{figure}
	\begin{center}
		\includegraphics[width=.6\textwidth]{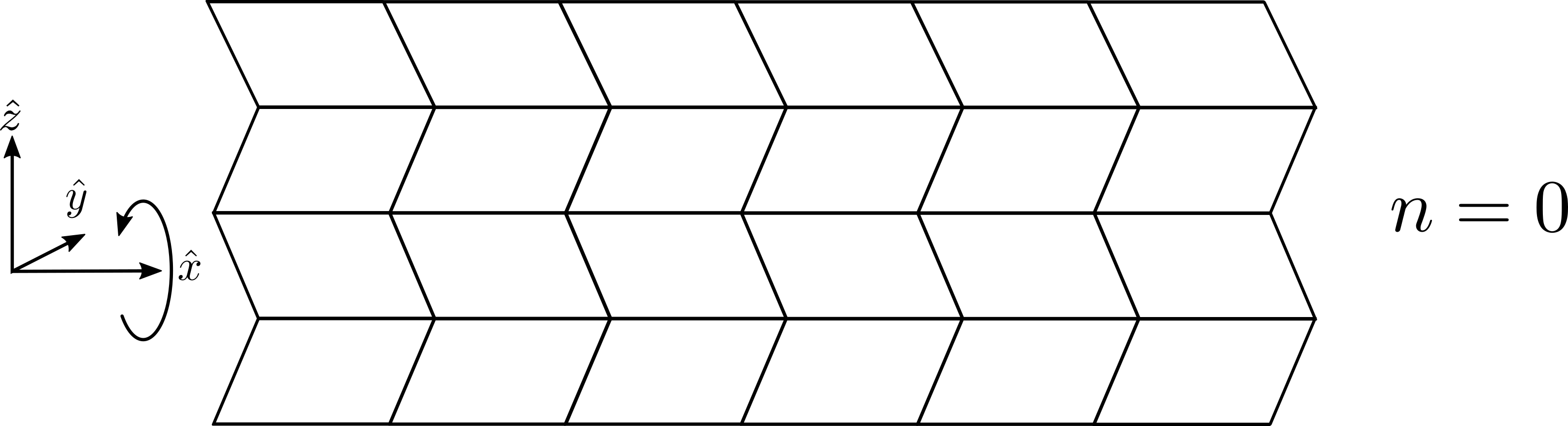}
		\caption{A sketch of the bending along constant time slices which is at the origin of the continuous set of solutions from which the $n=0$ solution is part.}
		\label{chap5:fig:case_n_0}
	\end{center}
\end{figure}
 
If $K=1$, the operation of moving from one cell to the temporally following one takes the initial cell to visit {\it all} the other cells before coming back to the initial one. This fact is what gives the rigidity to the structure, and enforces all the $\psi^T_x$ to be constant. Hence, for $K=1$, the prism described above has a regular polygon for a basis.

If $K>1$, on the other hand, this procedure produces exactly $K$ independent closed cycles of cells. The extrinsic geometry structure needs to be periodic only modulo $K$.
Considering groups of $K$ spatially consecutive cells as a single unit, we find again the same regular structure as the one discussed above for the regularly quadrangulated torus, the only difference being that the fundamental cells are now not necessarily planar polygons.  As a consequence, one expects that a regular solution, $\tl G_{t,x} = \E^{\I \alpha_{t,x} \sigma_z}$ with $\alpha_{t,x} = ( \bar\varphi + \tfrac{\varphi}{N_t} t + \psi^T x)$, can be deformed to another neighbouring solution by adding first-order perturbations of the type 
\begin{equation}
\alpha_{t,x} \mapsto \alpha_{t,x} +  \eps\sum_{m=1}^{k-1}\alpha_m\sin\left( \frac{2\pi }{k} m x \right).
\end{equation}
where $\epsilon\ll1$. We can check that, at first order in $\eps$, these are still solutions of the equation of motion, at least if $N_x$ is even and $K$ is odd, showing the existence of a continuous space of solutions. In full generality, this fact is imprinted in the zeros of the one-loop determinant \footnote{Similar redundancies arise in the Regge calculus treatment of \cite{Bonzom:2015ans}.}. A constant-time section is sketched in figure \ref{chap5:fig:constant_time_slice}.

\begin{figure}[!htb]
	\begin{center}
		\includegraphics[width=.6\textwidth]{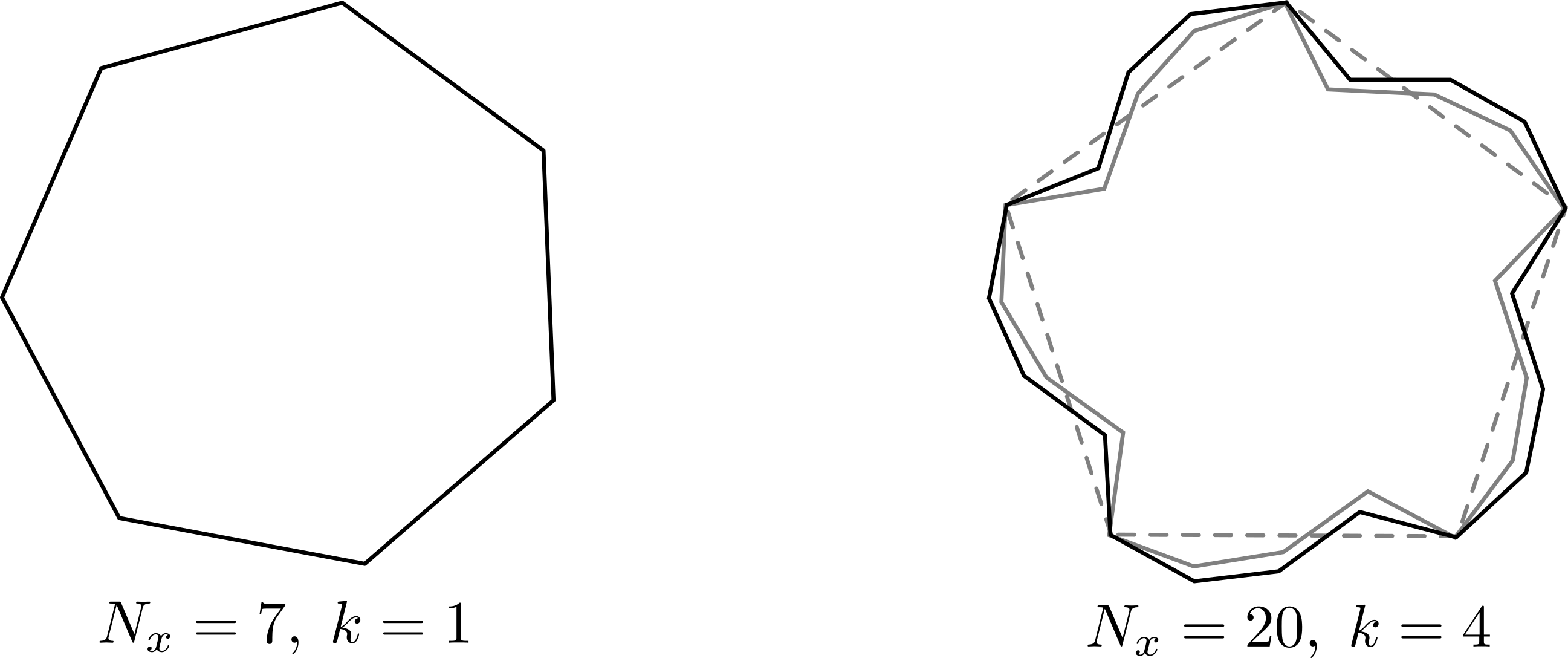}
		\caption{The reconstruction of a single time slice for $K=1$ and $K>1$. In the second case, two infinitesimally close solutions of the saddle point equations are shown.}
		\label{chap5:fig:constant_time_slice}
	\end{center}
\end{figure}

Finally, we comment on the interpretation of negative values of $n$. Under the change $n\mapsto-n$, there are no major changes in the geometric interpretation apart from $\psi^T\mapsto-\psi^T$, globally.
The presence of two sectors of solutions for the dihedral angles is well-known in the Ponzano-Regge model, and is in general attributed to the contribution of two oppositely oriented geometries. 
Indeed, fixing the boundary metric of a manifold, the saddle point analysis is supposed to determine the corresponding classical conjugated momentum (provided the chosen intrinsic metric admits one). The sign of the momentum cannot, however, be determined by this analysis due to time-reversal invariance. In gravity, such momentum is precisely the extrinsic curvature here encoded in the $\psi^T$.

\subsubsection{Foldings}

In this paragraph, we focus on the remaining solution we did not consider. Recall that this solution is parametrized by the choice
\begin{equation}
	\psi^{L}_{t,x} = \psi^L_{t,x+1} = \f{\pi}{2} \quad \text{and} \quad \psi^T_{t+1,x} = - \psi_{t,x}^{T} \; .
\end{equation}

We have already seen that the analogue solution $\psi^T_{t,x}=\pi$ (which is only possible if $N_x$ is even) contains a continuum of solutions, which falls outside the $X$ and $Z$-family classification and that we chose to ignore by considering $N_x$ odd.

These equations imply
\begin{equation}
	\tl G_{t,x+1} = \tl G_{t,x} \E^{\I \psi^T \sigma_3}
	\qquad\text{and}\qquad
	\tl G_{t+1,x} = \tl G_{t,x} \E^{\I \f{\pi}{2} \sigma_1},
\end{equation}
as well as 
\begin{equation}
	\tl G_{t,x+1}\E^{\I \f{\pi}{2} \sigma_1} = \tl G_{t+1,x+1} = \tl G_{t+1,x} \E^{-\I \psi^T \sigma_3}  .
\end{equation}
Extending these solutions homogeneously on a spacial slice, we see that the geometry encoded is that of a folding, i.e. a dihedral angle of $2 \f{\pi}{2} = \pi$ along a line of equal-time spatial edges of the quadrangulation. 
In particular, the difference in sign of $\psi^T_{t,x}$ from one time-slice to the next across the folding, means that the ``inside'' and the ``outside'' of the cylinder get swapped across the folding itself. Of course, periodicity in time enforces an even number $2m<N_t$ of such foldings. The case $m=1$ is depicted in figure \ref{chap5:fig:fig_folding}.
\begin{figure}[t]
	\begin{center}
		\includegraphics[width=.25\textwidth]{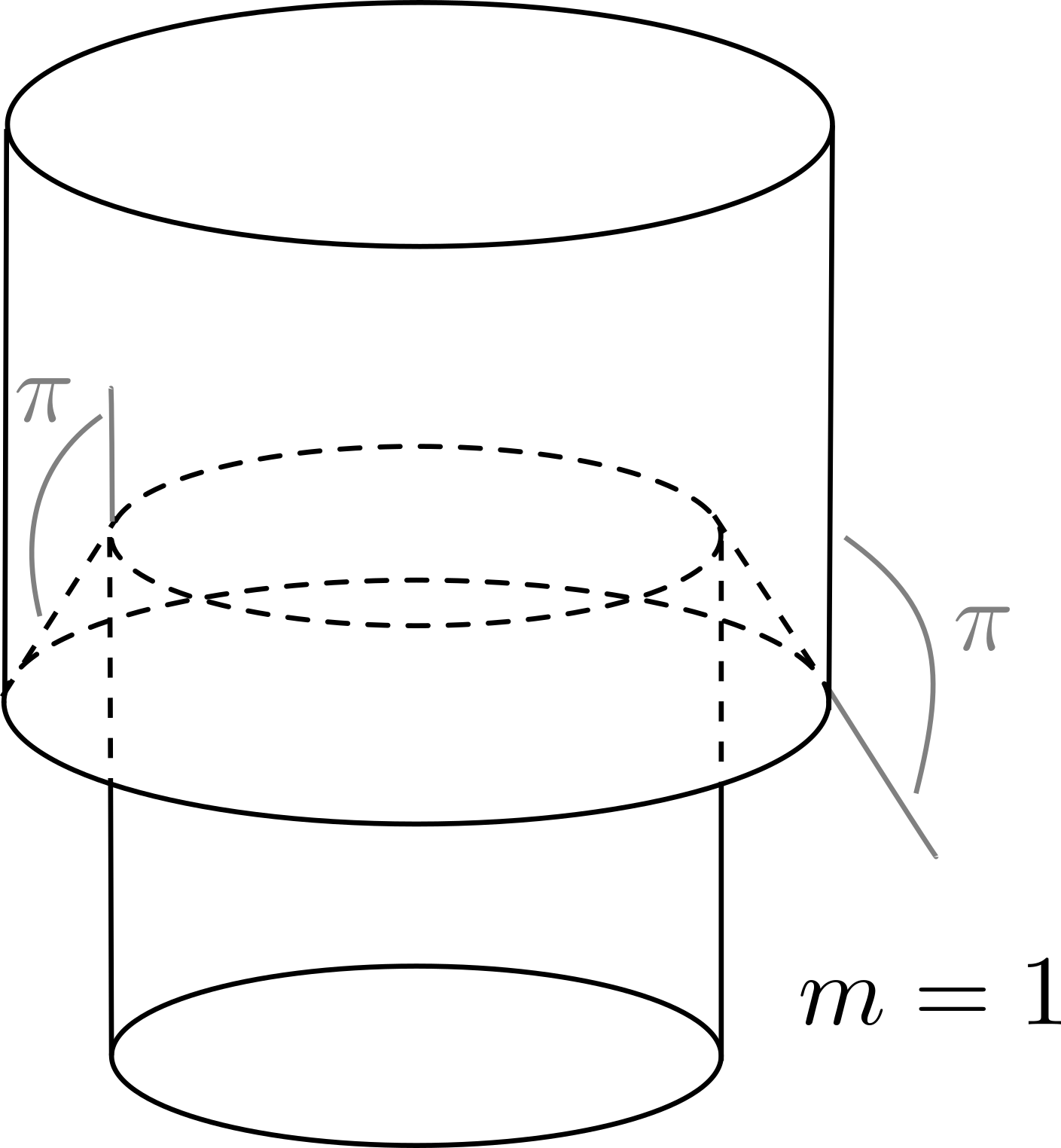}
		\caption{A schematic representation of the folding for $m=1$.}
		\label{chap5:fig:fig_folding}
	\end{center}
\end{figure}

The on-shell value of the LS action of such configurations, for $\psi^T=\tfrac{2\pi}{N_x}n$, is given by
\begin{equation}
	S|_o = - 2\I \pi  T(N_t^+ - N_t^-)n   - 2\I  \pi L N_x m,
\end{equation}
where $N_t^\pm$ are the number of time slices with positive and negative values of $\psi^T_{t,x}$, respectively.
E.g., if $m=0$, $N_t^+=N_t$ and $N_t^-=0$.
Thus, we find that the on-shell action effectively ``sees'' a shorter cylinder of inverse temperature
\begin{equation}
	\beta = (N_t^+ - N_t^-) T \;.
\end{equation}
The second term in the action simply counts the number of foldings. 

It is interesting to note that the foldings happen in the temporal direction. That is, the case with $N_t$ foldings corresponds to no temporal evolution at all. This, in a sense, breaks the unitarity of the temporal evolution and it is a case that must be studied in more detailed.

\subsection{Hessian of the action}

We now come back to the actual computation of the amplitude. In order to obtain the one-loop determinant, we need to develop the action to quadratic order around the solution of interest and calculate the associated Hessian.

First, fix a solution $o$ of the saddle point equations  $(G^o_{t,x}, \varphi^o)$ by
\begin{equation}
	G_{t,x}^o = \E^{\I \left(\bar\varphi+\frac{\varphi^o}{N_t} t + \psi^T_o x  \right)\sigma_3},
	\qquad\text{with}\qquad
	\psi^T_o = \frac{\pi}{N_x}n
	\qquad\text{and}\qquad
	\varphi^o = - N_\gamma \psi^T_o  + \pi n',
\end{equation}
and where $\bar\varphi$ is an arbitrary global rotation parameter due to the remaining gauge symmetry. 

This solution is valid for any $n$ and $K$. If $N_x$ is odd, $n\neq0$, and $K=1$, this solution is isolated, modulo the gauge parameter $\bar\varphi$, and the Hessian of the action is non degenerate.  As we will prove in a following section this is not the case for $n=0$ or $K>1$, confirming the claims of the previous sections

Consider the following parametrization of the linear perturbations around $(G^o_{t,x}, \varphi^o)$: $\vec a_{t,x}\in\mathbb R^3$ and $\phi \in \mathbb R$, such that
\begin{equation}
	G_{t,x} = G^o_{t,x} \E^{\vec a_{t,x} . \vec \sigma} 
	\qquad\text{and}\qquad
	\varphi = \varphi^o + \phi.
\end{equation}

The expansion of the action to second order reads
\begin{align}
	S & = S_o
	+ \left(%
	\frac{1}{2}\left.\frac{\pp^2  S}{\pp{a^j_{s,y} }\pp{a^k_{t,x} }}\right|_o  a^j_{s,y} a^k_{t,x}%
	+ \left.\frac{\pp^2  S}{\pp\phi \pp{a^k_{t,x} } }\right|_o       a^k_{t,x} \phi%
	+ \frac{1}{2}\left.\frac{\pp^2  S}{\pp\phi^2} \right|_o  \phi^2 %
	\right) + \mathrm o(a^2, \phi^2, a \phi)\notag\\
	& = S_o
	+\frac{1}{2} \left(%
	H^{j;k}_{s,y;t,x} a^j_{s,y} a^k_{t,x} %
	+ (H^{k}_{\phi; t,x} + H^{k}_{t,x; \phi})    a^k_{t,x} \phi%
	+  H_{\phi;\phi} \phi^2 %
	\right) + \mathrm o(a^2, \phi^2, a \phi),
\end{align}
where in the last line we have introduced the following notation $H_{\alpha,\beta}$ for the Hessian matrix:
\begin{subequations}\label{Hessiandef1}
	\begin{align}
	H^{j;k}_{s,y;t,x} & = \frac{1}{2}\left.(\nabla_j^{s,y} \nabla_k^{t,x} + \nabla_k^{t,x} \nabla_j^{s,y} ) S \right|_o \\
	H^{k}_{t,x; \phi} & =  \left.  \nabla_k^{t,x} \pp_\phi S \right|_o  \\
	H^{k}_{\phi; t,x} & = \left. \pp_\phi \nabla_k^{t,x}  S \right|_o  =  H^{k}_{t,x; \phi}  \\
	H_{\phi;\phi} & = y\left. \pp^2_\phi S \right|_o 
	\end{align}
\end{subequations}
where $\nabla_k^{t,x}|_o =\pp/\pp a^k_{t,x}$ and $j,k=1,2,3$ are indices for the $\su(2)$ Lie algebra components. Repeated indices are summed over. Since the action only involves correlations between the nearest neighbours, the only contribution to Hessian term $H^{j;k}_{s,y;t,x}$  is for $(s,y) = (t \pm 1,x)$ and $(s,y) = (t,x \pm 1)$.

The explicit form of the Hessian can be worked out by further deriving the first order derivation of the action. Organizing the $(1 + N_t\times N_x )$-dimensional perturbation vector as
\begin{equation}
	\bfa^T = \Big(\phi, (\vec a_{t=1,x=1}, \vec a_{t=1, x=2} , \cdots) , \cdots, (\vec a_{t=N_t,x=1}, \vec a_{t=N_t, x=2}, \cdots)\Big)^T,
\end{equation}
the Hessian matrix of second derivatives of the action can be put into the form (empty entries are vanishing entries)
\begin{equation}
\mathbf{H} = 
\mat{c|ccccc}%
{%
	F & \DD & \DD & \DD &\cdots& \DD  \\\hline%
	\DD^T &\GG & \CC & &   &\CC^T_\gamma \\%
	\DD^T&\CC^T& \GG &\CC&&\\%
	\DD^T&& \CC^T&\GG&\ddots&\\%
	\vdots&&&\ddots&\ddots&\CC\\%
	\DD^T&\CC_\gamma&&&\CC^T&\GG\\%
}.
\end{equation}
This is a matrix made of $(1+ N_t)\times(1+N_t)$ blocks built as follows. The details of the computation can be found in appendix \ref{app:hessian_computation}. Here, we only recapitulate the results.
\medskip

In the top left corner there is a $1\times1$-dimensional block (remember that $L$ and $T$ denote the values of the spins associated to the edges along the space and time direction, respectively),
\begin{equation}
	F\,=\,H_{\phi\phi}\,=\, 2L N_x \; .
\end{equation}
The first row and the first column are occupied by the element $\DD$ and its transpose respectively, with $\DD$ the $(3N_x)$-dimensional covector
\begin{equation}
\DD=\overbrace{D\otimes \cdots \otimes D}^{N_x-\text{times}}\; .
\end{equation}
The three components, $k=1,2,3$, of $D$ are
\begin{equation}
D_k\,=\, H^{k}_{t,x; \phi}
\,=\, \frac{4 \I L}{N_t} \delta^k_2 .
\end{equation}

Similarly, the blocks $\GG$ are ${3N_x\times 3N_x}$ blocks. They encode the spatial coupling of perturbations on a given time-slice. They are defined by 
\begin{equation}
\GG =  \mat{ccccc}%
{%
	A &B&&& B^T\\
	B^T&A&B&&\\
	&B^T&A&\ddots&\\
	&&\ddots&\ddots&B\\
	B&&&B^T&A
}
\end{equation}
with $A$ and $B$ given by the following $3\times3$ matrices 
\begin{align}
A^{jk}&\,=\, H^{j;k}_{t,x;t,x}
\,=\,  4 T \left( \delta^{jk} - \delta^k_3\delta^j_3 \right) + 4 L \left( \delta^{jk} - \delta^k_1\delta^j_1 \right) ,\\
B^{jk} &\,=\, H^{j;k}_{t,x-1;t,x}\,=\,
-2 T\left(  R_z(-2\psi^T_o)^k{}_j  + \I R_z(-2\psi^T_o)^k{}_i \epsilon^{ij3} -   \delta^k_3\delta^j_3    \right) .
\end{align}
($R_z(\alpha)$ denotes the matrix describing the rotation around the $z$--axis by an angle $\alpha$.). Explicitly in terms of a $3\times 3$ matrices, we have
\begin{equation}
	A = 4 \begin{pmatrix}
	T & 0   & 0 \\
	0 & T+L & 0 \\
	0 &  0  & L  
	\end{pmatrix}
\end{equation}
and
\begin{equation}
	B = - 2 T \E^{-2 \I \psi^T_o} 
	\begin{pmatrix}
	1 & -i & 0 \\
	i & 1  & 0 \\
	0 & 0  & 0
	\end{pmatrix}
	\; .
\end{equation}

The blocks $\CC$ are also ${3N_x\times 3N_x}$ blocks. They encode the coupling between subsequent time slices, 
\begin{equation}
\CC =  \mat{cccc}%
{%
	\,C\, &&&\\
	&\,C\,&&\\
	&&\,C\,&\\
	&&&\,\ddots
},
\end{equation}
with $C$ given by the $3\times3$ matrix 
\begin{equation}
C^{jk}\,=\,H^{j;k}_{t-1,x;t,x} \, =\, -2 L \Big(  \delta^{kj} + \I\epsilon^{kj1} - \delta^k_1\delta^j_1 \Big) .
\end{equation}
Explicitly
\begin{equation}
	C = -2 L
	\begin{pmatrix}
	0 & 0 & 0   \\
	0 & 1 & -\I \\
	0 & i & 1    
	\end{pmatrix} \; .
\end{equation}

The first and last time-slices of the cylinder, however, couple under a twist of $N_\gamma$ units. This is encoded in the shifted $\CC$ block we named $\CC_\gamma$:
\begin{equation}
\CC_\gamma = 
\begin{array}{c} 1\\2\\\vdots\\N_\gamma-1\\ N_\gamma\\ N_\gamma+1\\\vdots\\N_x \end{array}
\mat{cccc|cccc}%
{%
	&&&&\,C\,&&&\\
	&&&&&\,C\,&&\\
	&&&&&&\,\ddots&\\
	&&&&&&&\,C\,\\
	\hline
	\,C\,&&&&&&&\\
	&\,C\,&&&&&&\\
	&&\,\ddots&&&&&\\
	&&&\,C\,&&&&\\
}.
\end{equation}

Developed at second order around the solution, the action can now be written as
\begin{align*}
	S &= S_o + \frac12 \bfa. \mathbf{H} \bfa + \mathrm O(\bfa^2) \\
	  &=S_{o} + \f{1}{2} \phi F \phi +  \f{1}{2}\sum_{t,x=0}^{N_t-1,N_x-1} 2 \phi D. \vec{a}_{t,x} + \vec{a}_{t,x} .A. \vec{a}_{t,x} + \vec{a}_{t,x-1} .B. \vec{a}_{t,x} + \vec{a}_{t-1,x} .C. \vec{a}_{t,x}
\end{align*}
The one-loop determinant is therefore simply given by the determinant of $\mathbf H$.

Notice that $\mathbf H$ is essentially a band matrix, but that its entries (almost) do not depend on the $(t,x)$ labels.  Such matrices can be diagonalized via a Fourier transform in $(t,x)$. We have however to introduce a `twist' due to the shifts appearing in the blocks $\CC_\gamma$. This will enable us to compute $\det(\mathbf H)$.\\

\subsection{Twisted Fourier transform}

The determinant of the Hessian is readily computed by performing a Fourier Transform. However, for it to respect the peculiar boundary conditions we imposed on the lattice, we have to consider a particular twisted definition. For $\alpha = 1,2,3$ we consider the Fourier component $\hat a_{\omega,k}^{\alpha}$ defined by
\begin{equation*}
	a_{\omega,k}^{\alpha} = \frac{1}{\sqrt{N_t N_x}} \sum_{t,x=0}^{N_t-1,N_x-1} \E^{\I\frac{2\pi}{N_x}k x} \E^{\I \frac{2\pi}{N_t} \left(\omega -\f{\gamma}{2 \pi} k\right)t} a_{t,x}^{\alpha} \; ,
\end{equation*}
where $\omega \in [0,N_t-1]$ denotes the time modes and $k \in [0,N_x-1]$ the spatial modes. Introducing the notation
\begin{equation}
	\psi  = \frac{2\pi }{N_x},
	\qquad \text{and} \qquad
	\chi_{\omega,k} = \frac{2\pi}{N_t} \left( \omega - \frac{ \gamma}{2\pi}k\right)
\end{equation}
we have
\begin{equation}
	a_{\omega,k}^{\alpha} = \frac{1}{\sqrt{N_t N_x}} \sum_{t,x=0}^{N_t-1,N_x-1} \E^{\I\psi k x} \E^{\I \chi_{\omega,k}t} a_{t,x}^{\alpha} \; .
\end{equation}
Note that in these notation, $\psi= 2 \psi^{T}_{o}$. To keep the notation a bit light, we use the same letter for the Fourier component and the direct component. The name of the indices always allow us to know which object is considered without any problem.

The consequence of such a twist is that the Brillouin zone is reciprocally twisted, but in the spacial direction only
\begin{equation}
	a_{\omega+N_\gamma, k+N_x}^{\alpha} = a_{\omega,k}^{\alpha} = a_{\omega + N_t, k}^{\alpha} \; .
\end{equation}
The main difference with the untwisted Fourier Transform is that the periodic conditions do not involve any phase. Without considering the twist, the time direction periodic condition reads $a_{N_t,k}^{\alpha} = \E^{\I \gamma k}\omega_{0,k}$. Since this phase depends on $k$, we would not have been able to diagonalize the Hessian with this transformation.

Now, recall that by definition of the perturbation, $a_{t,x}^{\alpha}$ is real. Therefore, we can, as usual, introduce the complex conjugate of $a_{\omega,k}^{\alpha}$ by considering negative modes 
\begin{equation}
	\bar{a}_{\omega,k}^{\alpha} = a_{-\omega,-k}^{\alpha} \; .
\end{equation}
In term of modes in the Brillouin zone, this relations reads
\begin{equation}
	a_{N_\gamma - \omega + N_t \Theta(\omega-N_\gamma)  , N_x -k}^{\alpha}= \bar{a}^\alpha_{E, p}
\end{equation}
where $\Theta$ is the Heaviside function with the choice $\Theta(0) = 0$. Excluded the zeroth modes, the above equation always relates modes at two different momenta, unless
\begin{subequations}
	\begin{equation*}
	\Big(\;k = N_x/2 \qquad\text{and}\qquad \omega = N_\gamma/2\;\Big),
	\end{equation*}
	or
	\begin{equation*}
	\qquad\quad\Big(\;k=N_x/2\qquad\text{and}\qquad \omega=(N_\gamma+N_t)/2\;\Big).
	\end{equation*}
\end{subequations}
Both these cases are excluded with the requirement of having $N_x$ odd. Note also that the first case implies at least $K=2$.

Restricting ourselves to the case where $N_x$ is odd, we can then always consider $\hat a^{\alpha}_{\omega,k}$ as independent complex variables by doubling each degree of freedom. The only exception is for the zeroth mode, which is real, as always. We will make use of this property to explicitly compute the determinant of the Hessian. At the end of the day, the Hessian will be diagonal by block. All blocks but one are $3 \times 3$ blocks. The remaining block is $4 \times 4$, and corresponds to the coupling between the zeroth Fourier mode and the remaining bulk information and their self-interactions.
\medskip

We now quickly derive the expression of the Hessian in Fourier modes. From its expression, it is immediate to see that the space-time independent perturbation $\phi$ is simple in the sense that the Hessian expression in the Fourier basis is exactly the same as in the direct basis. This is obvious since these terms do not involve any coupling with the lattice direction. The same can be said for the terms only involving a given lattice position. At the end of the day, we really need to explicitly compute two contributions to the Hessian in Fourier modes. That is the one coming from the coupling of spatial and temporal slices, i.e. the terms of the form $\Delta_{s} = \sum_{t,x} \vec{a}_{t,x-1} .B. \vec{a}_{t,x}$ and $\Delta_{t} = \sum\limits_{t,x}\vec{a}_{t-1,x} .C. \vec{a}_{t,x}$.

Starting with $\Delta_{s}$, we explicitly expand the action of the matrix $B$
\begin{equation*}
	\vec{a}_{t,x-1} .B. \vec{a}_{t,x} = -2T \E^{-\I \psi n}(a_{t,x-1}^1 a_{t,x}^1 + a_{t,x-1}^2 a_{t,x}^2 + i (a_{t,x-1}^2 a_{t,x}^1 - a_{t,x-1}^1 a_{t,x}^2))
\end{equation*}
Going to the Fourier components, the expression becomes
\begin{equation*}
	\begin{split}
		\vec{a}_{t,x-1} .B. \vec{a}_{t,x}
		= \f{-2T \E^{-\I \psi n}}{N_t N_x} \sum_{\omega_1,k_1} \sum_{\omega_2,k_2} \Big(\bar{a}_{\omega_1,k_1}^1& a_{\omega_2,k_2}^1 + \bar{a}_{\omega_1,k_1}^2 a_{\omega_2,k_2}^2 +  \I (\bar{a}_{\omega_1,k_1}^2 a_{\omega_2,k_2}^1 - \bar{a}_{\omega_1,k_1}^1 a_{\omega_2,k_2}^2) \Big) \\ & \E^{\I \psi k_1} \E^{\I \psi (k_2-k_1)x} \E^{\I \chi_{\omega_2-\omega_1,k_2-k_1}t} \;.
		\end{split}
\end{equation*}
Taking into account the sums over $x$ and $t$, we get that $k_2 = k_1$ and that $\omega_2 = \omega_1$. In Fourier mode, $\Delta_s$ becomes
\begin{equation}
	\Delta_s = -2T \E^{-\I \psi n}\sum_{\omega,k=0}^{N_t-1,N_x-1}  \Big(\bar{a}_{\omega,k}^1 a_{\omega,k}^1 + \bar{a}_{\omega,k}^2 a_{\omega_2,k}^2 +  \I (\bar{a}_{\omega,k}^2 a_{\omega,k}^1 - \bar{a}_{\omega,k}^1 a_{\omega,k}^2) \Big) \E^{\I \psi k}
\end{equation}

It is possible to simplify this expression using the symmetry
\begin{equation*}
\Delta_{s} = \sum_{t,x=0}^{N_t-1,N_x-1} \f{1}{2} \left(\vec{a}_{t,x-1} .B. \vec{a}_{t,x} + \vec{a}_{t,x} .B. \vec{a}_{t,x-1}  \right) \;.
\end{equation*}

At the end of the day, we obtain
\begin{equation}
	\Delta_s =4 \sum_{\omega,k=0}^{N_t-1,N_x-1} 
		a^*_{\omega,k}
			\begin{pmatrix}
			-T \E^{- \I \psi n} \cos(\psi k) & -T \E^{- \I \psi n} \sin(\psi k) & 0 \\
			T \E^{- \I \psi n} \sin(\psi k) & -T \E^{- \I \psi n} \cos(\psi k)  & 0 \\
							0				&				0					& 0
			\end{pmatrix}
		a_{\omega,k}
\end{equation}
where $a^*_{\omega,k}$ is the complex transpose of $a_{\omega,k}$.

A similar computation can be done for $\Delta_t$ and returns
\begin{equation}
	\Delta_t = 4\sum_{\omega,k=0}^{N_t-1,N_x-1} 
		a^*_{\omega,k}
			\begin{pmatrix}
			0	&				 0 									& 				0 							\\
			0	& -L \E^{- \I \psi n} \cos(\chi_{\omega,k})			& -L \E^{- \I \psi n} \sin(\chi_{\omega,k}) \\
			0	& L \E^{- \I \psi n} \sin(\chi_{\omega,k})			& -L \E^{- \I \psi n} \cos(\chi_{\omega,k})
			\end{pmatrix}
		a_{\omega,k} \; .
\end{equation}

Adding these two terms plus the one depending on only one lattice site returns
\begin{equation}
	\footnotesize
	H_{\omega,k} = 4
	\begin{pmatrix}
		T\left[1- \E^{-\I \psi n }\cos(\psi  k) \right] & - T\E^{-\I \psi n}\sin(\psi  k)   																	& 0 														\\
		T\E^{-\I \psi n}\sin(\psi k)   					& T\left[1 - \E^{-\I \psi n}\cos(\psi k)\right] + L\left[1- \cos\left (\chi_{\omega,k}\right)\right] 	&  -L\sin\left (\chi_{\omega,k}\right)						\\
		0 							 					&  L\sin\left ( \chi_{\omega,k}\right)  																&  L\left [1 -\cos\left ( \chi_{\omega,k}\right)\right]
	\end{pmatrix} \; ,		
	\normalsize
\end{equation}
where we recall that
\begin{equation*}
	\psi  = \frac{2\pi }{N_x} \; ,
	\qquad
	\chi_{\omega,k} = \frac{2\pi}{N_t} \left( \omega - \frac{ \gamma}{2\pi}k \right)
	\qquad\text{and}\qquad
	\gamma = \frac{2\pi N_\gamma}{N_x} \;.
\end{equation*}

The remaining elements we did not compute in Fourier mode are the one coming from the perturbation $\phi$. It is immediate to see that the self-interacting term does not change in Fourier mode whereas the coupling of the field $\varphi$ with the lattice site translates into a coupling with the zeroth Fourier mode. At the end of the day, we get
\begin{equation}
S = S_o + \f{1}{2} \left( \phi^2 F + 2 \phi D.a_{0,0} + \bar{a}_{0,0} H_{0,0} a_{0,0} \right) + \f{1}{2} \sum_{t,x \neq 0,0}^{N_t-1,N_x-1} \bar{a}_{\omega,k} H_{\omega,k} a_{\omega,k} \; ,
\end{equation}
where $a_{0,0}$ is the zeroth Fourier mode. As such, this perturbation computation is not well-defined. Indeed, looking at $H_{\omega,k}$ carefully, we see that for $(\omega,k) = (0,0)$ the third row and third column of the matrix are identically zero. That is not surprising at all. Recall that our system is still globally invariant by a gauge transformation along the direction $z$. For the following computation, we introduce the matrix $H^{'}_{0,0}$ corresponding to $H_{0,0}$ but where the last row and column are removed.

\subsection{One-loop expansion of the Ponzano-Regge amplitude}

We are now ready to compute the one-loop expansion of the Ponzano-Regge amplitude. Recall that the amplitude reads
\begin{equation}
	\la Z_{PR}^{\cK}| \Psi_{coh}  \ra 
	=  \left[\frac{1}{\pi}\int_{0}^{2\pi} \dd \varphi \,\sin^2\left(\varphi\right)  \prod_{(t,x)} \int_{\SU(2)} \dd G_{t,x} \right]  \, \E^{-S (G_{t,x},\varphi)} \; .
\end{equation}

We will focus on the contribution of the $n$-th term of the saddle point without any possible folding. The contribution is given by
\begin{align}
	\la Z_{PR}^{\cK}| \Psi_{coh}  \ra ^\text{1-loop}_{o,n} 
	& =2\pi\times \frac{1}{\pi}\sin^2\left( \frac{\gamma n}{2} \right) \times \E^{-S_o} \left( \int_{\mathbb R}\dd \phi \int_{\mathbb R^2} \dd a_{0,0} \;  \E^{\frac12 F \phi^2 + \phi D.a_{0,0} + \frac12  a^{*}_{0,0} H'_{0,0} a_{0,0}} \right)\times\notag\\
	&\hspace{5cm}\times\prod_{(\omega,k)\neq(0,0)}\left( \int_{\mathbb C^3} \dd a_{\omega,k}  \,  \E^{ \frac12  a^*_{\omega,k} \hat H_{\omega,k} \hat a_{\omega,k}   } \right)^{1/2} \; .
\end{align}
As we said previously, we have introduced the matrix $\hat H'_{0,0}$ due to the residual global gauge symmetry by a transformation along any element along the $z$ direction. This corresponds to the choice of $\bar \varphi$ previously introduced. This exact symmetry of the action produces the $2\pi$ volume factor coming from an integral over it.

The second factor in the one-loop expression corresponds to the evaluation of the measure term at the critical solution, given by $\varphi_{o} = \f{\gamma}{2} n$ while the exponential $\E^{\I S_{o}}$ corresponds to the evaluation of the action at the critical solution. Recall that the on-shell action is given by equation \eqref{chap5:eq:classical_action}. In our case, the critical solution imposes the (demi) dihedral angles $\psi^{T} = \f{\pi}{N_x}n$ and $\psi^{L} = 0$. The on-shell action is then
\begin{equation*}
	S_o = \I \sum_{t,x=0}^{N_t-1,N_x-1} T \f{2 \pi}{N_x}n = \I 2 \pi n N_t \; .
\end{equation*}
Hence, the corresponding term of the one-loop amplitude is just a sign factor $(-1)^{2 T N_t n}$. Therefore, although the on-shell action takes formally the expected form of an on-shell discretized GBY action, the consequence of the discreteness of the lengths, the contribution is just a sign factor. Note that the on-shell action can really be interpreted as a discretized GBY term only in the case $n=1$. Indeed, as we previously mentioned, this is the only case that is embeddable in $\R^{3}$ and therefore providing us with the usual geometrical picture. 

Finally, the remaining factors are the Gaussian integrals on the linear perturbations, which have been approximated, as usual, to be integrals over the full real line, rather than over their original compact spaces. Note that we have used the trick previously mentioned of doubling the degrees of freedom in the first Brillouin zone and hence we have taken the square root in the second series of integrals. We are now left the computation of these integrals. Both are computed using the usual expression of a Gaussian integral in terms of the determinant
\begin{equation}
	\int_{-\infty}^{\infty} \E^{-\f{1}{2} x^{t}.A.x } \dd^n x = \sqrt{\f{(2\pi)^n}{\det(A)}} \; ,
\end{equation}
where $x^t = (x_1,...,x_n)$ and $A$ and $n$ square matrix.

\subsubsection{Computation of the first Gaussian integral}

We focus first on the integration over $\phi$ and $\hat a_{0,0}$
\begin{equation*}
	\cN_{(0,0),\phi} \equiv \int_{\R}\dd \phi \int_{\R^2} \dd a_{0,0} \;  \E^{\frac12 F \phi^2 + \phi {} D. a_{0,0} + \frac12 a^*_{0,0} \hat H'_{0,0} a_{0,0} } 
\end{equation*}

To compute this integral, it is easier to rewrite everything in only one matrix. Consider the vector $a' = (\phi,a^1_{0,0},a^2_{0,0})$. The previous integral becomes
\begin{equation*}
	\cN_{(0,0),\phi} = \int_{\R^3} \dd^3 a' \E^{\f{1}{2} (a')^*.A.a'}
\end{equation*}
where the matrix $A$ is
\begin{equation*}
	A=4
	\begin{pmatrix}
			\f{L N_x}{2}				&	0					&	\f{\I L}{N_t}			\\
			0					&	T(1-\E^{\I \psi n})	&	0			\\
			\f{\I L}{N_t}		&	0					&	T(1-\E^{\I \psi n})
	\end{pmatrix} \; .
\end{equation*}
 Therefore, the Gaussian integral is
\begin{equation}
		\cN_{(0,0),\phi} = \left(\frac{(2\pi)^3}{64 LT(1-\E^{-i\psi n})\left( \frac{L}{N_t^2} + \frac{TN_x}{2}(1- \E^{-i\psi n}) \right) }\right)^{1/2},
\end{equation}
a result which is independent of $\gamma$.

\subsubsection{Computation of the second Gaussian integral}

Each $(\omega,k)$ term gives
\begin{align*}
	\mathcal N_{\omega,k}& \equiv \left( \int_{\mathbb C^3} \dd a_{\omega,k}  \,  \E^{ \frac12  a^*_{\omega,k} H_{\omega,k} a_{\omega,k}   } \right)^{1/2} 
	\,=\,
	\left( \frac{(2\pi)^3}{\det H_{\omega,k} }\right)^{1/2} \notag\\
	& = \left(\frac{(2\pi)^3}{64 LT \E^{-i \psi n} (2- 2\cos \chi_{\omega,k})\Big( (L+T) (\cos(n\psi) - \cos(p \psi) ) + i L \sin(n\psi)  \Big) }\right)^{1/2}.
\end{align*}
As expected, the result is even in $(\omega,k)\mapsto(-\omega,-k)$. 

The last step is to perform the multiplication over the $(\omega,k)\neq(0,0)$. To do so, we first reorganize the product as follows
\begin{equation*}
\prod_{(\omega,k)\neq(0,0)} \mathcal N_{\omega,k} = \left(\prod_{\omega=1}^{N_t-1}\mathcal  N_{\omega,0}\right)\left(  \prod_{k=1}^{N_x-1}\prod_{\omega=0}^{N_t-1} \mathcal N_{\omega,k}\right).
\end{equation*}

This product can be computed using the following lemma (see appendix \ref{app:proof_formula_product_cos} for the proof)
\begin{lemma}
	We consider two integers $N$ and $M$. We denote the greatest common divisor (GCD) between $N$ and $M$ by $K$. We define two integers $n$ and $m$ such that
	\begin{equation*}
	N = K \, n \qquad \text{and} \qquad M = K \, m \; .
	\end{equation*}
	The following relation holds for all $x$ and $z$ complex numbers
	\begin{equation*}
	\prod_{k=0}^{N-1} \left(2 z + 2 \cos\left( \f{2 \pi M}{N}k + x\right) \right)= \left(2 \left( T_{n}(z) - (-1)^{n} \cos(n x) \right)\right)^K
	\end{equation*}
	where $T_{n}$ is the n-Chebyshev polynomial of the first kind.
	\label{lemma:cosine_product}
\end{lemma}

Applying this lemma for $z=1$, $x=\f{\gamma k}{N_t}$ $M=1$ and $N = N_t$ returns
\begin{equation*}
	\prod_{\omega=0}^{N_t-1} (2-2\cos\chi_{\omega,k}) = 2- 2\cos(\gamma k).
\end{equation*}
From this, one finds
\begin{equation*}
	\prod_{\omega=1}^{N_t-1} (2-2\cos\chi_{\omega,0})
	= \lim_{k\to 0}\frac{\prod_{\omega=0}^{N_t-1} (2-2\cos\chi_{\omega,k})}{2 - 2\cos\chi_{0,k}}
	=\lim_{k\to 0}\frac{ 2- 2\cos(\gamma k)}{2 - 2\cos\chi_{0,k}}
	= N_t^2 \; .
\end{equation*}

Finally, putting everything together, we find that the products over $(\omega,k)$ give
\begin{align}
	\prod_{(\omega,k)\neq(0,0)} & \mathcal N_{\omega,k} =
	\left(\frac{(2\pi)^3}{64 LT \E^{-i \psi n}}\right)^{\frac{N_x N_t -1}{2}}%
	\left(\frac{1}{(L+T)(\cos(n\psi)-1) + i L \sin(n\psi)}\right)^{\frac{N_t-1}{2}} \frac{1}{N_t} \times\notag\\%
	&\prod_{k=1}^{N_x-1} \left(   \frac{1}{\Big(2-2\cos(\gamma k)\Big)\Big((L+T)(\cos(n\psi)-\cos(k\psi)) + i L \sin(n\psi)\Big)^{N_t}}  %
	\right)^\frac{1}{2}  .
\end{align}
This formula holds whenever $(N_x/2 , N_\gamma/2)$ and $(N_x/2, (N_\gamma+N_t)/2)$ are {\it not} in $\mathbb N \times \mathbb N$, since in these cases we would be overcounting one real mode. In other words, these conditions ensure that the only real Fourier mode in the first Brillouin zone is $(\omega,k)=0$. In our case of $N_x$ odd, a simple rearrangement of the terms return
\begin{align*}
	\prod_{(\omega,k)\neq(0,0)} \mathcal N_{\omega,k}  \stackrel{N_x\text{ odd}}{=}&%
	\left(\frac{(2\pi)^3}{64LT \E^{-i \psi n}}\right)^{\frac{N_x N_t -1}{2}}%
	\left(\frac{1}{(L+T)(\cos(n\psi)-1) + i L \sin(n\psi)}\right)^{\frac{N_t-1}{2}} \frac{1}{N_t} \\%
	&\times \prod_{k=1}^{\frac{N_x-1}{2}}  \frac{1}{2-2\cos(\gamma k)} \\
	&\times \prod_{k=1}^{\frac{N_x-1}{2}} \left(\frac{1}{(L+T)(\cos(n\psi)-\cos(k\psi)) + i L \sin(n\psi)}  %
	\right)^{N_t} .
\end{align*}

\subsubsection{One-loop-amplitude and interpretation}

Finally, the total one-loop amplitude is
\begin{align}
	\la Z_{PR}^{\cK}| \Psi_{coh}  \ra^\text{1-loop}_{o}
	& = \sum_{n=1}^{N_x-1} (-1)^{2T N_t n} \times {\cal A}(n) \times 	{\cal D}(\gamma, n)  \nonumber\\
	&=  \sum_{n=1}^{N_x-1} (-1)^{2T N_t n}  {\cal A}(n) \big(2-2\cos(\gamma n)\big)  \times  \prod_{k=1}^{\frac{N_x-1}{2}}       \frac{1}{ 2-2\cos(\gamma k) }\;.
	\label{chap5:eq:amp_final}
\end{align}

In this expression, the first factor is the contribution of the on-shell LS action $S|_o$. The factor ${\cal A}_{\text{LS}}(n)$, does {\it not} depend on the twisting angle $\gamma$ but carries the dependence on the winding number. It is obtained by combining all the expressions of the one-loop computation that do not depend on the twist. Explicitly we have
\begin{align}
	{\cal A}(n)  =&\frac12 \left(\frac{(2\pi)^3}{64LT(1-\E^{-i\psi n})\left( \frac{L}{N_t^2} + \frac{TN_x}{2}(1- \E^{-\I\psi n}) \right) }\right)^{1/2} 
	\notag\\
	&\times
	\left(\frac{(2\pi)^3}{64LT \E^{-\I \psi n}}\right)^{\frac{N_x N_t -1}{2}}%
	\left(\frac{1}{(L+T)(\cos(n\psi)-1) + \I L \sin(n\psi)}\right)^{\frac{N_t-1}{2}} \frac{1}{N_t} 
	\notag\\%
	&\times%
	\prod_{p=1}^{\frac{N_x-1}{2}} \left(   \frac{1}{(L+T)(\cos(n\psi)-\cos(p\psi)) + \I L \sin(n\psi)}  %
	\right)^{N_t},
\end{align}
As we see, it is a rather intricate function of the spins $L$ and $T$, the lattice sizes $N_{x}$ and $N_{t}$ and of the label $n$. By explicitly performing the summation over the Fourier modes $k$ using lemma \ref{lemma:cosine_product}, this expression can be greatly simplified. Introducing the simple trigonometric function
\begin{equation}
	c_{n}=\cos n\psi+\f{iL}{T+L} \sin n\psi 
\end{equation}
the factor $\cA(n)$ becomes
\begin{eqnarray}
	\cA(n)
	&=&
	\f{1}{64}\left[\f{iL - (L+T)\tan\f{\psi n}{2} }{i L - (L +4TN_t^3) \tan\f{\psi n}{2}}\right]^{\f12}
	\\
	&&\times
	\,
	\left[
	\f{2(2\pi)^{3}e^{in\psi}}{LT(L+T)}
	\right]^{\f{N_{x}N_{t}}2}
	\Big[2T_{N_{x}}(a_{n})-2\Big]^{-\f{N_{t}}2}
	\; .
	\nn
\end{eqnarray}
From this expression one can first check the reality condition
\begin{equation}
\mathcal A(-n) = \overline{\mathcal A}(n)
\,.
%\,,
\end{equation}
This is consistent with a Hamilton-Jacobi functional, which cannot distinguish a momentum from its opposite, and more specifically with a first-order Einstein-Cartan formulation of General Relativity, of which the Ponzano-Regge model is a quantization. 

Moreover, ${\cal A}(n)$ is, in the large $N_x, N_t$ limit (but already true when they are larger than 5),  overwhelmingly peaked in modulus at the minimal and maximal values of $n$, that is at $n=1$ and $n=\lfloor\frac{N_x-1}{2}\rfloor$ as illustrated by the plots on  figure \ref{chap5:fig:ALS_plot}.
\begin{figure}[htb!]
	\begin{center}
		\includegraphics[height=5cm]{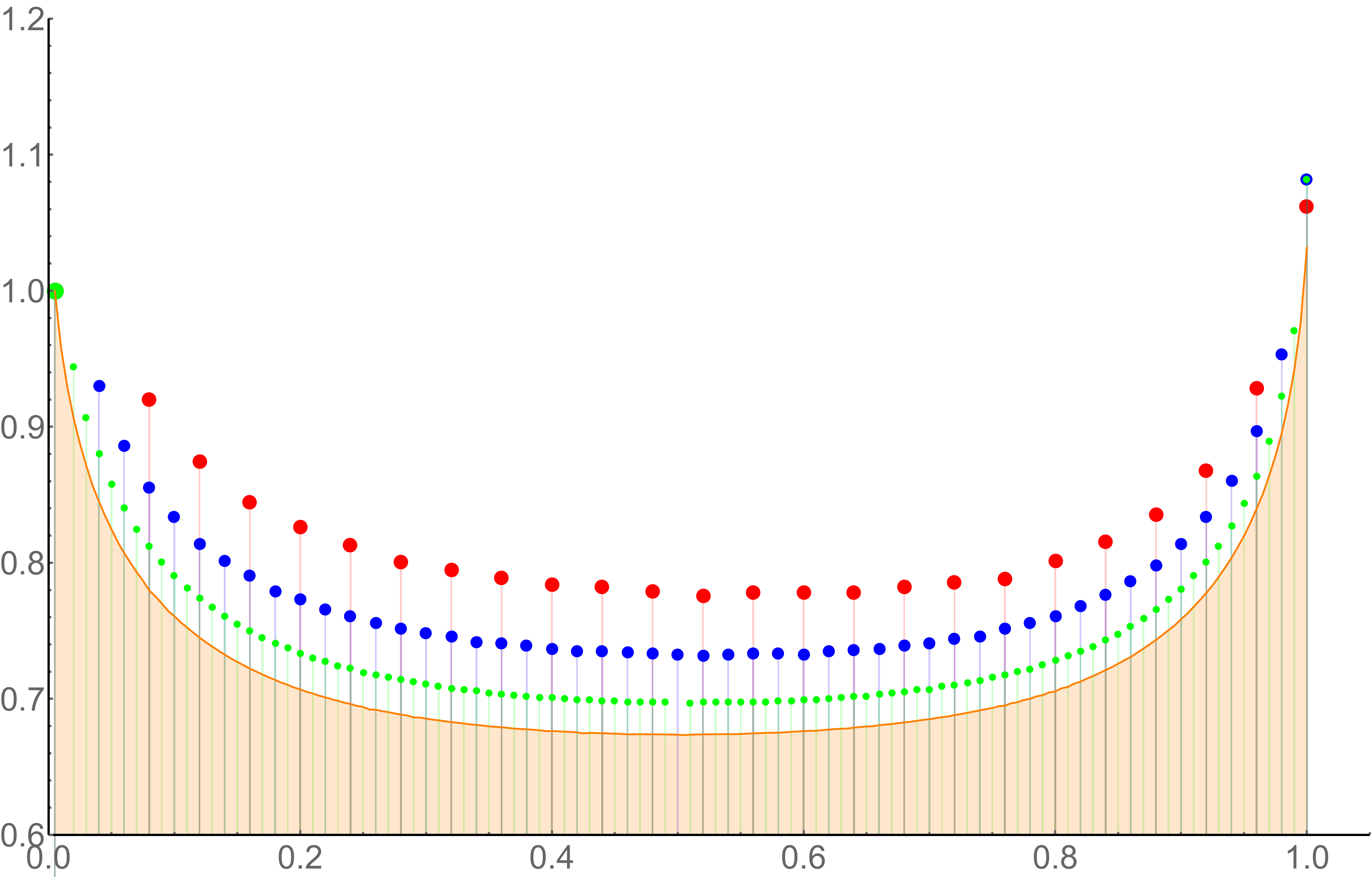}
	\end{center}
	\caption{Plots of $\frac{\log(|\mathcal A(1)|)}{\log(|\mathcal A(n)|)}$ for $n$ running from $1$ to $\frac{N_x-1}{2}$ with the parameters $N_t=20$, $L=8$, $T=8$. The four plots correspond to $N_x=50,100,200,400$ (red, blue, green, orange). The $x$-axis corresponds to $\frac{2n}{N_x-1}$, which runs from 0 to 1.}
	\label{chap5:fig:ALS_plot}
\end{figure}
This behaviour is entirely due to the factor with the Chebyshev polynomial, $\big{(}2T_{N_{x}}(c_{n})-2)\big{)}^{-\f{N_{t}}2}$. We can actually plot this at fixed $N_{x},N_{t}$ as a function of the continuous variable $x=n\psi\in[0,\pi]$, as in figure \ref{chap5:fig:chebplot}.
\begin{figure}[htb!]
	\begin{center}
		\includegraphics[height=4cm]{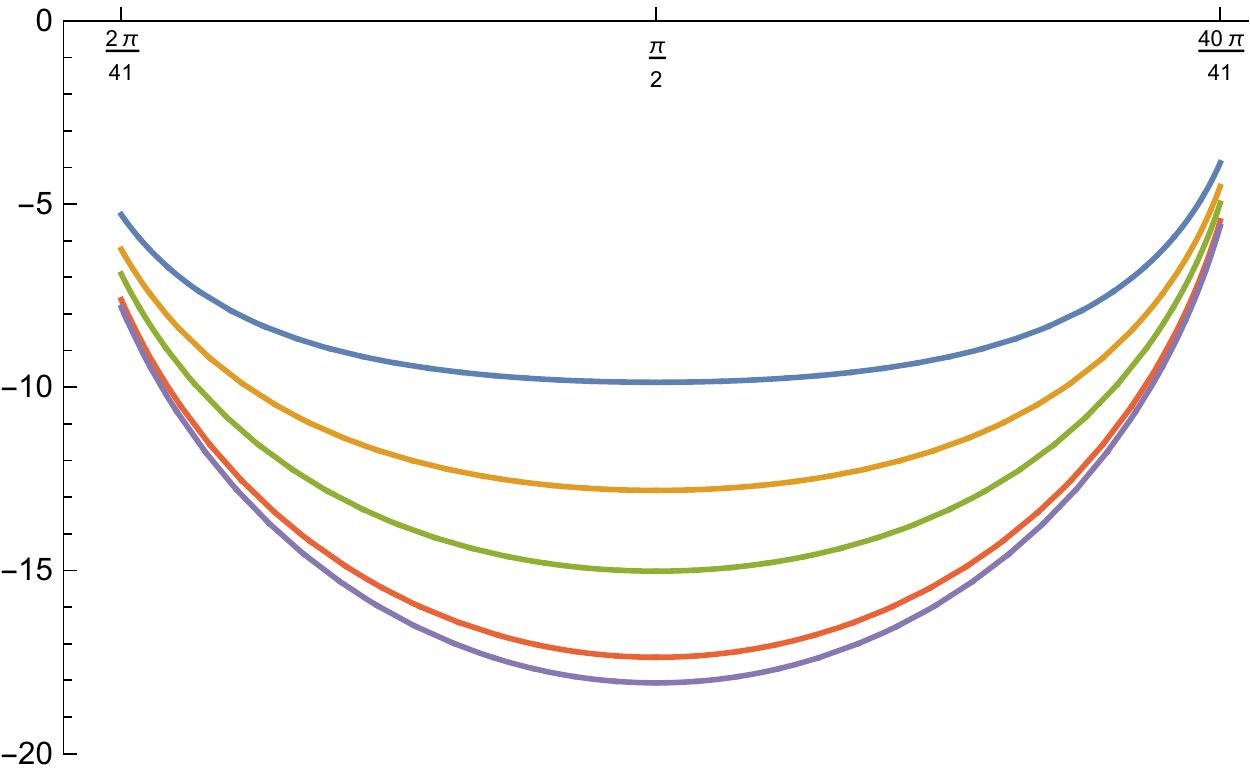}
	\end{center}
	\caption{ 
		Plot of ${\f{-N_t}2}\log\big{|}2T_{N_{x}}(a_{n})-2)\big{|}$ in terms of the continuous angle variable $x=n\psi\in\left[\f{2\pi}{N_x},\pi-\f{\pi}{N_x}\right]\subset[0,\pi]$ for the odd lattice size $N_{x}=41$ and $N_{t}=1$ and for spins $T=5$ and $L=5,10,20,100$. As $L$ increases and thus the ratio $\f TL$ goes to 0, the curves get more and more curved and goes to the limit function (lowest curve).}
	\label{chap5:fig:chebplot}
\end{figure}

In order to better understand the asymptotic behaviour at large $N_{x}$, let us focus on this factor. The function $\big{(}2T_{N_{x}}(c_{n})-2)\big{)}$ is well-behaved, both in modulus and phase, as one can see on figure \ref{chap5:fig:chebmodargplot}. The moot point is that the first winding number $n=1$ corresponds to the angle $x=\psi=\f{2\pi}{N_{x}}$ which goes to $x\rightarrow 0$ as $N_{x}$ grows large but not fast enough so that the asymptotic of $\big{(}2T_{N_{x}}(a_{1})-2)\big{)}$ be simply $\big{(}2T_{N_{x}}(0)-2)\big{)}=0$. Indeed, the two appearances of the lattice size $N_{x}$ conspire to give a non-trivial asymptotic
\begin{equation}
\big{[}2T_{N_{x}}(c_{n})-2)\big{]}^{\f12}\sim e^{(1+i)\sqrt{\f{\pi\lambda N_{x}n}{2}}}
\,,\quad\forall n\ll N_{x}
\,.
\end{equation}
This also gives the behaviour for large winding numbers $n\lesssim\lfloor\frac{N_x-1}{2}\rfloor$ since the function $\big{(}2T_{N_{x}}(x)-2)\big{)}$ is (almost) symmetric under reflections $x\leftrightarrow \pi-x$. In fact, the symmetry of $\big{(}2T_{N_{x}}(a(x))-2)\big{)}$ depends on the sign of $(-1)^{N_{x}}$. Under the reflection $x\rightarrow \pi-x$, it changes to its complex conjugate for even $N_{x}$ while it further gets an extra minus sign for odd $N_{x}$. This means that the modulus  $\big{|}2T_{N_{x}}(a_{n})-2)\big{|}$ is exactly symmetric under $n\leftrightarrow \frac{N_x}{2}-n$ for even $N_{x}$ while it is slightly skewed under the exchange $n\leftrightarrow \frac{N_x-1}{2}-n$ for odd $N_{x}$ due to the $\f12$ shift. This explains the peakedness of the amplitude pre-factor ${\cal A}(n)$ on the two limiting winding numbers  $n=1$ and $n=\lfloor\frac{N_x-1}{2}\rfloor$.

\begin{figure}[htb!]
	\begin{center}
		\includegraphics[height=40mm]{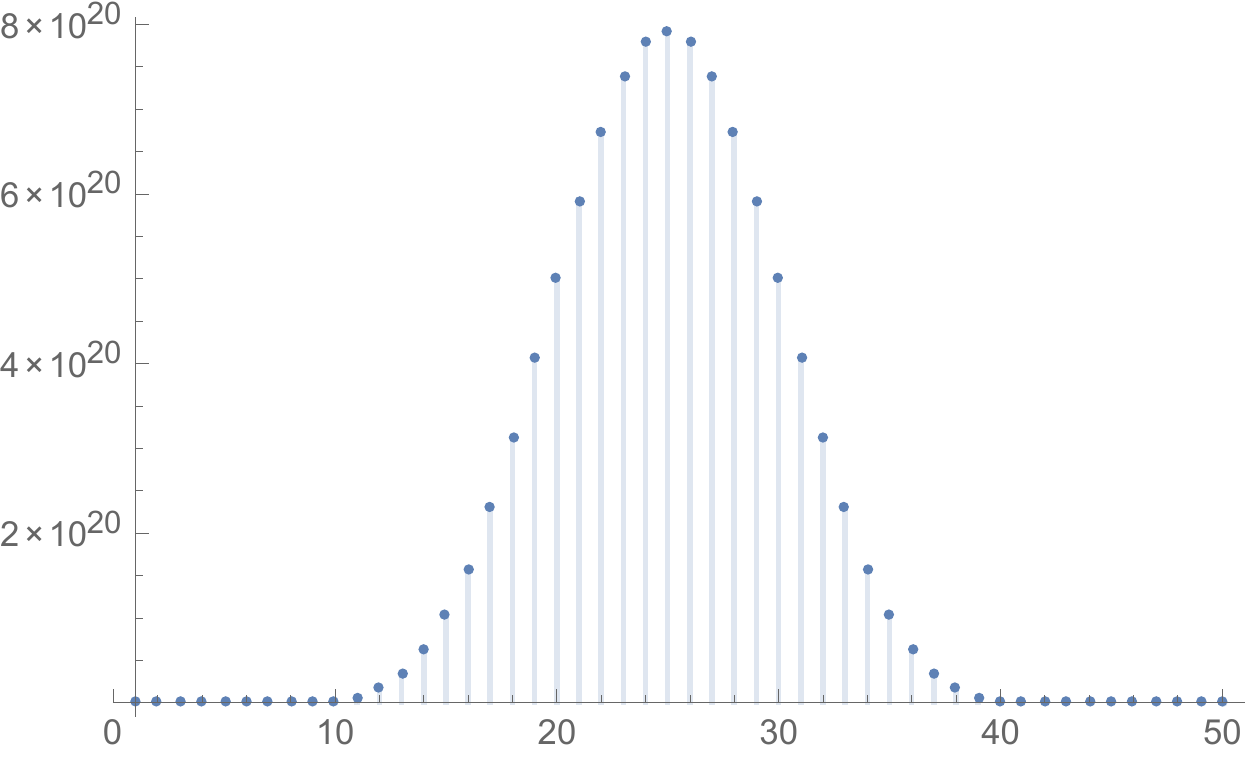}
		\hspace*{2mm}
		\includegraphics[height=40mm]{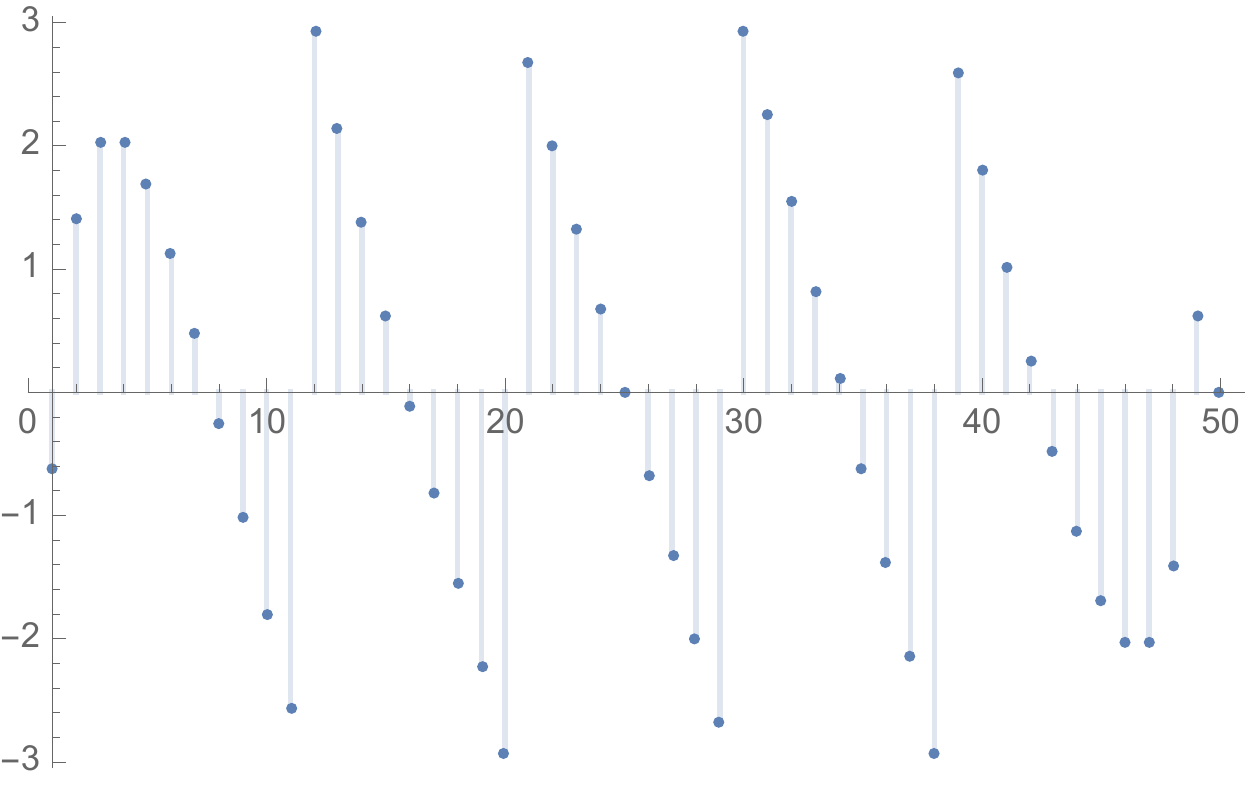}
	\end{center}
	\caption{Plots of the modulus (on the left) and argument (on the right) of $\big{(}2T_{N_{x}}(a_{n})-2)\big{)}$ for $N_{x}=100$ with $n$ running from 1 to 50, for spin parameters $L=T=1$.}
	\label{chap5:fig:chebmodargplot}
\end{figure}

The intuition behind the peakedness at the minimal and maximal winding numbers, $n=1$ and $n=\lfloor\frac{N_x-1}{2}\rfloor$, is that 
these solutions reconstruct locally almost-flat geometries (although one can be visualized as being folded onto itself), and at flat geometries the Hessian degenerates. This mechanism is analogous to Ditt-invariance \cite{Dittrich:2012jq,Rovelli:2011fk}. This peakedness can be used to argue that the slightest (semiclassical) knowledge of the extrinsic curvature, such as the fact that it is non-Planckian as in the maximal $n$ case, collapses the result onto the desired classical solution at $n=1$.\footnote{The very same physical argument allows to discard the folded solutions mentioned above.}

To summarize, the role of the amplitude pre-factor ${\cal A}(n)$ is to select the first winding number $n=1$ in the asymptotic limit, which corresponds to the semi-classical embedding of the torus surface in flat $\R^{3}$ space. The higher winding modes $n\ge 2$ however only allow for local embeddings (immersions) and seem to represent non-perturbative modes (instantons). We will see next that it also allows to reproduce the expected semi-classical partition function for three-dimensional quantum gravity as a function of the twist angle $\gamma$, namely the BMS character.

The last term is certainly the most interesting contribution to the amplitude
\begin{equation}
	\cD(\gamma,n)
	=  4 \sin^2 \left( \frac{\gamma n}{2} \right )  \times  \prod_{k=1}^{\frac{N_x-1}{2}}       \frac{1}{ 2-2\cos(\gamma k) } 
	= \big(2-2\cos(\gamma n)\big)  \times  \prod_{k=1}^{\frac{N_x-1}{2}}       \frac{1}{ 2-2\cos(\gamma k) } \; .
\end{equation}
This term is the only one depending on the twist angle. There are two contributions to this term. The first one comes from the measure factor from the integration over the remaining bulk information, and the last one is directly the Hessian contribution. As we see, the measure factor cancels exactly the contribution of the Fourier mode $k=n$. For the first winding mode $n=1$, we get a truncated product starting at $n=2$,
\begin{equation}
	\cD(\gamma,n=1)
	=   \prod_{k=2}^{\frac{N_x-1}{2}}       \frac{1}{ 2-2\cos(\gamma k) } 
	\; .
\end{equation}
This reproduces the results derived in \cite{Bonzom:2015ans} obtained from the path integral of Regge calculus, and fitting the 1-loop quantum General Relativity and he BMS character calculations of \cite{Barnich:2015mui,Barnich:2014kra,Oblak:2016eij}.

This beautiful interpretation needs to be put in balance against the fact that the product over $k=1,\dots,\frac12(N_x-1)$ is actually computable exactly, and gives\footnote{It comes from evaluating the polynomial $(X^{N_{x}}-1)/(X-1)$ at $X=1$.}

\begin{equation}
	\prod_{k=1}^{\frac{N_x-1}{2}} \big(2-2\cos(\gamma k) \big)=  \prod_{k=1}^{\frac{N_x-1}{2}}   4\sin^{2}\left(\f{\pi N_{\gamma} k}{N_{x}} \right)
	=
	\begin{cases}
		N_x \quad\textrm{if}\,\, K=\mathrm{GCD}(N_{\gamma},N_{x})=1 \\
		0 \quad\,\,\,\,\textrm{if}\,\, K=\mathrm{GCD}(N_{\gamma},N_{x})>1
	\end{cases}
	\;.
\end{equation}
Note, however, that the the derivation is only valid for $K=1$. It is interesting however that this derivation give rise to a divergence if we consider $K>1$, which is a case where the number of saddle is infinite, and hence the full analysis might provides a divergence.
\\

At this point two remarks are necessary.

First, it seems that the closed formula for the product over the Fourier modes kills the dependence of the Ponzano--Regge partition function in the twist angle $\gamma$. However, one key point is the tremendous difference in behaviour between the case $K=1$ and the case $K>1$. Taking $K=1$ gives a constant finite  result (simply $N_{x}$) for the product over Fourier modes $k$, leaving us simply with the measure factor $\sin^2\left(\f{\gamma n}2\right)$, while having $K>1$ leads formally to  a divergent amplitude. This divergence is actually due to a continuum of stationary points, which requires a finer analysis. As seen previously, in the asymptotic limit $(N_{\gamma},N_{x})\rightarrow \infty$, the case $K=1$ corresponds to an irrational value $\gamma\in2\pi(\R\setminus \Q)$, while $K>1$ corresponds to rational values $\gamma\in2\pi\Q$. This, again, corresponds to the behaviour obtained via computation in the continuum.

Secondly, even if the product over the spatial modes simplifies, giving the final result for the Ponzano-Regge partition function as a simple number, what is truly important is its explicit mode decomposition. Indeed the product formula \eqref{chap5:eq:amp_final} for  $\cD(\gamma,n)$ promises an interesting limit  $N_{x}\to\infty$ in terms of the inverse squared Dedekind $\eta$ function, which would establish explicitly the bridge between the Ponzano-Regge model for three-dimensional quantum gravity and the AdS${}_{3}$/CFT${}_{2}$ correspondence. However, the Dedekind function is only well-defined for $\gamma$ on the upper complex half-plane. As we recall previously, the perturbative one-loop calculations also need to introduce a regularization  $\gamma\mapsto\gamma + \I \epsilon^+$ to obtain meaningful amplitudes \cite{Giombi:2008vd,Barnich:2015mui}. In the present framework, the angle $\gamma$ is hardcoded into the calculation as a geometrical property of the lattice, as a ratio $\gamma=\f{2\pi N_{\gamma}}{N_{x}}$, which seems to make it unfeasible to extend to complex values. We will see in the next chapter that a natural complex regularization arises when considering a slightly more general boundary state.

Finally, we conclude this chapter with a important remark. In true, the result obtain in \eqref{chap5:eq:amp_final} is, in a sense, beautiful. This result basically tells us two things. First, by pushing the boundary to infinity, we do recover exactly the BMS character. It is interesting to point out that this is, again, not a true continuum limit, since we computed the amplitude in a large spins regime. Hence, the BMS structure seems also present for a intrinsically discrete boundary. Secondly, if the boundary is not at infinity, then non perturbative quantum corrections naturally arise. It is rather appealing to see the sum as taking into account all the possible geometries and possibly, topologies compatible with the boundary condition.
	
	\newpage
	~
	\thispagestyle{empty}
	
	\chapter{Exact evaluation of the Ponzano-Regge Amplitude for a Superposition of Coherent States}
\label{chap6}

In the previous chapter, we computed the Ponzano-Regge amplitude in the WKB approximation for a coherent boundary state. In this chapter, we want to go one step further, and try to exactly compute the amplitude of three-dimensional quantum gravity. To do so, we will consider boundary states construct as the generating function of spin-network. These states were first introduced in \cite{Freidel:2012ji}, where they were shown to allow an exact evaluation of the spin-network function on the sphere. Similar state where used in \cite{Bonzom:2015ova} for the Ponzano-Regge model on the sphere to show the duality between quantum gravity on the 3-ball and two copies of the Ising model. In this chapter, we want to use the same type of state in the torus topology to do an exact computation of the amplitude. This will allow us to define a completely generalized and regularized version of the BMS character. As in the previous chapter, it will feature non perturbative quantum corrections to the classic BMS character, on top of providing an expression valid for any twist parameter. This is to be compared with the previous computation where twist parameters leading to irrational angles were excluded from the computation. It is also worth noting that it is only in the following case that a true continuum limit can be taken. We will not, however study it presently. Instead, we focus on deriving the exact amplitude for quantum gravity on the torus.

\section{Generating function of spin-network state}

The generating function with weight $\omega$ is straightforwardly defined from the coherent spin-network state $\Psi_{coh}$ defined in \eqref{chap4:eq:general_coherent_SN} by
\begin{equation}
	\Psi_{\{y_l\}}(\{g_l\}) = \sum_{\{j_l\}} \omega(\{j_l\}) \left(\prod_{l} y_l^{2 j_l}\right) \Psi_{coh}(\{g_l\}) \; .
\end{equation}
As usual, the definition of the generating function is done by introducing dual variables for the spins. It is the generating function of coherent intertwiners at fixed spins $T_{t,x}$ and $L_{t,x}$. The dual parameters are coupling constants $\{y_l\}$ living on the links of the dual cellular decomposition. Since the sums over the spins are infinite, the weight $\omega$, which is a combinatorial factor depending on the spins must be carefully chosen to ensure convergence. It was shown in \cite{Freidel:2012ji} that the choice
\begin{equation}
	\omega = \prod_{n} \f{(J_n+1)!}{\prod_{l} (2 j_l)!} \; , \quad J_n = \sum_{l,n\in l} j_l
	\label{chap6:eq:weight_general}
\end{equation}
leads to a non-vanishing radius of convergence with respect to the coupling constants for the generating function. Up to the node factorial contribution $(J_n+1)!$n this weight defines a Poisson distribution on the spins. More particularly, such a choice for the weight reproduces a Gaussian as generating function for the spin-network \cite{Freidel:2012ji,Bonzom:2012bn,Dupuis:2011fz}. 

The boundary state is defined outside of the radius of convergence as the analytic continuation of the series. Such a coherent state is interpreted as a superposition of discrete geometries on the 2D boundary with different edge lengths \cite{Bonzom:2015ova}. The probability distribution of those edge lengths is given by the  weight $ \omega\big{(}\{j_l\}\big{)}$, the couplings $\prod_{l} y_l^{2 j_l}$ and the norm of the LS spin networks.
Indeed, we can compute the expectation value of any observable $\cO$ of the spins on a state $\Psi_{\{y_{l},\xi^{n}_{l}\}}$, 
\begin{equation}
\la\cO[\{j_l\}]\ra_{\{y_{l},\xi^{n}_{l}\}}
=
\f{\la \Psi_{\{y_{l},\xi^{n}_{l}\}}|\cO[\{j_l\}]|\Psi_{\{y_{l},\xi^{n}_{l}\}}\ra}{\la \Psi_{\{y_{l},\xi^{n}_{l}\}}|\Psi_{\{y_{l},\xi^{n}_{l}\}}\ra}
=
\f{\sum_{\{j_l\}}\cO[\{j_l\}]\cP_{\{y_{l},\xi^{n}_{l}\}}[\{j_l\}]}{\sum_{\{j_l\}}\cP_{\{y_{l},\xi^{n}_{l}\}}[\{j_l\}]}
\,,
\end{equation}
where the probability distribution for the spins is computed directly in terms of the norm of the coherent intertwiners:
\begin{equation}
\cP_{\{y_{l},\xi^{n}_{l}\}}(\{j_l\})
=
\omega\big{[}\{j_l\}\big{]}^{2} \left(\prod_{l} y_l^{4 j_l}\right)
\,\prod_{l}\f1{(2j_{l}+1)}
\,\prod_{n} \big{\la} \iota_{n}[\{\xi_{l}\}_{l \ni n}]
\,\big{|}\,
\iota_{n}[\{\xi_{l}\}_{l \ni n}]\big{\ra}
\,.
\end{equation}
The norm of a coherent intertwiner was computed in \cite{Freidel:2010tt,Bonzom:2012bn} from its generating function, here keeping the node index $n$ implicit:
\begin{align}
\big{\la} \iota[\{\xi_{l}\}]
\,\big{|}\,
\iota[\{\xi_{l}\}]\big{\ra}
&=
\int_{\SU(2)} \dd g\,
\prod_{l}\la\xi_{l}|g|\xi_{l} \ra^{2j_{l}}
\,, \\
\int_{\SU(2)} \dd g\,\,
e^{\sum_{l}\la\xi_{l}|g|\xi_{l} \ra}
&=
\sum_{J\in\N}
\f1{J!(J+1)!}
\left(\f{1}{2}\sum_{l,\tilde{l}}\big|[\xi_{l}|\xi_{\tilde{l}}\ra\big|^{2}\right)^{J}
\,,
\end{align}
from which we can extract the term of power $2j_{l}$ in each spinor $\xi_{l}$ to get the norm of a coherent intertwiner. Aside the exact expression, it is actually straightforward to derive its behaviour at large spins by computing the saddle point approximation of the integral over $\SU(2)$ as one rescales homogeneously the spins by a large factor, $j_{l}\mapsto \lambda j_{l}$ with $\lambda \rightarrow +\infty$. As shown in \cite{Livine:2007vk}, the maximum of the integrand $\prod_{l}\la\xi_{l}|g|\xi_{l} \ra^{2j_{l}}$ is necessarily reached at $g=\id$ and this is a stationary point if and only if the spinors satisfy the closure condition
\begin{equation}
\vcC=\sum_{l}j_{l}\f{\la\xi_{l}|\vsigma|\xi_{l}\ra}{\la\xi_{l}|\xi_{l}\ra}
=\sum_{l}j_{l}\f{\vu_{\xi_{l}}}{|\vu_{\xi_{l}}|}
=0
\,.
\end{equation}
Let us stress that the closure conditions are invariant under homogeneous rescaling of the spins $j_{l}\mapsto \lambda j_{l}$.
Thus, as confirmed by numerics, the behaviour is very different whether  the closure condition is satisfied or not. On the one hand, when the closure vector vanishes, $\vcC=0$, it means geometrically that the edge vectors around the node actually close and form a polygon. In that case, the norm $\big{\la} \iota[\{\xi_{l}\}] \,\big{|}\,\iota[\{\xi_{l}\}]\big{\ra}$ generically behaves as $\lambda^{-\f32}\prod_{l}\la\xi_{l}|\xi_{l}\ra^{2j_{l}}$ .  On the other hand, the case when the closure vector does not vanish, the integral is exponentially suppressed and the coherent intertwiner norm decreases as $\exp[-\lambda^{2}|\vcC|^{2}]$.

This means that the effect of the coherent intertwiner norm in the spin probability distribution is two-fold. First, it selects the spins configuration such that the closure conditions are satisfied around every node $n$.
Second, putting aside the algebraic term and focusing on the leading order exponential behaviour, it produces spinor norm factors $\la\xi_{l}|\xi_{l}\ra^{2j_{l}}$ that can be reabsorbed in the couplings $y_{l}\rightarrow Y_{l}=\sqrt{\la\xi^{n}_{l}|\xi^{n}_{l}\ra}\, y_{l}$ upon normalizing all the spinors $\xi^{n}_{l}\rightarrow \sqrt{\la\xi^{n}_{l}|\xi^{n}_{l}\ra}\,\hat{\xi}^{n}_{l}$.

Putting all the contributing factors together, combinatorial weight, couplings and intertwiner norms, using the Stirling formula approximating the factorials in the combinatorial weight at large spins, assuming that the spins satisfy the closure conditions around every node, and finally focusing on the exponential factors and considering the algebraic factors as sub-leading,  the spin probability distribution behaves at large spins as:
\begin{equation}
\cP_{\{y_{l},\xi^{n}_{l}\}}[\{j_l\}]
\sim
e^{2\Phi[\{j_l\}]}
\qquad\textrm{with}\quad
\Phi[\{j_l\}]
=
\sum_{l }2j_{l}(\ln Y_{l}-\ln j_{l} ) +\sum_{n}J_{n}\ln J_{n}
\,.
\end{equation}
The extrema of this probability distribution is given by a vanishing derivative with respect to each spin $j_{l}$:
\begin{equation}
\forall l\in\Gamma\,,\qquad
\pp_{j_{l}}\Phi=0
\quad\Longleftrightarrow\quad
\ln \f{Y_{l}^{2}J_{s(l)}J_{t(l)}}{4 j_{l}^{2}}=0
\quad\Longleftrightarrow\quad
\f{j_{l}^{2}}{J_{s(l)}J_{t(l)}}=\f{Y_{l}^{2}}4
\,.
\end{equation}
This equation for the peaks of the spin  probability distribution has the crucial property of being invariant homogeneous rescaling of the spins $j_{l}\mapsto \lambda j_{l}$. This scale-invariance of the extrema of the   probability distribution defined by this class of coherent spin networks was put forward in \cite{Bonzom:2015ova}. The fixed point equations for the spins induce non-trivial constraints on the couplings $Y_{l}$. But, once  the couplings allow for a solution for the spins $\{j_{l}\}$, then there exist a whole line of solution obtained by arbitrary global rescalings of the spins.
In the case of a planar 3-valent graph, these constraints determine a triangulation whose angles are fixed in terms of the couplings $Y_{l}$ and whose edge lengths give the spins \cite{Bonzom:2015ova}. Furthermore, in that case, it is understood that such fixed point couplings are related to the critical couplings of the Ising model on the considered graph \cite{Bonzom:2015ova,Bonzom:2019dpg}.

So we have three types of behaviour for the probability distribution. Either, the couplings $Y_{l}$ are within the convergence radius of the series defining the coherent spin network and the spin probability distribution is peaked on low spins. Or the couplings $Y_{l}$ make the series diverge in which case the spin probability distribution favors large spins. And finally, in between, there is a critical regime, where the couplings $Y_{l}$ lead to a line of maxima of the spin probability distributions where the spins are the edge lengths, up to arbitrary global rescaling, of a planar 2D cellular decomposition whose angles are determined by the critical couplings. Such a line of stationary points corresponds to a pole for the coherent spin network wave-function \cite{Freidel:2010tt,Bonzom:2012bn,Bonzom:2015ova,Bonzom:2019dpg}. 

We will explicitly show this property in the following while focusing on the case of interest, the torus. It is the same state that was used in \cite{Bonzom:2015ova} to show that the Ponzano-Regge model on the sphere with a triangulation is related to two copies of the Ising model on the same triangulation.

It is easy to write the generating function in a more explicit form using the discretization of the torus introduced in chapter \ref{chap4}. Recall that we discretize the torus with $N_t$ horizontal slices and $N_x$ vertical ones such that the boundary lattice is a square lattice. With this choice of discretization, the weight takes the particular form
\begin{equation*}
	\omega = \prod_{t,x = 0}^{N_t-1,N_x-1} \frac{\lambda_{t,x}^{2 L_{t,x}} \tau_{t,x}^{2T_{t,x}}}{(2 L_{t,x})!(2T_{t,x})!} (J_{t,x}+1)! \quad \text{where} \quad J_{t,x} := T_{t,x} + L_{t,x} + T_{t,x-1} + L_{t-1,x} \; .
\end{equation*}
To differentiate vertical and horizontal links, we introduced two different names for the coupling, namely $\lambda_{t,x}$ and $\tau_{t,x}$ respectively instead of $y_l$. The generating function of coherent spin networks then trivially takes the form
\begin{align}
	\Psi_{\{\lambda,\tau\}}(\{g_{t,x}^v,g_{t,x}^h\}) = &\sum_{T_{t,x}} \sum_{L_{t,x}} \prod_{t,x} \left( \frac{\lambda_{t,x}^{2 L_{t,x}} \tau_{t,x}^{2T_{t,x}}}{(2 L_{t,x})!(2T_{t,x})!} (J_{t,x}+1)!\right) \times \nonumber \\ & ~~ \int_{\SU(2)} \prod_{t,x=0}^{N_t-1,N_x-1} \dd G_{t,x} \; \la \uparrow |G^{-1}_{t,x+1} g_{t,x}^h G_{t,x} | \uparrow \ra^{2 T_{t,x}} \la + |G^{-1}_{t+1,x} g_{t,x}^v G_{t,x} | + \ra^{2 L_{t,x}} \; ,
	\label{chap6:eq:boundary_state_scale_free}
\end{align}
with the explicit expression of the coherent spin-network state $\Psi_{coh}$.

In the next section, we will focus on the large spins behaviour for the boundary state to have a better understanding of its geometry. To do so, the previous expression will be useful. However, for the actual computation of the Ponzano-Regge amplitude, it is more practical to rewrite the state making the Gaussian integral explicit. To do so, we rewrite \eqref{chap6:eq:boundary_state_scale_free} in terms of integrations over spinors instead of group elements. This transformation is based on the following lemma, see \cite{Freidel:2012ji} for the proof
\begin{lemma}
	Consider a spinor $|w \ra$. From $|w\ra$, construct the group element
	\begin{equation*}
	g(w) = |\uparrow\ra\la w| + |0][w| \; , \quad \text{with} \;\; g^{\dagger}(w)g(w) = \la w | w \ra \; .
	\end{equation*}
	Consider a function $F$ such that $F(g(w))$ is homogeneous of degree $2J$ in $|w \ra$
	\begin{equation*}
	F(g(c w)) = c^{2J} F(g(w)) \quad \text{for} \; c \in \R \;.
	\end{equation*}
	Then the following equality holds
	\begin{equation}
	(J+1)! \int_{SU(2)} \text{d}g F(g) = \int_{\C^2} \f{\text{d}^4 w}{\pi^2} e^{-\la w | w \ra} F(g(w))
	\label{chap6:eq:int_SU(2)_to_spinor}
	\end{equation}
	where $\text{d}g$ is the Haar measure on $\SU(2)$.
\end{lemma}

This lemma can be applied to the generating function of the spin-network. We denote by $F$ the integrand of the $\SU(2)$ integration
\begin{equation}
	F(\{G_{t,x}\}) = \la \uparrow |G^{-1}_{t,x+1} g_{t,x}^h G_{t,x} | \uparrow \ra^{2 T_{t,x}} \la + |G^{-1}_{t+1,x} g_{t,x}^v G_{t,x} | + \ra^{2 L_{t,x}} \; .
\end{equation}
The $g^{h,v}_{t,x}$ are parameters of the function, and we do not need to look at them any further for now. This function can easily be extended to be over $\UU(2)$ and not $\SU(2)$ defining its extension to $\UU(2)$ by
\begin{equation}
	F(\{G_{t,x}\}) = \la \uparrow |G^{\dagger}_{t,x+1} g_{t,x}^h G_{t,x} | \uparrow \ra^{2 T_{t,x}} \la + |G^{\dagger}_{t+1,x} g_{t,x}^v G_{t,x} | + \ra^{2 L_{t,x}} \; ,	
\end{equation}
since $G^{-1}=G^{\dagger}$ for an $\SU(2)$ element. It is clear that $F$ is a polynomial in $G_{t,x}$ of degree $2J_{t,x} = 2 (T_{t,x} + L_{t,x} + T_{t,x-1} + L_{t-1,x})$ for all pairs $(t,x)$. Hence, the lemma can be applied. Conveniently, the weight $\omega$ comes with the necessary $(J_{t,x}+1)!$ factor. Taking it into account, the integration over $\SU(2)$ is replaced by an integration over $\C^2$ applying the previous lemma. We get
\begin{align*}
	\int_{\C^2} \prod_{t,x = 0}^{N_t-1,N_x-1} \f{\text{d}^4 w_{t,x}}{\pi^2} e^{-\la w_{t,x} | w_{t,x} \ra}  &\la \uparrow |G^{\dagger}(w_{t,x+1}) g_{t,x}^h G(w_{t,x}) | \uparrow \ra^{2 T_{t,x}} \\ & \la + |G^{\dagger}(w_{t+1,x}) g_{t,x}^v G_{t,x}(w_{t,x}) | + \ra^{2 L_{t,x}}\; .
\end{align*}

The last step to obtain the Gaussian integration is to consider the sum over the spins. Considering for example a given spin $L_{t,x}$, the remaining factor we need to take into account is $\sum\limits_{L_{t,x}} \f{\lambda_{t,x}^{2 L_{t,x}}}{(2 L_{t,x})}$ such that the sum can be done explicitly and return an exponential. At the end of the day, all the sums can be taken care of in the same way, and we obtain a purely Gaussian state for the generating function of spin-network
\begin{align*}
	\Psi(\{g_{t,x}^v,g_{t,x}^h\}) = \int_{\C^2} \prod_{t,x = 0}^{N_t-1,N_x-1} \f{\text{d}^4 w_{t,x}}{\pi^2} \text{exp}\Big(-\la w_{t,x} | w_{t,x} \ra + &\tau_{t,x} \la \uparrow |G^{\dagger}(w_{t,x+1}) g_{t,x}^h G(w_{t,x}) | \uparrow \ra \\+ &\lambda_{t,x} \la + |G^{\dagger}(w_{t+1,x}) g_{t,x}^v G_{t,x}(w_{t,x}) | + \ra \Big) \; .
\end{align*}

Note that while the integration over the spinorial variables in a Gaussian, it is not straightfoward to have an explicit evaluation for it in the general case. Indeed, the dependency on the lattice site of the group elements $g_{t,x}^{h,v}$ makes the Hessian of the Gaussian integration complicated. However, when evaluating the boundary state on the Ponzano-Regge partition function, most of these parameters are fixed to the identity. Recall that it only remains one integration over $\SU(2)$ that can be made homogeneous on the whole lattice. Before moving on to the actual computation, we focus in the next section on the geometry of the boundary state at the saddle point. 

\subsection{Critical Regime and Scale Invariance of the Boundary State}

In this section, we look at a saddle point approximation of the generating function starting from equation \eqref{chap6:eq:boundary_state_scale_free}. More particularly, we look at saddle point approximation on the arguments of the sums in \eqref{chap6:eq:boundary_state_scale_free} on a large spins limit. Calling the arguments of the sums $I$, we have
\begin{align*}
	I_{\lambda,\tau,g^{h,v}}(\{T_{t,x}\},\{L_{t,x}\},\{G_{t,x}\}) = \prod_{t,x}& \left( \frac{\lambda_{t,x}^{2 L_{t,x}} \tau_{t,x}^{2T_{t,x}}}{(2 L_{t,x})!(2T_{t,x})!} (J_{t,x}+1)!\right) \\
	&\begin{split}
	\int_{SU(2)} \prod_{t,x=0}^{N_t-1,N_x-1} \text{d} G_{t,x} &\la \uparrow |G^{-1}_{t,x+1} g_{t,x}^h G_{t,x} | \uparrow \ra^{2 T_{t,x}}\\& \la + |G^{-1}_{t+1,x} g_{t,x}^v G_{t,x} | + \ra^{2 L_{t,x}}.
	\end{split}
\end{align*}
It is a function of $G_{t,x}$ and of the spins $L_{t,x}$ and $T_{t,x}$ while the coupling constants and the discrete connection group elements are parameters.

In the large spins limit, we can make use of the Stirling formula $n! \sim \sqrt{2 \pi n } \frac{e^n}{n^n}$ to greatly simplified the previous expression. Only keeping the exponential contribution to the Stirling formula, we get
\begin{equation*}
	I \underset{T,L >>1}{\sim} e^{I_1} \int_{SU(2)} \left( \prod_{t,x} \text{d} G_{t,x} \right) e^{-I_2}
\end{equation*}
where
\begin{align*}
	I_1(\{T_{t,x},L_{t,x}\}) = \sum_{t,x=0}^{N_t-1,N_x-1} &2 L_{t,x} \ln(\lambda_{t,x}) + 2 T_{t,x} \ln(\tau_{t,x}) + J_{t,x} \ln(J_{t,x}) \\ &- 2L_{t,x} \ln(2L_{t,x}) - 2T_{t,x} \ln(2T_{t,x}) \\
\end{align*}
and
\begin{align*}
	I_2(\{T_{t,x},L_{t,x},G_{t,x}\}) = -\sum_{t,x=0}^{N_t-1,N_x-1} &2 T_{t,x} \ln(\la \uparrow |G^{-1}_{t,x+1} g_{t,x}^h G_{t,x} | \uparrow \ra) \\ +& 2 L_{t,x} \ln(\la + |G^{-1}_{t+1,x} g_{t,x}^v G_{t,x} | + \ra)	\; .
\end{align*}
The saddle point approximation is done on the group variables $G_{t,x}$ and for the spins variables $T_{t,x}$ and $L_{t,x}$ and not on the parameters $g_{t,x}^h$ and $g_{t,x}^v$  or on the complex coupling constants $\lambda_{t,x}$ and $\tau_{t,x}$.

The factor $I_2$ is the action that was used in the previous chapter \ref{chap5} and is the contribution coming from the coherent state. Therefore, we already studied its the saddle point approximation in details in chapter \ref{chap5}. We know that it is peaked on the geometrical contribution given by the dihedral angles between the square of the discretization. See equation \eqref{chap5:eq:gluing_equation}. The remaining part that we need to study here is the contribution coming from the weight on the spin variables. This study was already partially done in \cite{Bonzom:2015ova}, where the geometrical meaning of the weight was understood. Considering a triangulation of the sphere, it was shown that, at the saddle point, the weight relates the spins to the coupling parameter in such a way that they encode the $2D$ inner angles of the triangles. Recall that we have three dual links for a triangle, hence three couplings parameters, corresponding to the tree angles of the triangle. At the saddle point the coupling parameters were naturally scale invariant. Here, we want to consider at the same time the contribution from the weight and from the coherent states. We will recover the scale invariance property, but not the relation with in terms of angles.

Note however, that we do not really perform a saddle point approximation, but only a stationary point. That is, we only consider real spins here, coming from the $\SU(2)$ recoupling theory. Nothing, however, stops us from looking at the analytical continuation of such an object. In that case, it was shown to lead to complex geometry \cite{Bonzom:2019dpg} when considering the sphere. Here, we consider only real geometry, which is easier to understand.
\medskip

The stationary point equations read
\begin{align*}
	\partial_{L_{t,x}} (I_1+I_2) = 0 \implies  2 \ln(\lambda_{t,x}) = &\ln\left(\f{(2L_{t,x})^2}{J_{t,x}J_{t+1,x}}\right) \\ &- 2\ln\Big(\la + |G^{-1}_{t+1,x} g_{t,x}^v G_{t,x} | + \ra\Big)\; ,\\ 
\end{align*}
and
\begin{align*}
	\partial_{T_{t,x}}(I_1+I_2) = 0 \implies  2 \ln(\tau_{t,x}) = &\ln\left(\f{(2T_{t,x})^2}{J_{t,x}J_{t,x+1}}\right) \\&- 2\ln\Big(\la + |G^{-1}_{t+1,x} g_{t,x}^v G_{t,x} | + \ra\Big)\; .
\end{align*}
On shell of the saddle point equations for $G_{t,x}$, that is, given the relation \eqref{chap5:eq:gluing_equation}, these equations can be easily solved and return a relation between spins, dihedral angles and coupling parameters
\begin{equation}
	\lambda_{t,x}^2 = \f{(2L_{t,x})^2}{J_{t,x}J_{t+1,x}} e^{-i 2\psi_{t,x}^L} \; ,
	\label{chap6:eq:saddle_lambda}
\end{equation}
\begin{equation}
	\tau_{t,x}^2 = \f{(2T_{t,x})^2}{J_{t,x}J_{t,x+1}} e^{-i 2\psi_{t,x}^T} \; .
	\label{chap6:eq:saddle_tau}
\end{equation}

From \eqref{chap6:eq:saddle_lambda} and \eqref{chap6:eq:saddle_tau} we see that, on-shell and in the asymptotic limit, the couplings are left unchanged under global rescaling of the spins. Since the spins are related to the length of the discretization, we recover the interpretation that the saddle point contribution is scale-invariant. That is, we can freely rescale all the spins, i.e. all the lengths by the same global parameter without changing the value of the coupling constants. Note that the phase of the coupling constants are related to the dihedral angles between two neighbourhood squares of the discretization. That is, the coupling constants readily encode at the same time the information about the extrinsic curvature and the scale of the torus. It is not possible however, to interpret the module of the coupling constant at the saddle point as also being related to the angles of the square as it was done in \cite{Bonzom:2015ova}.

Note that we have also used the fact that the spins $T_{t,x}$ determining the time intervals are constant in space $T_{t,x}=T_{t}$ and the spins $L_{t,x}$ determining the space intervals are constant in time $L_{t,x}=L_{x}$, as  resulting from the  closure constraints induced by the saddle point equations in $g_{t,x}$.

We must not forget that the couplings $\lambda_{t,x}$ and $\tau_{t,x}$ are given a priori and that the stationary point equations determine extremal spin configurations $\{L_{x},T_{t}\}$ in terms of those couplings.
A first important point is that such stationary point do not always exist. Indeed, the existence of a solution to the stationary point equations requires the couplings to satisfy some constraints. We will refer to couplings that fulfil those conditions as {\it critical couplings}.
The second important point is that, once a set of critical couplings has been chosen and one solution for the spin configuration $\{L_{x},T_{t}\}$ has been identified, then any arbitrary global rescaling of  these spins still gives a solution to the stationary point equations. This means that we actually a whole line of stationary points. This induced scale invariance for critical couplings then leads to a divergence of the series defining the coherent spin network wave-function. For non-critical couplings, there is no finite real stationary points, we lose the scale-invariant line of stationary points in the spins and the dominant contributions are given  either by low spins or  or by spins growing to infinity, respectively corresponding to an absolutely convergent or totally divergent series.

So let us start by highlighting necessary constraints satisfied by the critical couplings. Since the phases of the couplings give the dihedral angles determined by the stationary points in $g_{t,x}$, the stationary point equations in the spins will give conditions on the modulus of the couplings.

Let us focus on the stationary point equations involving $T_{t}, T_{t+1}, L_{x}, L_{x+1}$:
\begin{align*}
|\lambda_{t,x}| = \f{2L_{x}}{\sqrt{J_{t,x}J_{t+1,x}}}
\,,\quad
|\tau_{t,x}| = \f{2T_{t}}{\sqrt{J_{t,x}J_{t,x+1}}} 
\,,\quad
\\
|\lambda_{t,x+1}| = \f{2L_{x+1}}{\sqrt{J_{t,x}J_{t+1,x+1}}}
\,,\quad
|\tau_{t+1,x}| = \f{2T_{t+1}}{\sqrt{J_{t+1,x}J_{t+1,x+1}}}
\,.
\end{align*}
We can re-write these equations directly in terms of the ratios of $L$'s over $T$'s:
\begin{equation*}
|\tau_{t,x}| =
\f1{\sqrt{(1+A)(1+C)}}
\,,\qquad
|\tau_{t+1,x}| =
\f1{\sqrt{(1+B)(1+D)}}
\,,
\end{equation*}
\begin{equation*}
|\lambda_{t,x}| =
\f{\sqrt{AB}}{\sqrt{(1+A)(1+B)}}
\,,\qquad
|\lambda_{t,x+1}| =
\f{\sqrt{CD}}{\sqrt{(1+C)(1+D)}}
\,,\nn
\end{equation*}
with the notation:
\begin{align*}
A:=\f{L_{x}}{T_{t}}&
\,,\quad
B:=\f{L_{x}}{T_{t+1}}
\,,\quad
C:=\f{L_{x+1}}{T_{t}}
\,,\quad
D:=\f{L_{x+1}}{T_{t+1}}
\,,\quad
\\
&\textrm{satisfying the relation}\,\,
AD=BC
\,.
\end{align*}
These expressions directly imply that the couplings satisfy a polynomial equation:
\begin{equation}
\label{eq:critical}
1+|\lambda_{t,x}|^2|\lambda_{t,x+1}|^2+|\tau_{t,x}|^2|\tau_{t+1,x}|^2
=
2|\lambda_{t,x}||\lambda_{t,x+1}||\tau_{t,x}||\tau_{t+1,x}|
+|\lambda_{t,x}|^2+|\lambda_{t,x+1}|^2+|\tau_{t,x}|^2+|\tau_{t+1,x}|^2
\,.
\end{equation}
This looks  similar to the equation for the zeroes of the 2D Ising model. In light of the recently uncovered relation between the geometrical stationary point of coherent spin networks on triangulations and the critical couplings of the 2D Ising model on dual graphs \cite{Bonzom:2015ova,Bonzom:2019dpg}, it would probably be enlightening to investigate this similarity further and understand if there is indeed a straightforward mapping between critical couplings of coherent spin networks and 2D Ising model on the square lattice.

Then, once the couplings satisfy such a condition, it is possible to compute the spin ratios A,B,C,D in terms of the couplings. Extending this from node to node and enforcing the periodicity conditions of the twisted torus, we finally get a spin configuration solution on the whole square lattice. 

\bigskip

The rest of this chapter focuses on the explicit and {\it exact} computation for the amplitude of the Ponzano-Regge model. In this context, it will be complicated to keep all the couplings independent: we will again perform a Fourier Transform to compute the integral. To do so, we require that the couplings parameters do not depend on the lattice position. In the following we consider the homogeneous but anisotropic case where
\begin{equation*}
\tau_{t,x} = \tau  \qquad \text{and} \qquad \lambda_{t,x} = \lambda \qquad \forall (t,x) \; .
\end{equation*}
From the point of view of the continuum limit, this choice also makes sense. If the couplings are kept depending on the lattice position, they will become field. Hence, they cannot really be considered coupling anymore. However, this choice completely destroys the interpretation of being the generating function of the coherent spin networks. Further studies keeping the coupling arbitrary are left for the future.

\section{Three-dimensional Quantum Gravity Amplitude with Boundary}

Now that we have a better understanding of the boundary state geometry at the saddle point, it is time to look at the actual computation of the Ponzano-Regge amplitude. Taking into account the necessary gauge-fixing of the model, the amplitude reads, making the dependencies on $\lambda$ and $\tau$ explicit 
\begin{equation}
	\la Z_{PR}^{\cK}| \Psi  \ra_{\lambda,\tau} = \f{1}{\pi} \int_{0}^{2\pi} \text{d}\varphi \sin^{2}(\varphi) \int_{\C^2}\prod_{t,x=0}^{N_t-1,N_x-1} \frac{\text{d}^4 w_{t,x}}{\pi^2} e^{- S_{\lambda,\tau}[\{w_{t,x}\},\varphi]} \; .
	\label{chap6:eq:PR_amplitude_starting}
\end{equation}
with 
\begin{align}
	S_{\lambda,\tau}[\{w_{t,x}\},\varphi] = \sum_{t,x} \la \omega_{t,x} | \omega_{t,x} \ra \nonumber - &\tau \la \uparrow | \omega_{t,x+1} \ra \la \omega_{t,x} | \uparrow \ra -  \tau \la \uparrow | \omega_{t,x+1} ] [ \omega_{t,x} | \uparrow \ra \\ &- \lambda e^{i \frac{\varphi}{N_t}} \la + | \omega_{t+1,x} \ra \la \omega_{t,x} | + \ra - \lambda e^{-i \frac{\varphi}{N_t}} \la + | \omega_{t+1,x} ] [ \omega_{t,x} | + \ra\; .
\end{align}

In the action, $\lambda$ and $\tau$ are parameters, whereas they are the variables of the amplitude. We begin this section with the computation of the Gaussian integral. Then, we look at a particular case, allowing us to recover the result of the previous chapter in the right limit. Finally, we exactly compute, for the first time, the amplitude of quantum gravity with finite boundary on the torus. The computation are fairly technical, hence, most of the section starts with a proposition containing the result.

\subsection{Gaussian Integral and Interpretation}

The first step of the computation is to perform the Gaussian integration on the spinors. This can be done without any complication. To do so, we again consider a twisted Fourier Transform. We denote the Fourier component by $\tilde{w}_{\omega,k}$, defined by
\begin{equation}
	\tilde{w}_{\omega,k}^a = \f{1}{\sqrt{N_t N_x}} \sum_{t,x=0}^{N_t-1,N_x-1} w_{t,x}^a e^{i \f{2\pi}{N_x}k} e^{i \left(\f{2\pi}{N_t}\omega - \f{1}{N_t}\gamma k \right)} \; ,
\end{equation}
where $a=0,1$. As previously, we denote by $\omega$ the time modes, running from $0$ to $N_t-1$ and by $k$ the spatial modes, running from $0$ to $N_x-1$. We recall that $\gamma = \f{2 \pi N_\gamma}{N_x} $ is the discrete twist angle. 

With this transformation, the action takes a simple form
\begin{equation}
	S_{\lambda,\tau}[\{\tilde{w}_{\omega,k} \},\varphi] = \sum_{\omega,k=0}^{N_t-1,N_x-1} \la \tilde{w}_{\omega,k} |Q_{\omega,k}(\varphi)| \tilde{w}_{\omega,k}\ra
	\label{chap6:eq:arg_exp_action}
\end{equation}
where
\begin{equation*}
Q_{\omega,k}(\varphi) = \id - \tau e^{i \f{2\pi}{N_x}k \sigma^z} - \lambda e^{i \left( \f{\varphi}{N_t} + \f{2\pi}{N_t}\omega - \f{1}{N_t}\gamma k \right) \sigma^x} \; .
\end{equation*}
For the sake of keeping the notations simple, we did not make the dependency of $Q_{\omega,k}$ on the coupling constants $\lambda$ and $\tau$ explicit. 

In this form, the Gaussian integration is apparent and can be done exactly.
\begin{prop}
	The Gaussian integration over the spinors $w_{t,x}$ in equation \eqref{chap6:eq:PR_amplitude_starting} gives
	\begin{equation}
	\la Z_{PR}^{\cK}| \Psi  \ra_{\lambda,\tau} = \f{1}{\pi} \int_{0}^{2\pi} \text{d}\varphi \sin^{2}(\varphi) \prod_{\omega,k = 0}^{N_t-1,N_x-1} \f{1}{\det(Q_{\omega,k}(\varphi))} 
	\label{chap6:eq:starint_point_amp}
	\end{equation}
	with
	\begin{align*}
	\det(Q_{\omega,k}(\varphi)) 
	= 1 + &\lambda^2 + \tau^2  -2\tau \cos\left(\f{2\pi}{N_x}k\right) - 2\lambda \cos\left(\f{\varphi}{N_t} + \f{2\pi}{N_t}\omega - \f{1}{N_t}\gamma k\right) \\ + & 2 \lambda \tau \cos\left(\f{2\pi}{N_x}k\right) \cos\left(\f{\varphi}{N_t} + \f{2\pi}{N_t}\omega - \f{1}{N_t}\gamma k\right) \; .
	\end{align*}
\end{prop}
The proof is straightforward given the expression of the arguments of the exponential in formula \eqref{chap6:eq:arg_exp_action}.

\medskip

In this chapter, we are computing the amplitude of the Ponzano-Regge model exactly given the considered boundary conditions. Basically, we are computing the partition function integrated on a specified class of boundary data. This is just a number that does not, usually, give a lot of information. It is however possible to extract a lot from this number if the partition function is written in a correct form, that is with a apparent pole structure. We want to look at the poles of the amplitude as if they were the dispersion relation of the theory, therefore characterizing the theory we are working on. This is the whole point of this chapter, to compute the Ponzano-Regge amplitude, which is by definition finite, in such a form that it has apparent poles.

Now working from the expression of the Ponzano-Regge amplitude as a trigonometric integral, the next step is to identify the poles in $\varphi$ and calculate their residue  in order to compute this integral as a contour integral in the complex plane. The class angle $\varphi$ encodes the holonomy along the non-contractible cycle of the solid torus and describes the curvature of the geometry. Poles in $\varphi$, corresponding to zeros of the determinants, $\det(Q_{\omega,k}(\varphi))=0$, can be loosely thought as resonances in the connection.

From the perspective of standard (free) quantum field theory, zeroes of the determinants $\det[Q_{\omega,k}(\varphi)]=0$ would correspond to the classical dispersion relation between the frequency $\omega$ and the momentum $k$. Here, not only this relation depends on the curvature angle $\varphi$, but we further have to integrate over that angle. Before computing that integral, it is nevertheless enlightening to investigate the physical meaning of those $\varphi$-dependant dispersion relations. Let us start with the case of a trivial connection in the solid torus, $\varphi=0$. In that case, the determinant simply reads
\begin{align*}
	\det(Q_{\omega,k})
	= 1 + \lambda^2 + \tau^2  -&2\tau \cos\left(\f{2\pi}{N_x}k\right) - 2\lambda \cos\left(\f{2\pi}{N_t}\omega - \f{1}{N_t}\gamma_k\right) \\ + &2 \lambda \tau \cos\left(\f{2\pi}{N_x}k\right) \cos\left(\f{2\pi}{N_t}\omega - \f{1}{N_t}\gamma k\right) \; .
\end{align*}

In the spirit of viewing the zeroth of the determinant as the dispersion relation, we first look at the zeroth mode configuration $(\omega,k) = (0,0)$, returning the square mass term. Here we have
\begin{equation}
M_0^2 = 1 + \lambda^2+\tau^2 - 2\tau - 2\lambda + 2\lambda \tau = \left(1-\tau-\lambda \right)^2 \;.
\end{equation}	

Once we have the mass term, we can easily rewrite the determinant such that the mass term is explicit
\begin{align*}
\det(Q_{\omega,k}) = M_0^2 + &\tau\left(2-2\cos\left(\f{2\pi}{N_x}k\right)\right) + \lambda \left(2-2\cos\left(\f{2\pi}{N_t}\omega - \f{1}{N_t}\gamma k\right) \right)\\ - & 2\lambda\tau \left(2-2\cos\left(\f{2\pi}{N_x}k\right) \cos\left(\f{2\pi}{N_t}\omega - \f{1}{N_t}\gamma k\right)\right)
\end{align*}

We recognize here the discrete Laplacians $\Delta_k$ and $\Delta_{\omega}$ in the spatial and time direction respectively
\begin{align*}
\Delta_{k} &= 2-2\cos\left(\f{2\pi}{N_x}k\right) \; ,\\
\Delta_{\omega} &= 2-2\cos\left(\f{2\pi}{N_t}\omega - \f{1}{N_t}\gamma k\right) \; .
\end{align*}
Note that the Laplacian in the time direction is twisted with a term 
$\gamma k$. This is completely natural due to the Fourier Transform we are considering because of the periodic conditions of the lattice. Besides the mass term, we also recognize a coupling term between the Laplacian making the dispersion relation non trivial.

Let us remember that the critical coupling equation was simply $|\tau|+|\lambda|=1$. This means that the case of real positive critical couplings\footnotemark, $\lambda+\tau=1$ with $\lambda, \tau\ge 0$, corresponds to a vanishing  mass $M_{0}=0$ for the boundary modes propagating on the background defined by a bulk curvature $\varphi=0$.
\footnotetext{
	Real critical couplings naturally occur in the asymptotic limit, as $N_{x}, N_{t}$ grows to $\infty$, since the dihedral angles between the faces of the cylinder go to 0 in this continuum limit.
	% A similar consideration applies to the opposite non-geometrical regime of very high winding number $n\sim N_x$.
}
This remarkable fact extends in the complex plane to possible phases between $\tau$ and $\lambda$. Indeed, looking at the (complex) mass term for an arbitrary class angle $\varphi$, we get:
\begin{equation}
M_{\varphi}^{2}
=
\det[Q_{\omega,k}(\varphi)]\Big|_{\omega=k=0}
=
\big{(}1-\tau-\lambda e^{i\f\varphi{N_{t}}}\big{)}\big{(}1-\tau-\lambda e^{-i\f\varphi{N_{t}}}\big{)}
\,.
\end{equation}
Thus  any critical couplings with $\tau\in\R$ satisfy  $\tau+|\lambda|=1$ and thus correspond to a vanishing mass for some value of the angle $\varphi$ depending on the relative phase of $\tau$ and $\lambda$. 
This provides a neat interpretation of criticality of the boundary state as massless 0-modes on the boundary leading to a divergence of the partition function.

\bigskip

To do the actual computation, it is easier to rewrite the determinant in a semi-factorize form
\begin{equation}
\det(Q_{\omega,k}(\varphi)) = (\tau^2 + \lambda^2 - 1) +\left(-2  + 2\tau \cos\left(\f{2\pi}{N_x}k\right) \right)  \left(-1 + \lambda \cos\left(\f{\varphi}{N_t} + \f{2\pi}{N_t}\omega - \f{1}{N_t}\gamma k\right) \right) \; .
\label{chap6:eq:starting_point_det}
\end{equation}
In the previous expression, we have factorized as much as possible the contribution of the spatial direction, coming with the parameters $\tau$ and of the temporal direction, coming with the parameters $\lambda$. 

From this expression we see that a really peculiar condition arises. Consider the sub-space of the parameters $\lambda$ and $\tau$ defined by the condition
\begin{equation*}
\tau^2 + \lambda^2 = 1 \;.
\end{equation*}
The question of how, and if, this condition is related to the mass term remains a mystery. It is clear that they coincide in the limit where one of the parameters is $1$ whereas the other one is $0$, but in general this is not the case. The point is that this condition completely decouples the temporal and spatial direction on the torus, in the sense that they factorized. Indeed, under this condition, the determinant clearly becomes
\begin{equation}
\det(Q_{\omega,k}(\varphi)) = \left(-2  + \tau \cos\left(\f{2\pi}{N_x}k\right) \right)  \left(-1 + \lambda \cos\left(\f{\varphi}{N_t} + \f{2\pi}{N_t}\omega - \f{1}{N_t}\gamma k\right) \right) \; .
\end{equation}
The condition $\tau^{2}+\lambda^{2} = 1$ effectively decouples the spatial modes from the (twisted) temporal modes. Although this happens in momentum space and not directly in coordinate space on the original lattice, we can consider it as similar to the zero-spin recoupling channel for boundary spin network states at fixed spins considered in \cite{Dittrich:2017rvb}. Indeed it plays the same role of decoupling the spatial and temporal structures of the boundary state, trivializing in some sense the role of the number of time slices $N_{t}$ and thereby allowing allowing to focus solely on the effect of the twisted periodic conditions and how the Ponzano-Regge amplitude depends on the twisted angle $\gamma$. 

It is crucial to stress that this decoupling condition is not related to the criticality condition $|\tau|+|\lambda|=1$. In fact, criticality corresponds to extremal couplings satisfying $\tau^2 + \lambda^2 = 1$. For instance, considering real couplings $\tau,\lambda\in [0,1]$ satisfying $\tau^2=1- \lambda^2 $, critical couplings are $\lambda=0$ and $\lambda=1$. Thus, assuming $\tau^2 + \lambda^2 = 1$ for complex couplings does not restrict to specific behaviour of the partition functions and allow to study the transition from a convergent boundary state to the critical regime. We will actually show that we recover in that critical regime the asymptotic BMS character formula for the 3D quantum gravity path  integral.

We also discuss in the following the general calculation of the Ponzano-Regge amplitude when $\tau^2 + \lambda^2 \ne 1$. Evaluating the partition function as a contour integral over the class angle $\varphi$ in the complex plane, the poles and residues are not regularly spaced as in the decoupled case, but we can still write it as a finite sum over poles, which can be interpreted as a regularized and deformed BMS character.

\medskip

In order to compute the Ponzano-Regge amplitude, we will use the same lemma \ref{prop:cosine_product} as in the previous section. We recall it
\begin{lemma}
	Consider two integers $N$ and $M$ and denote the greatest common divisor (GCD) between $N$ and $M$ by $K$. Define then two integers $n$ and $m$ such that
	\begin{equation*}
	N = K \, n \qquad \text{and} \qquad M = K \, m \; .
	\end{equation*}
	The following relation holds for all $x$ and $a$ complex numbers
	\begin{equation*}
	\prod_{k=0}^{N-1} \left(2 a + 2 \cos\left( \f{2 \pi M}{N}k + x\right) \right)= \left(2 \left( T_{n}(a) - (-1)^{n} \cos(n x) \right)\right)^K
	\end{equation*}
	where $T_{n}$ is the n-Chebyshev polynomial of the first kind.
	\label{prop:cosine_product}
\end{lemma}

Using the explicit expression of the Chebyshev polynomial in terms of the analytic continuation of hyperbolic functions makes the notations simpler. That is, we denote 
\begin{equation*}
T_{n}(x) = \cosh(n X_x) \qquad \text{with} \quad X_{x} = \text{arccosh}(x) \;.
\end{equation*}
We recall that under a parity transformation, the Chebyshev polynomial transforms following
\begin{equation*}
T_{n}(-x) = (-1)^{n} T_{n}(x) \; . 
\end{equation*}
\medskip

As usual we denote by $K$ the GCD between $N_x$ and $N_\gamma$. We also introduce $n_x$ and $n_\gamma$ such that
\begin{equation*}
N_x = K n_x \qquad \text{and} \qquad N_\gamma = K n_\gamma \; .
\end{equation*}

\subsection{Ponzano-Regge determinant and multi-variate Chebytchev polynomial}

Before going deeper into the interpretation of the determinant formula and the computation, we are making a brief remark. The formula of the determinant has a nice expression in terms of trace of $SU(2)$ elements, allowing a possible link with multi-variate Chebytchev polynomials. We consider pairs of $SU(2)$ element $h^\tau_{k}$ and $h_{\omega,k}^\lambda$ defined by
\begin{align*}
	h_{k}^\tau &= \cos\left(\f{2\pi}{N_x}k\right) \id + i \sigma_{z} \sin\left(\f{2\pi}{N_x}k\right) = e^{i \f{2\pi}{N_x}k \sigma_z} \\
	h_{\omega,k}^\lambda &= \cos\left(\f{\varphi}{N_t} + \f{2\pi}{N_t}\omega - \f{1}{N_t}\gamma k\right) \id + i \sigma_{x} \sin\left(\f{\varphi}{N_t} + \f{2\pi}{N_t}\omega - \f{1}{N_t}\gamma k\right) = e^{i \f{\varphi}{N_t} + \f{2\pi}{N_t}\omega - \f{1}{N_t}\gamma_k \sigma_x} \; .
\end{align*}
The trace of $h_{k}^\tau$, $h_{\omega,k}^\lambda$ and $h_{k}^\tau h_{\omega,k}^\lambda$ is
\begin{align*}
	\tr(h_{k}^\tau) &= 2\cos\left(\f{2\pi}{N_x}k\right) \; , 
	\\ 
	\tr(h_{\omega,k}^\lambda) &= 2 \cos\left(\f{\varphi}{N_t} + \f{2\pi}{N_t}\omega - \f{1}{N_t}\gamma k\right) \; ,
	\\
	\tr(h_{k}^\tau h_{\omega,k}^\lambda) &= 2 \cos\left(\f{2\pi}{N_x}k\right) \cos\left(\f{\varphi}{N_t} + \f{2\pi}{N_t}\omega - \f{1}{N_t}\gamma k\right) \; .
\end{align*}
In terms of the traces, the determinant is
\begin{equation}
	\det(Q_{\omega,k}(\varphi)) =  1 + \lambda^2 + \tau^2 - 2\tau \tr(h_{k}^\tau) - 2\lambda \tr(h_{\omega,k}^\lambda) +2\lambda\tau \tr(h_k^\tau h_{\omega,k}^\lambda) \; .
	\label{chap6:eq:det_SU(2)_element_expression}
\end{equation}

Using this expression of the determinant, we can easily look for a generalization of our boundary state. Instead of considering a rectangular lattice, we can consider a tilted lattice. Let's introduce the spinor $|\theta \ra$ associated to the vector $\vec{u}_{\theta} = (\cos(\theta),0,\sin(\theta))$. The case $\theta = 0$ corresponds to the rectangular case we have been working on in this thesis. Replacing $|+\ra$ by $|\theta\ra$ corresponds to replacing $h^{\lambda}_{\omega,k}$ by
\begin{equation*}
h^{\lambda}_{\omega,k}(\theta) = \cos\left(\f{\varphi}{N_t} + \f{2\pi}{N_t}\omega - \f{1}{N_t}\gamma k\right) \id + i \sigma_{\theta} \sin\left(\f{\varphi}{N_t} + \f{2\pi}{N_t}\omega - \f{1}{N_t}\gamma k\right) \; .
\end{equation*}
where $\sigma_{\theta} = \vec{\sigma}.\vec{u}_{\theta}$.

\bigskip

With these notations, it is possible to relate our determinant to a generalization of the Chebyshev polynomial to multi-variate polynomial.

Indeed, the generating function of the Chebyshev polynomial of the second kind $U_{n}$ is
\begin{equation}
\cF(u,v) = 	\sum_{n=0}^{\infty} U_n(u) v^{n} = \f{1}{1-2vu+v^2} \; .
\end{equation}

Consider now the case where either $\lambda$ or $\tau$ is equal to $0$. Say $\lambda = 0$. The determinant becomes
\begin{equation}
\det(Q_{\omega,k}(\varphi))|_{\lambda = 0} = 1 - 2\tau \tr(h_{k}^\tau) + \tau^2 \; ,
\end{equation}
which is clearly related to the generating function $\cF$ at $u = \tr(h_{k}^\tau) $ and $v = \tau$
\begin{equation}
\cF(\tr(h_{k}^\tau),\tau) = \f{1}{\det(Q_{\omega,k}(\varphi))|_{\lambda = 0}}
\end{equation}
and the same holds for $\tau = 0$ and any $\lambda$.

It seems therefore appealing to define 2-variate Chebytchev polynomials by defining their generating function to be equal to the inverse of formula \eqref{chap6:eq:det_SU(2)_element_expression}. In the same way, n-variate Chebytchev polynomial can be defined. We leave the study of such polynomial to future work.

\subsection{Free Ponzano-Regge amplitude: pole and exact formula}

\begin{prop}
	The sub-space span by the constraint
	\begin{equation}
	\lambda^2 + \tau^2 = 1
	\label{chap6:eq:constraint_trivial_disp_rel}
	\end{equation}
	is the sub-space where the spatial and temporal directions are decoupled and the dispersion relation, given by
	\begin{equation}
	\varphi_{k}^{\pm} = \gamma k \pm i N_t X_{\lambda^{-1}} \; , \quad \gamma = \f{2 \pi N_\gamma}{ N_x} 
	\end{equation}
	is linear in $k$.
	
	The mode expansion of the Ponzano-Regge model reads
	\begin{align*}
	Z_{\lambda^2+\tau^2=1} = \la Z_{PR}^{\cK}| \Psi  \ra_{\lambda^2+\tau^2=1} = &\f{2^{N_t(N_x-2)-K}}{(\lambda \tau)^{N_t N_x}} \f{1}{\sh^{2N_t}\left(\f{1}{2}N_x X_{\f{1}{\tau}}\right)} \\ & \left(\f{n_x}{\sh(N_t X_{\f{1}{\lambda}})} \prod_{k=1}^{n_x-1} \f{1}{2(\ch(N_t X_{\f{1}{\lambda}}) - \cos(\gamma k+i N_t X_{\f{1}{\lambda}}))}\right)^K \\ & L_{K-1}\left( \f{\ch^K(n_x N_t X_{\f{1}{\lambda}})}{\sh^K(n_x N_t X_{\f{1}{\lambda}})} \right) \; ,
	\end{align*}
	where $K = GCD(N_x,N_\gamma)$ and $n_x = \f{N_x}{K}$. $L_n$ is the Legendre polynomial of degree $n$.
	
\end{prop}
Note that the critical condition $\tau = 0$, $\lambda = 1$ is pathologic. On top of the divergence of the amplitude, the argument of the Legendre polynomial becomes infinite and as such the previous formula does not capture all the poles of the amplitude. This case will be studied in the next section and related to the previous results and the BMS character.

In spirit, this subspace of trivial dispersion relation should correspond to the $J=0$-intertwiner in the s-channel introduced in chapter \ref{chap4}. Recall that with such a choice, the intertwiner completely decoupled the links in the vertical and in the horizontal direction, hence corresponding to this particular sub-space.
\medskip

Under the constraint \eqref{chap6:eq:constraint_trivial_disp_rel}, the Ponzano-Regge amplitude becomes 
\begin{align*}
	Z_{\lambda^2+\tau^2=1} = \f{1}{\pi} \int_{0}^{2\pi} &\text{d}\varphi \sin^{2}(\varphi) \\ & \prod_{\omega,k=0}^{N_t-1,N_x-1} \f{1}{\Big(-2 + 2 \tau \cos\left(\f{2\pi}{N_x}k\right) \Big)  \Big(-1 + \lambda \cos\left(\f{\varphi}{N_t} + \f{2\pi}{N_t}\omega - \f{1}{N_t}\gamma_k\right) \Big)}  \; .
\end{align*}
The product over the Fourier modes can easily be computed using the proposition \ref{prop:cosine_product}. The product over $-2 + 2 \tau \cos\left(\f{2\pi}{N_x}k\right)$ gives
\begin{equation*}
	\prod_{\omega,k=0}^{N_t-1,N_x-1} \Big(-2 + 2 \tau \cos\left(\f{2\pi}{N_x}k\right) \Big) = 2^{N_t}\tau^{N_x N_t} \left( \ch(N_x X_\f{1}{\tau}) - 1 \right)^{N_t} = 2^{2N_t} \tau^{N_x N_t} \sh^{2N_t}\left( \f{1}{2}N_x X_\f{1}{\tau}\right) \; .
\end{equation*}
For the product over $-1 + \lambda \cos\left(\f{\varphi}{N_t} + \f{2\pi}{N_t}\omega - \f{1}{N_t}\gamma k\right)$, the product over the temporal modes returns
\begin{equation*}
	\prod_{\omega=0}^{N_t-1} \Big(-1 + \lambda \cos\left(\f{\varphi}{N_t} + \f{2\pi}{N_t}\omega - \f{1}{N_t}\gamma k\right)\Big) = \f{\lambda^{N_t}}{2^{N_t}} \left(2\left(\ch(N_t X_\f{1}{\lambda}) -\cos\left( \varphi - \gamma k \right) \right) \right) \; .
\end{equation*}

Combining these two results, the Ponzano-Regge amplitude is then
\begin{equation}
	Z_{\lambda^2+\tau^2=1} = \f{2^{N_t(N_x-2)}}{(\lambda \tau)^{N_t N_x}} \f{1}{\sh^{2N_t}\left(\f{1}{2}N_x X_\f{1}{\tau}\right)} \f{1}{\pi} \int_{0}^{2\pi} \text{d} \varphi \sin^{2}(\varphi) \prod_{k=0}^{N_x-1} \f{1}{2(\ch(N_t X_\f{1}{\lambda})-\cos(\varphi-\gamma k))} \;.
\label{chap6:eq:starting_point_residue_case_trivial_disp_rel}
\end{equation}
As expected, we end up with an expression that does not explicitly depend on the temporal modes. The remaining product can also be dealt with in the same way using proposition \ref{prop:cosine_product}. Recall however that we are interested in getting a mode expansion of the amplitude relating the spatial model and the continuous parameter $\varphi$. Note that the same argument is not true for the temporal mode. Indeed, it is immediate to see from the initial amplitude that the temporal modes $\omega$ only contribute with a translation of $2\pi$ to the poles. Hence, they do not matter as previously claimed.

Instead, we will compute the integral over $\varphi$ using the residue theorem. For the residue theorem to be easily applied, we will restrict the computation to the case where $K=1$, since all the poles are then simple. For general $K$, the poles will be of order $K$, and finding a general expression with the residue theorem, if possible, is much harder. However, once the computation for $K=1$ is found, we will compare it with the computation of the amplitude using the usual proposition \ref{prop:cosine_product}. This will allow us to recover the exact pole structure of the any $K$ case up to a particular limit for the coupling parameters. This case will be discussed in more details in the next section.
\\

We start with the case $K=1$ and the residue computation. Note that the condition $K=1$ is, as always, the one corresponding to the irrational angles in the continuum, like in chapter \ref{chap2} under the double scaling limit
\begin{equation}
	N_x \rightarrow \infty \; , \; \; \quad N_\gamma \rightarrow \infty \qquad \text{with} \quad \f{2\pi N_\gamma}{N_x} \rightarrow \gamma \in [0,2\pi] \; .
\end{equation}
It is the exact same condition for $K$ that already appear in the computation of the previous chapter, and for the quantum Regge calculus approach. Here, we will see that we can go beyond this case without loosing the convergence of the amplitude.

The residue theorem is easily applied since the poles are apparent from equation \eqref{chap6:eq:starting_point_residue_case_trivial_disp_rel}. They number to $2 N_x$, at the positions
\begin{equation}
\varphi_{k}^{\pm} = \gamma k \pm i N_t X_\f{1}{\lambda} \;.
\end{equation}
Assuming $K>1$, it is immediate to see that
\begin{equation}
\varphi_{k}^{\pm} = \varphi_{k+n_x}^{\pm} \quad \forall \; k \; ,
\end{equation}
and the poles will not be simple anymore, but of order $K$. It is interesting to note that this feature is due to the sub-space of the coupling constant we are considering and hence it will not be necessary to restrict the residue theorem to the sole case $K=1$ in the general case. Indeed, the non triviality of the dispersion relation will imply that the poles are not linear in $k$ anymore, and therefore that they are always distinct. Note that the case $\lambda = 1$ is again peculiar. Indeed, in that case, we have $X_{1} = 0$ and the poles are thus real and are not simple anymore but of order two. Moreover, in that case, the spatial mode $k=0$ is also not present due to the sine square in the integral that will cancel it. It corresponds to the vanishing winding number of the previous chapter.

To compute the integral by residue, we introduce the complex variable $z$ defined by
\begin{equation*}
	z = e^{i \varphi} \; , \quad \text{d}z = i z \text{d}\varphi \;.
\end{equation*}
Under this change of variable, the integration contour becomes the whole $U(1)$. In the complex variable the poles are\footnote{Sadly, with these notations, $z_{k}^{+}$ (resp. $z_k^-$) corresponds to $\varphi_k^-$ (resp. $\varphi_k^+$).}
\begin{equation*}
	z_{k}^{\pm} = e^{i \gamma_k} e^{\pm N_t X_\f{1}{\lambda}} \; ,
\end{equation*}
and we are only interested in the poles whose norms are smaller than $1$ (since the integration contour is $U(1)$).

Under the condition that we consider the principal value for the inverse cosine hyperbolic function, the following inequalities hold for all $k$
\begin{equation}
	|z_k^{+}| \geq 1 \quad \text{and} \quad |z_k^{-}| \leq 1 \; ,
\end{equation}
and the inequalities can only saturate if and only if $|\lambda| \leq 1$. The values for $\lambda$ that saturate the inequalities are particular points where the analytic continuation of the generating function we consider is not defined anymore. Thus, we will exclude all those configuration in the rest of this chapter. 

To compute the integral by residue, we rewrite the amplitude in terms of the complex variables. It takes the factorized form
\begin{equation}
	Z_{\lambda^2+\tau^2=1}^{K=1} = \f{2^{N_t(N_x-2)}}{(\lambda \tau)^{N_t N_x}} \f{1}{\sh^{2N_t}\left(\f{1}{2}N_x X_{\tau}\right)}\f{1}{\pi} \int_{U(1)} \f{\text{d}z}{i} \f{-1}{4}\f{(z^2-1)^2}{z^2} z^{N_x-1} \prod_{k=0}^{N_x-1} \f{-e^{i \gamma k}}{(z-z_k^+)(z-z_k^-)} \; .
	\label{chap6:eq:case_trivial_disp_rel_amp_complex}
\end{equation}
Due to the integration contour, only the poles $z_k^{-}$ matter for the integration since the norms of the $z_{l}^+$ are all bigger than one. If $N_x > 2$, $z = 0$ is not a pole. We will assume $N_x > 2$ in the following. Applying the residue theorem to the integral leads to
\begin{align*}
	Z_{\lambda^2+\tau^2=1}^{K=1} =  &\f{2^{N_t(N_x-2)}}{(\lambda \tau)^{N_t N_x}} \f{1}{\sh^{2N_t}\left(\f{1}{2}N_x X_{\tau}\right)} 2 \\ &\sum_{n=0}^{N_x-1} \f{-1}{4}\f{((z_n^-)^2 -1)^2}{(z_n^-)^2} (z_n^-)^{N_x-1} \f{-e^{i \gamma_n}}{z_n^- -z_n^+} \prod_{\substack{k=0 \\ k \neq n}}^{N_x-1} \f{-e^{i \gamma_k}}{(z_n^- -z_k^+)(z_n^- -z_k^-)} \; ,
\end{align*}
where we used the well-known expression of the residue for simple poles $\text{Res}(f,c) = \underset{z \rightarrow c}{\text{limit}} (z-c)f(z)$ with $c$ a simple pole of the function $f$. 

We are now ready to go back to our geometrical variable $\varphi$. The amplitude reads in the geometrical variable
\begin{equation}
	\begin{split}
	Z_{\lambda^2+\tau^2=1}^{K=1} = &\f{2^{N_t(N_x-2)}}{(\lambda \tau)^{N_t N_x}} \f{1}{\sh^{2N_t}\left(\f{1}{2}N_x X_{\f{1}{\tau}}\right)} \f{1}{\sh(N_t X_\f{1}{\lambda})} \\ &\sum_{n=0}^{N_x-1} \sin^2\left(\gamma n + i N_t X_\lambda \right)  \prod_{k=1}^{N_x-1} \f{1}{2(\ch(N_t X_\f{1}{\lambda}) - \cos(\gamma k+i N_t X_\f{1}{\lambda}))} \; .
	\end{split}
	\label{chap6:eq:case_trivial_disp_rel_amp_final_sum}
\end{equation}
Note that we used the periodicity condition of the cosine function to absorb the dependency of $n$ from the argument of the product. Again, this is only possible because the dispersion relation is trivial, i.e. linear. Hence, this will not apply to the generalized case. This formula should be familiar. It is indeed really similar to the one obtained in the previous chapter \eqref{chap5:eq:amp_final}. However, we see that there is a natural shift in the complex plane in the cosine. Hence, if $K$ is not one, the previous formula does not give rise to a divergence. In that sense, it regularized the equation obtain \eqref{chap5:eq:amp_final}.

We can go one step further and explicitly compute the sum over $n$ which return a factor $N_x/2$.\footnote{This result holds only if $N_\gamma \neq N_x/2$ In that case, the sine does not depend on $n$ explicitly, and the sum returns $-N_x \sh^2(N_t X_\lambda)$.} After doing this sum, the amplitude is simply
\begin{align}
	Z_{\lambda^2+\tau^2=1}^{K=1} =\f{2^{N_t(N_x-2)-1}}{(\lambda \tau)^{N_t N_x}} \f{1}{\sh^{2N_t}\left(\f{1}{2}N_x X_\f{1}{\tau}\right)} \f{N_x}{\sh(N_t X_\f{1}{\lambda})} \prod_{k=1}^{N_x-1} \f{1}{2(\ch(N_t X_\f{1}{\lambda}) - \cos(\gamma k+i N_t X_\f{1}{\lambda}))} \; . 
	\label{chap6:eq:case_trivial_disp_rel_amp_final}
\end{align}

Next section is dedicated to the study of \eqref{chap6:eq:case_trivial_disp_rel_amp_final_sum} and \eqref{chap6:eq:case_trivial_disp_rel_amp_final} in the spirit of its comparison with the computation done in the previous chapter. 

The expression of the amplitude is in accordance with the geometrical interpretation of having decoupled spatial and temporal direction. Indeed, the amplitude can be seen coming from $N_t$ circle carrying the coupling $\tau$ and one closed loop carrying the coupling $\lambda$. Note that even though the directions are decoupled, each part of the amplitude still knows about the other direction of discretization. Indeed, looking at the amplitude, both couplings always comes with the weight $N_t$ or $N_x$. As expected, $\tau$ comes with $N_x$ since each circle is in fact constitute of $N_x$ links, each one carrying the coupling $\tau$. The same happens for $\lambda$. If we define the amplitude of a circle construct from $N_x$ links each one carrying $\tau$ by
\begin{equation}
	Z_{circle}^{\tau,N_x} = \f{2^{N_x-2}}{\tau^{N_x}} \f{1}{\sh^{2}\left(\f{1}{2}N_x X_\f{1}{\tau}\right)}
\end{equation}
and the amplitude coming from one closed loop construct from $N_x$ times $N_t$ links carrying $\lambda$ by
\begin{equation}
	Z_{closed-loop}^{\lambda,N_t,N_x} = \f{N_x}{2 \lambda^{N_t N_x} \sh(N_t X_\f{1}{\lambda})} \prod_{k=1}^{N_x-1} \f{1}{2(\ch(N_t X_\f{1}{\lambda}) - \cos(\gamma_{k}+i N_t X_\f{1}{\lambda}))} \;
\end{equation}

With these notations, $Z_{\lambda^2+\tau^2=1}^{K=1}$ becomes simply
\begin{equation}
	Z_{\lambda^2+\tau^2=1}^{K=1} = \left(Z_{circle}^{\tau,N_x}\right)^{N_t} \times Z_{closed-loop}^{\lambda,N_t,N_x} \;.
\end{equation} 
These notations will come in handy to understand the geometrical picture of the any $K$ case. This viewpoint on the amplitude is really similar to the graphical decomposition of the amplitude provided in figure \ref{chap4:fig:graph_integration} for the $J=0$ s-channel intertwiner in the case $K=1$. We will see next that this viewpoint can even be extended to the any $K$ case.
\\

To obtain the mode decomposition for the any $K$ case, it is needed to compute the amplitude using proposition \ref{prop:cosine_product} until the end. The product over the spatial and temporal modes previously computed return, making the dependency on $K$ explicit
\begin{equation*}
	\prod_{\omega,k} \Big(-2 + 2 \tau \cos\left(\f{2\pi}{N_x}k\right) \Big) = 2^{N_t}\tau^{N_x N_t} \left( \ch(N_x X_\f{1}{\tau}) - 1 \right)^{N_t} = 2^{2N_t} \tau^{N_x N_t} \sh^{2N_t}\left( \f{1}{2}N_x X_\f{1}{\tau}\right) \; .
\end{equation*}
and
\begin{equation*}
	\prod_{\omega,k} \Big(-1 + \lambda \cos\left(\f{\varphi}{N_t} + \f{2\pi}{N_t}\omega - \f{1}{N_t}\gamma k\right)\Big) = \f{\lambda^{N_t N_x}}{2^{N_t N_x}} \left(2\left(\ch(N_t n_x X_\f{1}{\lambda}) -\cos\left( n_x \varphi \right) \right) \right)^K \; .
\end{equation*}

The amplitude is then
\begin{equation*}
	Z_{\lambda^2+\tau^2=1} = \f{2^{N_t(N_x-2)-K}}{(\lambda \tau)^{N_t N_x}} \f{1}{\sh^{2N_t}\left(\f{1}{2}N_x X_\f{1}{\tau}\right)} \f{1}{\pi} \int_{0}^{2\pi} \text{d} \varphi \sin^{2}(\varphi) \left(\f{1}{\ch(N_t X_\f{1}{\lambda})-\cos(\varphi-\gamma k)}\right)^K \;.
\end{equation*}

The integral can then be exactly computed expanding the fraction into an infinite sum. We get
\begin{align*}
	\int_{0}^{2 \pi} \text{d}\varphi \sin^{2}(\varphi) &\f{1}{(\ch(n_x N_t X_{\lambda})-\cos(n_x \varphi))^{K}} 
	= \\
	&\sum_{n=0}^{\infty} (-1)^n \binom{K+n-1}{n} \f{1}{\ch^{K+n}(n_x N_t X_{\lambda})} \int_{0}^{2 \pi} \text{d} \varphi \sin^{2}(\varphi) \cos^{n}(n_x \varphi) \; .
	\label{chap6:eq:computation_amp_case_trivial_disp_rel_1}
\end{align*}
This is a well-defined operation for $n_x$ and $N_t$ big enough for the hyperbolic cosine to be bigger than one if $\lambda \neq 1$. The latest being a case that we ignore for now. 

The integral over the cosine and sine can then easily be done. To do so, we write the trigonometric functions in their exponential counterparts. We have
\begin{equation*}
	\int_{0}^{2 \pi} \text{d} \varphi \sin^{2}(\varphi) \cos^{n}(n_x \varphi)= \int_{0}^{2 \pi} \text{d} \varphi \f{-1}{4}(e^{2i\varphi}+e^{-2i\varphi}-2) \sum_{k=0}^{n} \binom{n}{k} e^{i(n-2k)n_x \varphi}
\end{equation*}
Assuming $n_x > 2$, the computation of the integral is immediate and only the terms $2k = n$ contribute. We find, substituting $n$ by $2n$
\begin{align*}
	\int_{0}^{2\pi}\text{d} \varphi \sin^{2}(\varphi) \cos^{2n}(n_x \varphi) &= \f{\pi}{2^{2n}} \binom{2n}{n} 
	\\
	\int_{0}^{2\pi}\text{d} \varphi \sin^{2}(\varphi) \cos^{2n+1}(n_x \varphi) &= 0
\end{align*}

Putting everything together, the Ponzano-Regge amplitude reads
\begin{equation*}
	Z_{\lambda^2+\tau^2=1} = \f{2^{N_t(N_x-2)-K}}{(\lambda \tau)^{N_t N_x}} \f{1}{\sh^{2N_t}\left(\f{1}{2}N_x X_{\tau}\right) \ch^K(n_x N_t X_{\lambda})} \sum_{n=0}^{\infty} \frac{(K+2n-1)!}{(n!)^{2}(K-1)!} \f{1}{2^{2n} T_{N_t n_x}(\lambda^{-1})^{2n}}
\end{equation*}

The sum happens to be exactly computable and returns the hypergeometric function ${}_{2}F_1\left(\frac{K}{2},\frac{K+1}{2};1;\frac{1}{T_{N_t n_x}(\lambda^{-1})^{2}} \right)$
\begin{equation*}
	Z_{\lambda^2+\tau^2=1} = \f{2^{N_t(N_x-2)-K}}{(\lambda \tau)^{N_t N_x}} \f{1}{\sh^{2N_t}\left(\f{1}{2}N_x X_{\tau}\right) \ch^K(n_x N_t X_{\lambda})} ~_2F_1(\f{K}{2},\f{K-1}{2};1;\f{1}{\ch^2(n_x N_t X_{\lambda})})
\end{equation*}
which in turn can be expressed with the help of the Legendre polynomials
\begin{equation}
	Z_{\lambda^2+\tau^2=1} = \f{2^{N_t(N_x-2)-K}}{(\lambda \tau)^{N_t N_x}} \f{1}{\sh^{2N_t}\left(\f{1}{2}N_x X_\f{1}{\tau}\right)} \f{1}{\sh^K(n_x N_t X_{\f{1}{\lambda}})} L_{K-1}\left( \f{\ch^K(n_x N_t X_\f{1}{\lambda})}{\sh^K(n_x N_t X_\f{1}{\lambda})} \right) \;
	\label{chap6:eq:case_decoupled_full_simplified}
\end{equation}
where $L_{K-1}$ is the Legendre polynomial of order $K-1$. We can now compare this expression for $K=1$ with \eqref{chap6:eq:case_trivial_disp_rel_amp_final} to get the relation
\begin{equation}
	\f{\sh(N_t X_\f{1}{\lambda})}{\sh(N_x N_t X_\f{1}{\lambda})} = N_x\prod_{k=1}^{N_x-1} \f{1}{2(\ch(N_t X_\f{1}{\lambda}) - \cos(\gamma k+i N_t X_\f{1}{\lambda}))} \; .
\end{equation}
This formula can be directly proven using the relation between the Chebytchev polynomial of the first kind and the analytic continuation of the $SU(2)$ character function for complex class angle.

The key point is that this relation is true for any $N_x \in \N$. In particular, it is then true for $n_x = \f{N_x}{K}$. We can therefore express the factor $\f{1}{\sh^K(n_x N_t X_{\lambda})}$ from formula \eqref{chap6:eq:case_decoupled_full_simplified} in terms of a mode decomposition. This implies that the amplitude in the any $K$ case can be written in the form
\begin{align*}
	Z_{\lambda^2+\tau^2=1}& = \f{2^{N_t(N_x-2)-K}}{(\lambda \tau)^{N_t N_x}} \f{1}{\sh^{2N_t}\left(\f{1}{2}N_x X_{\f{1}{\tau}}\right)} \\ & \left(\f{n_x}{\sh(N_t X_{\f{1}{\lambda}})} \prod_{k=1}^{n_x-1} \f{1}{2(\ch(N_t X_{\f{1}{\lambda}}) - \cos(\gamma k+i N_t X_{\f{1}{\lambda}}))}\right)^K  L_{K-1}\left( \f{\ch^K(n_x N_t X_{\f{1}{\lambda}})}{\sh^K(n_x N_t X_{\f{1}{\lambda}})} \right) \; .
\end{align*}
This formula is a straightforward, simplistic even, generalization of the $K=1$ case, up to the Legendre polynomial renormalization. The key point here is that, in general, the Legendre polynomial is just a non-vanishing complex number. It does not cause any divergence to the amplitude, and thus does not provide us with any poles. Excluding the particular case $\lambda = 1$, where the argument of the Legendre polynomial diverges, the previous formula {\it is} the mode decomposition for the amplitude for the any $K$ case. When $\lambda = 1$, the mode decomposition is not complete and we will discuss this case in more details later on this chapter.
\medskip

The role of $K=\mathrm{GCD}(N_{\gamma},N_{x})$ is crucial in considering the flow under refinement of boundary lattice and the continuum limit. As underlined previously, the case $K=1$ corresponds to an irrational twist angle while the case $K\rightarrow\infty$ corresponds to a rational twist angle. More precisely, the twist angle is defined in the discrete setting from the shift as $\gamma=2\pi N_{\gamma}/N_{x}$. Then we consider the two cases, at fixed number of time slices $N_{t}$:
\begin{itemize}
	\item the limit towards $\gamma\notin 2\pi\Q$: \\
	\noindent
	We construct a sequence of pairs of coprime integers $(N_{\gamma}^{(p)},N_{x}^{(p)})_{p\in\N}$ such that the limit of the ratios $2\pi N_{\gamma}^{(p)}/N_{x}^{(p)}$ converge towards $\gamma$ when $p$ goes to infinity. Then we define the Ponzano-Regge amplitude for a twisted torus with irrational twist as the limit of the partition function $Z_{\lambda^2+\tau^2=1}$ for the square lattice with $N_{x}^{(p)}$ spatial nodes and a gluing shift $N_{\gamma}^{(p)}$. Since their two numbers are coprime by definition, their GCD is always 1, so the partition function is always given by $Z^{K=1}_{\lambda^2+\tau^2=1}$ and does not depend on the twist angle $\gamma$ at the end of the day:
	\begin{equation}
	Z^{K=1}_{\lambda^2+\tau^2=1}
	%[\gamma]
	=
	\f{2^{N_t(N_x-2)-1}}{(\lambda \tau)^{N_t N_x}}
	\f{1}{\sh^{2N_t}\left(\f{1}{2}N_x X_{\tau^{-1}}\right)}
	\f{1}{\sh(N_{x}N_t X_{\lambda^{-1}})}
	=
	\cA\,{ z}^{PR}
	%\cA\,\f{1}{2\sh(N_{x}N_t X_{\lambda^{-1}})}
	\,,
	\nn
	\end{equation}
	where we have distinguished the pre-factor $\cA$ from a reduced Ponzano-Regge amplitude $z^{PR}$, with the following asymptotics as $N_{x}$ is sent to infinity:
	\begin{equation}
	\cA=
	\f{2^{N_t(N_x-2)}}{(\lambda \tau)^{N_t N_x}}
	\f{1}{\sh^{2N_t}\left(\f{1}{2}N_x X_{\tau^{-1}}\right)}
	\underset{N_{x}\rightarrow\infty}{\sim}
	\f1{2^{2N_{t}}}\,\left[\f{2}{(1+\lambda)\lambda} \right]^{N_{x}N_{t}}
	\,,
	\end{equation}
	\begin{equation}
	z^{PR}
	%[\gamma]
	=
	\f{1}{2\,\sh(N_{x}N_t X_{\lambda^{-1}})}
	\underset{N_{x}\rightarrow\infty}{\sim}
	\f1{2}\,\left[\f{1+\tau}{\lambda} \right]^{N_{x}N_{t}}
	\,.
	%\nn
	\end{equation}
	
	\item the limit towards $\gamma\in 2\pi\Q$: \\
	\noindent
	We choose the fundamental representation of the twist angle as a fraction $\gamma=2\pi n_{\gamma}/n_{x}$ with $n_{\gamma}$ and $n_{x}$ coprime. Then the infinite refinement limit is taken by considering the sequence of lattices with $(N_{\gamma},N_{x})=(Kn_{\gamma},Kn_{x})$ as the integer $K$ grows to infinity. The partition function is given by $Z^{K}_{\lambda^2+\tau^2=1}$ where we factorize out the same pre-factor $\cA$ as in the previous case:
	\begin{align}
	Z^{K}_{\lambda^2+\tau^2=1}
	%[\gamma]
	&=
	\f{2^{N_t(N_x-2)-K}}{(\lambda \tau)^{N_t N_x}}
	\f{1}{\sh^{2N_t}\left(\f{1}{2}N_x X_{\tau^{-1}}\right) \sh^K(n_x N_t X_{\lambda^{-1}})}
	\,
	L_{K-1}\left( \f{\ch(n_x N_t X_{\lambda^{-1}})}{\sh(n_x N_t X_{\lambda^{-1}})} \right) \\
	&=
	\cA z_{K}^{PR}
	\,.
	\nn
	\end{align}
	The reduced Ponzano-Regge amplitude now has a different behaviour as $K$ is sent to infinity while keeping $n_{x}$ fixed and finite:
	\begin{align}
	z^{PR}_{K}
	&=
	\left[\f{1}{2\,\sh(n_x N_t X_{\lambda^{-1}})}\right]^{K}
	\,
	L_{K-1}\left( \f{\ch(n_x N_t X_{\lambda^{-1}})}{\sh(n_x N_t X_{\lambda^{-1}})} \right)
	\\
	&\underset{K\rightarrow\infty}{\sim}
	\f1{\sqrt{\pi K}}\,\ch\left(\f{n_x N_t X_{\lambda^{-1}}}2\right)\,
	\left[\f{1}{2\,\sh\f{n_x N_t X_{\lambda^{-1}}}2}\right]^{2K}
	\,.
	\nn
	\end{align}
	We see that this reduced Ponzano-Regge amplitude has a different scaling than the irrational case with $K=1$. The power $K$ is natural from a geometric perspective: the vertical lines on the twisted torus form $K$ large loops (of length $n_{x}N_{t}$) as illustrated in figure \ref{chap6:fig:drawingK}.
	
\end{itemize}
\begin{figure}[h!]
	\begin{center}
		\begin{tikzpicture}[scale=.8]
		%%%% lattice %%%%
		\foreach \i in {0,2,4}{
			\foreach \j in {0,...,3}{
				\draw (\i,\j) node[color=red] {$\bullet$};
				\draw[<-,color=red] (\i,\j-.5) --(\i,\j+.5);
			}
		}
		
		\foreach \i in {1,3,5}{
			\foreach \j in {0,...,3}{
				\draw (\i,\j) node[color=blue] {$\bullet$};
				\draw[<-,color=blue] (\i,\j-.5) --(\i,\j+.5);
			}
		}
		
		\foreach \i in {-2,0,2}{
			\draw[rounded corners=3 pt,->,color=red] (\i,4.5+0.3) --(\i,4+0.1+0.3)-- (\i+2,4-0.1+0.3)--(\i+2	,3.5+0.3)   ;
		}
		
		\foreach \i in {-1,1,3}{
			\draw[rounded corners=3 pt,->,color=blue] (\i,4.5+0.3) --(\i,4+0.1+0.3)-- (\i+2,4-0.1+0.3)--(\i+2	,3.5+0.3)   ;
		}
		
		\draw (-2,4.5+0.3) node[above]{$4$};
		\draw (-1,4.5+0.3) node[above]{$5$};
		\draw (0,-0.5) node[below]{$0$};
		%		\draw[loosely dotted, line width=1.5pt] (0.4,-1.72)--++(1.1,0);
		\draw (1,-0.5) node[below]{$1$};
		\draw (2,-0.5) node[below]{$2$};
		\draw (3,-0.5) node[below]{$3$};
		\draw (4,-0.5) node[below]{$4$};
		%		\draw[loosely dotted, line width=1.5pt] (2.4,-1.72)--++(.2,0);
		%		\draw[loosely dotted, line width=1.5pt] (3.5,-1.72)--++(1.7,0);
		\draw (5,-0.5) node[below]{$5$};
		
		\draw(-.9,0) node{$N_t$-1};
		\draw(-0.7,3) node{$0$};
		
		\end{tikzpicture}
	\end{center}
	\caption{Vertical lines of the lattice for $N_{\gamma}=2$, $N_{x}=6$ and any $N_t$. There is $K=2$ large loops on the twisted torus, represented in red and blue}
	\label{chap6:fig:drawingK}
\end{figure}
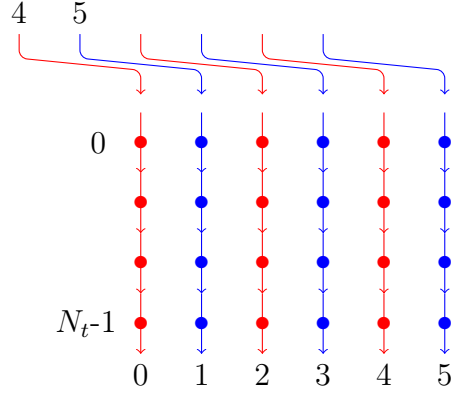

\subsection{General Residue Formula for the Ponzano-Regge Amplitude}

In this section, we focus on computing the Ponzano-Regge amplitude in the most general case, without assuming the condition $\lambda^{2}+\tau^{2}=1$ decoupling the temporal and spatial modes.
The computation is again done using the residue theorem, that allows us to obtain the mode expansion of the amplitude.
The starting point is equation \eqref{chap6:eq:starint_point_amp},
\begin{equation}
\la Z_{PR}^{\cK}| \Psi  \ra_{\lambda,\tau}
=
\f{1}{\pi} \int_{0}^{2\pi} \dd \varphi \sin^{2}(\varphi) \prod_{\omega,k} \f{1}{\det\big(Q_{\omega,k}(\varphi)\big)} 
\nn
\end{equation}
where we recall that the mode determinants are
\begin{align}
\det\big(Q_{\omega,k}(\varphi)\big)
\,=\,
(\tau^2 + \lambda^2 - 1) +\Big[2 - 2 \tau \cos\left(\f{2\pi}{N_x}k\right) \Big]  \Big[1 - \lambda \cos\left(\f{\varphi}{N_t} + \f{2\pi}{N_t}\omega - \f{1}{N_t}\gamma k\right) \Big] \; .
\end{align}
We assume that $\lambda^{2}+\tau^{2}\ne 1$, so that the determinants are not straightforwardly factorizable. We can nevertheless write them as:
\begin{equation}
\det\big[Q_{\omega,k}(\varphi)\big]
\,=\,
\f{\tau\lambda}2\left[2\tau^{-1} - 2 \cos\left(\f{2\pi}{N_x}k\right) \right]
\left[2 C_{k} - 2\cos\left(\f{\varphi}{N_t} + \f{2\pi}{N_t}\omega - \f{1}{N_t}\gamma k\right)\right]
\; .
\end{equation}
where the coefficients $C_k$ are (to keep the notation a bit lighter, we keep the dependency on $\lambda$ and $\tau$ implicit)
\begin{equation}
C_k = \lambda^{-1}\left[1 + \f{\tau^2 + \lambda^2 - 1}{2 - 2 \tau \cos \left(\f{2\pi}{N_x}k\right)}\right]
\; .
\end{equation}

In the decoupled case when assuming  $\tau^2 + \lambda^2 = 1$, the second term of those coefficients drop out and the $C_k$ become all equal to $\lambda^{-1}$.
Now, in the general case with  $\tau^2 + \lambda^2 \ne 1$, the coefficients $C_{k}$ remain all different. We can nevertheless still perform the products and get the pole decomposition in the class angle $\varphi$.
The product over spatial modes remains unchanged,
\begin{equation}
\prod_{\omega,k}
\left[
2 \tau^{-1} - 2 \cos\left(\f{2\pi}{N_x}k\right)
\right]
=
2^{N_t} \Big[ \ch(n_x X_{\tau^{-1}}) - 1 \Big]^{N_t}
=
2^{2N_t} \sh^{2N_t}\left( \f{1}{2}n_x X_{\tau^{-1}}\right)
\; .
\end{equation}
The product over the temporal modes is slightly modified,
\begin{equation}
\prod_{\omega,k} \left[2 C_{k} - 2\cos\left(\f{\varphi+2\pi\omega - \gamma k}{N_{t}}\right) \right]
=
\prod_{k}2\Big[ \ch(N_t X_{C_k}) - \cos\left(\varphi - \gamma k\right) \Big]
\; .
\end{equation}
The full Ponzano-Regge amplitude thus reads as a trigonometric integral,
\begin{equation}
\la Z_{PR}^{\cK}| \Psi  \ra_{\lambda,\tau}
=
\f{2^{N_t(N_x-2)}}{(\lambda \tau)^{N_t N_x}}
\f{1}{\sh^{2N_t}\left(\f{1}{2}N_x X_{\tau}\right)}
\f{1}{\pi} \int_{0}^{2\pi} \text{d} \varphi \sin^{2}(\varphi)
\prod_{k=0}^{N_x-1} \f{1}{2\left[ \ch(N_t X_{C_k}) - \cos\left(\varphi - \gamma k\right) \right]}
\;.
\end{equation}
The poles can be directly read from this expression:
\begin{equation}
\varphi_{k}^{\pm} = \gamma k \mp i N_t X_{C_k}
\,.
\end{equation}
As illustrated in figure \ref{chap6:fig:generalpoles}, these are not linear anymore as in the decoupled case with  $\tau^2 + \lambda^2 = 1$. Moreover the dependency on $k$ through the coefficients $C_k$ lifts the degeneracy of the poles due to the GCD $K$. Now, whatever the GCD between the shift $N_{\gamma}$ and the number of spatial nodes $N_{x}$, the poles in $\varphi$ remain simple and there is no degeneracy. This allows for a direct evaluation of the integral by residue, summarized in the following proposition.
\begin{figure}[h!]
	\includegraphics[scale=0.45]{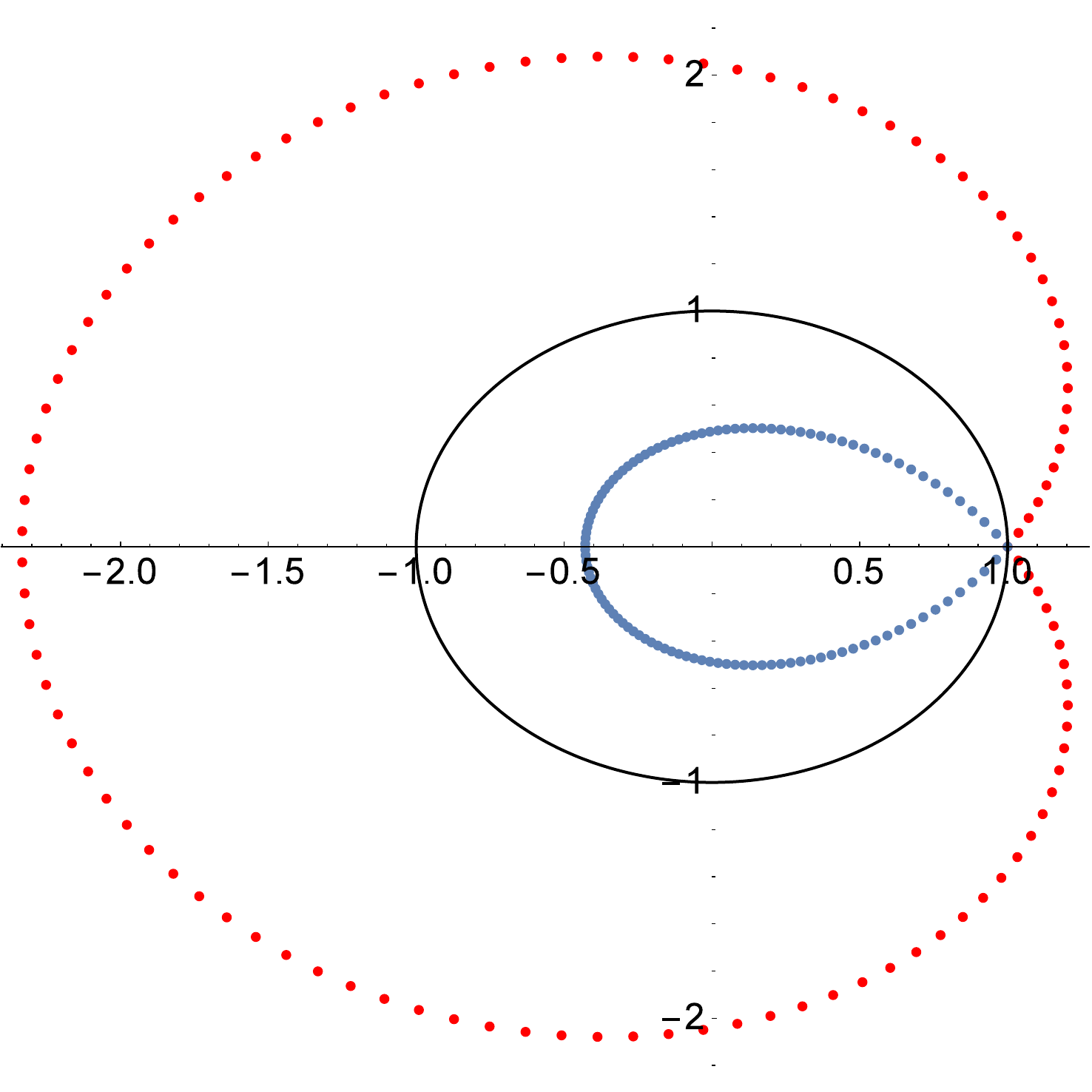}
	\hspace{8mm}
	\includegraphics[scale=0.45]{pole_lambda_06_tau_04_Nx_111_Ngamma_1_Nt_1-eps-converted-to.pdf}
	\includegraphics[scale=0.45]{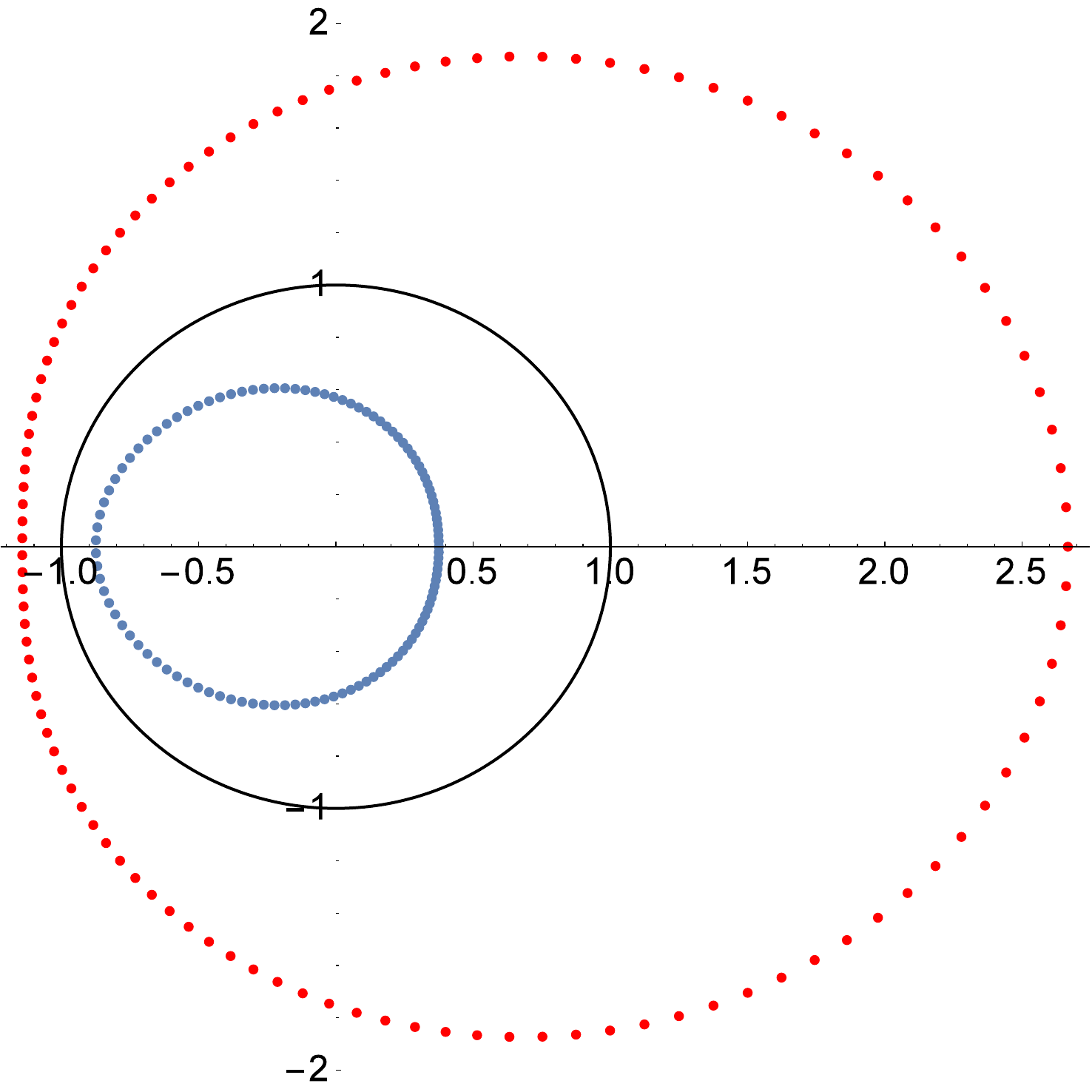}
	\hspace{16mm}
	\includegraphics[scale=0.45]{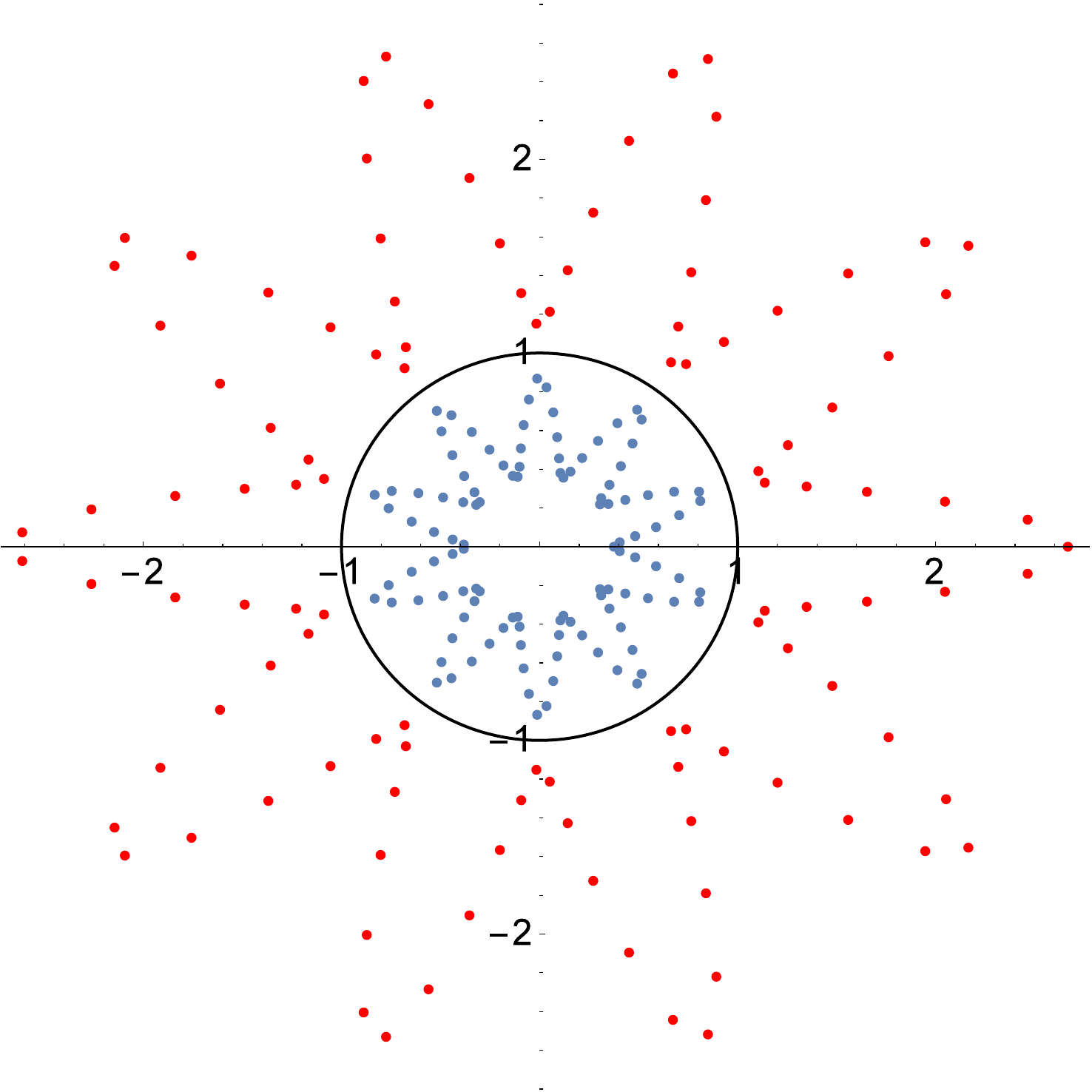}
	\caption{Plot of the poles in the complex planes, $z_{k}^{\pm} = e^{i \gamma k} e^{\pm N_t X_{C_k}}$, for a square lattice on the twisted torus determined by the numbers of nodes $N_{t}=1$ and $N_{x}=111$, for, in order $N_{\gamma}=1$ and critical couplings $(\lambda,\tau)=(0.6,0.4)$, $N_{\gamma}=11$ and critical couplings $(\lambda,\tau)=(0.6,0.4)$,  $N_{\gamma}=1$ and non-critical couplings $(\lambda,\tau)=(1.6,0.4)$, $N_{\gamma}=11$ and non-critical couplings $(\lambda,\tau)=(1.6,0.4)$. The red dots represent $z^{+}_{k}$ while the blue ones $z^{-}_{k}$. They are related by inversion with respect to the unit circle  (in black), since they have inverse modulus for equal phases. For critical couplings, on the top plots, $z=1$ is a pole, which leads to a divergent integral, while there is no pole on the unit circle for non-critical couplings as shown on the bottom plots.
	}
	\label{chap6:fig:generalpoles}
\end{figure}
\begin{prop}
	The Ponzano-Regge partition function on the twisted solid torus for a coherent spin network boundary state on the square lattice, with couplings satisfying $\lambda^2 + \tau^2 \ne 1$, can be written as a finite sum with a clear pole structure in the twist angle $\gamma$:
	\begin{align}
	\label{chap6:eq:Zgeneral}
	\la Z_{PR}^{\cK}| \Psi  \ra_{\lambda,\tau}
	=
	\f{2^{N_t(N_x-2)}}{(\lambda \tau)^{N_t N_x}} &\f{1}{\sh^{2N_t}\left(\f{1}{2}N_x X_{\tau^{-1}}\right)} \sum_{n=0}^{N_x-1} \f{\sin^2\left(\gamma n + i N_t X_{C_k} \right)}{\sh(N_t X_{C_n})} \nn \\
	&\prod_{\substack{k=0 \\ k \neq n}}^{N_x-1} \f{1}{2\big[\ch(N_t X_{C_k}) - \cos(\gamma(n-k)+i N_t X_{C_n}\big]} \; .
	\end{align}
\end{prop}

\begin{proof}
	
	To compute the integral by residue, we  perform the change of variable $z = e^{i \varphi}$ as before. The poles in the complex plane are then $z_{k}^{\pm} = e^{i \gamma k} e^{\pm N_t X_{C_k}}$, with
	\begin{equation}
	e^{X_{C_k}} = C_k + \sqrt{C_k^2-1} \; .
	\end{equation}
	The factorized form of the amplitude is exactly the same as before and the integration over the unit circle $\text{U}(1)$ gets contributions only from the poles $z_{k}^{-}$ whose norm is lesser than 1:
	\begin{align}
	\la Z_{PR}^{\cK}| \Psi  \ra_{\lambda,\tau}
	&=
	\f{2^{N_t(N_x-2)}}{(\lambda \tau)^{N_t N_x}} \f{1}{\sh^{2N_t}\left(\f{1}{2}N_x X_{\tau}\right)} \f{1}{\pi} \int_{U(1)} \f{\text{d}z}{i} \f{-1}{4}\f{(z^2-1)^2}{z^2} z^{N_x-1} \prod_{k=0}^{N_x-1} \f{-e^{i \gamma k}}{(z-z_k^+)(z-z_k^-)} 
	\nn\\
	&=
	\f{2^{N_t(N_x-2)}}{(\lambda \tau)^{N_t N_x}} \f{1}{\sh^{2N_t}\left(\f{1}{2}N_x X_{\tau}\right)} 2 \sum_{n=0}^{N_x-1} \f{-1}{4}\f{((z^{-}_n)^2-1)^2}{(z^{-}_n)^2} \f{- e^{i \gamma n}}{z^{-}_n - z^{+}_n} \prod_{\substack{k=0 \\ k\neq n}}^{N_x-1} \f{- z_{n}^{-} e^{i \gamma k}}{(z_n^- -z_k^+)(z_n^- -z_k^-)}
	\,,\nn
	\end{align}
	which gives the announced result.
	
\end{proof}

We do not know how to re-sum over the pole label $n$ due to non-linearity of the poles in the complex plane, but the expression \eqref{chap6:eq:Zgeneral} has a clear physical interpretation as an expansion over the poles in the twist angle $\gamma$ of the Ponzano-Regge partition function. This allows a clear comparison with other formulas for the partition function of 3D quantum gravity, such as the BMS character formula which we discuss below in the next section.

%%%%%%%%%%%
\subsection{The Continuum Limit: Recovering the BMS character from the Ponzano-Regge model}
%%%%%%%%%%%

We would like to compare the continuum limit of our Ponzano-Regge partition function formulas with the computation of the partition function of asymptotic flat 3D gravity for a solid twisted  torus as a BMS character as derived in  \cite{Barnich:2015mui}:
\begin{equation}
Z_{\textrm{3D flat gravity}}^{\textrm{asympt}}[\gamma] = \chi_{\text{BMS}}[\gamma]
\,,\qquad\textrm{with}\quad
\chi^{BMS}[\gamma]= e^{S_0}\prod_{k=2}^{\infty} \f{1}{2-2\cos(\gamma k)}\,,
\end{equation}
where  the  character $\chi^{BMS}$ of the Bondi-Metzner-Sachs (BMS) group in the ``vaccuum'' representation is evaluated on a super-rotation of (Poincar\'e) angle $\gamma$.
Here, $S_{0}$ is the on-shell action. For a real twist angle $\gamma\in\R$, the infinite product above is just a formal definition. It actually has poles at every rational twist angle, $\gamma\in2\pi\Q$. The BMS character is in fact well-defined on the upper complex plane, and can be identified as the Dedekind $\eta$-function up to a factor. 

This formula was also re-derived as an asymptotic limit of  Regge calculus for discretized 3D gravity in \cite{Bonzom:2015ans}. It was further recovered as leading order of a WKB approximation of the Ponzano-Regge model with LS spin network boundary states in \cite{Dittrich:2017rvb}. These LS spin networks consist in coherent intertwiners with fixed spins and allowed the study of the Ponzano-Regge amplitude by a saddle point approximation in the large spin regime (i.e. large edge lengths), which gives:
\begin{equation}
Z^{\textrm{1-loop PR}}
\,=\,
e^{i S_{0}} \left[\sum_{n=1}^{N_x-1} A_{n}\,4\sin^2\f{\gamma n}{2}\right]\, \prod_{k=1}^{\f{N_x-1}{2}} \f{1}{2-2 \cos(\gamma k)} \; .
\end{equation}
We refer to this leading order behaviour of the Ponzano-Regge amplitude as 1-loop by similarity with quantum field theory calculations.
The twist angle $\gamma$ is defined from the discrete shift as $2\pi N_{\gamma}/N_{x}$.
The integer $n$ labels the saddle point and describes geometrically the winding number of the embedding of the boundary 2-torus (in the spatial direction).
The coefficient $A_{n}$ is a rather complicated pre-factor which depends on the winding number but is independent from the twist parameter.
%
%Due to the discrete structure of the computation, the on-shell action only contributed with a sign that we do not consider further. Geometrically, one can see the winding number as the number of time the cylinder winds around itself before identifying the top and the bottom to get the torus.
%
There are two main differences with the original BMS formula derived in  \cite{Barnich:2015mui}:
\begin{itemize}
	
	\item There is a sum over winding numbers $n$ for the spatial geometry of the torus. This is a non-perturbative effect.
	
	\item The product over modes start at $k=1$ instead of $k=2$. As explained in \cite{Oblak:2015sea}, this corresponds to a massive mode, i.e. working with a massive BMS representation\footnotemark. This is also a non-perturbative effect.
	
\end{itemize}
\footnotetext{
	The massive BMS representation with mass $M>0$ and spin $J$ is constructed as the induced representation based on the co-adjoint orbit  of constant supermomentum $P=M-c/24$, where  $c$ is the BMS central charge. It represents a BMS particle with rest mass $M>0$.
	Its character  evaluated on a group element $(f,v)\in\textrm{BMS}_{3}=\textrm{Diff}^{+}(S^{1})\ltimes\textrm{Vect}(S^{1})$, with the superrotation $f$ and supertranslation $v$ are both given by functions over the circle $S^{1}$, can be computed  \cite{Oblak:2015sea} and depends only on the Poincar\'e rotation angle $\gamma$ of $f$ and the 0-mode of $v$:
	\begin{equation}
	\chi^{BMS}_{M,J}[(f,v)]=e^{iJ\gamma}e^{iv_{0}(M-c/24} \,\f1{\prod_{k\ge 1}2-2 \cos k(\gamma+i\eps)}\,,
	\nn
	\end{equation}
	where  $\eps\rightarrow 0^{+}$ is a regulator and the angle $\gamma$ is assumed to be non-vanishing.
	The vaccuum representation is the induced representation with vanishing spin based on the co-adjoint orbit with supermomentum $P=-c/24$, i.e. with vanishing mass. Its character is similarly computed in \cite{Oblak:2015sea}:
	\begin{equation}
	\chi^{BMS}_{vaccuum}[(f,v)]=e^{-iv_{0}c/24} \,\f1{\prod_{k\ge 2}2-2 \cos k(\gamma+i\eps)}\,.
	\nn
	\end{equation}
}
It is nevertheless possible to recover the BMS character formula as a truncation of $Z^{\textrm{1-loop PR}}$ in the asymptotic limit $N_{x}\rightarrow\infty$. In fact, the factor $\sin^2{\gamma n}/{2}$ for each winding number $n$ kills a factor in the product determinant, leading to
\begin{equation}
Z^{\textrm{1-loop PR}}
\,=\,
\sum_{n=1}^{N_x-1}A_{n}  \prod_{\substack{k=1 \\ k\neq n}}^{\f{N_x-1}{2}} \f{1}{2-2 \cos(\gamma k)} \; .
\end{equation}
If we truncate the sum to a trivial winding number $n=1$, which is the only allowed embedding of the ``spacetime'' cylinder in Euclidean 3D space, then this leads back as wanted to a finite version of the character for the vaccuum representation of the BMS group, which can be considered as the massless limit case of the massive representation (see \cite{Barnich:2015uva,Oblak:2015sea,Oblak:2016eij} for the classification of the BMS co-adjoint orbits and their corresponding induced representation).
We can then take the asymptotic limit $N_{x}\rightarrow\infty$  corresponding to a torus whose radius grows to infinity.
The difficulty in making sense of this limit is that, in the Ponzano-Regge context, we necessarily deal with a real twist angle defined by the combinatorics of the boundary lattice and that the infinite product is ill-defined in that case.

\medskip

The present work considerably improves on this previous 1-loop approximation of the Ponzano-Regge partition function. First, we compute the exact partition function for a quantum boundary state, with no saddle point approximation or large spin limit. Second, we naturally get an imaginary shift, which regularizes the amplitude.
Let us look into this exact Ponzano-Regge amplitude in more details.

Starting with the decoupling ansatz, with couplings satisfying $\lambda^{2}+\tau^{2}=1$, it is convenient to write the Ponzano-Regge amplitude in a factorized form distinguishing the amplitude pre-factor and the sum over poles as previously,
\begin{equation}
\la Z_{PR}^{\cK}| \Psi  \ra_{\lambda,\tau}= \cA\,z^{PR}
\qquad\textrm{with}\quad
\cA=
\f{2^{N_t(N_x-2)}}{(\lambda \tau)^{N_t N_x}}
\f{1}{\sh^{2N_t}\left(\f{1}{2}N_x X_{\tau^{-1}}\right)}\,,
\end{equation}
and the reduced Ponzano-Regge amplitude:
\begin{equation}
z^{PR}
=
\sum_{n=0}^{N_x-1} \f{\sin^2\left(\gamma n + i N_t X_{\lambda^{-1}} \right)}{\sh(N_t X_{\lambda^{-1}})}
\prod_{k=1}^{N_x-1} \f{1}{2(\ch(N_t X_{\lambda^{-1}}) - \cos(k\gamma+i N_t X_{\lambda^{-1}}))}
\end{equation}
The pre-factor $\cA$ plays the role of the exponential of the on-shell action, while the reduced Ponzano-Regge $z^{PR}$ encodes the whole pole structure of the amplitude. 
It is clear that we would like to take the critical limit $\lambda\rightarrow 1$, and thus $\tau\rightarrow0$, to recover the BMS character formula.
Geometrically, this corresponds to keeping the time span of the cylinder fixed while sending its radius to $\infty$, thus looking at the asymptotic limit in space while remaining at finite time.
In this limit, $X_{\lambda^{-1}}\rightarrow0$, the amplitude pre-factor diverges due to the factor $\tau^{-N_{x}N_{t}}$, while the  only divergent term in the reduced Ponzano-Regge amplitude is the factor $\sh(N_t X_{\lambda^{-1}})$:
\begin{equation}
\sh(N_t X_{\lambda^{-1}})\,z^{PR}
\underset{\substack{\lambda\rightarrow1 \\\tau\rightarrow0}}\sim
\left(\sum_{n=1}^{N_x-1} \sin^2\gamma n\right) \,
\prod_{k=1}^{N_x-1} \f{1}{2 - 2\cos (\gamma k )}
=
\sum_{n=1}^{N_x-1}\cos^{2}\f{\gamma n}2 \prod_{\substack{k=1 \\ k\neq n}}^{\f{N_x-1}{2}} \f{1}{2-2 \cos(\gamma k)}
\,,
\end{equation}
where we recognize once again the term with trivial winding number $n=1$ as the BMS character for the vaccuum representation\footnotemark.
\footnotetext{
	Let us nevertheless point out that the pre-factor $\cos^{2}\f{\gamma n}2$ depends on the rotation angle $\gamma$. Although it seems possible  re-absorb it as a superposition of massive, vaccuum and massless BMS characters, it does not have a direct natural interpretation in terms of BMS representation theory.
}
From this perspective, it is intriguing that the non-perturbative sum over winding numbers creates the equivalent of a mode $k=1$,
which would otherwise be absent from the semi-classical calculation, leading to a BMS character for a massive representation.

\medskip

To be more precise, we should take the infinite refinement limit $N_{t}\rightarrow \infty$ and describe the relative scaling of the coupling $\lambda$ to ensure this limit. If $\lambda=1-\eps$ with $\eps\rightarrow 0$, then $X_{\lambda^{-1}}\sim\sqrt{2\eps}$. So if $\lambda$ goes to 1 faster than $N_{t}^{-2}$, i.e. if $1-\lambda\propto 1/N_{t}^{2+\sigma}$ with $\sigma>0$, then we are clearly in the case described above.

However, there is a critical regime of the scaling limit, considering $\lambda=1-\alpha^{2} /2N_{t}^{2}$ for $\sigma=0$, in which case the couplings still converge to their critical value, but then we keep a finite imaginary shift in the poles:
\begin{equation}
z^{PR}
\quad\underset{N_{t}\rightarrow\infty}{\overset{\lambda=1-\f{\alpha^{2}}{2N_{t}^{2}}}{\sim}}\quad
%\underset{\textrm{critical regime}}\sim
\sum_{n=1}^{N_x-1} \f{\sin^2(\gamma n+i\alpha)}{\sh \alpha} \,
\prod_{k=1}^{N_x-1} \f{1}{2(\ch\alpha - \cos (\gamma k +i\alpha))}
\,.
\end{equation}
%As we truncate to the first winding number $n=1$,
Ignoring the sum over winding modes, and taking the limit $N_{x}\rightarrow\infty$, we obtain a deformation of the BMS character regularized by this complex shift:
\begin{equation}
\chi_{\textrm{reg}}[\gamma,\alpha]=
\prod_{k=1}^{N_x-1} \f{1}{2(\ch\alpha - \cos (\gamma k +i\alpha))}
\,.
\end{equation}
Interestingly, comparing to the derivation of the BMS character presented in \cite{Oblak:2015sea}, this amounts to the usual BMS character evaluated on the super-rotation defined by the twist angle $\gamma$ composed with a Bogoliubov transformation\footnotemark~ mixing the negative mode $-k$ with the positive mode $+k$.
\footnotetext{%
	The interpretation in terms of Bogoliubov transformations begets the question of unitarity. In fact, the 2D boundary state that we are considering does not correspond to an initial or final canonical state, but describes the geometry of the ``time-like'' boundary and thereby determines the flow of time and the evolution of the bulk geometry. The natural question in that context is which boundary states ensures a unitary evolution for the geometry of the disk.
}
Indeed, the superrotation $R_{\gamma}$ acts on supermomenta $v_{k}$ as $v_{k}\mapsto e^{i\gamma k}$, while we define a Bogoliubov transformation pairing the modes propagating in opposite directions:
\begin{equation}
B_{\alpha}=\mat{cc}{e^{-\alpha} & 2i\sh \f\alpha 2 \\ 2i\sh \f\alpha 2 & e^{+\alpha}} \qquad \textrm{acting on 2-vectors}\quad \mat{c}{v_{k}\\ v_{-k}}
\,.
\end{equation}
The resulting character $\chi_{BMS}(R_{\gamma}B_{\alpha})$ is the inverse determinant of $\id-R_{\gamma}B_{\alpha}$ acting on the vector space of supermomenta:
\begin{equation}
\chi_{BMS}(R_{\gamma}B_{\alpha})
=
\prod_{k\ge1}
\f1{\det_{k}(\id_{2}-R_{\gamma}B_{\alpha})}
=
\prod_{k\ge1} \f{1}{2(\ch\alpha - \cos (\gamma k +i\alpha))}
\,.
\end{equation}

\medskip

Now, although the details of the previous limits and calculations do not go through in the general case beyond the decoupled ansatz, when $\lambda^2+\tau^2\ne 1$, the logic of going to critical couplings to recover the BMS character formula still applies. Indeed, in the general case, the amplitude pre-factor does not change at all, but the reduced Ponzano-Regge amplitude acquires a less regular pole structure:
\begin{equation}
z^{PR}[\lambda,\tau]
=
\sum_{n=0}^{N_x-1} \f{\sin^2\left(\gamma n + i N_t X_{C_k} \right)}{\sh(N_t X_{C_n})}  \prod_{\substack{k=0 \\ k \neq n}}^{N_x-1} \f{1}{2\big[\ch(N_t X_{C_k}) - \cos(\gamma(n-k)+i N_t X_{C_n}\big]}
\,,
\end{equation}
with the coefficients 
\begin{equation}
C_k = \lambda^{-1}\left[1 + \f{\tau^2 + \lambda^2 - 1}{2 - 2 \tau \cos \left(\f{2\pi}{N_x}k\right)}\right]
\,.
\nn
\end{equation}
Choosing critical real positive couplings, with $\lambda+\tau=1$, it appears that the coefficients $C_{k}$ go to 1 for very low momenta $k\ll N_{x}$.  Thus the coefficients $X_{C_{k}}$ go to 0, and the general formula above once more reduces to a BMS character.
By periodicity, for very large momenta, for $k$ very close to $N_x$, the coefficient $C_k$ also approaches 1. The situation is however different for large momenta in the intermediate range,  with $k/N_x$ kept fixed in $]0,1[$, for example for $k\sim N_{x}/2$. In this case,  this approximation fails and we have to deal with the more complicated structure of the roots. Since the product over modes involves all the modes from $k=1$ to $k=N_{x}-1$, one can not ignore the effects of large momenta on the partition function, but could consider them as deep quantum gravity effects.

\medskip

Overall, we have discussed how to take the continuum limit of the boundary lattice. For critical couplings,  the exact Ponzano-Regge amplitude reads as a finitely truncated BMS character, which leads back to the BMS character---possibly with a complex shift regularizing it---in the infinite refinement limit $N_{x},N_{t}\rightarrow\infty$. Since the exact formulas and poles can be written explicitly, it is natural to wonder if  the  Ponzano-Regge amplitude can be identified as the character of a modified BMS group for finite boundary. Such a deformed BMS group would be identified as  the symmetry group of the Ponzano-Regge model with the 2D discrete boundary geometry, and would probably be the reason behind the simplification of the Ponzano-Regge partition function. This will be investigated in future work.

\bigskip

We conclude this chapter by mentioning that the result we obtained also motivates us to look at the symmetry group at the boundary of a quasi-local region. This fits with the abundance of recent work on edge modes for gravity. The first step towards this study is to consider a given spatial slice, and to study the symmetry of its boundary, i.e. the symmetry of a circle with handles. This is a work in progress!

%%%%%%%%%%%%%%%%%

\newpage
~
\thispagestyle{empty}

\pagestyle{Conclusion}
\chapter*{Conclusion}
\addcontentsline{toc}{part}{Conclusion}

The study of three-dimensional Euclidean flat gravity was at the centre of my work. A tremendous amount of work and progress have been made in the last decades on the study of quantum gravity.. We are still far, however,  from understanding everything. Nevertheless, one case is much more understood at least in absence of matter: the three-dimensional one. In that case, it is at least possible to write an explicit quantum theory for gravity. In this thesis, we considered a discrete approach of three-dimensional quantum gravity formulated via the Ponzano-Regge model \cite{PR1968}. Even though the model is fairly old, there is still a lot to explore. Especially so on the study of the Ponzano-Regge amplitude for a quasi-local region. While it is not something completely new, most of the previous studies focused on the trivial case: the three-dimensional ball with boundary given by the two-dimensional sphere. In that case, it has been shown that quantum gravity can be understood as dual to two copies of the Ising model \cite{Dittrich:2013jxa,Bonzom:2015ova}. This duality is a perfect example of quasi-local holography. What still needs to be studied is more complicated topology. A first step towards a more complicated one, the torus topology, was made by Bonzom and Dittrich in \cite{Bonzom:2015ans}. Using the (quantum) Regge calculus approach for General Relativity, they computed the partition function of three-dimensional flat gravity in a finite region of space-time with the torus topology. They obtained a beautiful result: they recovered the structure of the BMS character, i.e. the character of the asymptotic symmetry group of flat gravity \cite{Barnich:2014kra,Barnich:2015uva,Oblak:2015sea,Oblak:2016eij}. Recall that the partition function in the continuum was computed both in the AdS case \cite{Maloney:2007ud} and in the flat case \cite{Giombi:2008vd,Barnich:2015mui}. Hence, the computation of \cite{Bonzom:2015ans} showed that the BMS group might be extended, in an undefined way, to also be fully part of the symmetry group of flat gravity for a quasi-local region.The caveat of this computation being that it is still a perturbative one.

The main goal of my work was to try to go one step further, and focus on the computation of the quasi-local and fully {\it exact} amplitude for three-dimensional quantum gravity on the torus topology. To do so, we considered the Ponzano-Regge model. This model can be seen as a quantization of the Regge calculus based on a first order formulation of General Relativity. As a completely discrete model for quantum gravity formulated as a state sum, the Ponzano-Regge model is a spin foam model. It is, in fact, the first example of what will be called spin foam model a few decades later.

The Ponzano-Regge model is perfect to study quasi-local regions for gravity. It provides us with an exact definition of the bulk theory, and its divergence has long been understood, see chapter \ref{chap3}. In presence of boundaries, its partition function is naturally a functional of the boundary data. We must emphasize that, compared to the Regge calculus approach, the theory is always exactly solved in the bulk. There are no approximations whatsoever for the bulk theory. In the case of the 3-ball with a 2-sphere as boundary, this results in the fact that the Ponzano-Regge amplitude is just the spin network evaluation of the boundary state. In the case of the torus however, since the topology is non-trivial, one integration over the non-contractible cycle remains. The Ponzano-Regge model thus allows to solve exactly the bulk theory while projecting all of the non-trivial bulk information on the boundary. We truly obtain a theory that is entirely defined by its choice of boundary condition without any background parameters on top of it. As we have explained in the main text, the boundary conditions of the Ponzano-Regge model are encoded in the spin network states and, in this work, we have considered two classes of boundary state.
\\

The first one, presented in chapter \ref{chap5}, was chosen because of its suitable geometrical interpretation in the asymptotic limit: it is peaked on a given geometry. That is, the boundary state is chosen in the class of states that diagonalizes the geometry of the intrinsic metric on the boundary and intrinsically depends on the twist of the torus. This boundary state is supposed to be the closest analogue of the Gibbons-Hawking-York boundary action term in the continuum. This is indeed confirmed by the saddle points analysis, returning an on-shell action of the form length times extrinsic curvatures \eqref{chap5:eq:classical_action}. Hence, this choice of boundary state is perfect to see the link between our computation, the Regge calculus computation and the computation in the continuum. Due to the complicated structure of the boundary conditions, it is not possible however to exactly compute the amplitude. What we can do is a saddle points approximation. We again emphasize that the approximation is only at the level of the boundary. The bulk theory is {\it exactly} solved. At the end of the day, we recover at one-loop the beautiful formula \eqref{chap5:eq:amp_final}
\begin{equation*}
	\la Z_{PR}^{\cK}| \Psi_{coh}  \ra^\text{1-loop}_{o}
	=  \sum_{n=1}^{N_x-1} (-1)^{2T N_t n}  \cA(n) \big(2-2\cos(\gamma n)\big)  \times  \prod_{k=1}^{\frac{N_x-1}{2}} \frac{1}{ 2-2\cos(\gamma k) }\;.
\end{equation*}
Recall that the parameter $n$ is the winding number around the cylinder. Each part of this formula is of primordial importance. First, it is clear that, compared to the Regge approach, we obtain much more than just the BMS character. What we would like however, is to recover it when the boundary is pushed at infinity. And this is indeed the case. The factor $\cA(n)$ is such that the main contribution when the boundary is pushed at infinity is given by $n=1$. Thanks to the measure term, this case corresponds to the BMS character. Recall that the on-shell action is just the sign factor in the above formula. Hence, while we do recover the BMS character in the asymptotic limit for the boundary, the structure for a finite region is much more complicated and richer. On top of the usual BMS character, we see that non-perturbative quantum corrections arise with the winding number. It is however unclear what kind of corrections the winding number holds for. We saw that only the case $n=1$ is embeddable in $\R^{3}$. Hence, it might be that the winding number carries unwanted contributions. Or that it carries truly important information, such as a sum over all possible geometrical configurations compatible with the choice of boundary data, or the sum over (a subset of) admissible topologies. Indeed, only the case $n=1$ is embeddable in $\R^3$, hence one might see the others contributions as more complicated topology in $\R^3$. This question still remains open. This computation shows that the BMS structure is conserved for a quasi-local region, but as part of a larger set. The symmetry group for a finite region for flat gravity should thus be some extension of the BMS group, which we recovered in the asymptotic limit. It is interesting to note that the main difference between each contribution is basically that of which pole is suppressed by the measure factor. 

It is complicated however to truly understand these other contributions in the saddle approximation due to the complicated factor $\cA(n)$. It is first interesting to see if these quantum corrections are also present for an exact computation and that they are not only an artefact coming from the perturbative approach at the boundary. 
\\

After establishing the link between the Ponzano-Regge model and other approaches to the partition function of three-dimensional flat gravity, we performed an exact computation. To do so, it is natural to consider the generating function of the spin network state. Even though it is rather a mathematical reason, it was shown that, by carefully choosing the weight of the generating function, the Ponzano-Regge amplitude in the sphere can be written as a Gaussian. As a side note, it is through the use of this particular state that the duality with the Ising model was proven in the case of the sphere. It gives us another reason to try and see if some dualities can also be recovered for the torus (unfortunately, at this stage of my work, I have not yet succeeded in finding them). In practice, we did not consider the actual generating function in chapter \ref{chap6}. Indeed, due to the method of computation, we had to restrict the analysis to the anisotropic case. A study of the full generating function still needs to be performed. The result of this last chapter is, in a sense, the culmination of (most of) my work from the last three years. It can be synthesized in the following formula
\begin{equation*}
	\begin{split}
	\la Z_{PR}^{\cK}| \Psi_{\lambda,\tau}  \ra = &\f{2^{N_t(N_x-2)}}{(\lambda \tau)^{N_t N_x}} \f{1}{\sh^{2N_t}\left(\f{1}{2}N_x X_{\tau}\right)} \times \\ & \sum_{n=0}^{N_x-1} \f{\sin^2\left(\gamma n + i N_t X(C_n) \right)}{\sh(N_t X(C_n))}  \prod_{\substack{k=0 \\ k \neq n}}^{N_x-1} \f{1}{2(\ch(N_t X(C_k)) - \cos(\gamma(n-k)+i N_t X(C_n))} \; .
	\end{split}
\end{equation*}
with
\begin{equation*}
	C_k = \f{1}{\lambda} \left( 1 + \f{\lambda^2 + \tau^2 - 1}{2 - 2\tau \cos\left(\f{2\pi}{N_x}k\right)} \right) \;.
\end{equation*}
The beauty of this result is to be found in its sheer generality. We computed the amplitude of three-dimensional gravity for a finite region of a torus-like space-time valid for arbitrary $N_t$, $N_x$ and twist $N_\gamma$. Recall that one of the key problems of all the previous approaches is that it was truly only defined for an irrational twist angle. In the continuum case, it was necessary to keep track of the complex part of the modular parameter of the torus to regularize the amplitude. In the discrete approach however, while the discreteness of the lattice provides a natural regularization of the amplitude, it was also restricted to the case $K=1$, corresponding to an irrational angle in the suitable double scaling limit. Here, we went beyond this restriction and the amplitude is truly defined for any $K$. In exchange, we introduced coupling parameters. Depending on the value of the coupling parameters, the amplitude may or not converge. The key difference is that the divergence of the amplitude is restricted to isolated values of the coupling parameters. Hence, this computation provides us with a perfect starting point to look at the emergence of modular invariance in three-dimensional flat gravity and thus of conformal field theory. There is, however, a (big?) caveat. It is not the twist angle that acquires a complex part to recover the modular parameter; but rather, it is a global shift (by $N_t X(C_n)$) into the complex plane that made the amplitude defined for any twist. 

As we explained in the main text, it is possible to recover the vacuum BMS character from the limit $\lambda \rightarrow 0$ and $\tau \rightarrow 1$ as the $n=1$ contribution as in chapter \ref{chap5}. It is, however, much more natural and appealing to identify a massive BMS character. Indeed, we saw that this limit is part of a bigger family of solutions where the sums over $n$ can be performed explicitly. In the limit $\lambda \rightarrow 0$ and $\tau \rightarrow 1$, this returns an amplitude whose mode structure is closer to the massive BMS character than to the massless one. The biggest downside to this identification comes from the fact that we have not been able, at least for know, to make apparent the classical contribution of the action in the amplitude. Hence, we still miss this part from the character. However, recall that this contribution might just be a sign due to discreetness of the lattice. It is still unclear why we naturally see the massive character appear in this approach.

The global structure of the result, however, answers the previous question about the non-perturbative quantum corrections: it is not an artefact due to the perturbative approach at the boundary. There are true consequences of working on a quasi-local region. However, the interpretation of the corrections gets spoiled. The interpretation of the parameter $n$ as a winding number is not straightforward at all from the viewpoint of the residue computation and it is not clear how to get a geometrical interpretation of this parameter and of the corrections in this general case.

We conclude by pointing out that this result is also interesting from the viewpoint of quasi-local holographic dualities. The fact that we were able to do an exact computation suggests the existence of powerful hidden symmetry and possibly quasi-integrable structure. In fact, we already knew since chapter \ref{chap4} that integrable structures naturally arise even on the torus. Indeed, we related the Ponzano-Regge model on the torus with the 6-vertex model. It is appealing to look at the full generating function and see if some dualities with statistical models can be recovered, like the duality between the Ponzano-Regge model on the sphere and the Ising model.
\\

Although we have reached the end of my thesis, there is still plenty to be done regarding the work I have presented here. I have come to realise that what has been achieved in the last three years opens a lot of new leads that I hope to follow on my post-doctorate: the duality with different statistical models, the impact of the boundary state and the choice of polarization from the topological invariance perspective, the structure of the quasi-local symmetry group for gravity on the torus, related to the edge modes for gravity, going beyond the partition function and looking for general observables by adding matter and defects to the theory...
\newpage
To be continued ...

\newpage
~
\thispagestyle{empty}

\pagestyle{Regular}
\appendix

\chapter{Spin Recoupling and 4-valent Intertwiner Basis}
\label{App:4valent_intertwiner}
\numberwithin{equation}{chapter}
 
The irreducible unitary representations of the $\SU(2)$ Lie group are labeled by half-integers, called spins, $j\in\f12\N$. The corresponding representation space $V^j$ is $(2j+1)$-dimensional and spanned by basis vectors $|j,m\ra$ diagonalising the Casimir operator and the generator $J_{z}$:
\begin{equation}
	\vec{J}^2\,|j,m\ra
	=\Big{[}J_{z}^2+\f12\big{(}J_{-}J_{+}+J_{+}J_{-}\big{)}\Big{]}\,|j,m\ra
	=j(j+1)\,|j,m\ra
	\,,\qquad
	J_{z}\,|j,m\ra=m\,|j,m\ra
	\,.
\end{equation}
These states can be interpreted geometrically as quantized 3-vectors of length $j$. One can actually define coherent states, \`a la Perelomov, peaked on classical vectors wth minimal spread (e.g. \cite{Livine:2007vk}).

The spin $j$ representation is unitarily equivalent to its conjugate representation and the map is:
\begin{equation}
	\begin{array}{lclcl}
		\varsigma: &&V^j& \rightarrow &\overline{V^j} \\
		&& |j,m\ra &\mapsto &  D^j(\varsigma)|j,m\ra=(-1)^{j+m}\,|j,-m\ra
	\end{array}
\end{equation}
such that $\varsigma^2=(-1)^{2j}\,\id$. This maps applies to the representation of the $\SU(2)$ group elements:
\begin{equation}
	\overline{\la j,n|D^j(g)|j,m\ra}
	=
	\la j,n|D^j(\varsigma^{-1}) g \varsigma)| j,m\ra\,
\end{equation}
which is the expression for arbitrary spins of the matrix identity in the fundamental representation, for 2$\times$2 matrices:
\begin{equation}
	\eps^{-1}g\eps=\overline{g}
	\,,\qquad
	\eps=\varsigma_{(j=\f12)}=\mat{cc}{0& 1 \\ -1 &0}
	\,,\qquad
	\forall g =\mat{cc}{a & b \\ -\bar{b}&\bar{a}}\,\in\SU(2)
	\,.
\end{equation}
This map allows to define the bivalent intertwiners. This is equivalent to identifying the $\SU(2)$-invariant states in tensor products $V^{j_{1}}\otimes V^{j_{2}}$ of two spins. For such a state to exist, the two spins must be equal, $j_{1}=j_{2}=j$:
\begin{equation}
	|\omega_{j}\ra
	=
	\f1{\sqrt{2j+1}}\sum_{m}\varsigma|j,m\ra \otimes  |j,m\ra
	=
	\f1{\sqrt{2j+1}}\sum_{m}(-1)^{j-m}\,|j,m\ra \otimes  |j,-m\ra
	\,\quad\in V^j\otimes V^j
	\,,
\end{equation}
\begin{equation}
	\begin{array}{lclcl}
		\iota_{2}: &&  V^j \otimes V^j  & \rightarrow &\C \\
		&& |j,m\ra\otimes |j,n\ra &\mapsto &
		\la \omega_{j}\,|\,(j,m)(j,n)\ra
		\,=\, 
		\la j,m|D^j(\varsigma)|j,n\ra
	\end{array}
\end{equation}
The action of the $\su(2)$ generators vanishes on those states, $\vec{J}\,|\omega_{j}\ra\,=0$.

Trivalent intertwiners correspond to $\SU(2)$-invariant states in the tensor products of three spins, $V^{j_{1}}\otimes V^{j_{2}}\otimes V^{j_{3}}$. These only exist is the three spins satisfy triangular inequalities, $|j_{2}-j_{3}|\le j_{1}\le j_{2}+j_{3}$ or equivalently any of the other two circular sets of inequalities, and are then unique once the spins are given. The coefficients of those 3-valent intertwiners are the Wigner 3j-symbols, which are defined in terms of the Clebsh-Gordan coefficients or expicitly given as a sum of factorial factors by the Racah formula. These 3-valent intertwiners can be geometrically interpreted as quantized triangles, with the spins giving the three edge lengths.

Next, 4-valent intertwiners, or equivalently $\SU(2)$-invariant states in the tensor products of four spins, $V^{j_{1}}\otimes V^{j_{2}}\otimes V^{j_{3}}\otimes V^{j_{4}}$, can  be constructed from  3-valent intertwiners. One first chooses a pairing between the spins, say $(12)-(34)$ or we could have chosen $(13)-(24)$ or $(14)-(23)$. Then one recouples the two spins $j_{1}$ and $j_{2}$ to an intermediate spin $J$, also recouples the two other spins $j_{3}$ and $j_{4}$ to the same spin $J$ and finally glues the two 3-valent intertwiners together using the $\varsigma$ map, as illustrated on fig.\ref{4valent}. 
\begin{figure}[h!]
	\begin{center}
		\begin{tikzpicture}[scale=1]
			\draw (-.7,.7)-- (0,0) node[pos=0,left]{$j_{1}$};
			\draw (-.7,-.7)-- (0,0) node[pos=0,left]{$j_{2}$};
			\draw (.7,.7)-- (0,0) node[pos=0,right]{$j_{3}$};
			\draw (.7,-.7)-- (0,0) node[pos=0,right]{$j_{4}$};
			\draw (0,0) node[scale=1]{$\bullet$};
			
			\draw[very thick,->,>=stealth] (1.5,0)--(2.5,0);
			
			\draw (-.7+4,.7)-- (0+4,0) node[pos=0,left]{$j_{1}$};
			\draw (-.7+4,-.7)-- (0+4,0) node[pos=0,left]{$j_{2}$};
			\draw[decoration={markings,mark=at position 0.6 with {\arrow[scale=1.5,>=stealth]{>}}},postaction={decorate}] (4,0) -- (5,0) node[midway,above]{$J$};
			\draw (.7+5,.7)-- (0+5,0) node[pos=0,right]{$j_{3}$};
			\draw (.7+5,-.7)-- (0+5,0) node[pos=0,right]{$j_{4}$};
			%\draw (0+3,0) node[scale=1]{$\bullet$};			
		\end{tikzpicture}
	\end{center}
	\caption{We specify an intertwiner between the four spins $j_{1},..,j_{4}$ by pairing them two by two, say $j_{1}$ with $j_{2}$ and $j_{3}$ with $j_{4}$, and recoupling each pair of spins to an intermediate spin $J$. This defines a basis of 4-valent intertwiner states.}
	\label{4valent}
\end{figure}
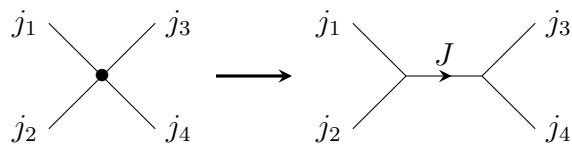
This provides us with the spin basis for 4-valent invariant states:
\begin{equation}
	|\iota^{J}_{(12)(34)}\ra
	=
	\sum_{m_{i},M} (-1)^{J+m}  |(j_{1}m_{1})(j_{2}m_{2})(j_{3}m_{3})(j_{4}m_{4})\ra 
	\la (j_{1}m_{1})(j_{2}m_{2})|J,M \ra
	\la (j_{3}m_{3})(j_{4}m_{4})|J,-M \ra
	\,,
\end{equation}
where the intermediate spin $J$ ranges from $\max(|j_{1}-j_{2}|,j_{3}-j_{4})$ to $\min(j_{1}+j_{2},j_{3}+j_{4})$, in order to satisfy the triangular inequalities.
One can extend this construction to $n$-valent intertwiners, defining a basis for $\SU(2)$-invariant states living in arbitrary tensor products of irreducible representations. Finally one can define coherent intertwiner states, with the semi-classical interpretation as quantized polygons or polyhedra \cite{Livine:2007vk,Bianchi:2010gc,Freidel:2013fia,Livine:2013tsa}.

\section{Intertwiners between spin 1/2 representations}

Let us write explicitly the basis of 4-valent intertwiners between four fundamental representations, that is between four spins $j_{1}=j_{2}=j_{3}=j_{4}=\f12$. Considering the triangular inequalities,  the Hilbert space $\cH^{(4)}_{\f12}$ of $\SU(2)$-invariant states in the tensor product $\big{(}V^{\f12}\big{)}^{\otimes 4}$ is two-dimensional.

We will use the notations of spin up and spin down for the two states in $V^{\f12}$, 
\begin{equation}
|\up\ra\,=\,\left|j=\tfrac12,m=\tfrac12\right\ra
\quad\text{and}\quad
|\down\ra\,=\,\left|j=\tfrac12,m=-\tfrac12\right\ra
\,.\nn
\end{equation}
First, two spins $\f12$ can recouple to either a spin 0 (scalar) or a spin 1 (vector). The spin 0 state is given by the $\varsigma$ map, while the spin 1 states are given by the corresponding Clebsch-Gordan coefficients:
\begin{equation}
|0\ra=\f1{\sqrt{2}}\,(|\up\down\ra-|\down\up\ra)
\,,\qquad
|1,+\ra=|\up\up\ra
\,,\quad
|1,0\ra=\f1{\sqrt{2}}\,(|\up\down\ra+|\down\up\ra)
\,,\quad
|1,-\ra=|\down\down\ra
\,.
\end{equation}
Now we choose one pairing, between $(12)-(34)$ or $(13)-(24)$ or $(14)-(23)$, which we can respectively identify as the channels $s$, $t$ or $u$. Choosing one pairing, say starting with  $(12)-(34)$, we define a basis for the 4-valent intertwiners:
\begin{equation}
|0\ra_s\equiv|0\ra_{(12)-(34)}
\equiv
|0\ra_{12}\otimes |0\ra_{34}
\equiv
\f12\Big{[}
|\up\down\up\down\ra-|\up\down\down\up\ra-|\down\up\up\down\ra+|\down\up\down\up\ra
\Big{]}\,,
\end{equation}
\begin{align}
|1\ra_{(12)-(34)}
&=&
\f1{\sqrt{3}}\Big{[}
|1,+\ra_{12}|1,-\ra_{34}
-|1,0\ra_{12}|1,0\ra_{34}
+|1,-\ra_{12}|1,+\ra_{34}
\Big{]}\\
&=&
\f1{\sqrt{3}}\Big{[}
|\up\up\down\down\ra+|\down\down\up\up\ra
-\f12\big{[}
|\up\down\up\down\ra+|\up\down\down\up\ra+|\down\up\up\down\ra+|\down\up\down\up\ra
\big{]}
\Big{]}
\nn
\end{align}
One can similarly define the intertwiner basis corresponding to the other two channels:
\begin{equation}
|0\ra_t\equiv|0\ra_{(13)-(24)}
=
\f12\Big{[}
|\up\up\down\down\ra-|\up\down\down\up\ra-|\down\up\up\down\ra+|\down\down\up\up\ra
\Big{]}\,,
\end{equation}
\begin{align}
|1\ra_{(13)-(24)}
&=&
\f1{\sqrt{3}}\Big{[}
|1,+\ra_{13}|1,-\ra_{24}
-|1,0\ra_{13}|1,0\ra_{24}
+|1,-\ra_{13}|1,+\ra_{24}
\Big{]}\\
&=&
\f1{\sqrt{3}}\Big{[}
|\up\down\up\down\ra+|\down\up\down\up\ra
-\f12\big{[}
|\up\up\down\down\ra+|\up\down\down\up\ra+|\down\up\up\down\ra+|\down\down\up\up\ra
\big{]}
\Big{]}
\nn
\end{align}
\begin{equation}
|0\ra_u\equiv|0\ra_{(14)-(23)}
=
\f12\Big{[}
|\up\up\down\down\ra-|\up\down\up\down\ra-|\down\up\down\up\ra+|\down\down\up\up\ra
\Big{]}\,,
\end{equation}
\begin{align}
|1\ra_{(14)-(23)}
&=&
\f1{\sqrt{3}}\Big{[}
|1,+\ra_{14}|1,-\ra_{23}
-|1,0\ra_{14}|1,0\ra_{23}
+|1,-\ra_{14}|1,+\ra_{23}
\Big{]}\\
&=&
\f1{\sqrt{3}}\Big{[}
|\up\down\down\up\ra+|\down\up\up\down\ra
-\f12\big{[}
|\up\up\down\down\ra+|\up\down\up\down\ra+|\down\up\down\up\ra+|\down\down\up\up\ra
\big{]}
\Big{]}
\nn
\end{align}
From these explicit formulae, one can easily  compute the unitary matrices mapping one channel onto another. This is especially useful when looking at the volume operator \cite{Feller:2015yta}.
In the present work, in order to analyze the twisted torus partition function for spins $\f12$ on the boundary, we are more interested in the decomposition of the identity on the intertwiner space $\cH^{(4)}_{\f12}$ by the $0$-states in the three channels:
\begin{equation}
\id=\f23
\big{[}
|0\ra_{s}{}_{s}\la 0|
+|0\ra_{t}{}_{t}\la 0|
+|0\ra_{u}{}_{u}\la 0|
\big{]}
\end{equation}
To prove this decomposition of the identity, it is enough to apply it to the orthonormal basis $|0\ra_{s},|1\ra_{s}$. Let us nevertheless not forget that the three states $|0\ra_{s},|0\ra_{t},|0\ra_{u}$ are not independent and form an over-complete basis:
\begin{equation}
|0\ra_{u}
=
|0\ra_{t}-|0\ra_{s}
\,.
\end{equation}
\chapter{Exact computation of the partition function for the 0-spin intertwiner in the s-channel}
\label{App:exact_computation_recouplingJ0}
\numberwithin{equation}{chapter}

In this appendix, we focus on the exact computation of the Ponzano-Regge partition function given in \eqref{chap4:eq:recouplingJ0_channel_general_amp} for a boundary spin network state on the square lattice, with a homogeneous spin $j$ and the 0-spin intertwiner in the s-channel at every vertex:
\begin{equation}
	\la \text{PR} | \Psi_{j,\iota^{s|0}}  \ra
	=
	\frac{d_{j}^{N_{t}}}{d_j^{N_t N_x}}
	\int_{\SU(2)} \dd g \; \chi_{j}(g^{p})^{K}
	= 
	\frac{1}{d_j^{N_t (N_x-1)}}
	\frac{2}{\pi} \int_{0}^{\pi} \dd\theta\; \sin^2(\theta) \; \chi_{j}(p \theta)^K
	\,,
	\nn
\end{equation}
Expanding the character in the $m$-basis, we express this integral over a random walk counting:
\begin{equation}
	\la \text{PR} | \Psi_{j,\iota^{s|0}}  \ra
	=
	\frac{1}{4\pi d_j^{N_t (N_x-1)}}
	\int_{0}^{2\pi} {d}\theta\,
	(2-\E^{2\I\theta}-\E^{-2\I\theta})
	\sum_{m_{1},..,m_{K}} \E^{2\I\sum_{k=1}^{K}m_{k}p\theta}
	\,,
\end{equation}
with the $m$'s running in integer steps from $-j$ to $+j$. We focus on the generic case for $p\ge 3$.  In that case, among the three terms coming from the measure factor $\sin^2\theta$, only the constant term contributes to a non-trivial random walk counting while the other two terms, in $\E^{2\I\theta}$ and $\E^{-2\I\theta}$, give vanishing integrals.
Thus, assuming $p\ge 3$, we see that the dependance on $p$ actually drops out and the renormalized Ponzano-Regge amplitude $d_j^{N_t (N_x-1)}\,\la \text{PR} | \Psi_{j,\iota^{s|0}}  \ra$  is always equal to the following integral:
\begin{equation}
\cI_{K} = \frac{1}{\pi} \int_{0}^{\pi} \text{d}\theta \chi^{j}\left(\theta\right)^{K} = \frac{1}{\pi} \int_{0}^{\pi} \text{d}\theta \left(\frac{\sin(d_j \theta)}{\sin\theta}\right)^{K}
\,.
\end{equation}

Using the binomial formula and the following series expansion
\begin{equation}
\frac{1}{(1-x)^{K}} = \sum_{n=0}^{\infty} \binom{n+K-1}{K-1}x^{n},
\end{equation}
we expand the both sine in the numerator and  in the denominator:
\begin{equation}
\cI_{K} = \frac{1}{\pi} \sum_{k=0}^{K} \sum_{n=0}^{\infty}
(-1)^{k}\binom{K}{k} \binom{n+K-1}{K-1}  \int_{0}^{\pi} \dd\theta\; \E^{\I\Big[(K-2k)d_j -K + 2n\Big]\theta}.
\label{eq:I_expression_1}
\end{equation}
where the two binomial coefficients are given by
\begin{equation}
	\begin{split}
	\binom{K}{k} \binom{n+K-1}{K-1} &= \frac{K}{k!(K-k)!}(n+1)(n+2)....(n+K-1) \\
	&= \frac{K}{2^{K-1}k!(K-k)!}(2n+2)(2n+4)....(2n+2K-2).
	\end{split}
	\label{eq:I_binomial_expression}
\end{equation}
The integration over the exponential is straightforward and gives a Kronecker delta:
\begin{equation}
	\int_{0}^{\pi} \dd\theta\; \E^{\I \Big[(K-2k)d_j -K + 2n\Big]\theta} = \pi \delta_{0,(K-2k)d_j -K -2n},
\end{equation}
which translates into the constraint for $n$:
\begin{equation}
	2n = (K-2k)d_j -K.
	\label{eq:I_n_constraint}
\end{equation}
Plugging this into  equation \eqref{eq:I_expression_1} leads to an expression of the integral $\cI_{K}$ as a sum: 
\begin{equation}
	\cI_{K}=\frac{K}{2^{K-1}} \sum_{k=0}^{[K/2]} \frac{(-1)^{k}}{k!(K-k)!} \prod_{m=1}^{K-1} \big[(K-2k)d_j + (2m-K)\big].
\end{equation}
It is natural to write this formula as a polynomial in $d_j$ of the form
\begin{equation}
	\cI_{K} = \frac{K}{2^{K-1}} \sum_{m=0}^{K-1} A_{m} d_j^{m},
\end{equation}
where  $A_{m}$ are polynomials in $K$. This coefficients can be computed for $m<K-1$ as:
\begin{equation}
	A_{m}
	= \sum_{k=0}^{K-1} \frac{(-1)^{k}}{k!(K-k)!} (K-2k)^{m} \sum_{n_1<n_2<...<n_{K-1-m}}\prod_{i=1}^{K-1-m}(2n_i-K)
	\,.
\end{equation}
The highest order coefficient $A_{K-1}$ is
\begin{equation}
	\frac{1}{2^{K-1}} A_{K-1} = \frac{1}{2^{K-1}} \sum_{k=0}^{K} (-1)^{k} \frac{(K-2k)^{K-1}}{k! (K-k)!}.
\end{equation}

These are in fact  known numbers, appearing in the Fourier series of powers of the cardinal sine $\frac{1}{n!}\Big(\frac{\sin(\theta)}{\theta}\Big)^{n}$. They can be identified as the value of the integral  $\frac{1}{\pi}\int_{0}^{\infty}\d x\Big(\frac{\sin(x)}{x}\Big)^k$ in \cite{SWANEPOEL2015} or as the coefficients of the Duflo map for $\SU(2)$ in \cite{Freidel:2005ec}. We will refer to them as the Freidel-Majid numbers $C_{n,s}$ after this latter work, defined by:
\begin{equation}
	C_{n,s} = \frac{1}{2^n} \sum_{k=0}^{n} (-1)^{k} \frac{(n-2k)^{s}}{k!(n-k)!}.
\end{equation}

All the coefficients $A_{K-2n}$, for $n\in\N$ vanish and we can focus on the terms  $A_{K-(2n+1)}$. They are expressed in terms of the Freidel-Majid numbers:
\begin{equation}
	\frac{1}{2^{K-1}} A_{K-(2n+1)} = 2C_{K,K-(2n+1)} \times \sum_{n_1<n_2<...<n_{K-1-m}}\prod_{i=1}^{K-1-m}(2n_i-K)
\end{equation}
We do not have an explicit closed formula for the remaining sums. An order by order investigation shows that they are proportional to a Pochhammer coefficient (defined as $(N)_{(n)} = N(N-1)...(N-n+1)$) times a polynomial of degree $(n-1)$, for instance:
\begin{equation}
	\begin{split}
		\frac{1}{2^{K-1}} A_{K-1} &= 2 C_{K,K-1} \,,\\
		\frac{1}{2^{K-1}} A_{K-3} &= -\frac{1}{6} (K)_{(3)} 2 C_{K,K-3} \,,\\
		\frac{1}{2^{K-1}} A_{K-5} &= \frac{1}{360} (K)_{(5)} (5K+2) 2 C_{K,K-5}\,.
	\end{split}
\end{equation}

\newpage
~
\thispagestyle{empty}

\chapter{Explicit computation of the Hessian}
\label{app:hessian_computation}
\numberwithin{equation}{chapter}

%------------
\paragraph*{\bf $\phi\phi$-term}
This is the simplest term. 
\begin{align}
H_{\phi;\phi} & = 
	\left.%
	-\frac{2L}{N_t} \sum_{t,x}  \left( \frac{\la +| G_{t+1,x}^{-1}  \E^{\I \frac{\varphi}{N_t}\sigma_3} \sigma_3^2 G_{t,x} | +\ra}{ \la +| G_{t+1,x}^{-1} \E^{\frac{\varphi}{N_t}\sigma_3} G_{t,x} | +\ra} 
	-%
	\frac{\la +| G_{t+1,x}^{-1}  \E^{\frac{\varphi}{N_t}\tau_z} \sigma_3 G_{t,x} | +\ra^2}{ \la +| G_{t+1,x}^{-1} \E^{\frac{\varphi}{N_t}\sigma_3} G_{t,x} | +\ra^2}\right)\right|_o \notag\\
	& =  \frac{2 L}{N_t} \sum_{t,x} \left( 1  - (\hat z. G^o_{t,x}\triangleright \hat x)^2 \right) \notag\\
	& = 2 L N_x\equiv F.
\end{align}
The topological $\varphi$-field has no kinetic term and a (positive) mass equal to one-half of the cylinder (spatial) circumference.\\

%------------
\paragraph*{\bf $\phi a$-term}
The mixed term is
\begin{align}
&  H^{k}_{t,x; \phi} = \notag\\
& = 
\left.%
-\frac{2L}{N_t}\left(%
\frac{\la +| G_{t+1,x}^{-1}  \E^{\frac{\varphi}{N_t}\sigma_3} \sigma_3 G_{t,x} \sigma^k | +\ra}{ \la +| G_{t+1,x}^{-1} \E^{\frac{\varphi}{N_t}\sigma_3} G_{t,x} | +\ra} 
-%
\frac{\la +| G_{t+1,x}^{-1}  \E^{\frac{\varphi}{N_t}\sigma_3} \sigma_3 G_{t,x} | +\ra\la +| G_{t+1,x}^{-1} \E^{\frac{\varphi}{N_t}\sigma_3} G_{t,x} \sigma^k| +\ra}{ \la +| G_{t+1,x}^{-1} \E^{\frac{\varphi}{N_t}\sigma_3} G_{t,x} | +\ra^2} \right)\right|_o + \notag\\
& \phantom{=} + \text{1 term from other link}\notag\\
& =  - \frac{2 L}{N_t}\Big( \la + | (G_{t,x}^o)^{-1} \sigma_z G^o_{t,x}\sigma^k | + \ra - %
(\hat z. G_{t,x}^o\triangleright \hat x) \delta^k_1  - \la + | \sigma^k (G_{t,x}^o)^{-1} \sigma_z G^o_{t,x} | + \ra +  (\hat z. G_{t,x}^o\triangleright \hat x) \delta^k_1  \Big) \notag\\
& = - \frac{2 L}{N_t}\Big( \la + | \sigma_z \sigma^k | + \ra - %
\la + | \sigma^k \sigma_z | + \ra \Big)  =  - \frac{4 \I L}{N_t} \epsilon_{3kj} \la + |  \sigma^j | + \ra \notag\\
& = \frac{4 iL}{N_t} \delta^k_2 \equiv D_k
\end{align}
Notice the independence from the position $(t,x)$.\\

%------------
\paragraph*{\bf $a a$-term}
Because of the near-neighbour interactions, this is the most complicated term.
We start from the ultralocal term (notice the symmetrization of the indices)
\begin{align}
	H^{j;k}_{t,x;t,x} %%
	& = 2T \left.\left(- \frac{\la \up |\sigma^{(j} \sigma^{k)} G_{t,x}^{-1}G_{t,x-1} |\up\ra }{\la \up | G_{t,x}^{-1}G_{t,x-1} |\up\ra} + \frac{\la \up | \sigma^k G_{t,x}^{-1}G_{t,x-1} |\up\ra \la \up | \sigma^j G_{t,x}^{-1}G_{t,x-1} |\up\ra }{\la \up | G_{t,x}^{-1}G_{t,x-1} |\up\ra^2} \right)\right|_o+\notag\\ 
	& \phantom{=} +  \text{3 other terms from other links}\notag\\
	& = 2T \left( {\la \up |\sigma^{(j} \sigma^{k)} |\up\ra  -\la \up | \sigma^k  |\up\ra \la \up | \sigma^j |\up\ra } \right)+  \text{3 other terms from other links}\notag\\
	& = 2T\left( \delta^{jk} - \delta^k_3\delta^j_3 \right) +  \text{3 other terms from other links}\notag\\
	& =  4T \left( \delta^{jk} - \delta^k_3\delta^j_3 \right) + L \left( \delta^{jk} - \delta^k_1\delta^j_1 \right) \equiv A^{jk}.
\end{align}
This term does not depend on the position $(t,x)$, either.

Then, we move on to the term which couples spatially separated cells.
In this case the two derivative commutes and one needs only to compute one of the two terms (we compute the second one). The result is 
\begin{align}
	H^{j;k}_{t,x-1;t,x} %%
	& = 2T \left.\left( \frac{\la \up |\sigma^{k} G_{t,x}^{-1}G_{t,x-1} \sigma^j|\up\ra }{\la \up | G_{t,x}^{-1}G_{t,x-1} |\up\ra} - \frac{ \la \up | \sigma^k G_{t,x}^{-1}G_{t,x-1} |\up\ra\la \up |  G_{t,x}^{-1}G_{t,x-1} \sigma^j |\up\ra  }{\la \up | G_{t,x}^{-1}G_{t,x-1} |\up\ra^2} \right)\right|_o \notag\\ 
	& = -2T\left(  \la \up |\E^{ \psi^T_o \tau_z}\sigma^{k} \E^{- \psi^T_o \tau_z} \sigma^j |\up\ra -   \la \up |\sigma^{k} |\up\ra \la \up |\sigma^{j} |\up\ra    \right)\notag\\
	& =  -2T\left(  R_z(-\psi^T_o)^k{}_i \la \up |\sigma^{i}\sigma^j |\up\ra -   \la \up |\sigma^{k} |\up\ra \la \up |\sigma^{j} |\up\ra    \right)\notag\\
	& =  -2T\left(  R_z(-\psi^T_o)^k{}_j  + i R_z(-\psi^T_o)^k{}_i \epsilon^{ij3} -   \delta^k_3\delta^j_3    \right) \equiv B^{jk},
\end{align}
where a rapid computation shows
\begin{equation}
	B = - 2T \E^{-i\psi^T_o} \mat{ccc}{ 1 & -i & 0 \\ i & 1 & 0 \\ 0 & 0 & 0  }.
\end{equation}
Similarly, one finds that
\begin{equation}
	H^{j;k}_{t,x+1;t,x} = B^{kj},
\end{equation}
which is just the transpose of the above matrix. Notice that this could have been deduced from $H^{j;k}_{t,x-1;t,x} = H^{k;j}_{t,x; t,x-1}$ and the independence of these matrices from the position $(t,x)$.

Finally, for timely separated neighbouring cells, one obtains
\begin{align}
	H^{j;k}_{t-1,x;t,x} %%
	& = 2L \left.\left(  \frac{\la + | \tau^k G_{t,x}^{-1}\E^{\frac{\varphi}{N_t} \sigma_3} G_{t-1,x} \sigma^j |+\ra }{\la + | G_{t,x}^{-1}\E^{\frac{\varphi}{N_t} \sigma_3} G_{t-1,x}  |+ \ra} - \frac{\la + | \tau^k G_{t,x}^{-1}\E^{\frac{\varphi}{N_t} \sigma_3} G_{t-1,x}  |+\ra \la + |G_{t,x}^{-1}\E^{\frac{\varphi}{N_t} \sigma_3} G_{t-1,x} \sigma^j  |+\ra }{\la + | G_{t,x}^{-1}\E^{\frac{\varphi}{N_t} \sigma_3} G_{t-1,x}  |+ \ra^2}   \right)\right|_o \notag\\ 
	& = -2L\Big(   \la + | \sigma^k \sigma^j |+\ra  - \la + | \sigma^k   |+\ra \la + |  \sigma^j  |+\ra   \Big) \notag\\ 
	& = -2L\Big(  \delta^{kj} + i\epsilon^{kj1} - \delta^k_1\delta^j_1 \Big) \equiv C^{jk},
\end{align}
where 
\begin{equation}
	C = - 2L \mat{ccc}{ 0 & 0 & 0 \\ 0  & 1 & -i \\ 0 & i & 1  }.
\end{equation}
And similarly
\begin{equation}
	H^{j;k}_{t+1,x;t,x} = C^{kj} .
\end{equation}

\newpage
~
\thispagestyle{empty}

\chapter{Proof of the lemma \ref{lemma:cosine_product}}
\label{app:proof_formula_product_cos}

We consider two integers $N$ and $M$ and two complexes $a$ and $x$. We define by $K$ the greatest common divisor of $N$ and $M$. From $K$ we define two integers $n$ and $m$ such that
\begin{align*}
N &= K n \; , \\
M &= K m \; .
\end{align*}

The following equality holds for all $N,M$ and all $a,x$
\begin{equation}
\prod_{k=0}^{N-1} \left(2a + 2\cos\left( \f{2 \pi M}{N} k + x \right) \right) =\left( 2 \left( T_{n}(a) - (-1)^{n} \cos(n x) \right) \right)^{K} \; .
\end{equation}
In the case where the $K \neq 1$, we see that only $n$ terms are different in the product, and so the product can be written as
\begin{equation*}
\left[\prod_{k=0}^{n-1} \left(2a + 2\cos\left( \f{2 \pi m}{n} k + x \right) \right)\right]^{K} \; ,
\end{equation*}
and the value of the product does not depend on $m$ anymore. Indeed, for $m$ and $n$ coprime, we can use the periodicity of the cosine to absorb the factor $m$.

Therefore, showing the equality
\begin{equation}
\prod_{k=0}^{n-1} \left(2a + 2\cos\left( \f{2 \pi}{n} k + x \right) \right) =2 \left( T_{n}(a) - (-1)^{n} \cos(n x) \right) \; 
\label{eq:proof_eq_ini}
\end{equation}
is enough to show the formula in the general case. The proof is done using the usual method of showing that the left hand side and the right hand side of \ref{eq:proof_eq_ini} are polynomials in $a$ with the same roots. This will show that the two sides are proportional.
\medskip

We introduce two functions of a complex variable $a$, with $x$ as a parameter
\begin{equation*}
f(a) = \prod_{k=0}^{n-1} \left(2a + 2\cos\left( \f{2 \pi}{n} k + x \right) \right)  \quad \text{and} \quad g(a) = 2 \left( T_{n}(a) - \cos(n x + n\pi) \right) \; .
\end{equation*}

Finding the roots of $f$ is straightforward. There are $N$ roots, parametrized by an integer $k \in [0,N-1]$
\begin{equation*}
a^{f}_{k} = - \cos\left( \f{2 \pi}{n} k + x \right)
\end{equation*}

For the roots of $g$, we use the expression of the Chebyshev polynomial in terms of the analytic continuation of cosine
\begin{equation*}
T_{n}(a) = \cos(n \arccos(a)) \;.
\end{equation*}

The equation $g(a) = 0$ is then equivalent to an equality between cosine, with solutions
\begin{equation*}
n \arccos(a_k^{g}) = \pm (n x + n\pi + 2\pi k) \quad \text{for} \; k \, \in \N \;,
\end{equation*}
that is
\begin{equation}
a_k^{g} = \cos\left(\f{2 \pi k}{n} + x + \pi\right) = - \cos\left(\f{2 \pi k}{n} + x \right) \; ,
\end{equation}
where we see that $k$ can be taken between $0$ and $n-1$.
\bigskip

This is enough to show that $f$ and $g$ are proportionals. Evaluating $f$ and $g$ at $a=1$ allows to show that the proportionality coefficient is 1.

%%%%%%%%%%%%%%%%%%%%%%%%%%%%%%%%%%%%%
%%%%%%%%% BIBLIOGRAPHY %%%%%%%%%%%%%%
%%%%%%%%%%%%%%%%%%%%%%%%%%%%%%%%%%%%%
\pagestyle{Bibliography}
\bibliography{./biblio/mypaper,./biblio/other_field,./biblio/gravity_holography_general,./biblio/PR_model}
\bibliographystyle{bib-style}

\end{document}